\documentclass[aps,pra,
onecolumn,tightenlines,
superscriptaddress,nofootinbib,longbibliography,11pt]{revtex4-2}

\usepackage{graphicx}
\usepackage{amsmath,amsfonts,amsbsy,amssymb,amsthm,mathtools}
\usepackage{mathrsfs}
\usepackage{tabularx}
\usepackage{booktabs}
\usepackage{color}
\usepackage{enumitem}
\usepackage{siunitx}
\usepackage{multirow}
\usepackage[theorems]{tcolorbox}
\usepackage{color}
\usepackage{tcolorbox}
\usepackage{booktabs}
\usepackage{graphicx}   %
\usepackage{environ}
\usepackage{comment}
\usepackage{xcolor}
\usepackage{stmaryrd}
\usepackage{booktabs}
\usepackage{tikz}
\usetikzlibrary{automata,arrows,positioning,calc}
\usepackage{tikz-cd}
\usepackage[normalem]{ulem}
\usepackage{subcaption}
\usepackage{ragged2e}
\usepackage{subfiles}

\usepackage{float}
\usepackage{placeins} %
\usepackage{tikzsymbols}
\usetikzlibrary{plotmarks,spy}

\newtheorem{theorem}{Theorem}
\newtheorem{proposition}[theorem]{Proposition}
\newtheorem{lemma}[theorem]{Lemma}

\theoremstyle{definition}
\newtheorem{definition}[theorem]{Definition}

\newtheorem{example}[theorem]{Example}
\newtheorem{remark}[theorem]{Remark}

\theoremstyle{remark}

\usepackage[colorlinks = true, citecolor=blue, linkcolor=blue, urlcolor=blue]{hyperref}
\usepackage[T1]{fontenc}
\usepackage{mwe}
\providecommand{\bra}{}
\providecommand{\ket}{}
\providecommand{\braket}{}
\renewcommand{\bra}[1]{\left\langle #1\right|}
\renewcommand{\ket}[1]{\left| #1 \right\rangle}

\renewcommand{\braket}[2]{\left\langle {#1}\middle|{#2}\right\rangle}
\newcommand{\brakets}[3]{\left\langle {#1}\middle| {#2}\middle| {#3}\right\rangle}
\newcommand{\abs}[1]{\left|{#1}\right|}
\newcommand{\nrm}[1]{\left\|{#1}\right\|}

\newcommand{\C}{{\mathbb{C}}} %
\newcommand{\R}{{\mathbb{R}}} %

\newcommand{\F}{{\mathbb{F}}} %
\newcommand{\Z}{{\mathbb{Z}}}
\renewcommand{\S}{{\mathbb{S}}}
\newcommand{\GL}{{\mathsf{GL}}}
\renewcommand{\O}{{\mathsf{O}}}

\newcommand{\spn}{{\textrm{span}\,}}

\newcommand{\diag}{{\textrm{diag}}}
\newcommand{\stab}{{\textbf{stab}}}

\newcommand{\lt}{{\mathtt{LT}\,}}
\newcommand{\lc}{{\mathtt{LC}\,}}
\newcommand{\lm}{{\mathtt{LM}\,}}
\newcommand{\lex}{{\mathtt{lex}\,}}
\newcommand{\grlex}{{\mathtt{grlex}\,}}
\newcommand{\grevlex}{{\mathtt{grevlex}\,}}
\newcommand{\muldeg}{{\mathtt{muldeg}\,}}

\newcommand{\at}[1]{{\textcolor{red}{Andrew: [#1]}}}
\definecolor{darkred}{RGB}{120,0,0}
\definecolor{darkpurple}{RGB}{120,0,120}

\begin{document}

\normalem
\newlength\figHeight 
\newlength\figWidth 

\title{BARC codes: general polynomial framework for coherent-state superposition codes}

\author{Praneel Gore}
\affiliation{Institute of Advanced Intelligence and Computing (IAIC),
Agency for Science, Technology and Research (A*STAR),
1 Fusionopolis Way, \#16-16 Connexis, Republic of Singapore 138632}

\author{Davit Aghamalyan}
\affiliation{SMT Pillar, Singapore University of Technology and Design,
Singapore 487372}

\author{Varun Narasimhachar}
\affiliation{Institute of Advanced Intelligence and Computing (IAIC),
Agency for Science, Technology and Research (A*STAR),
1 Fusionopolis Way, \#16-16 Connexis, Republic of Singapore 138632}

\author{Andrew Tanggara}
\affiliation{Centre for Quantum Technologies, National University of Singapore,
3 Science Drive 2, Singapore 117543}

\affiliation{Nanyang Quantum Hub, School of Physical and Mathematical Sciences,
Nanyang Technological University, Singapore 639673}

\date{\today}

\begin{abstract}
We introduce an algebraic framework for constructing a new class of quantum error-correcting codes, bosonic algebraically-restricted constellation (BARC) codes. The code states are finite superpositions of coherent-state constellations 
constrained by
symmetries of solution sets to multivariate complex polynomial systems. The constraints are associated with photon-gain and photon-loss errors, thereby imposing approximate Knill--Laflamme conditions for such errors.
For equally-weighted constellation points, we derive upper bounds on the minimum geometric separation between points, indicating how noise-resilient the code is.
In the single-mode case, we characterize degree-two polynomial solution sets using an orthogonal group symmetry to obtain explicit ellipsoidal and hexagonal {BARC codes}.
We benchmark representative instances of these families against spherical and cubature codes using the entanglement fidelity with optimal recovery under pure-loss noise.
\end{abstract}

\maketitle

\tableofcontents

\section{Introduction}

Quantum error correction (QEC) is essential for realizing reliable, large-scale quantum computers capable of addressing practically relevant problems \cite{NielsenChuang,Spencer2026,ErrorCorrectionZoo}. 
By redundantly encoding logical information, QEC enables physical errors to be detected and corrected without destroying the encoded quantum state, making the development of efficient error-correcting codes a central challenge on the path toward fault-tolerant quantum computation.

Bosonic QEC provides an attractive alternative to encodings built solely from collections of two-level systems. 
A single harmonic oscillator possesses an infinite-dimensional Hilbert space and can therefore accommodate highly redundant logical encodings within a single physical mode, allowing the structure of the encoding to be tailored to dominant physical noise processes and potentially reducing hardware overhead \cite{GKP2001,Albert2018,Grimsmo2020}. 
This observation has motivated several major families of single-mode bosonic codes. The Gottesman--Kitaev--Preskill (GKP) code encodes quantum information into grid states possessing discrete translation symmetry in oscillator phase space \cite{GKP2001}; its experimental relevance has since been underscored by real-time QEC beyond break-even in a superconducting-cavity implementation \cite{Sivak2023}. Cat codes, based on superpositions of coherent states, constitute another major family of single-mode bosonic encodings and have developed from early amplitude-damping constructions through dissipatively protected and optimized coherent-state codes, broader structural analyses, and experimentally accessible state-generation proposals \cite{Cochrane1999,Mirrahimi2014,Li2017,Albert2018,Erneman2026}. Binomial codes use finite superpositions of Fock states with binomial weights and can be constructed to correct photon loss and gain as well as dephasing  errors to prescribed orders \cite{Michael2016,Albert2018}. A broader unifying perspective is provided by rotation-symmetric bosonic codes, which include important cat- and binomial-code families and organize bosonic encodings according to discrete phase-space rotational symmetry \cite{Grimsmo2020,Totey2023}. Other single-mode approaches exploit squeezing as an additional resource, including squeezed-cat and squeezed-Fock-state encodings and related state-engineering proposals aimed at improving protection against relevant bosonic noise channels \cite{Schlegel2022,Hillmann2023,Zeng2026,Nauth2026}. Analogous and intrinsically multimode bosonic constructions have also been developed across several code families \cite{Royer2022,XuRotor2024,Lemonde2024,Aydin2026,XuTiger2025,Ahmed2026}.

Among bosonic encodings, coherent-state constructions are especially appealing because of their direct phase-space interpretation and their compatibility with experimentally mature cavity, circuit-QED, and photonic control techniques~\cite{Brune1992,  Brune1996,Deleglise2008,Auffeves2003,Ofek2016,Lemonde2024,LachanceQuirion2024,Putterman2025, Berdou2023,Reglade2024,Marquet2024,Guillaud2019}.
Experimental progress has established coherent-state bosonic QEC as more than a theoretical alternative. In a landmark experiment, Ofek \emph{et al.} demonstrated a cat-code logical qubit whose error-corrected lifetime surpassed that of the best constituent physical qubit, reaching the QEC break-even regime \cite{Ofek2016}. 
More recently, concatenated architectures have combined individually protected cat qubits with an outer repetition code, demonstrating a hardware-efficient route toward scaling bosonic QEC \cite{Putterman2025}. 
These developments motivate the search for coherent-state codes that combine strong error protection with structures amenable to preparation, control, and decoding.

Coherent-state bosonic codes also admit a natural geometric interpretation in phase space, where logical states are supported on discrete constellations of coherent amplitudes and their error-correction properties can be related to the geometry and moments of these constellations. 
A particularly natural geometric realization of this idea is provided by quantum spherical codes (QSCs) \cite{Jain2023}. 
In this framework, logical code states are superpositions of coherent states whose amplitudes form finite constellations on a fixed-energy sphere in phase space. 
The geometry of a constellation provides a natural notion of resolution: the minimum squared Euclidean separation \(d_E\) between distinct constellation points controls their asymptotic distinguishability and their robustness against sufficiently small phase-space perturbations. 
At the same time, photon-number-changing errors are represented by monomials in creation and annihilation operators. 
Because coherent states are eigenstates of annihilation operators, the corresponding Knill--Laflamme (KL) constraints acquire a direct algebraic form in the coherent-state amplitudes. Appropriate constellation moment conditions can therefore protect against photon-loss/gain and dephasing-type errors, while geometric separation provides complementary protection against small displacement-like perturbations in the large-energy regime \cite{Jain2023}. Quantum cubature codes (QCCs) subsequently generalize this geometric approach by constructing coherent-state QEC codes based on cubature formulas, which allow nonuniform weights and constellations distributed across multiple energy shells \cite{Yang2025}. 
This enlarged design space yields new single- and multimode constructions, demonstrating that relaxing the geometric restrictions on the constellation can reveal further bosonic encodings. 
We use representative single-mode QSCs and QCCs as benchmarks for our constructions under pure-loss noise in Section~\ref{sec:noise-benchmarking}.

{Despite this progress, systematically discovering weighted coherent-state codes in one or several modes remains a nontrivial problem; our explicit constructions below focus on the single-mode, uniform-weight case.}
Geometrically motivated frameworks such as QSCs and intrinsically multimode constructions such as tiger codes have uncovered rich families of bosonic codes \cite{Jain2023,XuTiger2025}, but they approach code discovery through specific geometric, symmetry, or algebraic structures. 
This motivates a complementary question: can the Knill--Laflamme conditions themselves be treated as the defining equations of the code space and solved systematically to uncover previously unknown coherent-state constellations? 
Such a formulation would provide a common more general framework for single- and multimode code discovery without requiring a particular constellation geometry to be imposed in advance.

Algebraic geometry and commutative algebra provide a natural language for such a program. Polynomial systems define algebraic varieties, while ideals, elimination theory, and Gr\"obner bases provide systematic tools for characterizing and solving coupled polynomial constraints \cite{Cox2015}. 

{In this work, we develop an algebraic framework for constructing bosonic algebraically-restricted constellation (BARC) codes by solving systems of polynomial moment equations for weighted coherent-state constellations.}
Shared weighted moments, multiple modes, and multiple energy shells are already part of the QCC framework~\cite{Yang2025}; our contribution is the formulation by equivalent real polynomial ideals, the degree-two solution-space parameterization, and the explicit symmetry-based construction procedures.
We simplify the polynomial generators and exploit symmetries of their real solution sets; expressing shared moment constraints in this way does not by itself define a code class disjoint from QCCs.

The framework applies generally to weighted coherent-state constellations in an arbitrary number of bosonic modes. 
However, in the present work, explicit constructions, classifications, and numerical examples focus on the uniformly-weighted single-mode setting. 
{We exploit symmetries of the polynomial solution space to construct code families, derive bounds on their constellation resolution, and characterize the real solution set of the single-mode, uniform-weight moment constraints through degree two.}

Concrete multimode code families obtained from the same polynomial framework will be investigated in a follow-up work.
The resulting real polynomial systems could also be
solved numerically to search for further constellations.
{We have not undertaken a systematic numerical root search or Gr\"obner-basis elimination to generate the reported examples; the numerical results below benchmark the explicit constellations obtained from the analytic constructions.}
To assess these constructions, we benchmark representative single-mode instances against selected QSC and QCC references using the optimized entanglement fidelity defined in Section~\ref{sec:noise-benchmarking} under bosonic pure-loss noise.
We compare equal constellation sizes using the best sampled amplitudes on a finite grid at one reference loss probability and the stated photon-budget reference-selection rule; amplitudes are then held fixed while recovery is optimized along each loss parameter.
Under this benchmarking, several of our constructed code families outperform the selected QSC and QCC references over all sampled loss probabilities, while other families outperform their selected references only in particular loss regimes.
%
%
%
%
%

The remainder of this manuscript is organized as follows.
Section~\ref{sec:coherent_state_codes} introduces coherent-state superpositions and their exact KL conditions.
Section~\ref{sec:constellations} formulates the single-mode, uniform-weight moment constraints as real polynomial equations.
Section~\ref{sec:moment_tensors} gives their moment-tensor representation and, in Section~\ref{sec:moment_symmetry}, their moment-map symmetries.
Section~\ref{sec:resolution_bounds} derives the general multimode resolution bound and its single-mode specialization.
Section~\ref{sec:degree_two} analyzes the constraints through degree two, and Sections~\ref{sec:ellipsoid_constellations} and~\ref{sec:nonellipsoid_constellations} present the ellipsoid and hexagonal constructions.
Section~\ref{sec:noise-benchmarking} compares the constructed families with the selected QSC and QCC references under pure-loss noise, and Section~\ref{sec:conclusions} gives conclusions and future directions.
Appendix~\ref{sec:weighted_constellation_polynomials} extends the polynomial formulation and moment-coordinate transformation to weighted multimode constellations.

\section{Codes from superpositions of coherent states}\label{sec:coherent_state_codes}\label{sec:weighted_KL}

We consider bosonic codes whose logical states are superpositions of finitely many coherent states. 
Cat codes and QSCs provide uniform-weight examples, while QCCs allow nonuniform weights and multiple energy shells~\cite{Jain2023,Yang2025}. 
Rotation symmetry also organizes cat and binomial codes, although binomial codes are defined using finite Fock-state superpositions~\cite{Michael2016,Grimsmo2020}.
Here we recall how coherent-state overlaps and constellation moments enter the KL conditions as in~\cite{Jain2023,Yang2025}, before treating the moment constraints as polynomial equations in Section~\ref{sec:constellations}.
{Appendix~\ref{sec:qsc_qcc_review} reviews the spherical-design and positive-cubature constructions of QSCs and QCCs, respectively.}

Let $P$ be the orthogonal projector onto a $K$-dimensional codespace. 
{An error set $\left\{E_\mu\right\}_\mu$ is exactly correctable if and only if there are constants $\lambda_{\mu,\nu}$ satisfying the KL conditions~\cite{Knill1997}}
\begin{equation}\label{eqn:KL_general}
    PE_\mu^\dag E_\nu P=\lambda_{\mu,\nu}P
    \quad\forall\mu,\nu .
\end{equation}
In an orthonormal logical basis this is the usual KL condition with a Kronecker delta. 
For linearly independent logical states $\left\{\ket{\mathcal{C}_i}\right\}_i$ spanning the codespace that may not be orthogonal, the equivalent form is
\begin{equation}\label{eqn:KL_Gram}
    \brakets{\mathcal C_i}{E_\mu^\dag E_\nu}{\mathcal C_j}
    =\lambda_{\mu,\nu}\braket{\mathcal C_i}{\mathcal C_j} .
\end{equation}
This form retains the generally nonzero overlaps between logical constellation states. 
Detectability of an operator $L$ requires $PLP=\lambda_LP$; correction requires this condition for every product $E_\mu^\dag E_\nu$.
{The constants $\lambda_{\mu,\nu}$ are independent of the logical states, and taking the adjoint of eqn.\ \eqref{eqn:KL_general} gives $\lambda_{\nu,\mu}=\lambda_{\mu,\nu}^*$, so the matrix $\left[\lambda_{\mu,\nu}\right]_{\mu,\nu}$ is Hermitian.}

Consider a set of weighted constellations $\left\{\mathcal{C}_j\right\}_j$ for a subset $\mathcal{C}_j = \left\{\left(w_\alpha,\alpha\right)\right\}$ where $w_\alpha\in\C$ is the weight of point $\alpha\in\C^n$ and all $\alpha\in\C^n$ component of elements of $\mathcal{C}_j$ and $\mathcal{C}_k$ are distinct (i.e. $\alpha\neq\beta$ for all $\left(w_\alpha,\alpha\right)\in\mathcal{C}_j$ and all $\left(w_\beta,\beta\right)\in\mathcal{C}_k$ if $j\neq k$).
A constellation $\mathcal{C}_j$ codestate over $n$ modes is defined as
\begin{equation}
    \ket{\mathcal{C}_j} = \frac{1}{\mathcal{N}_j} \sum_{\left(w_\alpha,\alpha\right)\in\mathcal{C}_j} w_\alpha \ket{\alpha}
\end{equation}
Here $\mathcal N_j>0$ is the positive normalization factor, and $\ket{\alpha}=\ket{\alpha_1}\otimes\dots\otimes\ket{\alpha_n}$ with $\ket{\alpha_i}$ being a coherent state at the $i$-th mode.
We assume that each constellation contains finitely many distinct amplitude points.

For multi-indices $p,q\in\Z_{\geq0}^n$, let $L_{p,q}\left(a^\dag,a\right)=\prod_{m=1}^n\left(a_m^\dag\right)^{p_m}a_m^{q_m}$ be the \emph{ladder monomial} with total degree $\abs{p}+\abs{q}$, where $\abs{p}=\sum_{m=1}^n p_m$.
The identities $a_m\ket{\alpha}=\alpha_m\ket{\alpha}$ and $\bra{\beta}a_m^\dag=\beta_m^*\bra{\beta}$ give
\begin{equation}
    \brakets{\beta}{L_{p,q}\left(a^\dag,a\right)}{\alpha}=\prod_{m=1}^n\left(\beta_m^*\right)^{p_m}\alpha_m^{q_m}\braket{\beta}{\alpha} .
\end{equation}
Therefore, the exact matrix elements entering eqn.\ \eqref{eqn:KL_Gram} are
\begin{equation}
    \brakets{\mathcal C_i}{L_{p,q}\left(a^\dag,a\right)}{\mathcal C_j}
    =\frac{1}{\mathcal N_i\mathcal N_j}
      \sum_{\substack{\left(w_\beta,\beta\right)\in\mathcal C_i\\
                       \left(w_\alpha,\alpha\right)\in\mathcal C_j}}
      w_\beta^*w_\alpha
      \prod_{m=1}^n\left(\beta_m^*\right)^{p_m}\alpha_m^{q_m}
      \braket{\beta}{\alpha} .
\end{equation}
In particular, the exact normalization and logical-state overlap are given by
\begin{equation}\label{eqn:weighted_normalization_Gram}
\begin{gathered}
    \mathcal N_j^2=\sum_{\substack{\left(w_\beta,\beta\right)\in\mathcal C_j\\\left(w_\alpha,\alpha\right)\in\mathcal C_j}}w_\beta^*w_\alpha\braket{\beta}{\alpha},\qquad
    \braket{\mathcal C_i}{\mathcal C_j}=\frac{1}{\mathcal N_i\mathcal N_j}\sum_{\substack{\left(w_\beta,\beta\right)\in\mathcal C_i\\\left(w_\alpha,\alpha\right)\in\mathcal C_j}}w_\beta^*w_\alpha\braket{\beta}{\alpha} ,
\end{gathered}
\end{equation}
where the coherent-state overlap satisfies $\abs{\braket{\beta}{\alpha}}=\exp\left(-\nrm{\beta-\alpha}_2^2/2\right)$.

Now, let $\mathcal A=\bigcup_j\left\{\alpha:\left(w_\alpha,\alpha\right)\in\mathcal C_j\right\}$ be the set of all coherent state amplitudes in the constellations of a code spanned by $\left\{\mathcal{C}_j\right\}_j$ (assuming that it contain at least two amplitudes).
The \emph{resolution} of this code is defined by
\begin{equation}\label{eqn:constellation_resolution}
    d_E=\min_{\substack{\alpha,\beta\in\mathcal A\\\alpha\neq\beta}}\nrm{\alpha-\beta}_2^2,\qquad 
    \epsilon=e^{-d_E/2} .
\end{equation}
Here $\nrm{\cdot}_2$ denotes the Euclidean norm on $\C^n$, so $\nrm{\alpha-\beta}_2^2=\sum_{m=1}^n\abs{\alpha_m-\beta_m}^2$ and $d_E$ is the minimum squared Euclidean distance between distinct constellation points.
Thus $\sqrt{d_E}$ is the minimum separation between distinct amplitudes, and $\epsilon$ bounds the modulus overlap $\abs{\braket{\alpha}{\beta}}$ for all amplitudes $\alpha\neq\beta$ in $\mathcal{A}$.

Define $W_j=\sum_{\left(w_\alpha,\alpha\right)\in\mathcal C_j}\abs{w_\alpha}^2$.
{We impose the same normalized moments on every logical constellation:}
\begin{equation}\label{eqn:weighted_moment_constraints}
    \frac{1}{W_j}
    \sum_{\left(w_\alpha,\alpha\right)\in\mathcal C_j}
    \abs{w_\alpha}^2\prod_{m=1}^n\left(\alpha_m^*\right)^{p_m}\alpha_m^{q_m}
    =c_{p,q}\quad\forall j .
\end{equation}
{To relate the moment constraints in eqn.\ \eqref{eqn:weighted_moment_constraints} to the KL conditions, we approximate the exact normalized matrix elements $\brakets{\mathcal C_i}{L_{p,q}\left(a^\dag,a\right)}{\mathcal C_j}$ by $c_{p,q}\delta_{i,j}$ and the exact logical-state overlaps $\braket{\mathcal C_i}{\mathcal C_j}$ by $\delta_{i,j}$.}
{This approximation discards the contributions containing coherent-state overlaps $\braket{\beta}{\alpha}$ with $\alpha\neq\beta$, including their contributions to the normalization factors; its errors are controlled by $\epsilon=e^{-d_E/2}$ in the large-separation limit.}
{To bound these errors, write $R=\max\left\{1,\max_{\alpha\in\mathcal A}\nrm{\alpha}_2\right\}$ and first separate equal- and unequal-amplitude contributions using the following exact identities:}
\begin{equation}
\begin{aligned}
    \mathcal N_j^2&=W_j+\sum_{\substack{\left(w_\beta,\beta\right),\left(w_\alpha,\alpha\right)\in\mathcal C_j\\\alpha\neq\beta}}w_\beta^*w_\alpha\braket{\beta}{\alpha},\\
    \brakets{\mathcal C_i}{L_{p,q}\left(a^\dag,a\right)}{\mathcal C_j}&=\frac{\delta_{i,j}}{\mathcal N_j^2}\sum_{\left(w_\alpha,\alpha\right)\in\mathcal C_j}\abs{w_\alpha}^2\prod_{m=1}^n\left(\alpha_m^*\right)^{p_m}\alpha_m^{q_m}+B_{i,j}^{p,q},\\
    B_{i,j}^{p,q}&=\frac{1}{\mathcal N_i\mathcal N_j}\sum_{\substack{\left(w_\beta,\beta\right)\in\mathcal C_i,\;\left(w_\alpha,\alpha\right)\in\mathcal C_j\\\alpha\neq\beta}}w_\beta^*w_\alpha\prod_{m=1}^n\left(\beta_m^*\right)^{p_m}\alpha_m^{q_m}\braket{\beta}{\alpha} .
\end{aligned}
\end{equation}
The normalization correction $\abs{\mathcal N_j^2-W_j}$ and the remainder magnitude $\abs{B_{i,j}^{p,q}}$ satisfy the following bounds:
\begin{equation}
\begin{gathered}
    \abs{\mathcal N_j^2-W_j}\leq\epsilon\left(\left(\sum_{\left(w_\alpha,\alpha\right)\in\mathcal C_j}\abs{w_\alpha}\right)^2-W_j\right) , \\
    \abs{B_{i,j}^{p,q}}\leq\frac{\epsilon R^{\abs{p}+\abs{q}}}{\mathcal N_i\mathcal N_j}\left(\sum_{\left(w_\beta,\beta\right)\in\mathcal C_i}\abs{w_\beta}\right)\left(\sum_{\left(w_\alpha,\alpha\right)\in\mathcal C_j}\abs{w_\alpha}\right) .
\end{gathered}
\end{equation}

A \emph{logical-state Gram matrix} is the $K\times K$ matrix $S$ with entries $S_{i,j}=\braket{\mathcal C_i}{\mathcal C_j}$, explicitly given in eqn.\ \eqref{eqn:weighted_normalization_Gram}.
Setting the unequal-amplitude contributions to zero replaces $\mathcal N_j^2$ by $W_j$ and $B_{i,j}^{p,q}$ by $0$, giving $c_{p,q}\delta_{i,j}$ and $S_{i,j}=\delta_{i,j}$ in this approximation.
The two bounds above quantify the errors in the moment approximation to the KL matrix elements caused by ignoring overlaps between coherent states with distinct amplitudes.
Now, fix the finite sets of distinct amplitudes and their nonzero weight lists, let their union be $\mathcal A_0$, and multiply every amplitude by the same real factor $\varrho\geq1$ and denote $\mathcal A\left(\varrho\right)=\varrho\mathcal A_0$ while the weights remain fixed.
Define
\begin{equation}
    d_{E,0}=\min_{\substack{\zeta,\xi\in\mathcal A_0\\\zeta\neq\xi}}\nrm{\zeta-\xi}_2^2>0,\qquad R_0=\max\left\{1,\max_{\zeta\in\mathcal A_0}\nrm{\zeta}_2\right\} .
\end{equation}
Then the definitions of $d_E$, $\epsilon$, and $R$ give the scaled resolution and errors
\begin{equation}
\begin{gathered}
    d_E\left(\varrho\right)=\varrho^2d_{E,0},\qquad \epsilon\left(\varrho\right)=e^{-\varrho^2d_{E,0}/2},\qquad R\left(\varrho\right)\leq\varrho R_0,\\
    \text{so that}\\
    R\left(\varrho\right)^{\abs{p}+\abs{q}}\epsilon\left(\varrho\right)\leq\varrho^{\abs{p}+\abs{q}}R_0^{\abs{p}+\abs{q}}e^{-\varrho^2d_{E,0}/2}\longrightarrow0\quad\text{as}\quad\varrho\to\infty .
\end{gathered}
\end{equation}
For fixed $p,q$, the exponential decrease therefore dominates the polynomial amplitude factor, and the normalization bound gives $\mathcal N_j^2\left(\varrho\right)\to W_j>0$, so the denominators in the bound on $B_{i,j}^{p,q}$ stay bounded away from zero.
This shows the decreasing coherent-state overlaps in the large-separation limit.
The common target moments also scale as $c_{p,q}\left(\varrho\right)=\varrho^{\abs{p}+\abs{q}}c_{p,q}\left(1\right)$, so their equality across logical constellations is preserved, rather than their numerical values being held fixed.

%

%
%
%
%

%

At finite separation, the exact detectability condition is $\brakets{\mathcal C_i}{L_{p,q}\left(a^\dag,a\right)}{\mathcal C_j}=\lambda_{p,q}\braket{\mathcal C_i}{\mathcal C_j}$ for every $i,j$.
Substituting the coherent-state matrix elements and eqn.\ \eqref{eqn:weighted_normalization_Gram} gives
\begin{equation}
\begin{aligned}
    &\frac{1}{\mathcal N_i\mathcal N_j}\sum_{\substack{\left(w_\beta,\beta\right)\in\mathcal C_i\\\left(w_\alpha,\alpha\right)\in\mathcal C_j}}w_\beta^*w_\alpha\prod_{m=1}^n\left(\beta_m^*\right)^{p_m}\alpha_m^{q_m}\braket{\beta}{\alpha}
    =\frac{\lambda_{p,q}}{\mathcal N_i\mathcal N_j}\sum_{\substack{\left(w_\beta,\beta\right)\in\mathcal C_i\\\left(w_\alpha,\alpha\right)\in\mathcal C_j}}w_\beta^*w_\alpha\braket{\beta}{\alpha} .
\end{aligned}
\end{equation}
Multiplying both sides by the nonzero factor $\mathcal N_i\mathcal N_j$ gives
\begin{equation}
    \sum_{\substack{\left(w_\beta,\beta\right)\in\mathcal C_i\\\left(w_\alpha,\alpha\right)\in\mathcal C_j}}w_\beta^*w_\alpha\prod_{m=1}^n\left(\beta_m^*\right)^{p_m}\alpha_m^{q_m}\braket{\beta}{\alpha}
    =\lambda_{p,q}\sum_{\substack{\left(w_\beta,\beta\right)\in\mathcal C_i\\\left(w_\alpha,\alpha\right)\in\mathcal C_j}}w_\beta^*w_\alpha\braket{\beta}{\alpha} .
\end{equation}
Subtracting the right-hand side and combining the two sums yields
\begin{equation}\label{eqn:weighted_KL_exact}
    \sum_{\substack{\left(w_\beta,\beta\right)\in\mathcal C_i\\\left(w_\alpha,\alpha\right)\in\mathcal C_j}}w_\beta^*w_\alpha\left(\prod_{m=1}^n\left(\beta_m^*\right)^{p_m}\alpha_m^{q_m}-\lambda_{p,q}\right)\braket{\beta}{\alpha}=0\quad\forall i,j .
\end{equation}
{This equation is equivalent to the exact detectability condition $PL_{p,q}P=\lambda_{p,q}P$ and retains every coherent-state overlap term; when $L_{p,q}=E_\mu^\dag E_\nu$, it is the corresponding exact KL condition for that error pair.}

Following Refs.\ \cite{Jain2023,Yang2025}, define the loss-detection distance $d_\downarrow$ as the largest integer such that every $L_{0,q}$ with $\abs{q}<d_\downarrow$ is detectable.
The loss-correction parameter $t_\downarrow$ is the largest integer such that the error set $\left\{L_{0,q}:\abs{q}<t_\downarrow\right\}$ is correctable.
Since $L_{0,p}^\dag L_{0,q}=L_{p,q}$, this is equivalent to detectability of every $L_{p,q}$ with $\abs{p},\abs{q}<t_\downarrow$.
The degree distance $d_\updownarrow$ is the largest integer such that every $L_{p,q}$ with $\abs{p}+\abs{q}<d_\updownarrow$ is detectable.
Thus the numbers of detectable and correctable photon losses are $d_\downarrow-1$ and $t_\downarrow-1$, respectively, and
\begin{equation}
    \left\lfloor\frac{d_\updownarrow+1}{2}\right\rfloor\leq t_\downarrow\leq d_\updownarrow\leq d_\downarrow .
\end{equation}
When these parameters are inferred from constellation moments alone, the corresponding detection and correction statements apply in the large-separation approximation; exact finite-separation statements require eqn.\ \eqref{eqn:weighted_KL_exact} for the relevant ladder monomials.

These finite-degree moment conditions concern the specified ladder monomials, not the entire pure-loss channel at a nonzero loss probability.
Writing $\eta$ for transmissivity, the exact pure-loss Kraus operators contain a photon-number-dependent factor:
\begin{equation}\label{eqn:moment_vs_finite_loss}
    K_l\left(\eta\right)=\sqrt{\frac{\left(1-\eta\right)^l}{l!}}\,\eta^{a^\dag a/2}a^l,\qquad 0<\eta<1.
\end{equation}
The factor $\eta^{a^\dag a/2}$ is not a finite-degree polynomial, so moment matching through a fixed degree does not imply exact correction of all products of these Kraus operators.
The overlap bounds above control the stated monomial matrix elements; they are not by themselves a uniform recovery-error bound for the full channel.
The full-channel benchmark is specified independently in Section~\ref{sec:noise-benchmarking} and Appendix~\ref{sec:beam-splitter-derivation}.

Our explicit code constructions and numerical benchmarks in the remainder of this paper focus on the single-mode, uniform-weight special case of the framework above.
Section~\ref{sec:constellations} formulates its moment constraints as polynomial equations, which we use to construct {BARC codes}.
The weighted multi-mode formulation remains available as a generalization, and the bounds on resolution $d_E$ also cover multiple modes.

\section{Finding constellations for BARC codes}\label{sec:constellations}\label{sec:single_mode_constellation_polynomials}

%

%

In this section, we formulate the moment constraints of Section~\ref{sec:coherent_state_codes} as real polynomial equations for single-mode constellations with uniform weights.
We write each coherent-state amplitude as $\alpha_j=u_j+\iota v_j$, separate the real and imaginary parts of the ladder-moment equations, and replace them by equivalent equations for the real monomial moments $\tau_{a,b}$.
We prove that the two sets of real polynomials generate the same ideal and hence have the same real solution set.
Appendix~\ref{sec:weighted_constellation_polynomials} gives the weighted multimode extension, including the corresponding change-of-coordinates lemma.
The moment equations describe the diagonal contributions to the KL conditions; exact finite-separation correction still requires the overlap terms in eqn.\ \eqref{eqn:weighted_KL_exact}, as discussed in Remark~\ref{rem:moment_vs_full_KL}.

We specialize the construction in Section~\ref{sec:coherent_state_codes} to one mode, uniform weights $w_\alpha=1$, and pairwise disjoint constellations $\mathcal C_k\subseteq\C$ with $\abs{\mathcal C_k}=\ell$ and $K\ell>1$.
Then $W_k=\ell$, and the normalized moment conditions in eqn.\ \eqref{eqn:weighted_moment_constraints} reduce to equality of the sums $\sum_{\alpha\in\mathcal C_k}\left(\alpha^*\right)^p\alpha^q$ across $k$.
We denote these unnormalized sums by $c_{p,q}$, so the corresponding normalized target moments in eqn.\ \eqref{eqn:weighted_moment_constraints} are $c_{p,q}/\ell$.
The exact normalization remains that of eqn.\ \eqref{eqn:weighted_normalization_Gram}, and exact finite-separation detectability still requires eqn.\ \eqref{eqn:weighted_KL_exact}.

{For $p,q\in\Z_{\geq0}$ and $c_{p,q}\in\C$, consider a polynomial $f_{p,q}\in\C\left[x_1,y_1,\dots,x_l,y_l\right]$ given by}
\begin{equation}
    f_{p,q}\left(x\right) = x_1^p y_1^q + x_2^p y_2^q + \dots + x_\ell^p y_\ell^q - c_{p,q} .
\end{equation}
For constellations $\mathcal{C}_1,\dots,\mathcal{C}_K\subseteq\C$ with $\abs{\mathcal{C}_k}=\ell$, let us define complex polynomial
\begin{equation}
    f_{p,q}\left(\mathcal{C}_k\right) = - c_{p,q} + \sum_{\alpha\in\mathcal{C}_k} \left(\alpha^*\right)^p \alpha^q .
\end{equation}

Assume the constellations are pairwise disjoint and contain
$\ell$ distinct points each. Set
\begin{equation}
\begin{gathered}
    \mathcal C=\bigcup_{k=1}^K\mathcal C_k,\qquad
    R=\max\left\{1,\max_{\alpha\in\mathcal C}\abs{\alpha}\right\},\qquad
    \epsilon=
    \max_{\substack{\alpha,\beta\in\mathcal C\\\alpha\neq\beta}}
    \abs{\braket{\beta}{\alpha}}
    =\exp\!\left(-\frac12
    \min_{\substack{\alpha,\beta\in\mathcal C\\\alpha\neq\beta}}
    \abs{\alpha-\beta}^2\right) .
\end{gathered}
\end{equation}
Here $R$ bounds the modulus of every coherent amplitude and is chosen to be at least $1$, while $\epsilon=e^{-d_E/2}$ is the largest overlap magnitude between distinct coherent states in $\mathcal C$.
For normalized uniform superpositions,
$f_{p,q}\left(\mathcal C_k\right)=0$ for all $k$ implies
\begin{equation}
    \brakets{\mathcal C_i}{L_{p,q}\left(a^\dag,a\right)}{\mathcal C_k}
    =\frac{c_{p,q}}{\ell}\delta_{i,k}
      +O\left(R^{p+q}\epsilon\right) ,
\end{equation}
as $\epsilon\to0$, with $K,\ell,p,q$ fixed.
Indeed, the exact normalization is
\begin{equation}
\begin{aligned}
    \mathcal N_k^2&=\sum_{\alpha,\beta\in\mathcal C_k}\braket{\beta}{\alpha}=\ell+\sum_{\substack{\alpha,\beta\in\mathcal C_k\\\alpha\neq\beta}}\braket{\beta}{\alpha}
    \quad\text{and}\quad
    \abs{\mathcal N_k^2-\ell}&\leq\ell\left(\ell-1\right)\epsilon .
\end{aligned}
\end{equation}
For these uniform superpositions, the distinct-amplitude contribution is
\begin{equation}
    B_{i,k}^{p,q}=\frac{1}{\mathcal N_i\mathcal N_k}\sum_{\substack{\beta\in\mathcal C_i,\;\alpha\in\mathcal C_k\\\alpha\neq\beta}}\left(\beta^*\right)^p\alpha^q\braket{\beta}{\alpha} .
\end{equation}
Since $f_{p,q}\left(\mathcal C_k\right)=0$, separating equal and distinct amplitudes gives
\begin{equation}
\begin{gathered}
    \brakets{\mathcal C_i}{L_{p,q}\left(a^\dag,a\right)}{\mathcal C_k}=\frac{c_{p,q}\delta_{i,k}}{\mathcal N_i\mathcal N_k}+B_{i,k}^{p,q}
    \\ \text{with}\\
    \abs{B_{i,k}^{p,q}}\leq\frac{\left(\ell^2-\ell\delta_{i,k}\right)R^{p+q}\epsilon}{\mathcal N_i\mathcal N_k},\qquad \abs{c_{p,q}}\leq\ell R^{p+q} .
\end{gathered}
\end{equation}
These bounds imply $\mathcal N_k=\sqrt\ell\left(1+O\left(\epsilon\right)\right)$ and the stated $O\left(R^{p+q}\epsilon\right)$ remainder as $\epsilon\to0$ with $K,\ell,p,q$ fixed.

For a prescribed list of pairs $\left(p_i,q_i\right)$, write
$f_i=f_{p_i,q_i}$ and $c_i=c_{p_i,q_i}$.
Each solution of $f_1=\dots=f_m=0$ satisfying $x_j=y_j^*$ specifies a constellation with amplitudes $\alpha_j=y_j$, provided these amplitudes are distinct.
Selecting $K$ suitable solutions with the same prescribed constants $c_i=c_{p_i,q_i}$ for $i=\in\left\{1,\dots,m\right\}$, imposes $\sum_{\alpha\in\mathcal C_k}\left(\alpha^*\right)^{p_i}\alpha^{q_i}=c_i$ for every $i$ and $k$.
For error correction, the targeted ladder monomials must include those appearing after normal ordering every product $E_\mu^\dag E_\nu$ of errors in the chosen error set.
We call the code spanned by the normalized constellation states $\ket{\mathcal C_1},\dots,\ket{\mathcal C_K}$ obtained from solutions of this common polynomial system a \emph{bosonic algebraically-restricted constellation} (BARC) code.
The name refers to the shared polynomial moment restrictions; exact finite-separation correctability still requires eqn.~\eqref{eqn:weighted_KL_exact}.
The same terminology applies to the weighted multimode construction in Appendix~\ref{sec:weighted_constellation_polynomials}.

The formal variables satisfy $x_j=\alpha_j^*$ and
$y_j=\alpha_j$. 
Write $\alpha_j=u_j+\iota v_j$ with
$u_j,v_j\in\R$, so that $x_j=u_j-\iota v_j$ and
$y_j=u_j+\iota v_j$. 
We write
$\alpha=\left(\alpha_1,\dots,\alpha_\ell\right)$ and similarly for $u,v$.
So, we now have $m$ complex polynomial equations of the form
\begin{equation}
    f_i\left(\alpha\right) = -c_i + \sum_{j=1}^l \left(\alpha_j^*\right)^{p_i} \alpha_j^{q_i} = 0 .
\end{equation}
Since $f_i\left(\alpha\right) = 0$ if and only if $\Re\left(f_i\left(\alpha\right)\right) = 0$ and $\Im\left(f_i\left(\alpha\right)\right) = 0$, solving for $f_i\left(\alpha\right) = 0$ for all $i\in\left[1,m\right]$ is equivalent to solving 
\begin{equation}
\begin{gathered}
    g_i^\Re\left(u,v\right) = \Re\left(f_i\left(\alpha\right)\right) = 0
    \quad\text{and}\quad
    g_i^\Im\left(u,v\right) = \Im\left(f_i\left(\alpha\right)\right) = 0
\end{gathered}
\end{equation}
for all $i\in\left[1,m\right]$, for some real polynomials $g_i^\Re,g_i^\Im\in \R\left[u,v\right]:=\R\left[u_1,v_1,\dots,u_l,v_l\right]$.
Below we give an explicit form of these real polynomials

\begin{proposition}\label{prop:real_polynomial_equations_1mode}
    Solving for each complex polynomial equation $f_i\left(\alpha\right) = -c_i + \sum_{j=1}^l \left(\alpha_j^*\right)^{p_i} \alpha_j^{q_i} = 0$ is equivalent to solving two real polynomial equations:
    \begin{equation}
    \begin{gathered}
        g_i^\Re\left(u,v\right) = - c_i^\Re + \sum_{j=1}^l \sum_{\substack{ 0\leq a\leq p_i \\ 0\leq b\leq q_i\\ a+b\text{ even} }} \binom{p_i}{a}\binom{q_i}{b} \left(-1\right)^{a + \frac{a+b}{2}} u_j^{p_i+q_i-a-b} v_j^{a+b} = 0 \\
        \text{and}\\
        g_i^\Im\left(u,v\right) = - c_i^\Im + \sum_{j=1}^l \sum_{\substack{ 0\leq a\leq p_i \\ 0\leq b\leq q_i\\ a+b\text{ odd} }} \binom{p_i}{a}\binom{q_i}{b} \left(-1\right)^{a+\frac{a+b-1}{2}} u_j^{p_i+q_i-a-b}v_j^{a+b} = 0 ,
    \end{gathered}
    \end{equation}
    where $c_i = c_i^\Re + \iota c_i^\Im$.
\end{proposition}
\begin{proof}
    Let $\alpha=\gamma+\iota\tau$ and then note that we can express each monomial of $\alpha$ and its conjugate $\alpha^*$ as
    \begin{equation}
    \begin{aligned}
        \left(\alpha^*\right)^p \alpha^q &= \left(\gamma-\iota \tau\right)^p \left(\gamma+\iota \tau\right)^q \\
        &= \left(\sum_{a=0}^p \binom{p}{a} \gamma^{p-a}\left(-\iota \tau\right)^a \right) \left(\sum_{b=0}^q \binom{q}{b} \gamma^{q-b}\left(\iota \tau\right)^b \right) \\
        &= \sum_{a=0}^p \sum_{b=0}^q \binom{p}{a}\binom{q}{b} \left(-1\right)^a \iota^{a+b} \gamma^{p+q-a-b}\tau^{a+b} .
    \end{aligned}
    \end{equation}
    Thus the real and imaginary parts of this monomial are given by
    \begin{equation}\label{eqn:real_imaginary_monomial_1mode}
    \begin{gathered}
        \Re\left(\left(\alpha^*\right)^p\alpha^q\right) = \sum_{a=0}^p \sum_{b=0}^q \binom{p}{a}\binom{q}{b} \left(-1\right)^a \Re\left(\iota^{a+b}\right) \gamma^{p+q-a-b}\tau^{a+b} \\
            = \sum_{\substack{ 0\leq a\leq p \\ 0\leq b\leq q\\ a+b\text{ even} }} \binom{p}{a}\binom{q}{b} \left(-1\right)^{a + \frac{a+b}{2}} \gamma^{p+q-a-b}\tau^{a+b} \\
        \text{and}\\
        \Im\left(\left(\alpha^*\right)^p\alpha^q\right) = \sum_{a=0}^p \sum_{b=0}^q \binom{p}{a}\binom{q}{b} \left(-1\right)^a \Im\left(\iota^{a+b}\right) \gamma^{p+q-a-b}\tau^{a+b} \\
            = \sum_{\substack{ 0\leq a\leq p \\ 0\leq b\leq q\\ a+b\text{ odd} }} \binom{p}{a}\binom{q}{b} \left(-1\right)^{a+\frac{a+b-1}{2}} \gamma^{p+q-a-b}\tau^{a+b} ,
    \end{gathered}
    \end{equation}
    since $\Re\left(\iota^t\right) = \left(-1\right)^{t/2}$ whenever $t$ is even and $\Re\left(\iota^t\right) = 0$ otherwise, whereas $\Im\left(\iota^t\right) = \left(-1\right)^{\left(t-1\right)/2}$ whenever $t$ is odd and $\Im\left(\iota^t\right) = 0$ otherwise.
    Solving for each complex polynomial equation $f_i\left(\alpha\right) = -c_i + \sum_{j=1}^l \left(\alpha_j^*\right)^{p_i} \alpha_j^{q_i} = 0$, therefore is equivalent to solving two real polynomial equations:
    \begin{equation}
    \begin{gathered}
        g_i^\Re\left(u,v\right) = - c_i^\Re + \sum_{j=1}^l \sum_{\substack{ 0\leq a\leq p_i \\ 0\leq b\leq q_i\\ a+b\text{ even} }} \binom{p_i}{a}\binom{q_i}{b} \left(-1\right)^{a + \frac{a+b}{2}} u_j^{p_i+q_i-a-b} v_j^{a+b} = 0 \\
        \text{and}\\
        g_i^\Im\left(u,v\right) = - c_i^\Im + \sum_{j=1}^l \sum_{\substack{ 0\leq a\leq p_i \\ 0\leq b\leq q_i\\ a+b\text{ odd} }} \binom{p_i}{a}\binom{q_i}{b} \left(-1\right)^{a+\frac{a+b-1}{2}} u_j^{p_i+q_i-a-b}v_j^{a+b} = 0 ,
    \end{gathered}
    \end{equation}
    which proves the claim.
\end{proof}

For the remainder of this section, take all $p,q\in\Z_{\geq0}$ with $p+q\leq t$, and assume $c_{p,q}=c_{q,p}^*$ and $c_{0,0}=\ell$.
Consider the ideal of $\R\left[u,v\right]$ generated by the corresponding real polynomials in Proposition~\ref{prop:real_polynomial_equations_1mode}:
\begin{equation}\label{eqn:real_constellation_ideal}
    I_t=\left\langle g_{p,q}^\Re,g_{p,q}^\Im:p+q\leq t\right\rangle .
\end{equation}
{For nonnegative integers $a,b$ with $a+b\leq t$, define the \emph{moment polynomials} by}
\begin{equation}\label{eqn:moment_polynomials}
    h_{a,b}=-\tau_{a,b}+\sum_{j=1}^l u_j^a v_j^b \;
\end{equation}
{The target moments $\tau_{a,b}$ are determined by the relations}
\begin{equation}\label{eqn:monomial_moment_constant_real}
\begin{gathered}
        c_{p,q}^\Re =
    \sum_{\substack{0\leq r\leq p+q\\r\text{ even}}}
    {\tau_{p+q-r,r}}
    \sum_{\substack{0\leq a\leq p\\0\leq b\leq q\\a+b=r}}
    \binom{p}{a}\binom{q}{b}
    {\left(-1\right)^{a+\frac{r}{2}}} \\
    \text{and}\\
    c_{p,q}^\Im =
    \sum_{\substack{0\leq r\leq p+q\\r\text{ odd}}}
    \tau_{p+q-r,r}
    \sum_{\substack{0\leq a\leq p\\0\leq b\leq q\\a+b=r}}
    \binom{p}{a}\binom{q}{b}
    \left(-1\right)^{a+\frac{r-1}{2}} .
\end{gathered}
\end{equation}
We can also equivalently define the constants $\tau_{a,b}$ in terms of the real and imaginary parts of complex constants $c_{p,q}$ for $p+q=s$ as
\begin{equation}\label{eqn:monomial_moment_constant_complex}
\begin{gathered}
    c_{p,q} = \sum_{r=0}^s A_{p,q;r}^s \tau_{s-r,r} \quad
    \text{for}\quad 
    A_{p,q;r}^s = \sum_{\substack{ 0\leq a\leq p \\ 0\leq b\leq q\\ a+b=r }} \binom{p}{a}\binom{q}{b} \left(-1\right)^a \iota^r .
\end{gathered}
\end{equation}
Let us denote the complex matrix $A^s = \left[A_{p,q;r}^s\right]_{\left(p,q\right),r} \in \C^{\left(s+1\right)\times\left(s+1\right)}$ so that for vectors $c=\left(c_{s,0},c_{s-1,1},\dots,c_{0,s}\right) \in\C^{s+1}$ and $\tau = \left(\tau_{s,0},\tau_{s-1,1},\dots,\tau_{0,s}\right)\in\C^{s+1}$ we have
\begin{equation}
    c = A^s\tau .
\end{equation}
Now we show that for any $s\geq0$, matrix $A^s$ is always invertible.

\begin{lemma}\label{lem:real_complex_polynomial_coordinate_tranform}
    Matrix $A^s$ defined in eqn.\ \eqref{eqn:monomial_moment_constant_complex} is invertible.
    Moreover $\left(A^s\right)^{-1}c$ is a real vector in $\R^{s+1}$ if and only if $c_{p,q}=c_{q,p}^*$.
\end{lemma}
\begin{proof}
    For fixed $s\geq0$, the monomials
    $\left\{u^{s-r}v^r:0\leq r\leq s\right\}$ form a basis of the
    complex vector space of homogeneous degree-$s$ polynomials
    in $u,v$.
    Set $X=u-\iota v$ and $Y=u+\iota v$.
    This change of variables is invertible, with
    \begin{equation}
        u=\frac{X+Y}{2}
        \quad\text{and}\quad
        v=\frac{Y-X}{2\iota} .
    \end{equation}
    Hence $\left\{X^pY^q:p+q=s\right\}$ is another basis of the same
    vector space. Expanding these monomials gives
    \begin{equation}
        \left(u-\iota v\right)^p\left(u+\iota v\right)^q
        =\sum_{r=0}^s A_{p,q;r}^s u^{s-r}v^r .
    \end{equation}
    Thus $A^s$ is a change-of-basis matrix and is invertible.
    
    Moreover, conjugating the coefficients of the preceding
    identity exchanges $p$ and $q$, so that
    $\left(A_{p,q;r}^s\right)^*=A_{q,p;r}^s$.
    If $\tau\in\R^{s+1}$ and $c=A^s\tau$, this identity
    implies $c_{p,q}=c_{q,p}^*$.
    Conversely, suppose $c_{p,q}=c_{q,p}^*$ and set
    $\tau=\left(A^s\right)^{-1}c$. Then, componentwise,
    \begin{equation}
    \begin{aligned}
        \left(A^s\tau^*\right)_{p,q}
        &=\sum_{r=0}^s A_{p,q;r}^s\tau_{s-r,r}^*
        =\left(\sum_{r=0}^s A_{q,p;r}^s\tau_{s-r,r}\right)^*
          =c_{q,p}^*=c_{p,q} .
    \end{aligned}
    \end{equation}
    Invertibility gives $\tau^*=\tau$, proving that
    $\tau\in\R^{s+1}$.
\end{proof}

{The conjugacy condition guarantees that the target moments $\tau_{a,b}$ are real; it does not by itself guarantee that the resulting polynomial system has a real solution.
The degree-zero condition gives $\tau_{0,0}=\ell$, so $h_{0,0}=0$ can be omitted from the generating set.}

Now, we show that the ideal generated by the moment polynomials is equal to the ideal $I_t$ generated by the $g_{p,q}^\Re,g_{p,q}^\Im$'s with $p+q\leq t$.

\begin{proposition}\label{prop:moment_poly_ideal_equals_original_ideal}
    For moment polynomials $\left\{h_{a,b} : a+b\leq t\right\}$ defined in eqn.\ \eqref{eqn:moment_polynomials}, it holds that
    \begin{equation}
        J_t = \left\langle h_{a,b} : a+b\leq t\right\rangle = I_t .
    \end{equation}
\end{proposition}
\begin{proof}
    First we show that $I_t\subseteq J_t$ by showing that each $g_{p,q}^\Re,g_{p,q}^\Im$ is a linear combination of the $h_{a,b}$'s.
    For shorthand, let $s=p+q$ and
    \begin{equation}
        A_{p,q,r}^s = \sum_{\substack{ 0\leq a\leq p \\ 0\leq b\leq q\\ a+b=r }} \binom{p}{a}\binom{q}{b} \left(-1\right)^a \iota^r ,
    \end{equation}
    so that the relation in eqn.\ \eqref{eqn:monomial_moment_constant_complex} becomes $c_{p,q} = \sum_{r=0}^s A_{p,q,r}^s \tau_{s-r,r}$.
    By this relation we can write the complex polynomials $f_{p,q}$ as
\begin{equation}
\begin{aligned}
    f_{p,q}&=-c_{p,q}+\sum_{j=1}^l\left(u_j-\iota v_j\right)^p\left(u_j+\iota v_j\right)^q\\
    &=-c_{p,q}+\sum_{r=0}^s A_{p,q,r}^s\sum_j u_j^{s-r}v_j^r\\
    &=\sum_{r=0}^s A_{p,q,r}^s h_{s-r,r}.
\end{aligned}
\end{equation}
    Since $g_{p,q}^\Re$ and $g_{p,q}^\Im$ are the real and imaginary part of $f_{p,q}$, respectively, then they can be expressed as a linear combination of the $h_{a,b}$'s.
    So, $I_t\subseteq J_t$.

    {Now we show that $J_t\subseteq I_t$ by showing that each $h_{a,b}=-\tau_{a,b}+\sum_{j=1}^l u_j^av_j^b$ is a linear combination of the $g_{p,q}^\Re,g_{p,q}^\Im$ polynomials.}
    Let $s=a+b$.
    For each $a,b\in\Z_{\geq0}$ with $a+b=s$, the
    invertibility of $A^s$ gives coefficients
    $\lambda_{p,q}^{a,b}=\left(\left(A^s\right)^{-1}\right)_{\left(a,b\right),\left(p,q\right)}$ such that
    \begin{equation}
        u_j^av_j^b
        =\sum_{p,q:p+q=s}\lambda_{p,q}^{a,b}
          \left(u_j-\iota v_j\right)^p\left(u_j+\iota v_j\right)^q .
    \end{equation}
    Then we have the relation
    \begin{equation}
        \sum_{j=1}^l \left(u_j-\iota v_j\right)^p \left(u_j+\iota v_j\right)^q = f_{p,q} + c_{p,q} ,
    \end{equation}
    which allows us to write
    \begin{equation}
        \sum_{j=1}^l u_j^av_j^b = \sum_{p,q:p+q=s} \lambda_{p,q}^{a,b} \sum_{j=1}^l \left(u_j-\iota v_j\right)^p \left(u_j+\iota v_j\right)^q 
        = \sum_{p,q:p+q=s} \lambda_{p,q}^{a,b} f_{p,q} + \sum_{p,q:p+q=s} \lambda_{p,q}^{a,b} c_{p,q} .
    \end{equation}
    Now note that for a matrix $A = \left[A_{\left(p,q\right),r}^s\right]_{\left(p,q\right),r}$ we have $c = A\tau$ by eqn.\ \eqref{eqn:monomial_moment_constant_complex} so that $\tau = A^{-1}c$, which gives us $\tau_{a,b} = \sum_{p,q:p+q=s} \lambda_{p,q}^{a,b} c_{p,q}$ where $A^{-1} = \left[\lambda_{p,q}^{a,b}\right]_{\left(a,b\right),\left(p,q\right)}$.
    Then we obtain
    \begin{equation}
    \begin{aligned}
            h_{a,b}
    = \sum_{p,q:p+q=s}\lambda_{p,q}^{a,b} f_{p,q} .
    \end{aligned}
    \end{equation}
    Taking real parts, since $h_{a,b}$ is real, gives
    \begin{equation}
        h_{a,b}
        =\sum_{p,q:p+q=s}
          \left(\Re\left(\lambda_{p,q}^{a,b}\right)g_{p,q}^\Re
              -\Im\left(\lambda_{p,q}^{a,b}\right)g_{p,q}^\Im\right) .
    \end{equation}
    Thus $h_{a,b}$ belongs to the ideal generated by the real polynomials $g_{p,q}^\Re,g_{p,q}^\Im$.
    Again, since $g_{p,q}^\Re$ and $g_{p,q}^\Im$ are the real and imaginary part of $f_{p,q}$, respectively, then $h_{a,b}$ is a linear combination of $g_{p,q}^\Re$ and $g_{p,q}^\Im$.
    Hence, $J_t\subseteq I_t$.
\end{proof}

\section{Real single-mode constellation polynomials with degree \texorpdfstring{$p+q\leq t$}{p+q ≤ t}}\label{sec:moment_tensors}

In this section, we organize the real moment constraints into symmetric tensors. 
For each degree $\xi$, the entries of the moment tensor encode the monomial sums $S_{a,b}$ with $a+b=\xi$.
We identify the corresponding target tensor and show that matching these tensors is equivalent to solving the moment polynomials from Section~\ref{sec:constellations}.
This formulation separates the constellation coordinates from their prescribed moments and prepares the study of moment fibers and their symmetries in Section~\ref{sec:moment_symmetry}.
We write $N_t=\left(t+1\right)\left(t+2\right)/2$ for the number of moment coordinates through degree $t$.

A logical code state $\ket{\mathcal{C}_k}$ corresponding to constellation $\mathcal{C}_k=\left\{\alpha_1^{\left(k\right)},\ldots,\alpha_\ell^{\left(k\right)}\right\}\subset\C$ is a uniform superposition of $\ell$ coherent states,
\begin{equation}
  \ket{\mathcal{C}_k} = \frac{1}{\mathcal{N}_k\sqrt{\abs{\mathcal{C}_k}}} \sum_{j=1}^\ell \ket{\alpha_j^{\left(k\right)}} .
\end{equation}
The common moment constraints associated with detecting ladder errors $L_{p,q}=\left(a^\dagger\right)^pa^q$ with $p+q\leq t$ are
\begin{equation}
    \sum_{j=1}^\ell\left(\alpha_j^{\left(k\right)*}\right)^p
       \left(\alpha_j^{\left(k\right)}\right)^q=c_{p,q}\quad\forall k .
\end{equation}
The constants are independent of $k$. Their relation to the full KL conditions, including finite-separation corrections, is specified in Remark~\ref{rem:moment_vs_full_KL}.

Consider $\alpha_j^{\left(k\right)}=u_j^{\left(k\right)}+\iota v_j^{\left(k\right)}$ and define the vector in the physical real plane $\R^2$
\begin{equation}
    q_j^{\left(k\right)} = u_j^{\left(k\right)}e_1+v_j^{\left(k\right)}e_2
    = \begin{bmatrix}
        u_j^{\left(k\right)}\\
        v_j^{\left(k\right)}
    \end{bmatrix}
    \quad\text{for}\quad
    e_1=\begin{bmatrix}1\\0\end{bmatrix}
    \quad\text{and}\quad
    e_2=\begin{bmatrix}0\\1\end{bmatrix} .
\end{equation}
For an integer $\xi\geq1$, let $\left(\R^2\right)^{\otimes\xi}$ denote the $\xi$-fold tensor product of $\R^2$ and define its symmetric subspace by
\begin{equation}\label{eqn:symmetric_tensor_space}
    \operatorname{Sym}^{\xi}\left(\R^2\right)
    = \left\{T\in\left(\R^2\right)^{\otimes\xi}:P_\pi T=T\text{ for every }\pi\in\mathfrak{S}_\xi\right\} ,
\end{equation}
where $P_\pi:\left(\R^2\right)^{\otimes\xi}\rightarrow\left(\R^2\right)^{\otimes\xi}$ is the permutation map that permutes the $\xi$ tensor factors according to $\pi$ and $\mathfrak{S}_\xi$ is the permutation group on $\left\{1,\dots,\xi\right\}$.
We also set $\operatorname{Sym}^0\left(\R^2\right)=\R$ and $q^{\otimes0}=1$.

For $r\in\left\{0,\dots,\xi\right\}$, define the unnormalized symmetric-tensor monomial
\begin{equation}\label{eqn:symmetric_tensor_monomial_basis}
    E_{\xi,r}
    = \sum_{\substack{S\subseteq\left\{1,\dots,\xi\right\}\\\abs{S}=r}}
    \bigotimes_{m=1}^{\xi}e_{\varepsilon_S\left(m\right)}
    \quad\text{for}\quad
    \varepsilon_S\left(m\right)=
    \begin{cases}
        2,&m\in S,\\
        1,&m\notin S.
    \end{cases}
\end{equation}
Thus, $E_{\xi,r}$ is the sum of all tensor words containing $\xi-r$ copies of $e_1$ and $r$ copies of $e_2$.
The set
\begin{equation}
    \mathcal{E}_\xi=\left\{E_{\xi,0},E_{\xi,1},\dots,E_{\xi,\xi}\right\}
\end{equation}
is a basis of $\operatorname{Sym}^{\xi}\left(\R^2\right)$, and for $q=ue_1+ve_2$ it holds that
\begin{equation}\label{eqn:pure_symmetric_tensor_expansion}
    q^{\otimes\xi}
    = \sum_{r=0}^{\xi}u^{\xi-r}v^rE_{\xi,r} .
\end{equation}

For a constellation $\mathcal{C}_k$ and $\left(a,b\right)\in\Z_{\geq0}\times\Z_{\geq0}$, its real moment sums
\begin{equation}\label{eqn:raw_real_moment_sum_code state}
    S_{a,b}^{\left(k\right)}
    = \sum_{j=1}^l\left(u_j^{\left(k\right)}\right)^a\left(v_j^{\left(k\right)}\right)^b
\end{equation}
and the corresponding moment polynomials $\left\{h_{p,q}:p+q\leq t\right\} \subseteq\R\left[u,v\right]$ introduced in eqn.\ \eqref{eqn:moment_polynomials} can be written as
\begin{equation}
    h_{a,b}^{\left(k\right)}
    = -\tau_{a,b}+S_{a,b}^{\left(k\right)} .
\end{equation}

\begin{definition}[Moment tensor and target moment tensor]\label{def:moment_and_target_tensors}
    For $0\leq\xi\leq t$, the \emph{$\xi$-th moment tensor} of $\mathcal{C}_k$ is
    \begin{equation}\label{eqn:moment_tensor_definition}
        M_\xi^{\left(k\right)}
        = \sum_{j=1}^l\left(q_j^{\left(k\right)}\right)^{\otimes\xi}
        \in\operatorname{Sym}^{\xi}\left(\R^2\right) 
        \quad\text{for}\quad q_j^{\left(k\right)} = u_j^{\left(k\right)} e_1 + v_j^{\left(k\right)} e_2 .
    \end{equation}
    The \emph{target $\xi$-th moment tensor} corresponding to the real constants $\tau=\left(\tau_{a,b}\right)_{a+b\leq t}$ is
    \begin{equation}\label{eqn:target_moment_tensor_definition}
        T_\xi\left(\tau\right)
        = \sum_{r=0}^{\xi}\tau_{\xi-r,r}E_{\xi,r}
        \in\operatorname{Sym}^{\xi}\left(\R^2\right) .
    \end{equation}
\end{definition}

By eqns.\ \eqref{eqn:pure_symmetric_tensor_expansion} and~\eqref{eqn:raw_real_moment_sum_code state}, the coordinates of $M_\xi^{\left(k\right)}$ in the basis $\mathcal{E}_\xi$ are the raw moments,
\begin{equation}\label{eqn:moment_tensor_coordinate_expansion}
\begin{gathered}
    M_\xi^{\left(k\right)}
    = \sum_{r=0}^{\xi}S_{\xi-r,r}^{\left(k\right)}E_{\xi,r}
    \quad\text{and}\quad
    \left[M_\xi^{\left(k\right)}\right]_{\mathcal{E}_\xi}
    = \left(S_{\xi,0}^{\left(k\right)},S_{\xi-1,1}^{\left(k\right)},\dots,S_{0,\xi}^{\left(k\right)}\right) .
\end{gathered}
\end{equation}
Consequently, the tuple of moment polynomials is the coordinate vector of the moment-tensor residual rather than the coordinate vector of $M_\xi^{\left(k\right)}$ itself:
\begin{equation}\label{eqn:moment_tensor_residual_expansion}
\begin{gathered}
    M_\xi^{\left(k\right)}-T_\xi\left(\tau\right)
    = \sum_{r=0}^{\xi}h_{\xi-r,r}^{\left(k\right)}E_{\xi,r} \\
    \text{and}\quad
    \left[M_\xi^{\left(k\right)}-T_\xi\left(\tau\right)\right]_{\mathcal{E}_\xi}
    = \left(h_{\xi,0}^{\left(k\right)},h_{\xi-1,1}^{\left(k\right)},\dots,h_{0,\xi}^{\left(k\right)}\right) .
\end{gathered}
\end{equation}

\begin{proposition}[Tensor form of the moment constraints]\label{prop:tensor_form_moment_constraints}
    Let $\tau=\left(\tau_{a,b}\right)_{a+b\leq t}\in\R^{N_t}$ and let the compatible complex constants $c=\left(c_{p,q}\right)_{p+q\leq t}$ be related to $\tau$ by eqns.\ \eqref{eqn:monomial_moment_constant_real},\eqref{eqn:monomial_moment_constant_complex}.
    For every constellation $\mathcal{C}_k$, the following statements are equivalent:
    \begin{enumerate}
        \item $h_{a,b}^{\left(k\right)}=0$ for every $a,b\geq0$ with $a+b\leq t$.
        \item $M_\xi^{\left(k\right)}=T_\xi\left(\tau\right)$ for every $0\leq\xi\leq t$.
        \item $f_{p,q}\left(\mathcal{C}_k\right)=0$ for every $p,q\geq0$ with $p+q\leq t$.
    \end{enumerate}
\end{proposition}
\begin{proof}
    By eqn.\ \eqref{eqn:moment_tensor_residual_expansion} and linear independence of the basis tensors $E_{\xi,0},\dots,E_{\xi,\xi}$, it holds that
    \begin{equation}
        M_\xi^{\left(k\right)}=T_\xi\left(\tau\right)
        \quad\Longleftrightarrow\quad
        h_{\xi-r,r}^{\left(k\right)}=0\quad\forall r\in\left\{0,\dots,\xi\right\} .
    \end{equation}
    Taking all $0\leq\xi\leq t$ proves equivalence of statements 1 and 2.
    Equivalence of statements 1 and 3 follows from Proposition~\ref{prop:moment_poly_ideal_equals_original_ideal}, or directly from the invertible degree-wise coordinate transformation $c=A\tau$ and $\Tilde{\Phi}_t=A\Phi_t$.
\end{proof}

\begin{remark}[Moment constraints and the full KL condition]\label{rem:moment_vs_full_KL}
    Proposition~\ref{prop:tensor_form_moment_constraints} characterizes the equality of the atomic ladder moments
    \begin{equation}
        \sum_{j=1}^l\left(\alpha_j^{\left(k\right)*}\right)^p\left(\alpha_j^{\left(k\right)}\right)^q=c_{p,q}
    \end{equation}
    across the logical constellations.
    For finite coherent-state separation, the exact KL matrix element also contains coherent-state overlap terms.
    These overlap terms must either cancel exactly or be shown to be sufficiently small; this point is made explicit below when constructing codes from fiber orbits.
\end{remark}

\subsection{Constellations from moment-map symmetry}\label{sec:moment_symmetry}

We now regard the moment constraints as the preimage of a prescribed moment vector under a polynomial map.
This preimage, called a moment fiber, contains all constellation coordinates with those prescribed moments.
We identify transformations that preserve a moment fiber and use their orbits to generate candidate logical constellations from one seed.
Polynomial invariants give sufficient conditions for preserving these moments, while the coherent-state overlap terms must still be checked as in Section~\ref{sec:coherent_state_codes}.

First, recall that for integer $t\geq2$ and real polynomials $g_{p,q}^\Re=\Re\left(f_{p,q}\right)$ and $g_{p,q}^\Im=\Im\left(f_{p,q}\right)$, we have the ideal
\begin{equation}\label{eqn:target_moment_ideal}
    J_t^\tau=\left\langle h_{a,b}:a+b\leq t\right\rangle
    =\left\langle g_{p,q}^\Re,g_{p,q}^\Im:p+q\leq t\right\rangle=I_t .
\end{equation}
Here $h_{a,b}=-\tau_{a,b}+\sum_{i=1}^l u_i^av_i^b$, as established in Proposition~\ref{prop:moment_poly_ideal_equals_original_ideal}.
The number of real moment coordinates with total degree at most $t$ is
\begin{equation}
    N_t=\abs{\left\{\left(a,b\right)\in\Z_{\geq0}^2:a+b\leq t\right\}}
    = \frac{\left(t+1\right)\left(t+2\right)}{2} .
\end{equation}
The degree-zero coordinate is $S_{0,0}=l$ and can be omitted whenever the number $l$ of constellation points is fixed.

\begin{definition}[Real moment map and moment fiber]\label{def:real_moment_map_fiber}
    The \emph{degree-$t$ real moment map} is
    \begin{equation}\label{eqn:real_moment_map}
        \Phi_t:\R^{2l}\longrightarrow\R^{N_t}
        \quad\text{with}\quad
        \Phi_t\left(u,v\right)=\left(S_{a,b}\left(u,v\right)\right)_{a+b\leq t} .
    \end{equation}
    For $\tau=\left(\tau_{a,b}\right)_{a+b\leq t}\in\R^{N_t}$, the preimage $F_\tau=\Phi_t^{-1}\left(\tau\right)$ is the \emph{fiber} of $\tau$ or a \emph{$\tau$-fiber}:
    \begin{equation}\label{eqn:real_moment_fiber}
        F_\tau
        = \Phi_t^{-1}\left(\tau\right)
        = \left\{\left(u,v\right)\in\R^{2l}:S_{a,b}\left(u,v\right)=\tau_{a,b}\quad\forall a+b\leq t\right\} ,
    \end{equation}
    which is the set of all solutions corresponding to constants $\tau$
\end{definition}
{By the definition of $J_t^\tau$ in eqn.\ \eqref{eqn:target_moment_ideal},}
\begin{equation}\label{eqn:variety_equals_moment_fiber}
    \mathbf{V}\left(J_t^\tau\right)=F_\tau=\Phi_t^{-1}\left(\tau\right) .
\end{equation}
Thus two ordered constellations belong to the same moment fiber if and only if they have identical real monomial moments through total degree $t$.

For fixed $l$, the degree-zero polynomial is $h_{0,0}=l-\tau_{0,0}$.
The compatibility condition $\tau_{0,0}=c_{0,0}=l$ makes this polynomial zero, so it can be omitted from the generating set.
Any other identically zero real or imaginary constraint can likewise be omitted. A nonzero constant generator must be retained: it implies $1\in I_t$ and hence $\mathbf V\left(I_t\right)=\varnothing$.
For $a+b>0$, the polynomial $h_{a,b}=S_{a,b}-\tau_{a,b}$ has a nonconstant part.
We retain the degree-zero coordinate in Definition~\ref{def:real_moment_map_fiber} to keep the indexing uniform, with $\tau_{0,0}=l$ understood.

In coordinates, the moment map of Definition~\ref{def:real_moment_map_fiber} is given by
\begin{equation}
    \Phi_t\left(u_1,v_1,\dots,u_l,v_l\right) = \left(S_{a,b}\right)_{a+b\leq t} 
    \quad\text{for}\quad
    S_{a,b} = \sum_{i=1}^l u_i^a v_i^b .
\end{equation}
{Then the variety of $J_t^\tau$ in eqn.\ \eqref{eqn:target_moment_ideal} (equivalently the variety of $I_t$ defined in eqn.\ \eqref{eqn:real_constellation_ideal}) can be written as}
\begin{equation}
\begin{gathered}
    \mathbf{V}\left(J_t^\tau\right) = \Phi_t^{-1}\left(\tau\right) := \left\{\left(u,v\right)\in\R^{2l} : S_{a,b}\left(u,v\right) = \tau_{a,b} ,\; a+b\leq t \right\} \\
    \text{for}\quad
    \tau = \left(\tau_{a,b}\right)_{a+b\leq t} \in\R^{N_t} .
\end{gathered}
\end{equation}

In the complex coordinate $\alpha_j = u_j+\iota v_j$, we define the \emph{complex moment map} as
\begin{equation}
    \Tilde{\Phi}_t\left(\alpha_1,\dots,\alpha_l\right) = \left(\Tilde{S}_{p,q}\right)_{p+q\leq t}
    \quad\text{for}\quad
    \Tilde{S}_{p,q} = \sum_{i=1}^l \left(\alpha_i^*\right)^p \alpha_i^q ,
\end{equation}
which $c$-fiber $\Tilde{\Phi}_t^{-1}\left(c\right)$ is defined similarly to the real moment map.
By Lemma~\ref{lem:real_complex_polynomial_coordinate_tranform}, then for coordinate relation $\alpha_i=u_i+\iota v_i$ between $\left(u,v\right)\in\R^{2l}$ and $\alpha\in\C^l$ the real and complex moment maps are related by a coordinate-transformation matrix $A$ as
\begin{equation}
\begin{gathered}
    \Tilde{\Phi}_t\left(\alpha\right) = A \Phi_t\left(u,v\right)
    \quad\text{for}\quad
    A = \bigoplus_{s=0}^t A^s
\end{gathered}
\end{equation}
where $A^s$ acts on the $\left(a,b\right)$-th entries of $\Phi_t\left(u,v\right)$ such that $a+b=s$.

\subsubsection{Fiber group symmetries}

We distinguish transformations that preserve every fiber of the moment map from transformations that preserve only one prescribed fiber.
For transformations with an induced map on the target moments, preserving that fiber amounts to fixing its target moment vector.
We apply this criterion to phase-space rotations and determine which rotations preserve a given set of complex moments.
These rotations generate candidate constellations with the same moments; their finite-separation KL conditions still require the overlap terms in eqn.\ \eqref{eqn:weighted_KL_exact}.

\begin{definition}
    For $t\geq2$, a group $G$ with group action on $\R^{2l}$ is a \emph{group symmetry} of moment map $\Phi_t$ if $\Phi_t\circ g = \Phi_t$, i.e. $\Phi_t\left(g\left(u,v\right)\right) = \Phi_t\left(u,v\right)$ for all $\left(u,v\right)\in\R^{2l}$.
\end{definition}

\begin{definition}[Fiber stabilizer and constant map]
    For a nonempty fiber $F_\tau$, define its linear stabilizer by
    \begin{equation}
        \stab_\tau=\left\{T\in\GL\left(2l,\R\right):TF_\tau=F_\tau\right\} .
    \end{equation}
    If a map $\varrho_t\left(T\right)$ on moment coordinates satisfies
    $\Phi_t\left(Tz\right)=\varrho_t\left(T\right)\Phi_t\left(z\right)$ for every $z\in\R^{2l}$,
    we call it the induced constant map. For an invertible $T$
    with invertible induced constant map, $T$ preserves $F_\tau$
    if and only if $\varrho_t\left(T\right)\tau=\tau$.
\end{definition}

\begin{proposition}[Cyclic fiber symmetries]\label{prop:cyclic_fiber_symmetries}
    Let $c\in\C^{N_t}$ specify a nonempty complex moment fiber.
    The rotation
    \begin{equation}
        R_\theta\left(u,v\right)=\left(\cos\theta\,u-\sin\theta\,v,
                       \sin\theta\,u+\cos\theta\,v\right)
    \end{equation}
    preserves this fiber if and only if
    \begin{equation}
        e^{\iota\left(q-p\right)\theta}c_{p,q}=c_{p,q}
        \quad\forall p+q\leq t .
    \end{equation}
    If the set $\left\{\abs{q-p}:p+q\leq t,\ p\neq q,\ c_{p,q}\neq0\right\}$
    is nonempty, let $D$ be its greatest common divisor.
    The preserving rotations form $\left\langle R_{2\pi/D}\right\rangle$.
    If this set is empty, every rotation preserves the fiber.
\end{proposition}

\begin{proof}
    Under $R_\theta$, the complex coordinates transform as
    \begin{equation}
    \begin{aligned}
        \alpha_j'&=\left(\cos\theta\,u_j-\sin\theta\,v_j\right)
             +\iota\left(\sin\theta\,u_j+\cos\theta\,v_j\right)\\
             &=e^{\iota\theta}\left(u_j+\iota v_j\right)
             =e^{\iota\theta}\alpha_j .
    \end{aligned}
    \end{equation}
    Thus $\Tilde R_\theta\alpha=e^{\iota\theta}\alpha$.
    The transformed real moments are
    \begin{equation}
        \tau_{a,b}\left(\theta\right)
        =\left[\Phi_t\left(R_\theta\left(u,v\right)\right)\right]_{a,b}
        =\sum_{j=1}^l\left(\cos\theta\,u_j-\sin\theta\,v_j\right)^a
                       \left(\sin\theta\,u_j+\cos\theta\,v_j\right)^b .
    \end{equation}
    For $s=p+q$, the degree-$s$ coordinate transformation of
    Lemma~\ref{lem:real_complex_polynomial_coordinate_tranform}
    therefore gives
    \begin{equation}
    \begin{aligned}
        {}\left[\Tilde\Phi_t\left(\Tilde R_\theta\alpha\right)\right]_{p,q}
        &=\sum_{r=0}^s A_{p,q;r}^s\tau_{s-r,r}\left(\theta\right)\\
        &=\sum_{j=1}^l
          \left(e^{-\iota\theta}\alpha_j^*\right)^p
          \left(e^{\iota\theta}\alpha_j\right)^q\\
        &=e^{\iota\left(q-p\right)\theta}
          \sum_{j=1}^l\left(\alpha_j^*\right)^p\alpha_j^q
         =e^{\iota\left(q-p\right)\theta}\left[\Tilde\Phi_t\left(\alpha\right)\right]_{p,q} .
    \end{aligned}
    \end{equation}
    Hence the induced constant map in complex coordinates is
    \begin{equation}
        \Tilde\Phi_t\left(\Tilde R_\theta\alpha\right)
        =\varrho_t\left(\Tilde R_\theta\right)\Tilde\Phi_t\left(\alpha\right)
        \quad\text{with}\quad
        \varrho_t\left(\Tilde R_\theta\right)
        =\diag\left(e^{\iota\left(q-p\right)\theta}\right)_{p+q\leq t} .
    \end{equation}
    In particular, if $\alpha\in\Tilde\Phi_t^{-1}\left(c\right)$, then
    $\Tilde R_\theta\alpha\in\Tilde\Phi_t^{-1}\left(c\right)$ if and only if
    \begin{equation}
        e^{\iota\left(q-p\right)\theta}c_{p,q}=c_{p,q}
        \quad\forall p+q\leq t .
    \end{equation}
    Nonemptiness of the fiber makes these conditions necessary
    for fiber preservation. Conversely, when they hold, they
    also hold with $-\theta$ in place of $\theta$. Applying the
    two inverse rotations therefore gives
    \begin{equation}
        \Tilde R_\theta\Tilde\Phi_t^{-1}\left(c\right)
        =\Tilde\Phi_t^{-1}\left(c\right) .
    \end{equation}

    Moments with $p=q$ or $c_{p,q}=0$ impose no restriction on
    $\theta$. For each remaining moment the condition is
    $\left(q-p\right)\theta\in2\pi\Z$.
    Write the distinct positive integers $\abs{q-p}$ occurring in
    these conditions as $d_1,\dots,d_r$, and let
    $D=\gcd\left(d_1,\dots,d_r\right)$.
    Since $d_i=Dm_i$ for integers $m_i$, the condition
    $D\theta\in2\pi\Z$ implies $d_i\theta\in2\pi\Z$ for every $i$.
    Conversely, B\'ezout's identity gives integers $b_1,\dots,b_r$
    such that $D=\sum_i b_i d_i$. Hence
    \begin{equation}
        d_i\theta\in2\pi\Z\quad\forall i
        \quad\Longrightarrow\quad
        D\theta=\sum_i b_i d_i\theta\in2\pi\Z .
    \end{equation}
    The preserving angles are therefore $\theta=2\pi k/D$
    modulo $2\pi$, giving the group $\left\langle R_{2\pi/D}\right\rangle$.
    If there are no such positive integers $d_i$, every
    moment is preserved for every $\theta$, as claimed.
\end{proof}

\begin{example}[Cyclic symmetry on fibers]
    For $t=3$, setting all moments with odd $q-p$ to zero
    ensures that the stabilizer contains $\left\{R_0,R_\pi\right\}$.
    Setting all off-diagonal moments to zero except possibly
    $c_{0,3}$ and $c_{3,0}$ ensures that it contains
    $\left\{R_0,R_{2\pi/3},R_{4\pi/3}\right\}$.
    The diagonal moments, including $c_{0,0}=l$ and $c_{1,1}$,
    remain subject to feasibility of the moment fiber.
\end{example}

\subsection{Polynomial invariants}\label{sec:polynomial_invariants}

\begin{definition}[{\cite[Section 7.2, Definition 7]{cox2015ideals}}]\label{def:polynomial_invariance}
    For a finite subgroup $G\leq\GL\left(l,\F\right)$ over a field $\F$, a polynomial $f\in\F\left[x_1,\dots,x_l\right]$ is \textit{$G$-invariant} if for every $g\in G$ the following polynomial identity holds:
    \begin{equation}\label{eqn:polynomial_invariance}
        f\left(x\right)=f\left(gx\right)\quad\text{in}\quad\F\left[x_1,\dots,x_l\right] .
    \end{equation}
    The set of all such polynomials is denoted as $\F\left[x_1,\dots,x_l\right]^G$.
\end{definition}

These are some of the nice properties of $G$-invariant polynomials:
\begin{itemize}
    \item It is shown in~\cite[Section 7.2, Proposition 9]{cox2015ideals} that $\F\left[x_1,\dots,x_l\right]^G$ itself is a commutative ring.

    \item A polynomial $f$ is $G$-invariant polynomial if and only if each of its homogeneous component $f_d$ is $G$-invariant (\cite[Chapter 7.2, Proposition 10]{cox2015ideals}).
    A \emph{homogeneous polynomial} is a polynomial $f=\sum_\alpha a_\alpha x^\alpha$ where all degrees $\alpha,\alpha'$ with nonzero $a_\alpha,a_{\alpha'}$ satisfy $\abs{\alpha}=\abs{\alpha'}$, i.e. total degree of all of its terms is equal.
    Homogeneous components $f_k$ of $f$ is defined by its unique decomposition $f = \sum_\alpha a_\alpha x^\alpha = \sum_{k=0}^d f_k$ where $f_k = \sum_{\alpha : \abs{\alpha} = k} a_\alpha x^\alpha$ is the sum over all terms of $f$ with degree $k$.

    \item Also, by~\cite[Section 7.2, Lemma 11]{cox2015ideals} If $G$ is generated by $A_1,\dots,A_r\in G$, then $\F\left[x_1,\dots,x_l\right]^G$ contains all polynomials satisfying
    \begin{equation}
        f\left(x\right) = f\left(A_1x\right) = \dots = f\left(A_rx\right) .
    \end{equation}
\end{itemize}

For $a+b>0$, the moment sum $S_{a,b}$ is homogeneous of degree $a+b$, whereas $h_{a,b}=S_{a,b}-\tau_{a,b}$ also has a degree-zero component unless $\tau_{a,b}=0$.
Since constants are invariant under every linear action, $h_{a,b}$ is $G$-invariant if and only if $S_{a,b}$ is $G$-invariant. 
Similarly, each $g_{p,q}^\Re$ or $g_{p,q}^\Im$ has a degree-$p+q$ component and a constant component.
For a group generated by matrices $A_1,\dots,A_r$, these invariance conditions can be checked on those generators.
They provide a sufficient condition for producing further solutions with the same target moments from one solution.
More generally, a group may preserve the common zero set while mixing its defining polynomials.
As we will see later, the degree-$t=2$ action of orthogonal group $\O\left(Z\right)$ in Proposition~\ref{prop:symmetry_group_transf_cordinate_variety} gives an explicit construction from invariant first and second moments.

For the compatible target moments, let
\begin{equation}
    V=\mathbf V\left(J_t^\tau\right)=\mathbf V\left(I_t\right)=F_\tau .
\end{equation}
Suppose a finite group $G\leq\GL\left(2l,\R\right)$ preserves $V$;
invariance of every $h_{a,b}$ is sufficient.
For a seed $\left(u_0,v_0\right)\in V$, its orbit is
\begin{equation}\label{eqn:invariant_solution_orbit}
    O_{\left(u_0,v_0\right)}^G=\left\{g\left(u_0,v_0\right):g\in G\right\}\subseteq V .
\end{equation}
Each orbit point defines a candidate constellation with the
same moments. Coordinate permutations only relabel an
equally weighted constellation and do not give additional
logical states. We therefore select distinct constellations
with well-separated supports, and assess their normalized
states using eqn.\ \eqref{eqn:KL_Gram}.

\section{Resolution bounds}\label{sec:resolution_bounds}

We derive geometric upper bounds on the squared Euclidean resolution of a union of coherent-state constellations.
We first compare volumes in the $2n$-dimensional real amplitude space of $n$ modes, and then express the enclosing radius in terms of the common centered moments.
The single-mode result is stated afterward as a specialization, with its explicit scalar moment expansion and a derivation from the multimode bound.
These bounds apply to the constellation geometry; they do not by themselves imply exact finite-separation KL conditions.
For $n$ bosonic modes, the amplitude space $\C^n$ is identified with $\R^{2n}$.
Consider $N>1$ distinct amplitudes with squared Euclidean resolution $d_E=\min_{\left(k,j\right)\neq\left(k',j'\right)}\nrm{\alpha_j^{\left(k\right)}-\alpha_{j'}^{\left(k'\right)}}_2^2$.
For any fixed center $\mu\in\C^n$, define
\begin{equation}\label{eqn:enclosing_radius_amplitudes}
    R=\max_{j,k}\nrm{\alpha_j^{\left(k\right)}-\mu}_2 .
\end{equation}
All amplitudes then lie in the ball of radius $R$ centered at $\mu$; the moment-based bound below chooses $\mu$ to be the common centroid determined by the first moments.

\begin{theorem}\label{thm:multimode_resolution_bound}
The squared Euclidean resolution $d_E$ satisfies
    \begin{equation}
        d_E\leq\frac{4R^2}{\left(N^{1/\left(2n\right)}-1\right)^2} .
    \end{equation}
\end{theorem}
\begin{proof}
Place a $2n$-dimensional ball of radius $\sqrt{d_E}/2$ at each amplitude. Their interiors are disjoint by the definition of $d_E$, and the triangle inequality places all of them inside a ball of radius $R+\sqrt{d_E}/2$. A $2n$-dimensional ball of radius $r$ has volume $v_{2n}r^{2n}$, where $v_{2n}=\pi^n/n!$. Volume comparison and cancellation of $v_{2n}$ give
    \begin{equation}
    \begin{aligned}
        N\left(\frac{\sqrt{d_E}}{2}\right)^{2n}
        &\leq \left(R+\frac{\sqrt{d_E}}{2}\right)^{2n} \\
        N^{\frac{1}{2n}}\left(\frac{\sqrt{d_E}}{2}\right) &\leq \left(R+\frac{\sqrt{d_E}}{2}\right)\\
        \frac{\sqrt{d_E}}{2}\left(N^{\frac{1}{2n}}-1\right) &\leq R \\
        \sqrt{d_E} &\leq \frac{2R}{N^{\frac{1}{2n}}-1} ,
    \end{aligned}
    \end{equation}
which proves the claim.
\end{proof}

For $K$ equally weighted constellations with $l$ points each, assume that all $N=Kl>1$ support points are distinct, and write their common unnormalized moments as $c_{p,q}=\sum_{j=1}^l\left(\alpha_j^{\left(k\right)*}\right)^p\left(\alpha_j^{\left(k\right)}\right)^q$ for multi-indices $p,q\in\Z_{\geq0}^n$ with $\abs{p}+\abs{q}\leq t$. Here $c_{0,0}=l$, and $\mu$ is the common centroid, with $\mu_m=c_{0,e_m}/l$, where $e_m$ is the $m$th standard basis vector. For integers $1\leq s\leq\left\lfloor t/2\right\rfloor$, define
\begin{equation}
    M_s^{\left(k\right)}=\sum_{j=1}^l{\nrm{\alpha_j^{\left(k\right)}-\mu}_2}^{2s} .
\end{equation}
To show that this is independent of $k$, apply the multinomial identity
\begin{equation}
    \left(z_1+\dots+z_n\right)^s
    =\sum_{\substack{r\in\Z_{\geq0}^n\\\abs{r}=s}}
        \binom{s}{r}\prod_{m=1}^n z_m^{r_m}
    \quad\text{with}\quad
    z_m=\abs{\alpha_{j,m}^{\left(k\right)}-\mu_m}^2 .
\end{equation}
Expanding each of the two factors in $\abs{\alpha_{j,m}^{\left(k\right)}-\mu_m}^{2r_m}$ by the binomial theorem gives
\begin{equation}
\begin{aligned}
    M_s^{\left(k\right)}
    &=\sum_{j=1}^l\sum_{\abs{r}=s}\binom{s}{r}
       \prod_{m=1}^n
       \left(\alpha_{j,m}^{\left(k\right)*}-\mu_m^*\right)^{r_m}
       \left(\alpha_{j,m}^{\left(k\right)}-\mu_m\right)^{r_m}\\
    &=\sum_{\abs{r}=s}\binom{s}{r}
       \sum_{0\leq p,q\leq r}\binom{r}{p}\binom{r}{q}
       \left(-\mu^*\right)^{r-p}\left(-\mu\right)^{r-q}
       \sum_{j=1}^l\left(\alpha_j^{\left(k\right)*}\right)^p\left(\alpha_j^{\left(k\right)}\right)^q\\
    &=\sum_{\abs{r}=s}\binom{s}{r}
       \sum_{0\leq p,q\leq r}\binom{r}{p}\binom{r}{q}
       \left(-\mu^*\right)^{r-p}\left(-\mu\right)^{r-q}c_{p,q}
     =M_s ,
\end{aligned}
\end{equation}
where powers and inequalities are componentwise, $\binom{s}{r}=s!/\prod_mr_m!$, and $\binom{r}{p}=\prod_m\binom{r_m}{p_m}$. The last equality uses $\abs{p}+\abs{q}\leq2\abs{r}=2s\leq t$. Thus $M_s$ is independent of $k$.

Choose $R=\max_{j,k}\nrm{\alpha_j^{\left(k\right)}-\mu}_2$. Since each summand is nonnegative, for every $j,k$ we have
\begin{equation}
    {\nrm{\alpha_j^{\left(k\right)}-\mu}_2}^{2s}
    \leq\sum_{j'=1}^l{\nrm{\alpha_{j'}^{\left(k\right)}-\mu}_2}^{2s}
    =M_s .
\end{equation}
Taking the maximum over $j,k$ and then the $2s$th root gives $R\leq M_s^{1/\left(2s\right)}$. Substituting this into Theorem~\ref{thm:multimode_resolution_bound} with $N=Kl$ gives
\begin{equation}\label{fig:inequality_for_multi_mode_universal_bound}
    d_E\leq\frac{4M_s^{1/s}}{\left(\left(Kl\right)^{1/\left(2n\right)}-1\right)^2} .
\end{equation}

For integer $t\geq2$ and real polynomials $g_i^\Re,g_i^\Im$ containing all constraints of total degree $p_i+q_i\leq t$ as given in Proposition~\ref{prop:real_polynomial_equations_1mode}, we consider $K$ solutions $\mathcal{C} = \left\{\alpha^{\left(k\right)}\right\}_{k=1}^K$ for $\left\{\left(u^{\left(k\right)},v^{\left(k\right)}\right)\right\}_{k=1}^K$ such that $\alpha^{\left(k\right)} = u^{\left(k\right)}+\iota v^{\left(k\right)}$, namely $g_i^\Re\left(u^{\left(k\right)},v^{\left(k\right)}\right)=g_i^\Im\left(u^{\left(k\right)},v^{\left(k\right)}\right)=0$ for all $i$ and $k$. Assume that all $N=Kl>1$ support points are distinct and that $c_{0,0}=l$. The squared Euclidean resolution $d_E$ of $\mathcal{C}$ is defined as
\begin{equation}
    d_E = \min_{\left(j,k\right)\neq\left(j',k'\right)} \abs{\alpha_j^{\left(k\right)} - \alpha_{j'}^{\left(k'\right)}}^2 ,
\end{equation}
which is the minimum squared distance between two distinct points within and across solutions $\left\{\alpha^{\left(k\right)}\right\}_k$.

Consider $\mu = \frac{c_{0,1}}{l} = \Bar{u}+\iota\Bar{v}$ for the scalars $\Bar{u}=c_{0,1}^\Re/l$ and $\Bar{v}=c_{0,1}^\Im/l$, let $1_l\in\R^l$ be the vector of all ones, and define
\begin{equation}
    \beta^{\left(k\right)} = \alpha^{\left(k\right)} - \mu 1_l = \left(u^{\left(k\right)}-\Bar{u}1_l\right) + \iota \left(v^{\left(k\right)}-\Bar{v}1_l\right) ,
\end{equation}
so that $\abs{\alpha_j^{\left(k\right)} - \alpha_{j'}^{\left(k'\right)}} = \abs{\beta_j^{\left(k\right)} - \beta_{j'}^{\left(k'\right)}}$. Then let us consider $M_s^{\left(k\right)} = \sum_{j=1}^l \abs{\beta_j^{\left(k\right)}}^{2s}$ for integers $1\leq s\leq\left\lfloor t/2\right\rfloor$, which satisfy the inequality
\begin{equation}
    \max_j \abs{\beta_j^{\left(k\right)}}^{2s} \leq \sum_{j=1}^l \abs{\beta_j^{\left(k\right)}}^{2s} = M_s^{\left(k\right)} ,
\end{equation}
which in turn implies the following upper-bound on the modulus of the furthest centered point in constellation $k$ as
\begin{equation}\label{eqn:radius_upper_bound}
    \max_j \abs{\beta_j^{\left(k\right)}} \leq \left(M_s^{\left(k\right)}\right)^{\frac{1}{2s}} .
\end{equation}
Now we show that the $\left(M_s^{\left(k\right)}\right)^{\frac{1}{2s}}$ term in the upper bound holds for all centered constellations $\left\{\beta_j^{\left(k\right)}\right\}_j$ and that it can be expressed in terms of the constants $c_{p,q}$ in the polynomials $\left\{f_{p,q} : p+q\leq t\right\}$.

\begin{proposition}\label{prop:radius_upper_bound}
For all $k\in\left\{1,\dots,K\right\}$ and all integers $1\leq s\leq\left\lfloor t/2\right\rfloor$, it holds that
    \begin{equation}
        M_s^{\left(k\right)} = M_s = \sum_{p,q=0}^s \binom{s}{p}\binom{s}{q} \left(-\mu^*\right)^{s-p} \left(-\mu\right)^{s-q} c_{p,q} .
    \end{equation}
\end{proposition}
\begin{proof}
First let us expand $\abs{\beta_j^{\left(k\right)}}^{2s}$ as
    \begin{equation}
    \begin{aligned}
        \abs{\beta_j^{\left(k\right)}}^{2s} &= \left(\left(\alpha_j^{\left(k\right)}\right)^* - \mu^*\right)^s \left(\alpha_j^{\left(k\right)} - \mu\right)^s \\
        &= \sum_{p,q=0}^s \binom{s}{p}\binom{s}{q} \left(-\mu^*\right)^{s-p} \left(-\mu\right)^{s-q} \left(\left(\alpha_j^{\left(k\right)}\right)^*\right)^p \left(\alpha_j^{\left(k\right)}\right)^q .
    \end{aligned}
    \end{equation}
Then, substituting this back into $M_s^{\left(k\right)}$ gives
    \begin{equation}
    \begin{aligned}
        M_s^{\left(k\right)} = \sum_{j=1}^l \abs{\beta_j^{\left(k\right)}}^{2s}
        &= \sum_{j=1}^l \sum_{p,q=0}^s \binom{s}{p}\binom{s}{q} \left(-\mu^*\right)^{s-p} \left(-\mu\right)^{s-q} \left(\left(\alpha_j^{\left(k\right)}\right)^*\right)^p \left(\alpha_j^{\left(k\right)}\right)^q \\
        &= \sum_{p,q=0}^s \binom{s}{p}\binom{s}{q} \left(-\mu^*\right)^{s-p} \left(-\mu\right)^{s-q} \underbrace{ \sum_{j=1}^l \left(\left(\alpha_j^{\left(k\right)}\right)^*\right)^p \left(\alpha_j^{\left(k\right)}\right)^q }_{=f_{p,q}\left(\alpha^{\left(k\right)}\right)+c_{p,q}=c_{p,q}} .
    \end{aligned}
    \end{equation}
Then, we can substitute in $c_{p,q}$ in place of the sum $\sum_j \left(\left(\alpha_j^{\left(k\right)}\right)^*\right)^p \left(\alpha_j^{\left(k\right)}\right)^q$, since $p+q\leq2s\leq t$ and $f_{p,q}\left(\alpha^{\left(k\right)}\right)=0$. The resulting expression is independent of $k$, proving the claim.
\end{proof}

\begin{remark}[Single-mode specialization]\label{thm:single_mode_resolution_bound}
Consider $K$ equally weighted constellations $\mathcal{C}^{\left(k\right)}=\left\{\alpha_j^{\left(k\right)}\right\}_{j=1}^l$ with common moments $c_{p,q}$ for $p+q\leq t$, and assume that all $N=Kl>1$ support points are distinct. Then for every integer $1\leq s\leq\left\lfloor t/2\right\rfloor$ its squared Euclidean resolution $d_E$ is bounded as
    \begin{equation}
        d_E \leq \frac{4}{\left(\sqrt{Kl}-1\right)^2} \left( \sum_{p,q=0}^s \binom{s}{p}\binom{s}{q} \left(-\mu^*\right)^{s-p} \left(-\mu\right)^{s-q} c_{p,q} \right)^{\frac{1}{s}} .
    \end{equation}
    To obtain this bound, set $n=1$ and $N=Kl$ in Theorem~\ref{thm:multimode_resolution_bound}, so that $\left(Kl\right)^{1/\left(2n\right)}=\sqrt{Kl}$ and the Euclidean norm of a complex amplitude is its modulus.
    With the common centroid $\mu=c_{0,1}/l$, eqn.\ \eqref{eqn:radius_upper_bound} and Proposition~\ref{prop:radius_upper_bound} give $R\leq M_s^{1/\left(2s\right)}$ and the scalar expression for $M_s$.
    Substituting these into the geometric bound yields
    \begin{equation}
    \begin{aligned}
        d_E&\leq\frac{4R^2}{\left(\sqrt{Kl}-1\right)^2}
            \leq\frac{4M_s^{1/s}}{\left(\sqrt{Kl}-1\right)^2}\\
           &=\frac{4}{\left(\sqrt{Kl}-1\right)^2}\left(\sum_{p,q=0}^s\binom{s}{p}\binom{s}{q}\left(-\mu^*\right)^{s-p}\left(-\mu\right)^{s-q}c_{p,q}\right)^{1/s} .
    \end{aligned}
    \end{equation}
    Equivalently, this is eqn.\ \eqref{fig:inequality_for_multi_mode_universal_bound} with $n=1$, followed by the substitution from Proposition~\ref{prop:radius_upper_bound}.
\end{remark}

\section{Real constellation polynomials with degree \texorpdfstring{$p+q\leq2$}{p+q at most 2}}\label{sec:degree_two}

We now specialize the moment constraints of Section~\ref{sec:moment_tensors} to $t=2$.
We reduce the real constellation equations to five generators and translate their real solution set to centered coordinates with a fixed Gram matrix.
We then show that every solution is obtained from one seed by the action of $\O\left(Z\right)$ and use this description to formulate a search for constellations with large resolution.

\subsection{Generators of real constellation polynomials of degree \texorpdfstring{$p+q\leq2$}{p+q ≤ 2}}

{For the nonconstant complex polynomials of total degree $1\leq p+q\leq 2$ we have $\left(p,q\right)\in\left\{\left(0,1\right),\allowbreak\left(1,0\right),\allowbreak\left(1,1\right),\allowbreak\left(0,2\right),\allowbreak\left(2,0\right)\right\}$, so we have 10 real constellation polynomial equations to solve.}
The explicit form of these real polynomials can be obtained using Proposition~\ref{prop:real_polynomial_equations_1mode}:
\begin{equation}\label{eqn:poly_generators_pq_leq_2}
\begin{gathered}
    g_{0,1} = -c_{0,1}^\Re + \sum_{j=1}^l u_j 
    ,\quad
    h_{0,1} = -c_{0,1}^\Im + \sum_{j=1}^l v_j
    ,\quad
    g_{1,0} = -c_{1,0}^\Re + \sum_{j=1}^l u_j 
    ,\quad
    h_{1,0} = -c_{1,0}^\Im - \sum_{j=1}^l v_j
    ,\\
    g_{1,1} = -c_{1,1}^\Re + \sum_{j=1}^l \left(u_j^2+v_j^2\right) 
    ,\\
    g_{0,2} = -c_{0,2}^\Re + \sum_{j=1}^l \left(u_j^2-v_j^2\right) 
    ,\quad
    h_{0,2} = -c_{0,2}^\Im + 2\sum_{j=1}^l u_jv_j
    ,\\
    g_{2,0} = -c_{2,0}^\Re + \sum_{j=1}^l \left(u_j^2-v_j^2\right) 
    ,\quad
    h_{2,0} = -c_{2,0}^\Im - 2\sum_{j=1}^l u_jv_j
    .
\end{gathered}
\end{equation}
The omitted imaginary-part equation is $h_{1,1}=-c_{1,1}^\Im=0$.
The compatibility condition here is $c_{p,q}=c_{q,p}^*$ from Lemma~\ref{lem:real_complex_polynomial_coordinate_tranform}; taking $p=q=1$ gives $c_{1,1}=c_{1,1}^*\in\R$, so $c_{1,1}^\Im=0$.

Let us now consider polynomials in $\R\left[u,v\right]$
\begin{equation}
\begin{gathered}
    g_u = \sum_{j=1}^l u_j ,\quad
    g_v = \sum_{j=1}^l v_j ,\\
    g_{u,v}^+ = \sum_{j=1}^l \left(u_j^2+v_j^2\right) ,\quad
    g_{u,v}^- = \sum_{j=1}^l \left(u_j^2-v_j^2\right) ,\quad
    g_{u,v}^\times = \sum_{j=1}^l u_jv_j .
\end{gathered}
\end{equation}
Then from the original generators in eqn.\ \eqref{eqn:poly_generators_pq_leq_2} we see that
\begin{equation}
\begin{gathered}
    g_{0,1} + c_{0,1}^\Re = g_u = g_{1,0} + c_{1,0}^\Re
    ,\quad
    h_{0,1} + c_{0,1}^\Im = g_v = -h_{1,0} - c_{1,0}^\Im ,\\
    g_{0,2} + c_{0,2}^\Re = g_{u,v}^- = g_{2,0} + c_{2,0}^\Re 
    ,\quad
    \frac{h_{0,2} + c_{0,2}^\Im}{2} = g_{u,v}^\times = -\frac{h_{2,0} + c_{2,0}^\Im}{2} .
\end{gathered}
\end{equation}
Now let us write the compatible moment constants as
\begin{equation}
\begin{gathered}
    c_1=c_{0,1}^\Re=c_{1,0}^\Re ,\quad
    c_2=c_{0,1}^\Im=-c_{1,0}^\Im ,\\
    c_3=c_{0,2}^\Re=c_{2,0}^\Re ,\quad
    c_4=\frac{c_{0,2}^\Im}{2}=-\frac{c_{2,0}^\Im}{2} .
\end{gathered}
\end{equation}
The nonzero scalar factors in the generators do not change the ideal, so
\begin{equation}
\begin{aligned}
    I&=\left\langle g_u-c_1,\;g_v-c_2,\;g_{u,v}^--c_3,\;g_{u,v}^\times-c_4,\;g_{1,1}\right\rangle
    =\left\langle g_{0,1},\;h_{0,1},\;g_{0,2},\;h_{0,2},\;g_{1,1}\right\rangle .
\end{aligned}
\end{equation}

Now let use define polynomials
\begin{equation}
\begin{gathered}
    h_u = \frac{g_{1,1} + g_{0,2}}{2} = \sum_{j=1}^l u_j^2 -\frac{c_{1,1}^\Re+c_3}{2}
    \quad\text{and}\quad
    h_v = \frac{g_{1,1} - g_{0,2}}{2} = \sum_{j=1}^l v_j^2 -\frac{c_{1,1}^\Re-c_3}{2} ,
\end{gathered}
\end{equation}
so that we can write the ideal $I$ as
\begin{equation}\label{eqn:qp_leq2_five_ideal_grobner}
\begin{gathered}
    I = \left\langle g_{0,1},\; h_{0,1},\; h_{0,2},\; h_u,\; h_v \right\rangle \\
    \text{for}\\
    g_{0,1} = -c_{0,1}^\Re + \sum_{j=1}^l u_j ,\quad
    h_{0,1} = -c_{0,1}^\Im + \sum_{j=1}^l v_j ,\quad
    h_{0,2} = -c_{0,2}^\Im + 2\sum_{j=1}^l u_jv_j ,\\
    h_u = \sum_{j=1}^l u_j^2 -\frac{c_{1,1}^\Re+c_3}{2}
    ,\quad
    h_v = \sum_{j=1}^l v_j^2 -\frac{c_{1,1}^\Re-c_3}{2}
    .
\end{gathered}
\end{equation}

\subsection{Affine variety of real constellation polynomials of degree \texorpdfstring{$p+q\leq2$}{p+q ≤ 2}}

Let $1_l$ be the vectors of all ones in $\R^l$ and $\left\langle\cdot,\cdot\right\rangle$ the dot product over $\R^l$.
Then, the variety $\mathbf{V}\left(I\right)\subseteq\R^{2l}$ consists of all $\left(u,v\right)\in\R^{2l}$ satisfying
\begin{equation}\label{eqn:variety_pq_leq2_constraints}
\begin{gathered}
    \left\langle u,u\right\rangle = c_3^+ = \frac{c_{1,1}^\Re+c_3}{2} ,\quad
    \left\langle v,v\right\rangle = c_3^- = \frac{c_{1,1}^\Re-c_3}{2} ,\quad
    \\ \left\langle u,1\right\rangle =  c_1 ,\quad
    \left\langle v,1\right\rangle =  c_2 ,\quad\text{and}\quad
    \left\langle u,v\right\rangle =  c_4 .
\end{gathered}
\end{equation}
For $\alpha_i = u_i+\iota v_i$, we can write the conditions above in terms of $\left(\alpha_1,\dots,\alpha_l\right)\in\C^l$ as
\begin{equation}\label{eqn:complex_constraints_pq_leq2}
\begin{gathered}
    \sum_{i=1}^l \abs{\alpha_i}^2 = \sum_{i=1}^l u_i^2+v_i^2 = \left\langle u,u\right\rangle+\left\langle v,v\right\rangle = c_{1,1}^\Re
    ,\\
    \sum_{i=1}^l \alpha_i = \sum_{i=1}^l u_i+\iota v_i = \left\langle u,1\right\rangle+\iota\left\langle v,1\right\rangle = c_1+\iota c_2
    ,\\
    \sum_{i=1}^l \alpha_i^* = \sum_{i=1}^l u_i-\iota v_i = \left\langle u,1\right\rangle-\iota\left\langle v,1\right\rangle = c_1-\iota c_2
    ,\\
    \sum_{i=1}^l \alpha_i^2 = \sum_{i=1}^l u_i^2 + 2\iota u_iv_i - v_i^2 = \left\langle u,u\right\rangle + 2\iota\left\langle u,v\right\rangle - \left\langle v,v\right\rangle = c_3 + 2\iota c_4
    ,\\
    \sum_{i=1}^l \left(\alpha_i^*\right)^2 = \sum_{i=1}^l u_i^2 - 2\iota u_iv_i - v_i^2 = \left\langle u,u\right\rangle - 2\iota\left\langle u,v\right\rangle - \left\langle v,v\right\rangle = c_3 - 2\iota c_4 .
\end{gathered}
\end{equation}
For $\ket{\mathcal C}=\left(\mathcal N\sqrt l\right)^{-1}\sum_{i=1}^l\ket{\alpha_i}$, where $\mathcal N^2=l^{-1}\sum_{i,j}\braket{\alpha_i}{\alpha_j}$, the mean photon number is
\begin{equation}
    \brakets{\mathcal C}{a^\dag a}{\mathcal C}
    =\frac{1}{\mathcal N^2l}\left(c_{1,1}^\Re+
      \sum_{i=1}^l\sum_{j\neq i}\alpha_i^*\alpha_j
      e^{-\frac12\left(\abs{\alpha_i}^2+\abs{\alpha_j}^2-2\alpha_i^*\alpha_j\right)}\right) .
\end{equation}
The first and second moments do not, in general, determine this exponentially weighted double sum or $\mathcal N$.
They determine the unweighted sum
\begin{equation}
    \sum_{i=1}^l\sum_{j\neq i}\alpha_i^*\alpha_j
    =c_1^2+c_2^2-c_{1,1}^\Re .
\end{equation}
For fixed $l$, let $R=\max_i\abs{\alpha_i}$ and $\epsilon=\max_{i\neq j}\abs{\braket{\alpha_i}{\alpha_j}}$.
In the limit $\epsilon\to0$, we have $\mathcal N^2=1+O\left(\epsilon\right)$ and
\begin{equation}
    \brakets{\mathcal C}{a^\dag a}{\mathcal C}
    =\frac{c_{1,1}^\Re}{l}+O\left(R^2\epsilon\right) .
\end{equation}
Thus the common second moment determines the leading mean photon number, while the coherent-state overlaps give constellation-dependent corrections.

\subsection{Centered solution space}

{Now let us consider variables $x_i=u_i-\Bar{u}$ and $y_i=v_i-\Bar{v}$, where $\Bar{u}=c_1/l$ and $\Bar{v}=c_2/l$.}
In vector form, $x=u-\Bar{u}1_l$ and $y=v-\Bar{v}1_l$, where $1_l\in\R^l$ is the all $1$ vector.
So for $x=\left(x_1,\dots,x_l\right)$ and $y=\left(y_1,\dots,y_l\right)$ we have 
\begin{equation}
\begin{gathered}
    \left\langle x,1\right\rangle = \sum_i \left(u_i - \Bar{u}\right) = 0 
    \quad\text{and}\quad
    \left\langle y,1\right\rangle = \sum_i \left(v_i-\Bar{v}\right) = 0
\end{gathered}
\end{equation}
by the third and fourth constraints in eqn.\ \eqref{eqn:variety_pq_leq2_constraints}.
Also, we can write the first two and the last constraints in eqn.\ \eqref{eqn:variety_pq_leq2_constraints} as
\begin{equation}
\begin{gathered}
    \left\langle u,u\right\rangle = \sum_i x_i^2 + 2 x_i\Bar{u} + \Bar{u}^2 
        = \left\langle x,x\right\rangle + 2\Bar{u}\underbrace{\left\langle x,1\right\rangle}_{=0} + \frac{c_1^2}{l} 
        = \left\langle x,x\right\rangle + \frac{c_1^2}{l} ,\\
    \left\langle v,v\right\rangle = \sum_i y_i^2 + 2 y_i\Bar{v} + \Bar{v}^2 
        = \left\langle y,y\right\rangle + 2\Bar{v}\underbrace{\left\langle y,1\right\rangle}_{=0} + \frac{c_2^2}{l} 
        = \left\langle y,y\right\rangle + \frac{c_2^2}{l} ,\\
    \left\langle u,v\right\rangle = \sum_i \left(x_i+\Bar{u}\right)\left(y_i+\Bar{v}\right) 
        = \left\langle x,y\right\rangle +\left\langle x,1\right\rangle\Bar{v} + \left\langle y,1\right\rangle\Bar{u} + \frac{c_1c_2}{l} 
        = \left\langle x,y\right\rangle + \frac{c_1c_2}{l} .
\end{gathered}
\end{equation}
Thus now we can write the constraints in eqn.\ \eqref{eqn:variety_pq_leq2_constraints} in terms of these new $\left(x,y\right)\in\R^{2l}$ coordinates as a subset $V'$ of $\R^{2l}$, given by
\begin{equation}\label{eqn:pq_leq2_transformed_coordinates_constraints}
\begin{gathered}
    V' = \left\{\left(x,y\right)\in Z\times Z : \left\langle x,x\right\rangle=\tau_x ,\; \left\langle y,y\right\rangle=\tau_y ,\; \left\langle x,y\right\rangle=\tau_{xy} \right\} \\
    \text{for}\\
    Z = \left\{w\in\R^l : \left\langle w,1\right\rangle = 0\right\} ,\quad
    \tau_x = \frac{c_{1,1}^\Re+c_3}{2} - \frac{c_1^2}{l} 
    ,\quad
    \tau_y = \frac{c_{1,1}^\Re-c_3}{2} - \frac{c_2^2}{l} 
    ,\quad
    \tau_{xy} = c_4 - \frac{c_1c_2}{l} ,
\end{gathered}
\end{equation}
where $Z$ an $l-1$-dimensional subspace of $\R^l$ that is orthogonal the 1-dimensional subspace $1\R$ spanned by the vector $1$.
So, for $\left(x,y\right)\in Z\times Z \subseteq\R^{2l}$, we can further write the equality constraints in $V'$ as
\begin{equation}
\label{eq:taus}
    \sum_{i=1}^l
    \begin{bmatrix}
        x_i \\
        y_i
    \end{bmatrix}
    \begin{bmatrix}
        x_i & y_i
    \end{bmatrix} 
    = T
    \quad\text{for}\quad
    T =
    \begin{bmatrix}
        \tau_x & \tau_{xy} \\
        \tau_{xy} & \tau_y
    \end{bmatrix}
\end{equation}
and therefore, the entire constraints over $\left(x,y\right)\in\R^{2l}$ can be expressed as
\begin{equation}\label{eqn:pq_leq2_transformed_coordinates_matrix_constraints}
\begin{gathered}
    M_{x,y}^\top 1 = 0 
    \quad\text{and}\quad
    M_{x,y}^\top M_{x,y} = T
    \quad\text{for}\quad
    M_{x,y} = 
    \begin{bmatrix}
        x_1 & y_1 \\
        \vdots & \vdots\\
        x_l & y_l
    \end{bmatrix} .
\end{gathered}
\end{equation}

Now let us consider the orthogonal group $\O\left(l,\R\right)$ over $\R^l$ and its subset:
\begin{equation}
    \O\left(Z\right) = \left\{Q\in\O\left(l,\R\right) : Q1=1 \right\} .
\end{equation}
We show that matrices in $\O\left(Z\right)$ is an orthogonal group over subspace $Z$ and that $Z$ is closed under $\O\left(Z\right)$.

\begin{lemma}
    Subset $\O\left(Z\right)$ of orthogonal group $\O\left(l,\R\right)$ over $\R^l$ satisfies:
    \begin{enumerate}
        \item 
        {$\O\left(Z\right)$ is a subgroup of $\O\left(l,\R\right)$ where $Q^\top Q = QQ^\top = I$ for all $Q\in\O\left(Z\right)$.}
        \item $Z$ is closed under the action of $\O\left(Z\right)$, i.e. for any $Q\in\O\left(Z\right)$, $Qx\in Z$ whenever $x\in Z$.
    \end{enumerate}
\end{lemma}
\begin{proof}
    First we show statement 1.
    {Since $\O\left(Z\right)\subseteq\O\left(l,\R\right)$ therefore any $Q\in\O\left(Z\right)$ satisfies $Q^\top Q = QQ^\top = I$.}
    The identity matrix $I\in\O\left(l,\R\right)$ is in $\O\left(Z\right)$ since $I1=1$ and it is the identity element of $\O\left(Z\right)$ since $IQ=Q$ for any $Q\in\O\left(Z\right)$.
    Also, $QQ'\in\O\left(Z\right)$ for any $Q,Q'\in\O\left(Z\right)$ because $QQ'1=Q1=1$.
    {For any $Q\in\O\left(Z\right)$, the identity $Q1=1$ implies $Q^\top1=Q^\top Q1=1$, so $Q^{-1}=Q^\top\in\O\left(Z\right)$.}
    Thus $\O\left(Z\right)$ is a subgroup of $\O\left(l,\R\right)$.

    To show statement 2, note that for any $x\in Z$ and any $Q\in\O\left(Z\right)$, it holds that $Qx\in Z$ since
    \begin{equation}
        \left\langle Qx,1\right\rangle = \left\langle Qx,Q1\right\rangle = \left\langle x,Q^\top Q 1\right\rangle = \left\langle x,1\right\rangle = 0 .
    \end{equation}
    Thus $Z$ is closed under the action of $Q\in\O\left(Z\right)$.
\end{proof}

{In an orthonormal basis $e_1,\dots,e_l$ with $e_1 = 1_l/\sqrt{l}$ and $e_2,\dots,e_l$ a basis of subspace $Z$, a matrix $Q$ in $\O\left(Z\right)$ takes the form}
\begin{equation}\label{eqn:block_Z_orthogonal_group}
    Q 
    = \begin{bmatrix}
        1 & 0 \\
        0 & Q_Z
    \end{bmatrix} ,
\end{equation}
which means that $Q_Z\in\O\left(l-1,\R\right)$ since $Q^\top Q= QQ^\top = I_l$ implies that $Q_Z^\top Q_Z= Q_ZQ_Z^\top = I_{l-1}$.
Therefore $\O\left(Z\right) \cong \O\left(l-1,\R\right)$.
It turns out the diagonal subgroup of $\O\left(Z\right)\oplus\O\left(Z\right)$ is a symmetry group of $Z\times Z$.
Namely, $Q\oplus Q \left(Z\times Z\right) = Z\times Z$ for all $Q\in\O\left(Z\right)$.

\begin{proposition}\label{prop:symmetry_group_transf_cordinate_variety}
    Consider the subset $G$ of $\O\left(Z\right)\oplus\O\left(Z\right)$, given by
    \begin{equation}
        G = \left\{Q\oplus Q : Q\in\O\left(Z\right)\right\} .
    \end{equation}
    The subset $G$ is a group under the usual matrix multiplication and for any $\left(\hat{x},\hat{y}\right)\in V'$, it holds that
    \begin{equation}
        V' = G\cdot\left(\hat{x},\hat{y}\right) = \left\{\left(x,y\right)\in\R^{2l} : \exists Q\in\O\left(Z\right) ,\; x=Q\hat{x} \;\text{and}\; y=Q\hat{y} \right\} .
    \end{equation}
\end{proposition}
\begin{proof}
    Let us denote elements of $G$ as $g_Q=Q\oplus Q$.
    Subset $G$ is a group under matrix multiplication with identity $g_I=I\oplus I$ and inverse $g_Q^{-1}=g_{Q^\top}$, since $g_Q=Q\oplus Q$ and $g_Qg_{Q'}\in G$ for all $Q\in\O\left(Z\right)$.
    
    Now we show that for any $\left(\hat{x},\hat{y}\right)\in V'$ it holds that $\left(Q\hat{x},Q\hat{y}\right)\in V'$.
    In the matrix form of eqn.\ \eqref{eqn:pq_leq2_transformed_coordinates_matrix_constraints}, we have $M_{Q\hat{x},Q\hat{y}} = QM_{\hat{x},\hat{y}}$.
    Since $Q^\top Q=I$ for all $Q\in\O\left(Z\right)$ and $M_{\hat{x},\hat{y}}^\top M_{\hat{x},\hat{y}}=T$ we have
    \begin{equation}
        M_{Q\hat{x},Q\hat{y}}^\top M_{Q\hat{x},Q\hat{y}} = M_{\hat{x},\hat{y}}^\top Q^\top Q M_{\hat{x},\hat{y}} = T .
    \end{equation}
    Hence, $\left(Q\hat{x},Q\hat{y}\right)$ is in $V'$ for any $Q\in\O\left(Z\right)$ so that $G\cdot\left(\hat{x},\hat{y}\right)\subseteq V'$.
    
    Now we show that for any $\left(\hat{x},\hat{y}\right),\left(x,y\right)\in V'$ there exists $Q\in\O\left(Z\right)$ such that $x=Q\hat{x}$ and $y=Q\hat{y}$ to obtain $V' \subseteq G\cdot\left(\hat{x},\hat{y}\right)$.
    From matrix constraints in eqn.\ \eqref{eqn:pq_leq2_transformed_coordinates_constraints} for elements of $V'$, we have
    \begin{equation}\label{eqn:pq_leq2_constraints2}
        M_{\hat{x},\hat{y}}^\top M_{\hat{x},\hat{y}} = T = M_{x,y}^\top M_{x,y} .
    \end{equation}
    For $a,b\in\R$, consider mappings $\phi_0\left(a,b\right) = a\hat{x}+b\hat{y}$ and $\phi\left(a,b\right)=ax+by$.
    
    Consider a map $\Phi:S_0\rightarrow S$ for $S_0=\spn\left\{\hat{x},\hat{y}\right\}$ and $S=\spn\left\{x,y\right\}$ given by
    \begin{equation}
        \Phi\left(\phi_0\left(a,b\right)\right) = \phi\left(a,b\right) 
    \end{equation}
    for any $a,b\in\R$ and $\phi_0\left(a,b\right) = a\hat{x}+b\hat{y}$ and $\phi\left(a,b\right)=ax+by$.
    Now we show that $\Phi$ is well-defined, namely we need to show that for all $a,b\in\R$, $\phi\left(a,b\right)=0$ whenever $\phi_0\left(a,b\right)=0$, i.e. $\ker\phi_0 \subseteq \ker\phi$.
    Now suppose that $a\hat{x}+b\hat{y}=\phi_0\left(a,b\right)=0$.
    Since $\left(\hat{x},\hat{y}\right)\in V'$, by eqn.\ \eqref{eqn:pq_leq2_constraints2} we have
    \begin{equation}
        T \begin{bmatrix}
            a\\ b
        \end{bmatrix} =
        M_{\hat{x},\hat{y}}^\top M_{\hat{x},\hat{y}} \begin{bmatrix}
            a\\ b
        \end{bmatrix} = 0 .
    \end{equation}
    The norm squared of $\phi\left(a,b\right)=ax+by$ is given by
    \begin{equation}
    \begin{aligned}
        \nrm{ax+by}^2 
        &= \begin{bmatrix}
            a & b
        \end{bmatrix} M_{x,y}^\top M_{x,y} \begin{bmatrix}
            a \\ b
        \end{bmatrix}
        = \begin{bmatrix}
            a & b
        \end{bmatrix} T \begin{bmatrix}
            a \\ b
        \end{bmatrix} = 0 ,
    \end{aligned}
    \end{equation}
    where we again use eqn.\ \eqref{eqn:pq_leq2_constraints2} since $\left(x,y\right)\in V'$.
    Since $\nrm{\phi\left(a,b\right)}^2=0$ therefore $\phi\left(a,b\right)=0$, and thus $\phi_0\left(a,b\right)=0$ implies $\phi\left(a,b\right)=0$ for all $a,b\in\R$.
    So $\Phi$ is well-defined.

    {Also, $\Phi$ is a linear isometry from $S_0$ onto $S$ since for $a,b,c,d\in\R$ we have}
    \begin{equation}
    \begin{gathered}
        \Phi\left(\phi_0\left(a,b\right)+\phi_0\left(c,d\right)\right) = \phi\left(a+c,b+d\right) = \phi\left(a,b\right)+\phi\left(c,d\right) = \Phi\left(\phi_0\left(a,b\right)\right) + \Phi\left(\phi_0\left(c,d\right)\right) \\
        \text{and}\\
        \left\langle\Phi\left(\phi_0\left(a,b\right)\right),\Phi\left(\phi_0\left(c,d\right)\right)\right\rangle = \left\langle\phi\left(a,b\right),\phi\left(c,d\right)\right\rangle =
        \begin{bmatrix}
            a&b
        \end{bmatrix} T \begin{bmatrix}
            c \\ d
        \end{bmatrix}
        = \left\langle\phi_0\left(a,b\right),\phi_0\left(c,d\right)\right\rangle .
    \end{gathered}
    \end{equation}
    {Since $\Phi:S_0\rightarrow S$ is an isometry, it maps an orthonormal basis $v_1,\dots,v_k$ of $S_0$ to an orthonormal basis $u_i=\Phi\left(v_i\right)$ of $S$. 
    Extend both bases to orthonormal bases of $Z$ and define $Q_Zv_i=u_i$ for $i=1,\dots,l-1$.
    Then $Q_Z|_{S_0}=\Phi$, and extending $Q_Z$ by the identity on $\spn\left\{1\right\}$ gives the required $Q\in\O\left(Z\right)$.}
\end{proof}

\subsection{Solutions in the original real-coordinates and the complex coordinates}

Recall that the transformed coordinates are defined as $x_i = u_i-\Bar{u}$ and $y_i = v_i-\Bar{v}$ for $\Bar{u}=c_1/l$ and $\Bar{v}=c_2/l$.
So, a $\left(\hat{x},\hat{y}\right)\in V'$ correspond to $\left(\hat{u},\hat{v}\right) \in \mathbf{V}\left(I\right)$ where $\hat{u}_i = \hat{x}_i+\Bar{u}$ and $\hat{v}_i = \hat{y}_i+\Bar{v}$.
By Proposition~\ref{prop:symmetry_group_transf_cordinate_variety}, for any $\left(u,v\right)\in\mathbf{V}\left(I\right)$ there exists $Q\in\O\left(Z\right)$ such that 
\begin{equation}
    u = Q\left(\hat{u}-\Bar{u}1\right) + \Bar{u}1 = Q\hat{u}
    \quad\text{and}\quad
    v = Q\left(\hat{v}-\Bar{v}1\right) + \Bar{v}1 = Q\hat{v}
\end{equation}
since $\Bar{u}=c_1/l$ and $\Bar{v}=c_2/l$ and $Q1=1$.
Thus affine variety $\mathbf{V}\left(I\right)$ can be described by a single solution $\left(\hat{u},\hat{v}\right) \in \mathbf{V}\left(I\right)$ as
\begin{equation}
    \mathbf{V}\left(I\right) = \left\{ \left(u,v\right)\in\R^{2l} : \exists Q\in\O\left(Z\right),\; u = Q\hat{u} \;\text{and}\; v = Q\hat{v} \right\} .
\end{equation}

{For a constellation $\hat{\alpha}=\left(\hat{\alpha}_1,\dots,\hat{\alpha}_l\right)$ where $\hat{\alpha}_i=\hat{u}_i+\iota \hat{v}_i$, the vector $\alpha = Q\hat{\alpha} = Q\hat{u}+\iota Q\hat{v}$ satisfies the same moment constraints and gives another $l$-point constellation whenever its entries are distinct.}

The action of $\O\left(Z\right)$ gives a search over constellations with fixed first and second moments. 
For a seed constellation $\hat\alpha$, choose $Q_0,\dots,Q_{K-1}\in\O\left(Z\right)$ and set $\alpha^{\left(k\right)}=Q_k\hat\alpha$.
A natural geometric objective is
\begin{equation}
\begin{gathered}
    \max_{Q_0,\dots,Q_{K-1},\,\delta\geq0}\ \delta\\
    \text{subject to}\\
    Q_k\in\O\left(Z\right) \quad\forall k ,\\
    \abs{\left(Q_k\hat\alpha\right)_j-\left(Q_{k'}\hat\alpha\right)_{j'}}^2
      \geq\delta \quad\forall \left(k,j\right)\neq\left(k',j'\right) .
\end{gathered}
\end{equation}
If the optimal value of $\delta$ is positive, it equals the resolution $d_E$ of the union of the resulting logical constellations.
The positivity qualification excludes repeated indexed amplitudes.

For a permutation matrix $P$, there is a permutation $\pi$ of $\left[l\right]$ such that $\left(P\hat\alpha\right)_j=\hat\alpha_{\pi\left(j\right)}$, and hence
\begin{equation}
    \sum_{j=1}^l\ket{\left(P\hat\alpha\right)_j}=\sum_{j=1}^l\ket{\hat\alpha_{\pi\left(j\right)}}=\sum_{j=1}^l\ket{\hat\alpha_j} .
\end{equation}
Thus a coordinate permutation gives exactly the same uniform superposition and normalization, so the resulting normalized states are identical and cannot represent distinct logical states.

A feasible $\delta>0$ gives a distance of at least $\sqrt\delta$ between every pair of amplitudes with different indices $\left(k,j\right)\neq\left(k',j'\right)$, so each constellation has distinct points and different logical constellations have disjoint supports; finite-amplitude performance still requires the normalized KL matrix elements or the noise benchmarking of Section~\ref{sec:noise-benchmarking}.
{The constructions in Sections~\ref{sec:ellipsoid_constellations} and~\ref{sec:nonellipsoid_constellations} provide explicit candidates for this search.}
{These constructions are not claimed to maximize $d_E$ over the full $\O\left(Z\right)$ search; optimizing a rotation within one chosen family does not establish optimality over all constellations with the same moments.}

\subsection{\texorpdfstring{$\O\left(Z\right)$}{O(Z)}-invariant polynomials}

Now let us consider polynomials $h_1,\dots,h_5$ in $\R\left[x,y\right]:=\R\left[x_1,y_1,\dots,x_l,y_l\right]$, given by
\begin{equation}\label{eqn:pq_leq2_polynomials_transformed_coord}
\begin{gathered}
    h_1 = \sum_{i=1}^l x_i ,\quad
    h_2 = \sum_{i=1}^l y_i ,\\
    h_3 = - \tau_x + \sum_{i=1}^l x_i^2 ,\quad
    h_4 = - \tau_y + \sum_{i=1}^l y_i^2 ,\quad
    h_5 = - \tau_{xy} + \sum_{i=1}^l x_iy_i ,
\end{gathered}
\end{equation}
for constants $\tau_x,\tau_y,\tau_{xy}$ given in eqn.\ \eqref{eqn:pq_leq2_transformed_coordinates_constraints}.
Now we show that these polynomials are the generators of $I$ given in eqn.\ \eqref{eqn:qp_leq2_five_ideal_grobner} ``translated'' to the new coordinates, so that their variety is equal to the translated variety $V'$ defined in eqn.\ \eqref{eqn:pq_leq2_transformed_coordinates_constraints}.

\begin{lemma}
    For polynomials $h_1,\dots,h_5$ given in eqn.\ \eqref{eqn:pq_leq2_polynomials_transformed_coord}, it holds that $\mathbf{V}\left(h_1,\dots,h_5\right) = V'$.
\end{lemma}
\begin{proof}
    Since any $\left(x,y\right)\in V'$ solves the polynomials $h_1,\dots,h_5$, we have $V'\subseteq\mathbf{V}\left(h_1,\dots,h_5\right)$.
    On the other hand, any $\left(x,y\right)\in\R^{2l}$ in the variety $\mathbf{V}\left(h_1,\dots,h_5\right)$ must satisfy $h_1\left(x,y\right)=\dots=h_5\left(x,y\right)=0$ which is equivalent to
    \begin{equation}
    \begin{gathered}
        \left\langle x,1\right\rangle=\left\langle y,1\right\rangle=0 ,\\
        \left\langle x,x\right\rangle=\tau_x ,\quad
        \left\langle y,y\right\rangle=\tau_y ,\quad
        \left\langle x,y\right\rangle=\tau_{xy} .
    \end{gathered}
    \end{equation}
    Therefore any $\left(x,y\right)\in\mathbf{V}\left(h_1,\dots,h_5\right)$ must be in $V'$ thus $\mathbf{V}\left(h_1,\dots,h_5\right) \subseteq V'$.
    \newline Hence, $\mathbf{V}\left(h_1,\dots,h_5\right) = V'$.
\end{proof}

\section{Ellipsoid BARC codes from cyclic subgroups of \texorpdfstring{$O\left(Z\right)$}{O(Z)}}\label{sec:ellipsoid_constellations}

We construct an ellipsoidal family of single-mode BARC codes by selecting constellations from the degree-two polynomial solution set characterized in Section~\ref{sec:degree_two}.

Consider a two-dimensional subspace of $Z\subseteq\R^l$ spanned by $c,s\in \R^l$, given by 
\begin{equation}\label{eqn:equal_angle_ellipsoid_solutions}
    c = \sqrt{\frac{2}{l}} \begin{bmatrix}
        \cos\left(0\right) \\
        \cos\left(\frac{2\pi}{l}\right) \\
        \vdots\\
        \cos\left(\frac{2\pi\left(l-1\right)}{l}\right)
    \end{bmatrix}
    \quad\text{and}\quad
    s = \sqrt{\frac{2}{l}} \begin{bmatrix}
        \sin\left(0\right) \\
        \sin\left(\frac{2\pi}{l}\right) \\
        \vdots\\
        \sin\left(\frac{2\pi\left(l-1\right)}{l}\right)
    \end{bmatrix} .
\end{equation}
The following lemma shows that $c,s$ are in $Z$ and orthonormal.

\begin{lemma}\label{lem:ellipsoid_pq_leq2_solution_subspace_basis}
    For $l\geq3$, vectors $c,s\in\R^l$ given in eqn.\ \eqref{eqn:equal_angle_ellipsoid_solutions} satisfy
    \begin{equation}
        \left\langle c,1\right\rangle = \left\langle s,1\right\rangle = 0 
        ,\quad
        \left\langle c,c\right\rangle = \left\langle s,s\right\rangle = 1 
        ,\quad\text{and}\quad
        \left\langle c,s\right\rangle = 0 .
    \end{equation}
    Namely, $c,s\in Z$ and they are orthonormal.
\end{lemma}

\begin{proof}
    First we show that $c,s\in Z$, i.e. $\left\langle c,1\right\rangle = \left\langle s,1\right\rangle = 0$.
    For the $l$-th root of unity $\omega = e^{\frac{2\pi\iota}{l}} = \cos\frac{2\pi}{l} + \iota\sin\frac{2\pi}{l}$, it holds that
    \begin{equation}\label{eqn:ellipsoid_first_root_identity}
        \sum_{k=1}^{l}\omega^{k-1} = \sum_{k=0}^{l-1}\omega^k = \frac{1-\omega^l}{1-\omega} = 0
    \end{equation}
    since $\omega\neq1$ for $l\geq2$.
    Let us denote the $k$-th entry of $c$ as $c_k$, and similarly $s_k$ for the $k$-th entry of $s$.
    Because $c_k + \iota s_k = \sqrt{2/l}\,\omega^{k-1}$, we have
    \begin{equation}
    \begin{aligned}
        \left\langle c,1\right\rangle + \iota\left\langle s,1\right\rangle = \sum_{k=1}^{l} \left(c_k+\iota s_k\right) = \sqrt{\frac{2}{l}}\sum_{k=1}^{l}\omega^{k-1} = 0
    \end{aligned}
    \end{equation}
    so it must hold that $\left\langle c,1\right\rangle = \left\langle s,1\right\rangle = 0$ and therefore $c,s\in Z$.
    
    Now we show that $c,s$ are orthonormal.
    By using trigonometric identities $\left(\cos\theta\right)^2 = \left(1+\cos\left(2\theta\right)\right)/2$ and $\left(\sin\theta\right)^2 = \left(1-\cos\left(2\theta\right)\right)/2$ and $\left(\cos\theta\right)\left(\sin\theta\right) = \left(\sin\left(2\theta\right)\right)/2$, we have
    \begin{equation}
    \begin{gathered}
        \left\langle c,c\right\rangle = \frac{2}{l} \sum_{k=1}^l \cos^2\left(\frac{2\pi\left(k-1\right)}{l}\right)
            = \frac{2}{l} \sum_{k=1}^l \frac{1+\cos\left(4\pi\left(k-1\right)/l\right)}{2}
            = 1+ \frac{1}{l}\sum_{k=1}^l \cos\frac{4\pi\left(k-1\right)}{l} ,\\
        \left\langle s,s\right\rangle = \frac{2}{l} \sum_{k=1}^l \sin^2\left(\frac{2\pi\left(k-1\right)}{l}\right)
            = \frac{2}{l} \sum_{k=1}^l \frac{1-\cos\left(4\pi\left(k-1\right)/l\right)}{2}
            = 1- \frac{1}{l}\sum_{k=1}^l \cos\frac{4\pi\left(k-1\right)}{l} ,\\
        \left\langle c,s\right\rangle = \frac{2}{l} \sum_{k=1}^l \cos\left(\frac{2\pi\left(k-1\right)}{l}\right)\sin\left(\frac{2\pi\left(k-1\right)}{l}\right)
            = \frac{2}{l} \sum_{k=1}^l \frac{\sin\left(4\pi\left(k-1\right)/l\right)}{2}
            = \frac{1}{l}\sum_{k=1}^l \sin\frac{4\pi\left(k-1\right)}{l} .
    \end{gathered}
    \end{equation}
    Then, consider again the $l$-th root of unity $\omega = e^{\frac{2\pi\iota}{l}} = \cos\frac{2\pi}{l} + \iota\sin\frac{2\pi}{l}$.
    Thus we have
    \begin{equation}\label{eqn:ellipsoid_second_root_identity}
        \sum_{k=1}^{l}\omega^{2\left(k-1\right)} = \sum_{k=0}^{l-1}\omega^{2k} = \frac{1-\omega^{2l}}{1-\omega^2} = 0
    \end{equation}
    since $\omega^2\neq1$ for $l\geq3$, which implies that $\sum_{k=1}^l \sin\frac{4\pi\left(k-1\right)}{l} = \sum_{k=1}^l \cos\frac{4\pi\left(k-1\right)}{l} =0$.
    Eqns.\ \eqref{eqn:ellipsoid_first_root_identity} and~\eqref{eqn:ellipsoid_second_root_identity} are the first and second root-of-unity identities used here.
    Hence
    \begin{equation}
        \left\langle c,c\right\rangle = \left\langle s,s\right\rangle = 1 
        \quad\text{and}\quad
        \left\langle c,s\right\rangle = 0 ,
    \end{equation}
    as claimed.
\end{proof}

Now assume that $T$ is positive definite, so that
$\tau_x>0$ and $\tau_x\tau_y-\tau_{xy}^2>0$.
Define
\begin{equation}\label{eqn:cholesky_G_matrix}
    \frac{l}{2}LL^\top=T \quad\text{for}\quad
    L=\sqrt{\frac{2}{l}}
    \begin{bmatrix}
        \sqrt{\tau_x}&0\\
        \tau_{xy}/\sqrt{\tau_x}&\sqrt{\tau_y-\tau_{xy}^2/\tau_x}
    \end{bmatrix} .
\end{equation}
Since $c,s$ are orthonormal, the vectors
\begin{equation}
    x^{\left(0\right)}=\sqrt{\tau_x}\,c \quad\text{and}\quad
    y^{\left(0\right)}=\frac{\tau_{xy}}{\sqrt{\tau_x}}c
             +\sqrt{\tau_y-\frac{\tau_{xy}^2}{\tau_x}}\,s
\end{equation}
satisfy $M_{x^{\left(0\right)},y^{\left(0\right)}}^\top M_{x^{\left(0\right)},y^{\left(0\right)}}=T$ and
therefore $\left(x^{\left(0\right)},y^{\left(0\right)}\right)\in V'$.
For $\mu=\left(c_1+\iota c_2\right)/l$, the corresponding seed constellation has
$\alpha_j^{\left(0\right)}=\mu+x_j^{\left(0\right)}+\iota y_j^{\left(0\right)}$.

Now let us consider a linear rotation map over $\spn\left\{c,s\right\}$ as
\begin{equation}
    R_\theta = \left(\cos\theta\right) cc^\top + \left(\sin\theta\right) cs^\top
        - \left(\sin\theta\right) sc^\top + \left(\cos\theta\right) ss^\top + \sum_{i=1}^{l-2} v_iv_i^\top
\end{equation}
where $v_1,\dots,v_{l-2}$ is an orthonormal basis on the orthogonal subspace of $\spn\left\{c,s\right\}$, so that $R_\theta c = \left(\cos\theta\right) c - \left(\sin\theta\right)s$ and $R_\theta s = \left(\sin\theta\right) c + \left(\cos\theta\right) s$, whereas $R_\theta z = z$ for all $z\in\left(\spn\left\{c,s\right\}\right)^\perp$.
Moreover, since $R_\theta R_\theta^\top = R_\theta^\top R_\theta = I$ it holds that $\left\{R_\theta\right\}_\theta$ is a subgroup of $O\left(Z\right)$.
Since $c,s\in Z$, the map $R_\theta$ also fixes $1$, so the polynomials in eqn.\ \eqref{eqn:pq_leq2_polynomials_transformed_coord} satisfy
\begin{equation}
    h_a\left(R_\theta x,R_\theta y\right)=h_a\left(x,y\right),\qquad a\in\left\{1,\dots,5\right\} .
\end{equation}
Indeed, fixing $1$ preserves the two first-moment sums, while orthogonality preserves $\left\langle x,x\right\rangle$, $\left\langle y,y\right\rangle$, and $\left\langle x,y\right\rangle$.
For the finite cyclic subgroup chosen below, these identities are precisely the $G$-invariance condition of Definition~\ref{def:polynomial_invariance} and eqn.\ \eqref{eqn:polynomial_invariance}; applying the orbit inclusion in eqn.\ \eqref{eqn:invariant_solution_orbit} to the defining polynomials of eqn.\ \eqref{eqn:pq_leq2_polynomials_transformed_coord} therefore generates constellations with the same prescribed first and second moments from the seed solution.
The cyclic subgroup chosen below selects finitely many such constellations; their common moments give the degree-two KL conditions in the large-separation approximation of Section~\ref{sec:coherent_state_codes}.
Therefore by Proposition~\ref{prop:symmetry_group_transf_cordinate_variety} $R_\theta$ maps a transformed-coordinate solution to another transformed-coordinate solution, i.e. $\left(R_\theta x, R_\theta y\right)\in V'$ if $\left(x,y\right)\in V'$.
Thus for equally spaced $l$ points $\left\{\left(\cos\left(2\pi\left(j-1\right)/l\right),\sin\left(2\pi\left(j-1\right)/l\right)\right)\right\}_{j=1}^l$ on the unit circle of $\R^2$ we have
\begin{equation}
\begin{gathered}
    R_\theta \begin{bmatrix}
            \cos\left(0\right) \\
            \cos\left(\frac{2\pi}{l}\right) \\
            \vdots\\
            \cos\left(\frac{2\pi\left(l-1\right)}{l}\right)
        \end{bmatrix}
        = \begin{bmatrix}
            \cos\left(\theta\right)\cos\left(0\right) - \sin\left(\theta\right)\sin\left(0\right) \\
            \cos\left(\theta\right)\cos\left(\frac{2\pi}{l}\right) - \sin\left(\theta\right)\sin\left(\frac{2\pi}{l}\right) \\
            \vdots\\
            \cos\left(\theta\right)\cos\left(\frac{2\pi\left(l-1\right)}{l}\right) - \sin\left(\theta\right)\sin\left(\frac{2\pi\left(l-1\right)}{l}\right)
        \end{bmatrix}
        = \begin{bmatrix}
            \cos\left(\theta\right) \\
            \cos\left(\theta+\frac{2\pi}{l}\right)\\
            \vdots\\
            \cos\left(\theta+\frac{2\pi\left(l-1\right)}{l}\right)
        \end{bmatrix}\\
    \quad\text{and}\quad\\
    R_\theta \begin{bmatrix}
        \sin\left(0\right) \\
        \sin\left(\frac{2\pi}{l}\right) \\
        \vdots\\
        \sin\left(\frac{2\pi\left(l-1\right)}{l}\right)
    \end{bmatrix}
        = \begin{bmatrix}
            \sin\left(\theta\right)\cos\left(0\right) + \cos\left(\theta\right)\sin\left(0\right) \\
            \sin\left(\theta\right)\cos\left(\frac{2\pi}{l}\right) + \cos\left(\theta\right)\sin\left(\frac{2\pi}{l}\right) \\
            \vdots\\
            \sin\left(\theta\right)\cos\left(\frac{2\pi\left(l-1\right)}{l}\right) + \cos\left(\theta\right)\sin\left(\frac{2\pi\left(l-1\right)}{l}\right)
        \end{bmatrix}
        = \begin{bmatrix}
            \sin\left(\theta\right) \\
            \sin\left(\theta+\frac{2\pi}{l}\right)\\
            \vdots\\
            \sin\left(\theta+\frac{2\pi\left(l-1\right)}{l}\right)
        \end{bmatrix} .
\end{gathered}
\end{equation}

\begin{figure}
    \centering
    \includegraphics[width=0.9\linewidth]{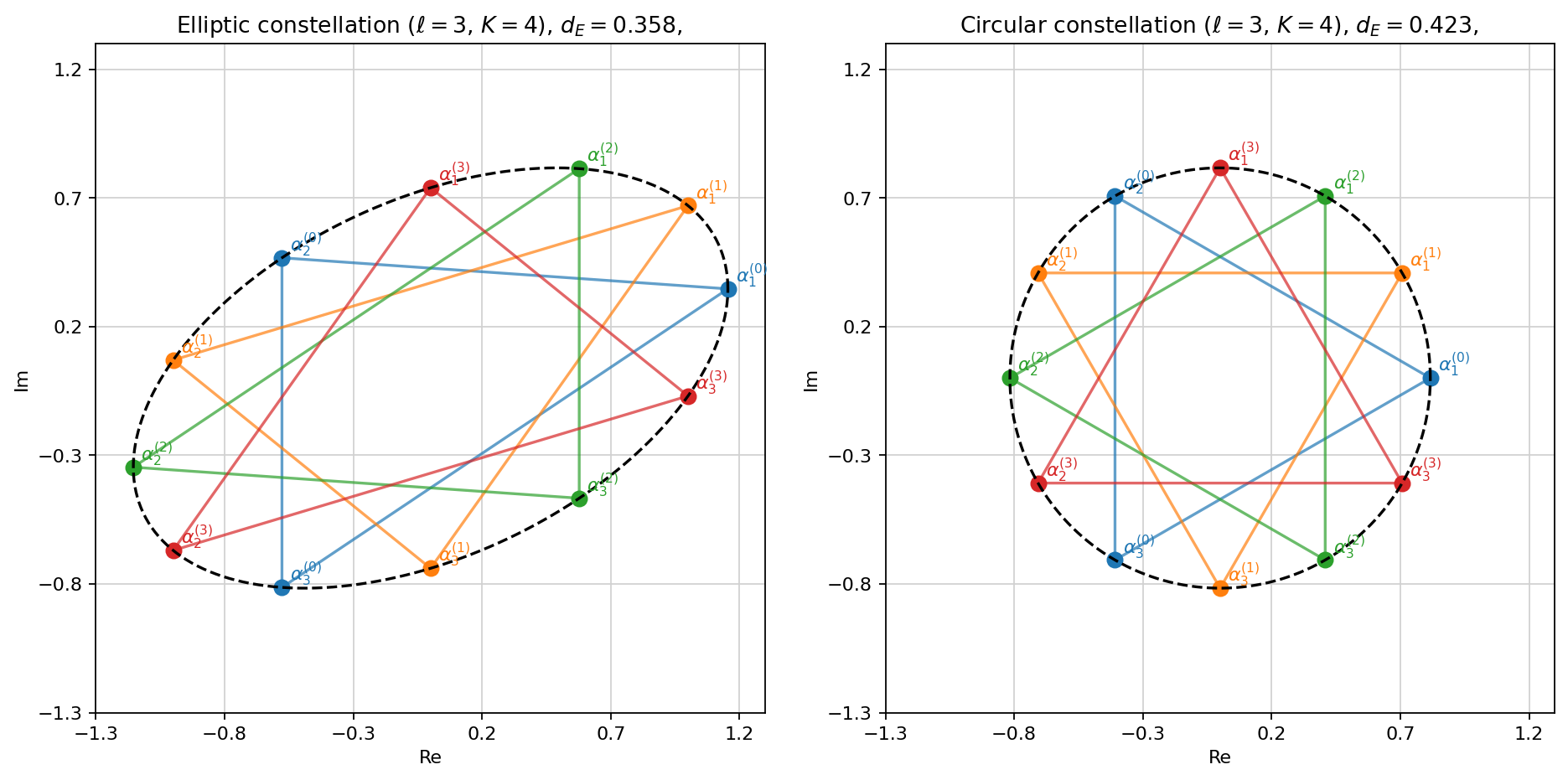}%
    \caption{
    {Ellipsoid constellations $\mathcal{C}^{\left(0\right)},\mathcal{C}^{\left(1\right)},\mathcal{C}^{\left(2\right)},\mathcal{C}^{\left(3\right)}$ spanning a $4$-dimensional {BARC code} encoded in a single mode, with minimum separation $\Delta=\sqrt{d_E}$, where each constellation state is a uniform superposition of $3$ coherent states.}
    The plots show the two real phase-space coordinates of that mode.
    {The left panel has $\left(\tau_x,\tau_y,\tau_{xy}\right)=\left(2,1,0.6\right)$, minimum separation $\Delta=\sqrt{d_E}\approx0.358$, and $\brakets{\mathcal C^{\left(k\right)}}{a^\dag a}{\mathcal C^{\left(k\right)}}\approx0.451$.}
    {The right panel has $\left(\tau_x,\tau_y,\tau_{xy}\right)=\left(1,1,0\right)$, $\Delta\approx0.423$, and energy $\brakets{\mathcal{C}^{\left(k\right)}}{a^\dag a}{\mathcal{C}^{\left(k\right)}}\approx0.142$ per constellation state.}
    The distance labels in the panels report $\Delta$, whereas $d_E$ denotes squared Euclidean resolution; $\left(\tau_x,\tau_y,\tau_{xy}\right)$ are as defined in \ref{eq:taus}.
    }
    \label{fig:ellispoid_constellations}
\end{figure}

For $N=lK$, consider the cyclic subgroup $C_N$ of $O\left(Z\right)$, given by
\begin{equation}
    C_N = \left\langle R_{2\pi/N}\right\rangle = \left\{R_0,R_{2\pi/N},\dots,R_{2\pi\left(N-1\right)/N}\right\} ,
\end{equation}
where $R_{2\pi\left(j-1\right)/N} = \left(R_{2\pi/N}\right)^{j-1}$ for $j=1,\dots,N$.
Then, let $x^{\left(k\right)} = R_{2\pi/N}^k x^{\left(0\right)}$ and $y^{\left(k\right)} = R_{2\pi/N}^k y^{\left(0\right)}$ for $k\in\left\{0,\dots,K-1\right\}$, which define constellations $\left\{\mathcal{C}^{\left(k\right)}\right\}_{k=0}^{K-1}$ where
\begin{equation}\label{eqn:pq_leq2_ellipsoid_constellations}
\begin{gathered}
    \mathcal{C}^{\left(k\right)} = \left\{\alpha_1^{\left(k\right)},\dots,\alpha_l^{\left(k\right)}\right\} \\
    \text{for}\\
    \alpha_j^{\left(k\right)} = \mu +
        \sqrt{\frac{2}{l}} \left(\vphantom{\left( \sqrt{\tau_x} + \iota \tau_{xy}\left(\tau_x\right)^{-\frac{1}{2}} \right) \cos\left(\frac{2\pi\left(j-1\right)}{l}+\frac{2\pi k}{lK}\right) + \iota \sqrt{\tau_y-\frac{\tau_{xy}^2}{\tau_x}} \sin\left(\frac{2\pi\left(j-1\right)}{l}+\frac{2\pi k}{lK}\right)} \left( \sqrt{\tau_x} + \iota \tau_{xy}\left(\tau_x\right)^{-\frac{1}{2}} \right) \cos\left(\frac{2\pi\left(j-1\right)}{l}+\frac{2\pi k}{lK}\right) \right.\\\left.\vphantom{\left( \sqrt{\tau_x} + \iota \tau_{xy}\left(\tau_x\right)^{-\frac{1}{2}} \right) \cos\left(\frac{2\pi\left(j-1\right)}{l}+\frac{2\pi k}{lK}\right) + \iota \sqrt{\tau_y-\frac{\tau_{xy}^2}{\tau_x}} \sin\left(\frac{2\pi\left(j-1\right)}{l}+\frac{2\pi k}{lK}\right)}
        + \iota \sqrt{\tau_y-\frac{\tau_{xy}^2}{\tau_x}} \sin\left(\frac{2\pi\left(j-1\right)}{l}+\frac{2\pi k}{lK}\right) \right) ,
\end{gathered}
\end{equation}
which two instances are shown in Fig.\ \ref{fig:ellispoid_constellations}.

Below we give lower and upper bounds on the squared Euclidean resolution
$d_E=\min_{\left(k,j\right)\neq\left(k',j'\right)}\abs{\alpha_j^{\left(k\right)}-\alpha_{j'}^{\left(k'\right)}}^2$
of a code defined by such constellations.

\begin{proposition}\label{prop:ellipsoid_resolution}
    For the ellipsoid constellations in
    eqn.\ \eqref{eqn:pq_leq2_ellipsoid_constellations}, the squared
    Euclidean resolution is bounded as
    \begin{equation}
        4\sigma_{\min}\left(L\right)^2\sin^2\frac{\pi}{lK}
        \leq d_E\leq
        4\sigma_{\max}\left(L\right)^2\sin^2\frac{\pi}{lK} ,
    \end{equation}
    where $\sigma_{\min}\left(L\right)$ and $\sigma_{\max}\left(L\right)$ are the
    smallest and largest singular values of $L$ (given in
    eqn.\ \eqref{eqn:cholesky_G_matrix}), respectively.
\end{proposition}

\begin{proof}
    First note that the set of points in the union of constellations 
    \begin{equation}
        \bigcup_{k=0}^{K-1}\mathcal C^{\left(k\right)}=\left\{\alpha_j^{\left(k\right)}:j\in\left\{1,\dots,l\right\},\;k\in\left\{0,\dots,K-1\right\}\right\}    
    \end{equation}
    corresponds to the image of an affine map $w\mapsto \left(\Bar{u},\Bar{v}\right) + Lw$ in $\R^2$ from a set of equally-spaced $N=Kl$ points on the unit sphere.
    Namely if $\alpha_j^{\left(k\right)} = u_j^{\left(k\right)} + \iota v_j^{\left(k\right)}$, then
    \begin{equation}
    \begin{gathered}
        \begin{bmatrix}
            u_j^{\left(k\right)}\\
            v_j^{\left(k\right)}
        \end{bmatrix}
            = \begin{bmatrix}
                \Bar{u}\\
                \Bar{v}
            \end{bmatrix} +
            L\begin{bmatrix}
                \cos\left(\frac{2\pi\left(j-1\right)}{l} + \frac{2\pi k}{N}\right) \\
                \sin\left(\frac{2\pi\left(j-1\right)}{l} + \frac{2\pi k}{N}\right)
            \end{bmatrix} ,\quad k\in\left\{0,\dots,K-1\right\}
        \\
        \text{where}\\
        \begin{bmatrix}
            u_j^{\left(0\right)}\\
            v_j^{\left(0\right)}
        \end{bmatrix}
        =
        \begin{bmatrix}
            \Bar{u}\\
            \Bar{v}
        \end{bmatrix} +
            \begin{bmatrix}
                x_j^{\left(0\right)}\\
                y_j^{\left(0\right)}
            \end{bmatrix}
        = \begin{bmatrix}
            \Bar{u}\\
            \Bar{v}
        \end{bmatrix} +
            \sqrt{\frac{l}{2}}L \begin{bmatrix}
                c_j\\
                s_j
            \end{bmatrix}
        = \begin{bmatrix}
            \Bar{u}\\
            \Bar{v}
        \end{bmatrix} +
            \sqrt{\frac{2}{l}} \begin{bmatrix}
                \sqrt{\tau_x} & 0\\
                \frac{\tau_{xy}}{\sqrt{\tau_x}} & \sqrt{\tau_y-\frac{\tau_{xy}^2}{\tau_x}}
            \end{bmatrix}
            \begin{bmatrix}
                \cos\frac{2\pi\left(j-1\right)}{l}\\
                \sin\frac{2\pi\left(j-1\right)}{l}
            \end{bmatrix} \\
        \text{for}\\
        x_j^{\left(0\right)} = \sqrt{\frac{2\tau_x}{l}} \cos\frac{2\pi\left(j-1\right)}{l}
        \quad\text{and}\\
        y_j^{\left(0\right)} = \tau_{xy} \sqrt{\frac{2}{l\tau_x}} \cos\frac{2\pi\left(j-1\right)}{l} + \sqrt{\frac{2}{l}}\sqrt{\tau_y-\frac{\tau_{xy}^2}{\tau_x}} \sin\frac{2\pi\left(j-1\right)}{l} .
    \end{gathered}
    \end{equation}
    So, writing $\Delta=\sqrt{d_E}$, the minimum separation of
    the union of constellations $\left\{\mathcal C^{\left(k\right)}\right\}_{k=0}^{K-1}$ is
    \begin{equation}
        \Delta=\min_{\left(k,j\right)\neq\left(k',j'\right)}
            \abs{\alpha_j^{\left(k\right)}-\alpha_{j'}^{\left(k'\right)}}
        =\min_{\substack{0\leq i,j\leq N-1\\i\neq j}}
            \nrm{L\left(\nu_i-\nu_j\right)}_2 ,
    \end{equation}
    where $\nu_j=\left(\cos\theta_j,\sin\theta_j\right)\in\R^2$ lies on
    the unit circle and $\theta_j=2\pi j/N$ for $j\in\left\{0,\dots,N-1\right\}$.
        
    Now let us consider $i\neq j$ with $\abs{i-j}=m\geq1$.
    Then it holds that
    \begin{equation}
    \begin{aligned}
        {\nrm{\nu_i-\nu_j}_2}^2 &= \left(\cos\theta_i-\cos\theta_j\right)^2 + \left(\sin\theta_i-\sin\theta_j\right)^2 
            \\ &= 2-2\cos\left(\theta_i-\theta_j\right) 
            = 4\sin^2\left(\frac{\theta_i-\theta_j}{2}\right) 
            = 4\sin^2\left(\frac{m\pi}{N}\right) .
    \end{aligned}
    \end{equation}
    Note that for any $z\in\R^2$ it holds that
    \begin{equation}
        \sigma_{\min}\left(L\right)\nrm{z}_2 \leq \nrm{Lz}_2 \leq \sigma_{\max}\left(L\right)\nrm{z}_2
    \end{equation}
    where $\sigma_{\min}\left(L\right) = \min_{z:\nrm{z}_2=1} \nrm{Lz}_2$ and $\sigma_{\max}\left(L\right) = \max_{z:\nrm{z}_2=1} \nrm{Lz}_2$ are the smallest and largest singular value of $L$.
    Since $\Delta$ is upper bounded by $\nrm{L\left(\nu_i-\nu_j\right)}_2$ for any $i\neq j$, we have
    \begin{equation}
    \begin{aligned}
        \Delta \leq \nrm{L\left(\nu_{j+1}-\nu_j\right)}_2 \leq \sigma_{\max}\left(L\right)\nrm{\nu_{j+1}-\nu_j}_2
            = 2 \sigma_{\max}\left(L\right) \sin\left(\frac{\pi}{N}\right) .
    \end{aligned}
    \end{equation}
    On the other hand, for any $i\neq j$ it holds that
    \begin{equation}
        \nrm{L\left(\nu_i-\nu_j\right)}_2 \geq \sigma_{\min}\left(L\right) \nrm{\nu_i-\nu_j}_2 \geq 2\sigma_{\min}\left(L\right) \sin\left(\frac{\abs{i-j}\pi}{N}\right) .
    \end{equation}
    For $1\leq m=\abs{i-j}\leq N-1$, the smallest value of
    $\sin\left(m\pi/N\right)$ is $\sin\left(\pi/N\right)$: use
    $\sin\left(m\pi/N\right)=\sin\left(\left(N-m\right)\pi/N\right)$ and monotonicity of sine
    on $\left[0,\pi/2\right]$. This value is attained by adjacent points.
    Therefore
    \begin{equation}
    \begin{aligned}
        \Delta = \min_{\substack{0\leq i,j\leq N-1\\i\neq j}} \nrm{L\left(\nu_i-\nu_j\right)}_2
            \geq \min_{\substack{0\leq i,j\leq N-1\\i\neq j}} 2\sigma_{\min}\left(L\right) \sin\left(\frac{\abs{i-j}\pi}{N}\right)
            = 2\sigma_{\min}\left(L\right) \sin\left(\frac{\pi}{N}\right) .
    \end{aligned}
    \end{equation}
    Squaring the nonnegative upper and lower bounds for $\Delta$ gives the claimed bounds for $d_E$.
\end{proof}

\section{Non-ellipsoid constellations}\label{sec:nonellipsoid_constellations}

This section first gives the moment identities for the hexagonal patch $\mathcal P_r$ in Lemma~\ref{lem:hexagonal_moments} and the transformation to identity second moment in Lemma~\ref{lem:planar_whitening}.
Using these two tools, we optimize the relative rotation of the $r=2$ patch, which is an example of the Hex1 family defined in Section~\ref{sec:hexagonal_codes}.
That subsection then specifies Hex1, Hex2, and Hex3: rotated full patches, balanced shifts followed by rewhitening, and alternating triangles of sixfold orbits, respectively.
Hex2 and Hex3 use selected lattice orbits and are not, in general, pairs of full patches $\mathcal P_r$; they meet the same first- and second-moment constraints by the identities stated below.

Consider $L=\sqrt{T}$ so that $LL^\top = T$ and real points $q_i^{\left(k\right)} = \left(u_i^{\left(k\right)},v_i^{\left(k\right)}\right) \in\R^2$ for $i\in\left[l\right]$ and $k\in\left[K\right]$.
So for $N=lK$ real points $w_i^{\left(k\right)} \in\R^2$ given by
\begin{equation}
    q_i^{\left(k\right)} = L w_i^{\left(k\right)}
\end{equation}
it holds that the following two conditions over the $w_i^{\left(k\right)}$'s
\begin{equation}\label{eqn:normalized_planar_moments}
    \sum_{i=1}^l w_i^{\left(k\right)} = 0
    \quad\text{and}\quad
    \sum_{i=1}^l w_i^{\left(k\right)} \left(w_i^{\left(k\right)}\right)^\top = I
\end{equation}
imply that that $\left\{q^{\left(k\right)}\right\}_k$ are in the centered solution $V'$ since the $q_i^{\left(k\right)}$'s satisfy 
\begin{equation}
\begin{gathered}
    \sum_{i=1}^l q_i^{\left(k\right)} = L \sum_i w_i^{\left(k\right)} = 0
    \\ \text{and} \\
    \sum_{i=1}^l q_i^{\left(k\right)} \left(q_i^{\left(k\right)}\right)^\top = L \left(\sum_{i=1}^l w_i^{\left(k\right)} \left(w_i^{\left(k\right)}\right)^\top\right) L^\top = LL^\top = T .
\end{gathered}
\end{equation}
For $b_1=\left(1,0\right)$, $b_2=\left(1/2,\sqrt3/2\right)$, and an integer $r\geq1$, define
\begin{equation}
    \mathcal P_r=\left\{mb_1+nb_2:m,n\in\Z,\;
       1\leq\max\left\{\abs{m},\abs{n},\abs{m+n}\right\}\leq r\right\} .
\end{equation}

\begin{lemma}\label{lem:hexagonal_moments}
    The set $\mathcal P_r$ has $3r\left(r+1\right)$ points, minimum
    Euclidean separation $1$, and zero first moment. Its second moment is
    \begin{equation}
        \sum_{w\in\mathcal P_r}ww^\top=c_rI_2
        \quad\text{for}\quad
        c_r=\frac{r\left(r+1\right)\left(5r^2+5r+2\right)}{8} .
    \end{equation}
\end{lemma}
\begin{proof}
    The $k$th hexagonal shell has $6k$ points and squared-norm
    sum $k\left(5k^2+1\right)$. Summing over $k=1,\dots,r$ gives the point
    count and $\sum_w{\nrm{w}_2}^2=2c_r$.
    The first moment vanishes by inversion symmetry, and the
    second moment is scalar by sixfold rotational symmetry.
    Distinct lattice points have separation at least $1$,
    and neighboring points in $\mathcal P_r$ attain it.
\end{proof}

\begin{lemma}\label{lem:planar_whitening}
    If $\sum_iq_i=0$ and $T=\sum_iq_iq_i^\top>0$, then
    $w_i=T^{-1/2}q_i$ satisfies
    \begin{equation}
        \sum_iw_i=0 ,\quad \sum_iw_iw_i^\top=I_2 ,
    \end{equation}
    and
    \begin{equation}
        \frac{\nrm{q_i-q_j}_2}{\sqrt{\lambda_{\max}\left(T\right)}}
        \leq\nrm{w_i-w_j}_2\leq
        \frac{\nrm{q_i-q_j}_2}{\sqrt{\lambda_{\min}\left(T\right)}} .
    \end{equation}
    In particular, if $T=cI_2$, then $w_i=q_i/\sqrt c$.
\end{lemma}
\begin{proof}
    The first- and second-moment identities in the statement follow from
    \begin{equation}
    \begin{aligned}
        \sum_iw_i&=T^{-1/2}\sum_iq_i=0,\quad\text{and}\\
        \sum_iw_iw_i^\top&=T^{-1/2}\left(\sum_iq_iq_i^\top\right)T^{-1/2}=T^{-1/2}TT^{-1/2}=I_2 .
    \end{aligned}
    \end{equation}
    Its smallest and largest singular values are $1/\sqrt{\lambda_{\max}\left(T\right)}$ and $1/\sqrt{\lambda_{\min}\left(T\right)}$, respectively.
    Since $w_i-w_j=T^{-1/2}\left(q_i-q_j\right)$, the singular-value inequalities give
    \begin{equation}
    \begin{aligned}
        \frac{\nrm{q_i-q_j}_2}{\sqrt{\lambda_{\max}\left(T\right)}}
        &=\sigma_{\min}\left(T^{-1/2}\right)\nrm{q_i-q_j}_2\\
        &\leq\nrm{T^{-1/2}\left(q_i-q_j\right)}_2=\nrm{w_i-w_j}_2\\
        &\leq\sigma_{\max}\left(T^{-1/2}\right)\nrm{q_i-q_j}_2
        =\frac{\nrm{q_i-q_j}_2}{\sqrt{\lambda_{\min}\left(T\right)}} .
    \end{aligned}
    \end{equation}
    If $T=cI_2$ with $c>0$, then $T^{-1/2}=c^{-1/2}I_2$, so $w_i=q_i/\sqrt c$ and $\nrm{w_i-w_j}_2=\nrm{q_i-q_j}_2/\sqrt c$.
\end{proof}

By Lemma~\ref{lem:hexagonal_moments}, the patch has zero first moment and second moment $c_rI_2$, so the scalar case of Lemma~\ref{lem:planar_whitening} gives the normalized patch $\mathcal P_r/\sqrt{c_r}$ satisfying eqn.\ \eqref{eqn:normalized_planar_moments}.
Pairing it with a rotated copy and applying $L$ gives exactly the Hex1 construction in eqn.\ \eqref{eqn:hex1_family}.
For $K=2$, $r=2$, and $T=cI_2$ with $c>0$, we have $l=18$ and $c_r=24$.
The case $r=1$ is a circular constellation; the $r=2$ patch used here contains several radii and provides the non-ellipsoid example.

Identify $\mathcal P_2$ with complex amplitudes and take
\begin{equation}
    \mathcal C^{\left(0\right)}=\sqrt{\frac{c}{24}}\,\mathcal P_2
    \quad\text{and}\quad
    \mathcal C^{\left(1\right)}=e^{\iota\theta}\mathcal C^{\left(0\right)} .
\end{equation}
Both constellations have the same first and second moments.
For this centered isotropic example, the unnormalized complex moments through degree two are
\begin{equation}
    c_{0,0}=18,\qquad c_{1,1}=2c,\qquad c_{0,1}=c_{1,0}=c_{0,2}=c_{2,0}=0 .
\end{equation}
Thus every off-diagonal moment $c_{p,q}$ with $p\neq q$ and $p+q\leq2$ vanishes, and Proposition~\ref{prop:cyclic_fiber_symmetries} with $t=2$ shows that every phase-space rotation preserves this moment fiber.
This justifies generating the second Hex1 constellation by rotation of the first; it does not assert that every such rotation preserves higher-degree moments or satisfies exact finite-separation KL conditions.
{Sixfold symmetry and reflection reduce the relative angle to $0<\theta\leq\pi/6$ for two disjoint constellations.}
On this interval,
\begin{equation}
    d_E\left(\theta\right)=\frac{c}{24}\min\left\{
       2-2\cos\theta,\;
       7-4\sqrt3\cos\left(\frac{\pi}{6}-\theta\right)\right\} .
\end{equation}
The two terms come from pairs at radii $1,1$ and $\sqrt3,2$, respectively; the remaining pair distances are larger.
The first term increases with $\theta$, and the second decreases.
Equating them therefore gives the optimal relative rotation
\begin{equation}\label{eqn:hex_r2_optimal_rotation}
    \theta^*=\arccos\frac{13}{14}
    =2\arctan\frac{1}{3\sqrt3}\approx21.79^\circ
    \quad\text{and}\quad d_E\left(\theta^*\right)=\frac{c}{168} .
\end{equation}
This improves on $d_E\left(\pi/6\right)=c\left(7-4\sqrt3\right)/24$.
The optimum is within this relative-rotation family; it does not establish optimality over all moment-preserving deformations.
For general $T>0$, rotating $\mathcal P_2/\sqrt{24}$ before applying $L$ still preserves $T$, but the optimal angle depends on $T$.

\subsection{Hexagonal BARC codes}\label{sec:hexagonal_codes}

\begin{figure}
    \centering
    \includegraphics[
        width=0.32\textwidth,
        height=0.25\textheight,
        keepaspectratio
    ]{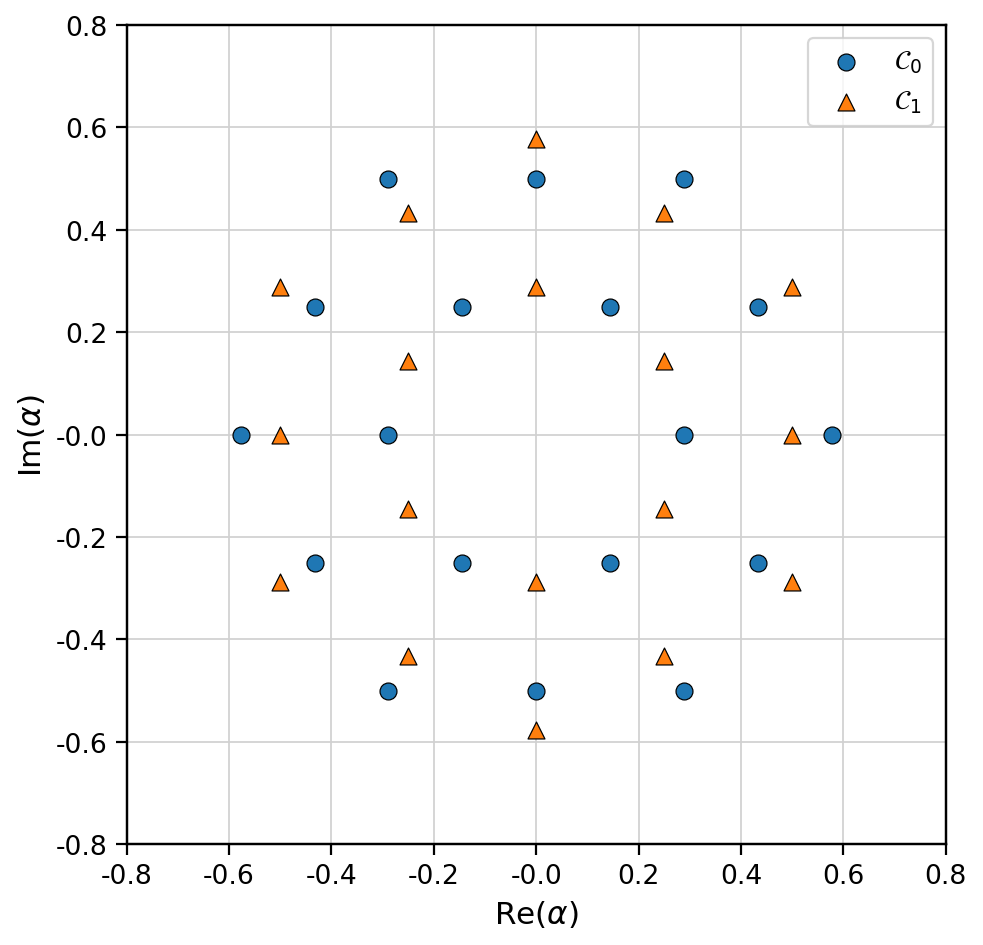}\hfill%
    \includegraphics[
        width=0.32\textwidth,
        height=0.25\textheight,
        keepaspectratio
    ]{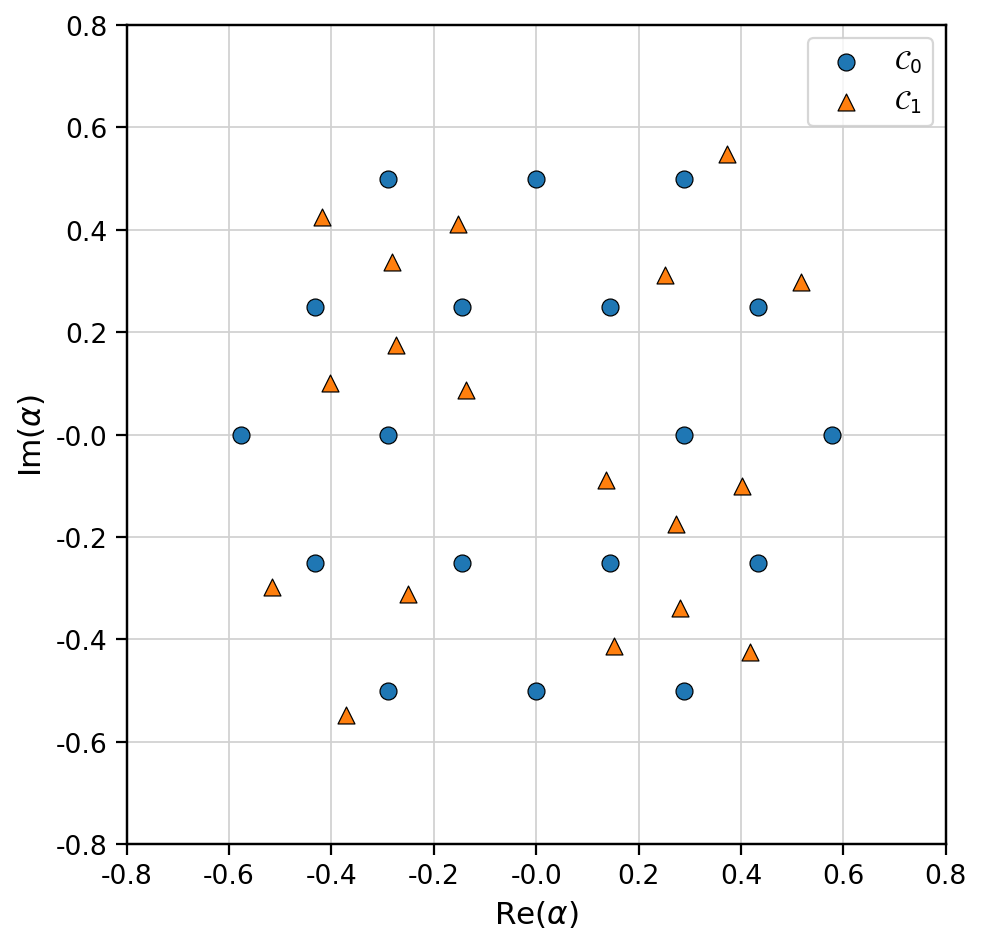}\hfill%
    \includegraphics[
        width=0.34\textwidth,
        height=0.215\textheight,
        keepaspectratio
    ]{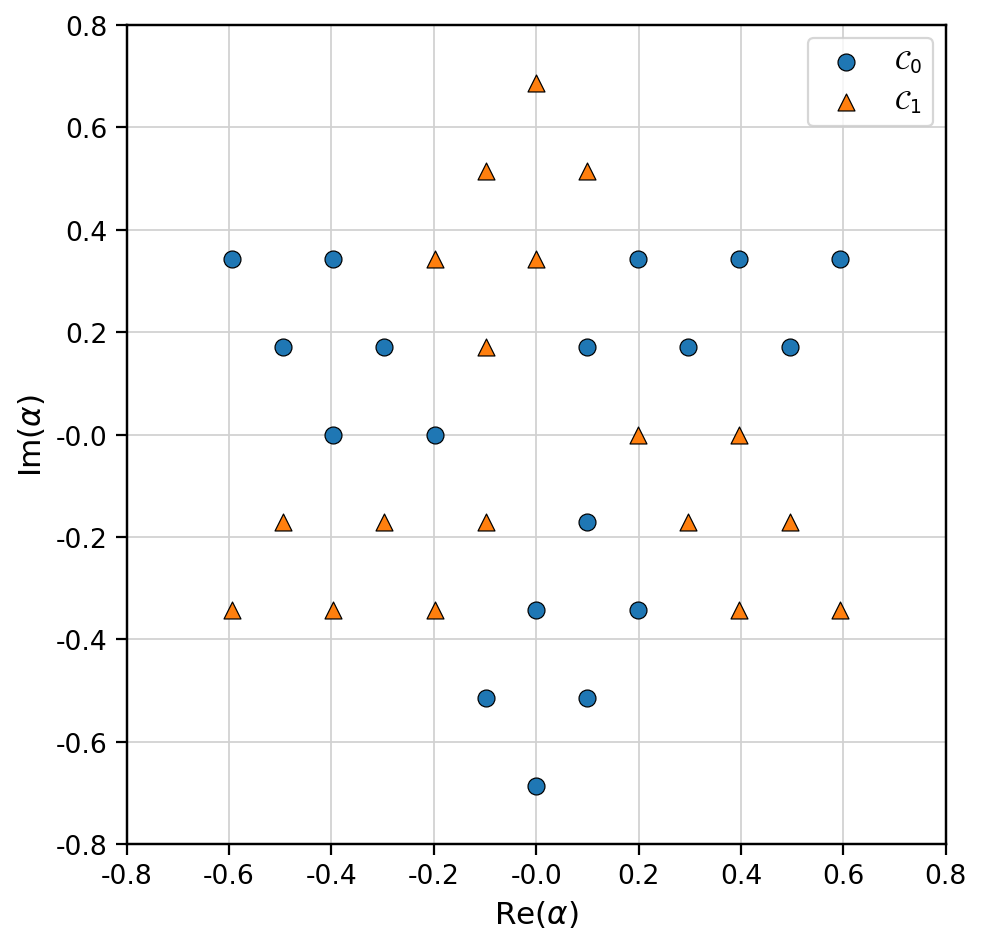}

    \caption{Constellations for single-mode Hex1-$18$, Hex2-$18$, and Hex3-$18$ {BARC codes} from Section~\ref{sec:hexagonal_codes}. Each construction gives two normalized constellation states, each a uniform superposition of $18$ coherent states. The panels illustrate their geometries at a fixed reference scale, where $\left(\tau_x, \tau_y, \tau_{xy}\right)=\left(2,2,0\right)$; the actual operating scale is varied and optimized in our numerical results reported in section \ref{sec:noise-benchmarking}. Left: Hex1-$18$ uses the origin-excluded radius-two patch and relative rotation $\pi/2$ (\ref{sec:Hex1-desc}). Center: Hex2-$18$ uses balanced opposite shifts followed by rewhitening (\ref{sec:Hex2-desc}). Right: Hex3-$18$ uses six chosen sixfold orbits with alternating triangle colors.(\ref{sec:Hex3-desc})}
    \label{fig:hexagonal_codes_figure}
\end{figure}

We describe {three two-constellation constructions of BARC codes} on the triangular lattice $\Lambda=\left\{mb_1+nb_2:m,n\in\Z\right\}$, with $b_1,b_2$ as above.
Let $U$ be rotation by $\pi/3$ in $\R^2$.
For a nonzero lattice point $\zeta$, define the two alternating triangles of its sixfold orbit by
\begin{equation}\label{eqn:hex_orbit_triangles}
\begin{gathered}
    A\left(\zeta\right)=\left\{\zeta,U^2\zeta,U^4\zeta\right\},\qquad B\left(\zeta\right)=UA\left(\zeta\right),\qquad O\left(\zeta\right)=A\left(\zeta\right)\cup B\left(\zeta\right),\\
    \text{so that}\\
    \sum_{x\in A\left(\zeta\right)}x=\sum_{x\in B\left(\zeta\right)}x=0
    \quad\text{and}\quad
    \sum_{x\in A\left(\zeta\right)}xx^\top=\sum_{x\in B\left(\zeta\right)}xx^\top=\frac32\nrm{\zeta}_2^2I_2 .
\end{gathered}
\end{equation}
To verify these identities, identify $\zeta$ with a complex number $z$ and put $\omega=e^{2\pi\iota/3}$.
The first-moment sum is $z\sum_{j=0}^2\omega^j=0$, and the complex second-moment sum is $z^2\sum_{j=0}^2\omega^{2j}=0$.
The latter gives equality of the two real squared-coordinate sums and vanishing mixed sum; since their sum is $3\abs{z}^2$, each squared-coordinate sum is $3\abs{z}^2/2$.
Rotating by $U$ gives the same identities for $B\left(\zeta\right)$.

In each construction below, $\mathcal W^{\left(0\right)}$ and $\mathcal W^{\left(1\right)}$ have $l$ points, zero first moment, and identity second-moment matrix.
For a positive-definite target matrix $T$ and amplitude scale $\varrho>0$, define
\begin{equation}\label{eqn:hex_target_map}
    \mathcal C^{\left(k\right)}=\left\{\varrho\left[\left(T^{1/2}w\right)_1+\iota\left(T^{1/2}w\right)_2\right]:w\in\mathcal W^{\left(k\right)}\right\},\qquad k=0,1 .
\end{equation}
Thus each real-coordinate constellation has zero first moment and second-moment matrix $\varrho^2T$.
Taking uniform coherent-state coefficients gives the moment constraints through degree two, so the Hex1, Hex2, and Hex3 constructions below, including their elliptical versions, are instances of BARC codes as defined in Section~\ref{sec:constellations}.

The isotropic constructions use $T=2I_2$, while the elliptical versions replace this by the target $T$ specified in the numerical section.
Configurations are retained only when the two supports are disjoint, as required in Section~\ref{sec:coherent_state_codes}.

\paragraph{Hex1: rotated full patches.}
\label{sec:Hex1-desc}
Using $\mathcal P_r$ and $c_r$ from Lemma~\ref{lem:hexagonal_moments}, set
\begin{equation}\label{eqn:hex1_family}
    \mathcal W^{\left(0\right)}=\frac{\mathcal P_r}{\sqrt{c_r}},\qquad
    \mathcal W^{\left(1\right)}=R_\theta\mathcal W^{\left(0\right)},\qquad l=3r\left(r+1\right),
\end{equation}
where $R_\theta$ here denotes the two-dimensional phase-space rotation through angle $\theta$.
Both moment matrices are $I_2$, since $R_\theta I_2R_\theta^\top=I_2$.
The reported Hex1-$18$ geometry takes $r=2$ and $\theta=\pi/2$.
Since the patch is invariant under $U$, this gives the same second constellation as a relative rotation by $\pi/6$.
It is therefore not the resolution-optimized member with $\theta=\theta^*$ from eqn.\ \eqref{eqn:hex_r2_optimal_rotation}.

\paragraph{Hex2: balanced shifts and rewhitening.}
\label{sec:Hex2-desc}
Choose $m$ nonzero lattice representatives $\zeta_1,\dots,\zeta_m$ with disjoint sixfold orbits, and set $l=6m$ and $\kappa=3\sum_{a=1}^m\nrm{\zeta_a}_2^2$.
Split each orbit into its two triangles, choosing one triangle for $\mathcal P_+$ and the other for $\mathcal P_-$.
Explicitly, for choices $\sigma_a\in\left\{0,1\right\}$,
\begin{equation}
    \mathcal P_+=\bigcup_{a=1}^m U^{\sigma_a}A\left(\zeta_a\right),\qquad
    \mathcal P_-=\bigcup_{a=1}^m U^{\sigma_a+1}A\left(\zeta_a\right),\qquad
    \mathcal W^{\left(0\right)}=\frac{\mathcal P_+\cup\mathcal P_-}{\sqrt\kappa} .
\end{equation}
Each subset has $l/2$ points and zero first moment, and $\mathcal W^{\left(0\right)}$ has identity second moment.
Let $h=\left(b_1+b_2\right)/\left(3\sqrt\kappa\right)$, the scaled shift to the center of an elementary lattice triangle, and define
\begin{equation}\label{eqn:hex2_family}
\begin{gathered}
    \widetilde{\mathcal W}^{\left(1\right)}=\left(\frac{\mathcal P_+}{\sqrt\kappa}+h\right)\cup\left(\frac{\mathcal P_-}{\sqrt\kappa}-h\right),\qquad
    H=\left(I_2+l hh^\top\right)^{-1/2},\qquad
    \mathcal W^{\left(1\right)}=H\widetilde{\mathcal W}^{\left(1\right)} .
\end{gathered}
\end{equation}
The opposite shifts leave the first moment zero because the two subsets have equal cardinality.
Writing $\mathcal W_\pm^{\left(0\right)}=\mathcal P_\pm/\sqrt\kappa$ and using their zero sums gives
\begin{equation}
\begin{aligned}
    \sum_{x\in\widetilde{\mathcal W}^{\left(1\right)}}xx^\top
    &=\sum_{x\in\mathcal W_+^{\left(0\right)}}\left(x+h\right)\left(x+h\right)^\top+\sum_{x\in\mathcal W_-^{\left(0\right)}}\left(x-h\right)\left(x-h\right)^\top\\
    &=I_2+l hh^\top,\\
    \sum_{w\in\mathcal W^{\left(1\right)}}ww^\top
    &=H\left(I_2+l hh^\top\right)H^\top=I_2 .
\end{aligned}
\end{equation}
Thus the second constellation is a rewhitened union of oppositely shifted subsets of the first constellation, not a translation of the entire first constellation.
The map $H$ is the transformation of Lemma~\ref{lem:planar_whitening} applied to the shifted set, whose second-moment matrix is $I_2+l hh^\top$ by the calculation above.

Choices $m=1,2,3$ give $l=6,12,18$.
The selected representatives and triangle assignments specify the individual member of the family.
For the $l=18$ pattern shown in Fig.\ \ref{fig:hexagonal_codes_figure}, one realization is
\begin{equation}\label{eqn:hex2_18_representatives}
    \left(\zeta_1,\zeta_2,\zeta_3\right)=\left(b_1,b_1+b_2,2b_1\right),\qquad \left(\sigma_1,\sigma_2,\sigma_3\right)=\left(1,0,1\right),\qquad\kappa=24 .
\end{equation}
For the $18$-point representatives in eqn.\ \eqref{eqn:hex2_18_representatives} above, the unshifted union $\mathcal P_+\cup\mathcal P_-$ equals $\mathcal P_2$, so Lemma~\ref{lem:hexagonal_moments} gives $\kappa=24$ and the first normalized constellation is the same patch as in Hex1.
For general representative lists this union need not be a full patch, and the orbit identities in eqn.\ \eqref{eqn:hex_orbit_triangles} supply its moments instead.

\paragraph{Hex3: alternating orbit colors.}
\label{sec:Hex3-desc}
Choose $m=l/3$ nonzero lattice representatives $\zeta_1,\dots,\zeta_m$ with disjoint sixfold orbits.
Assign opposite triangles of each orbit to the two logical constellations and alternate the assignment between successive chosen orbits:
\begin{equation}\label{eqn:hex3_family}
\begin{gathered}
    \sigma_a=\left(a-1\right)\bmod2,\qquad \kappa_3=\frac32\sum_{a=1}^m\nrm{\zeta_a}_2^2,\\
    \mathcal W^{\left(0\right)}=\frac{1}{\sqrt{\kappa_3}}\bigcup_{a=1}^m U^{\sigma_a}A\left(\zeta_a\right),\qquad
    \mathcal W^{\left(1\right)}=\frac{1}{\sqrt{\kappa_3}}\bigcup_{a=1}^m U^{\sigma_a+1}A\left(\zeta_a\right) .
\end{gathered}
\end{equation}
Each constellation contains three points from every chosen orbit, so it has $l=3m$ points.
Equation~\eqref{eqn:hex_orbit_triangles} gives zero first moment and identity second moment for both constellations, independently of the color assignment.
Moreover, $\mathcal W^{\left(1\right)}=U\mathcal W^{\left(0\right)}$.
The common factor $1/\sqrt{\kappa_3}$ is the scalar case of Lemma~\ref{lem:planar_whitening} applied separately to the two selected triangle unions, whose unnormalized second moments are $\kappa_3I_2$ by eqn.\ \eqref{eqn:hex_orbit_triangles}.
Unlike Hex1, neither constellation is required to be a full patch $\mathcal P_r$.
For isotropic $T=cI_2$, eqn.\ \eqref{eqn:hex_target_map} gives $\mathcal C^{\left(1\right)}=e^{\iota\pi/3}\mathcal C^{\left(0\right)}$, so Proposition~\ref{prop:cyclic_fiber_symmetries} with $t=2$ also applies to this pair.
For general positive-definite $T$, the map $T^{1/2}$ need not commute with $U$, so equality of the target moments follows from the matrix identities above rather than from an ordinary phase-space rotation of the physical amplitudes.

The cases $l=6,12,18$ use $m=2,4,6$ chosen orbits, respectively.
The seed representatives and their ordering remain part of the specification of a particular code; the construction does not prescribe a unique code from $l$ alone.
For the $l=18$ point pattern in Fig.\ \ref{fig:hexagonal_codes_figure}, one can use the effective triangle representatives $\xi_a=U^{\sigma_a}\zeta_a$ given by
\begin{equation}\label{eqn:hex3_18_representatives}
    \left(\xi_1,\dots,\xi_6\right)=\left(-b_1,b_1+b_2,-2b_1,2b_1+b_2,-3b_1+b_2,2b_1+2b_2\right),\qquad\kappa_3=51 .
\end{equation}
The corresponding orbit representatives in eqn.\ \eqref{eqn:hex3_family} are $\zeta_a=U^{-\sigma_a}\xi_a$.
These choices reproduce the displayed $18$-point geometries at the reference scale $\varrho=1$ in eqn.\ \eqref{eqn:hex_target_map}; specifying the other constellation sizes and reproducing all numerical operating points still requires the corresponding representative lists and amplitude scales.

\section{Numerical performance under bosonic pure-loss noise}\label{sec:noise-benchmarking}

{We now benchmark the new BARC codes against quantum spherical codes (QSCs) and quantum cubature codes (QCCs).}

Consider a code with $K$ code states $\left\{\ket{\psi_i}\right\}_{i\in\left[K\right]}$ in the bosonic Hilbert space (with one or more modes); this code is supposed to encode the quantum information of a $K$-dimensional Hilbert space (we'll call it a ``qu$K$it'' for convenience).
The benchmark is based on the following operational task:
\begin{enumerate}
    \item Prepare a qu$K$it--mode ``probe state'' of the form $\ket{\Gamma_K}=K^{-1/2}\sum\limits_{i\in\left[K\right]}\ket i\otimes\ket{\psi_i}$;
    \item Subject the mode to the relevant noise channel $\mathcal{N}$;
    \item Apply a recovery channel $\mathcal{R}$ on the noise-affected mode to map it to a qu$K$it, aiming to correctly restore the intended value $i$.
\end{enumerate}
The success at the task can be measured as the final state's fidelity with the two-qu$K$it maximally-entangled state
\begin{equation}
    \ket{\Phi_K}=\frac{1}{\sqrt K}\sum_{i\in\left[K\right]}\ket{i}\otimes\ket{i} .
\end{equation}
Defining the composite channel
\begin{equation}
    \Lambda = \mathcal{R} \circ \mathcal{N},
\end{equation}
this fidelity metric can be expressed as
\begin{equation}\label{eq:logical_entanglement_fidelity}
\begin{gathered}
    F^\star= \max_{\mathcal R\ {\rm CPTP}}\bra{\Phi_K}
       \left(\operatorname{id}\otimes\Lambda\right)
       \left(\ket{\Gamma_K}\bra{\Gamma_K}\right)\ket{\Phi_K}.%
\end{gathered}
\end{equation}

We use $F^{\star}$ as the optimized fidelity of the recovered prepared state with the specified maximally entangled target; the phrase ``entanglement fidelity'' in the numerical comparisons below refers to this particular prepared-state benchmark.
It is not, in general, the standard entanglement fidelity of an encoded logical channel~\cite{NielsenChuang}.
Indeed, let $S^\psi$ be the Gram matrix of the normalized constellation states, with $S^\psi_{i,j}=\braket{\psi_i}{\psi_j}$.
The reference marginal of the probe is
\begin{equation}\label{eq:benchmark-reference-marginal}
    \rho_A=\operatorname{Tr}_B\ket{\Gamma_K}\bra{\Gamma_K}
    =\frac1K\sum_{i,j\in\left[K\right]}\braket{\psi_j}{\psi_i}\ket{i}\bra{j}
    =\frac{\left(S^\psi\right)^{\mathsf T}}{K} .
\end{equation}
It equals $I_K/K$ precisely when the constellation states are orthonormal.
In that case $V\ket{i}=\ket{\psi_i}$ is an encoding isometry, $\ket{\Gamma_K}=\left(I\otimes V\right)\ket{\Phi_K}$, and eqn.\ \eqref{eq:logical_entanglement_fidelity} is the optimal recovered entanglement fidelity for the logical input initially maximally entangled with a reference.
For nonorthogonal constellation states, it instead tests recovery of the directly prepared probe toward the Bell target. Even in the absence of noise, nonorthogonal code states imply an upper bound
\begin{equation}\label{eq:benchmark-overlap-ceiling}
    F^{\star}\leq F\left(\rho_A,I_K/K\right)
    =\frac{\left(\operatorname{Tr}\sqrt{S^\psi}\right)^2}{K^2} .
\end{equation}
Here $F$ denotes the squared fidelity; in the $K=2$ case, the bound is $\left(1+\sqrt{1-s^2}\right)/2$, where $s:=\abs{\braket{\psi_0}{\psi_1}}$.

It is important to note the following in connection with our benchmarking operational task:
\begin{enumerate}
    \item From a purely coding-theoretic perspective, the optimal way to encode the qu$K$it would be to map each value $i$ to a corresponding member of \emph{an orthonormal basis} spanning the code space, rather than to the literal code state $\ket{\psi_i}$. But this would undermine the operational significance of the code states as the central primitives of our code construction, relegating their role to simply defining a subspace. Instead, we use the ``literally-encoded'' qu$K$it--mode probe state $\ket{\Gamma_K}$ as the initial state, with the express intention to privilege the literal code state states in our benchmarking.
    \item We take the preparation of this probe state for granted, ignoring both the practical challenge of preparing it and the details of its operational interpretation. On the latter note: it is tempting to interpret the state as the result of an ``encoding map'' acting on a pure logical--logical state (the $\ket{\Phi_K}$ defined above), but this fails because non-orthogonal code states don't admit such a map acting locally on only one of the qu$K$its and producing a pure probe state like $\ket{\Gamma_K}$ (orthogonal code states admit an isometry, as noted above). Success-heralded probabilistic preparation is possible, but its success probability is not included in eqn.\ \eqref{eq:logical_entanglement_fidelity}. Here we content ourselves with noting that (1) this preparation-related difficulty applies alike for both our candidate codes and their reference counterparts, and (2) the overall optimized fidelity of recovery back to the logical value $i$ is adversely affected by the very same structure---viz.\ code state non-orthogonality---as is the preparation success rate; as such, trends in the two quantities can reasonably be expected to correlate.
    \item As with the preparation, we ignore the practical challenges in implementing the recovery channel $\mathcal{R}$. This is standard practice in the theoretical literature on bosonic QEC \cite{Fletcher_2006,Albert2018,Jain2023}.
\end{enumerate}

The above benchmarking scheme applies in full generality, but in the results we report here, we restrict to
\begin{itemize}
    \item $K=2$ (i.e.,\ the logical information is a qubit);
    \item single-mode codes; and
    \item $\mathcal{N}\equiv\mathcal{N}_\gamma$ a pure-loss channel with loss probability $\gamma$.
\end{itemize}   

For a single bosonic mode, the pure-loss channel acts as $\mathcal{N}_\gamma \left(\cdot\right) = \sum
\limits_{\ell\in\mathbb Z_+} E_{\ell}\left(\cdot\right) E_{\ell}^{\dagger}$ with the Kraus operators (see Appendix \ref{sec:beam-splitter-derivation} for a derivation)
\begin{equation}
    E_{\ell} = \left(\frac{\gamma}{1-\gamma} \right)^{\ell/2}\frac{\hat{a}^{\ell}}{\sqrt{\ell !}}\left(1-\gamma\right)^{\hat{n}/2}.
\end{equation}

The optimization problem in \eqref{eq:logical_entanglement_fidelity} can be cast as an SDP, following the channel-optimization approach of Ref.\ \cite{Fletcher_2006}. Let $J_{\mathcal R}$ be the
unnormalized recovery Choi matrix in output--input order and set
$X=J_{\mathcal R}^{\mathsf T}/K$, with full transpose in that product basis; also define $M
=\left(\operatorname{id}\otimes\mathcal{N}_{\gamma}\right)
\left(\ket{\Gamma_K}\bra{\Gamma_K}\right)$.
Then the Bell overlap is $\operatorname{Tr}\left(MX\right)$ and the SDP is
\begin{equation}
\begin{aligned}
F^{\star}
= &&\max_{X} &&&\operatorname{Tr}\left(MX\right) \\
&&\text{s.t.} &&& X \succeq 0, \\
&&&&& \operatorname{Tr}_{A} X = I_B/K,
\end{aligned}
\label{eq:benchmark-recovery-sdp}
\end{equation}
where the subscripts $A$ and $B$ in the last condition refer to the first and second subsystems, respectively. The partial trace here is over the logical output. The factor $1/K$ is absorbed
in $X$; using the unscaled Choi matrix would instead require moving this
factor into the objective.

We consider $\ell\in\left\{6,12,18\right\}$, except for Hex1 and Elliptical Hex1, whose nontrivial variants\footnote{The only smaller Hex1 codes, with $\ell=6$, are subsumed in the QSC family, while the only smaller elliptical Hex1 codes, also with $\ell=6$, are contained within the Ellipse family.} require a minimum $\ell=18$. QSC references are circular single-mode $\ell$-gon constructions; QCC references are available at $\ell=12,18$. Where Elliptical Hex family is deformed Hex family such that it is anisotropic (with $\left(\tau_x,\tau_y, \tau_{xy}\right)=\left(2,1,0.5\right)$ used for benchmarking as presented in the following section).
The geometric definitions of Hex1, Hex2, and Hex3 are given in Section~\ref{sec:hexagonal_codes}, specifically eqns.\ \eqref{eqn:hex1_family},~\eqref{eqn:hex2_family}, and~\eqref{eqn:hex3_family}.
Their elliptical versions use the same normalized point patterns with the anisotropic target matrix in eqn.\ \eqref{eqn:hex_target_map}.
The labels Ellipse, Hex1, Hex2, Hex3, and the elliptical Hex variants used in the figures and tables denote the corresponding BARC code subfamilies.
Fig.\ \ref{fig:Code_family_tree.drawio.png} shows the collection of codes being benchmarked and their grouping into families.

\begin{figure}[H]
    \centering
    \includegraphics[width=0.75\linewidth]{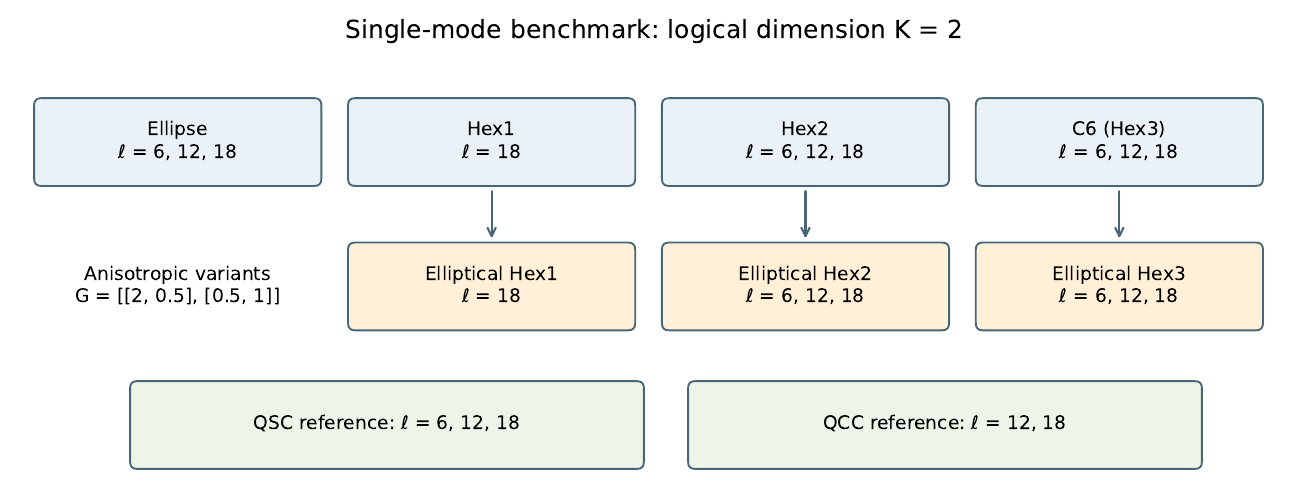}
    \caption{Benchmarked families and available constellation sizes. Arrows indicate the anisotropic hexagonal variants; references are the specific QSC and QCC constructions used here.} 
    \label{fig:Code_family_tree.drawio.png}
\end{figure}

In the remainder, we use $F^\star_C$ to mean the figure of merit defined in eqn.\ \eqref{eq:benchmark-recovery-sdp}, applied to a code family $C$ (the latter being inferred from context). In order to benchmark the {BARC code families} against corresponding QSCs and QCCs, we use the following methods of comparison.

\subsubsection{Equal-size comparisons at optimal amplitude}
\label{sec:comparison-method-1-desc}
For each geometrical configuration described in Table \ref{tab:equalell-geometry}, we pick a convenient scale-standardized constellation $\tilde{\mathcal C}$. A particular instance is then obtained by scaling this representative by a parameter $\alpha\in\mathbb R_+$ to yield a constellation $\mathcal C=\alpha\tilde{\mathcal C}$.

For every code family $C$ that we benchmark, we first independently optimize the scale parameter $\alpha$ over the range $\mathcal A_C:=\left[0.5,10\right]$ at reference loss $\gamma_0=0.095$, denoting the optimal value
\begin{equation}
 \alpha_C^\star=\underset{\alpha\in\mathcal A_C}{\operatorname{argmax}}\,
 F_{C}^\star\left(\alpha,\gamma_0\right),
 \label{eq:equalell-optimum}
\end{equation}
Then, with $\alpha$ fixed at this value, the loss parameter $\gamma$ is varied from [0.01, 0.5], and $F^\star$ is plotted as a function of $\gamma$. Note that we do not optimize $\alpha$ tailored to each distinct value of $\gamma$; this approach has precedent in the literature \cite{Yang2025}.

In order to quantify the relative performance of our new code to its reference counterpart, the latter is chosen to match the former's constellation size $\ell$ (i.e.,\ number of coherent states per code state) and logical dimension $K$. Following \cite{Yang2025}, we then define the metric 

For candidate $C$ and reference $B$ of equal $\ell$, we define  
\begin{equation}
 R_{C|B}\left(\gamma\right)=\frac{1-F_B^{\star}\left(\gamma\right)}{1-F_C^{\star}\left(\gamma\right)}.
 \label{eq:equalell-ratio}
\end{equation}

Where $R_{C|B}\left(\gamma\right)=1$ signifies candidate code performs same as reference code. $R_{C|B}\left(\gamma\right)>1$ signifies candidate code performs better than reference code at the corresponding value of $\gamma$. And $R_{C|B}\left(\gamma\right)<1$ signifies the candidate code performs worse than the reference code at the corresponding values of $\gamma$

This first method therefore compares each code at its own preferred sampled operating point, and its results are reported in Section~\ref{sec:method1-results}. The unequal photon budgets establish the operating-point baseline against which the candidate photon-budget results can be read.

\subsubsection{Equal-size comparisons under the candidate photon budget}
\label{sec:comparison-method-2-desc}
Rather than imposing equal energy universally, in this method we ask how well the reference can perform without exceeding the candidate's photon budget.
We retain each candidate's sampled optimum from
Eq.\ \eqref{eq:equalell-optimum}, but allow its reference to use less than the
candidate's mean photon number. This avoids forcing a reference away from
its independently optimized operating point merely to consume the full
candidate budget. The constellation size $\ell$ and logical dimension $K=2$ remain equal within each pair. QSC and QCC denote the particular constructions
in Table~\ref{tab:equalell-geometry}~\cite{Jain2023,Yang2025}.

The reference amplitude is selected according to the following procedure:
\begin{equation}
\begin{aligned}
E_C &\leftarrow \bar n_C\!\left(\alpha_C^\star\right),
&
\varepsilon_C &\leftarrow 10^{-9}\!\left(1+\abs{ E_C}\right),
\\
E_B &\leftarrow \bar n_B\!\left(\alpha_B^\star\right).
\end{aligned}
\end{equation}

\begin{equation}
\begin{aligned}
\textbf{if}\quad
E_B &\leq E_C+\varepsilon_C:
&
{\alpha^{\star}}_{B|C}
&\leftarrow \alpha_B^\star,
\\[1mm]
\textbf{else:}\quad
&
&
{\alpha^{\star}}_{B|C}
&\leftarrow
\text{the solution of }
\bar n_B\!\left({\alpha^{\star}}_{B|C}\right)=E_C.
\end{aligned}
\end{equation}

Thus the reference uses no more than the candidate's budget, up to the
energy tolerance. Its selected amplitude is fixed across the entire loss
sweep. The corresponding ratio is similar to the relation defined in \ref{eq:equalell-ratio}, after the energy budget choice.
\begin{equation}
 R^{\left(\leq E\right)}_{C|B}\left(\gamma\right)=
 \frac{1-F_{B}^\star\left(\alpha^{\star}_{B|C},\gamma\right)}
      {1-F_{C}^\star\left(\alpha_C^\star,\gamma\right)}.
 \label{eq:energy-budget-ratio}
\end{equation}

Here, $C$ denotes the candidate code and $B$ the equal-size reference code. The quantities
$F_{\mathrm{B},C}^\star$ and $F_{\mathrm{B},B}^\star$ are the corresponding
optimized entanglement fidelities at loss probability $\gamma$. Thus,
$R^{\left(\leq E\right)}_{C|B}$ is the ratio of reference to candidate infidelity, with
values greater than one favoring the candidate.

Thus, a cheaper independently optimized reference is retained; only a reference whose independent optimum exceeds the candidate budget is rescaled to an amplitude that matches the reference code's energy to the candidate code. Its results are reported in Section~\ref{sec:method2-results}.

The following two subsections present the results in the order of the methods defined in Sections~\ref{sec:comparison-method-1-desc} and~\ref{sec:comparison-method-2-desc}, first at independently optimized amplitudes and then under the candidate photon budget.

\subsection{Equal-size comparisons at optimal amplitude}
\label{sec:method1-results}

Unless otherwise stated, comparisons in this section use each code's independently optimized amplitude following \ref{sec:comparison-method-1-desc}, so photon budgets differ between candidate and reference. The selected amplitudes and corresponding photon numbers are listed in Table~\ref{tab:equalell-parameters} for all benchmarked codes. We show the Hex3, Elliptical Hex3 and Ellipse loss sweeps here; the Hex2, Hex1, Elliptical Hex2, and Elliptical Hex1 loss sweeps are collected in Appendix~\ref{sec:appendix-hex-loss-sweeps-method1}, and the relative-infidelity curves are collected in Appendix~\ref{sec:appendix-method1-extra-figures}. For referring to a code with code states consisting of  $\ell$ points with-in a family, the notation family-name-$\ell$ will be used (for example: Hex1-12 for code with 12 points per code state belonging to Hex1 family).

Under this method of comparison, QSC and QCC references generally have higher fidelity near the optimization point ($\gamma_0=0.095$), while the Hex3 and its Elliptical Hex3 consistently have higher fidelity compared to QCCs after a certain value of loss probability ($\gamma$).

\subsubsection{Elliptical Hex families}
\label{subsec:method1-elliptical-hex}

\begin{figure}[H]
\centering
\includegraphics[width=0.495\linewidth]{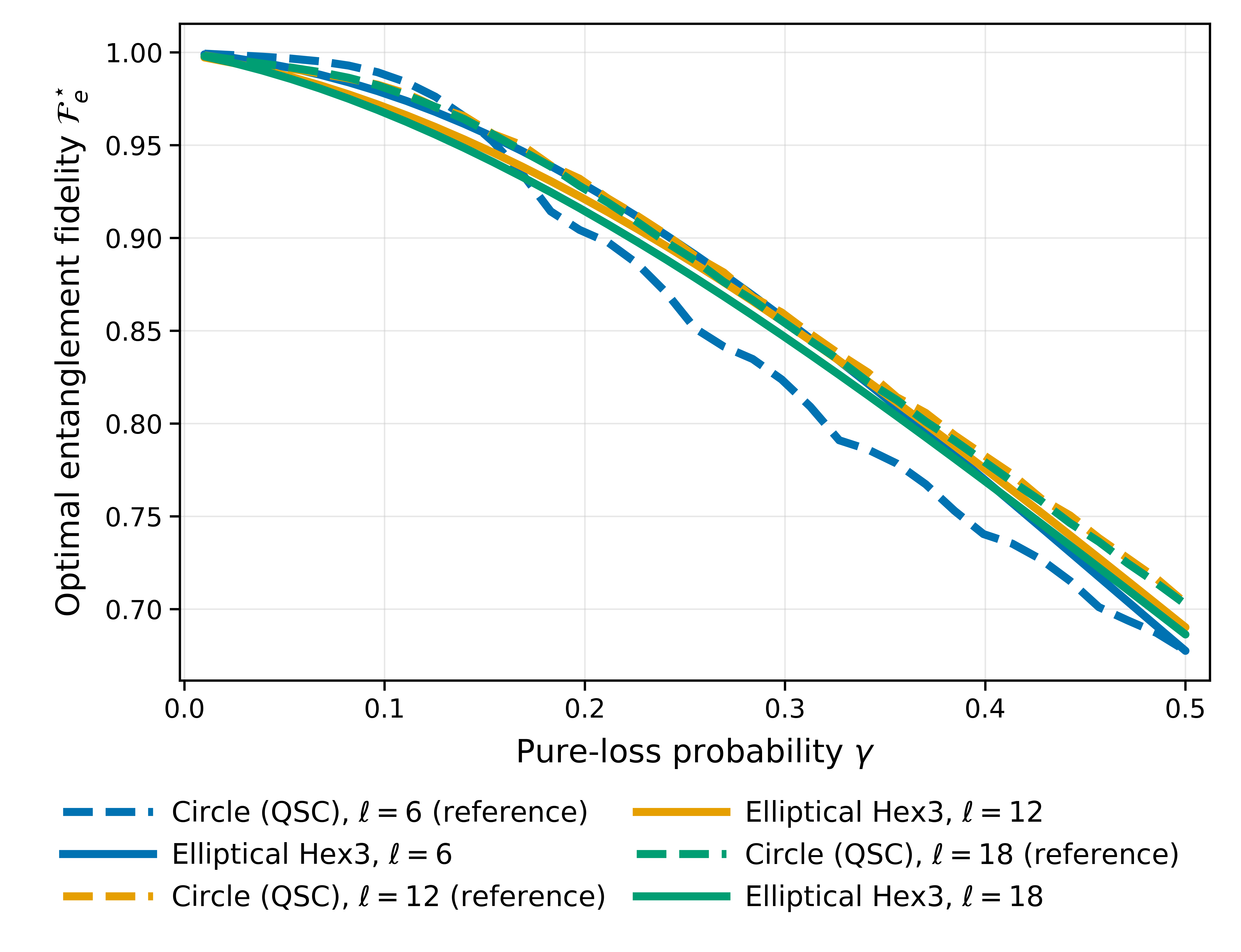}\hfill%
\includegraphics[width=0.495\linewidth]{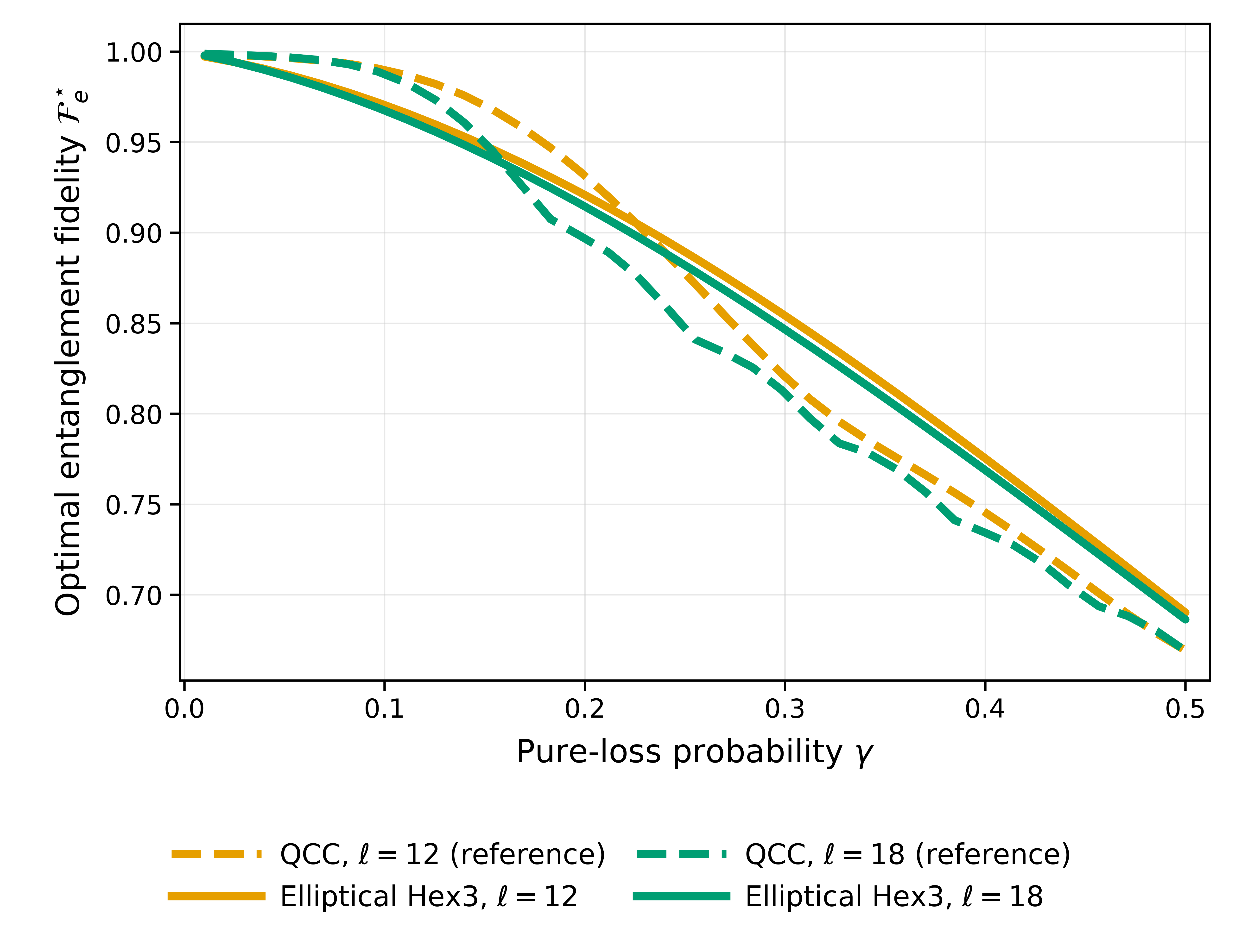}
\caption{Fixed-amplitude loss sweeps for Elliptical Hex3. The candidate is compared with QSC (left) and QCC (right). Each code uses its own sampled optimum at $\gamma_0=0.095$, held fixed across the loss sweep, while recovery is reoptimized at each loss probability. The Elliptical Hex2 and Elliptical Hex1 loss sweeps are provided in Section~\ref{sec:appendix-hex-loss-sweeps-method1}.}
\label{fig:method1-elliptical-hex3-loss}
\end{figure}

Elliptical Hex3-$6$ demonstrates the broadest QSC-relative gain in this group, with higher fidelity for approximately $\gamma=0.15$ onwards, while Elliptical Hex3-$12$ and -$18$ have closely comparable performance (Fig.\ \ref{fig:method1-elliptical-hex3-loss}). Elliptical Hex2-$6$ also has slightly higher fidelity compared to QSC-6 once loss is moderate ($\gamma \approx 0.25$), whereas larger Elliptical Hex2 and Elliptical Hex1 show no evident advantage over QSCs (Fig.\ \ref{fig:method1-elliptical-hex2-loss}) and  (Fig.\ \ref{fig:method1-elliptical-hex1-loss}). Every Elliptical Hex construction uses fewer photons than its QSC reference; Note that the Elliptical Hex2-$18$ and Elliptical Hex1 optima occur at the end of the sampled range for amplitude scale (Table~\ref{tab:equalell-parameters}).

Elliptical Hex3 produces the strongest QCC-relative result in this group, combining substantially lower photon use with higher fidelity at both available sizes for $\gamma\gtrsim0.23$ (Fig.\ \ref{fig:method1-elliptical-hex3-loss}; Table~\ref{tab:equalell-parameters}). Elliptical Hex2-$12$ and Elliptical Hex1 have a higher comparative fidelity over limited intermediate-loss regions, while QCC leads Elliptical Hex2-$18$ throughout the sweep (Figs.\ \ref{fig:method1-elliptical-hex2-loss} and~\ref{fig:method1-elliptical-hex1-loss}). The mean photon number for $12$-point Elliptical Hex2, the $18$-point Elliptical Hex2 and the selected Elliptical Hex1 at optimal amplitude is reported in Table~\ref{tab:equalell-parameters}.
\FloatBarrier

\subsubsection{Hex families}
\label{subsec:method1-hex}

\begin{figure}[H]
\centering
\includegraphics[width=0.495\linewidth]{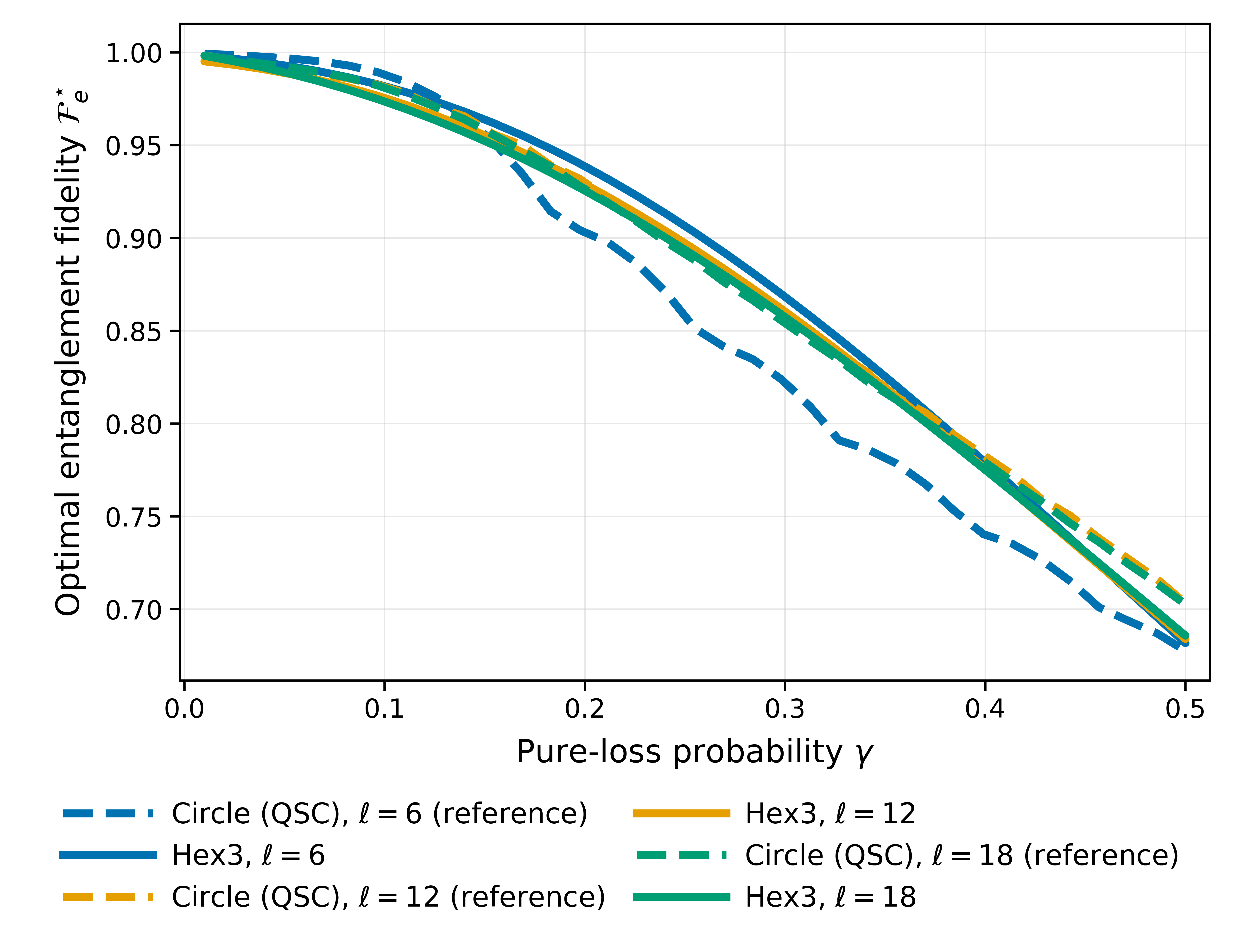}\hfill%
\includegraphics[width=0.495\linewidth]{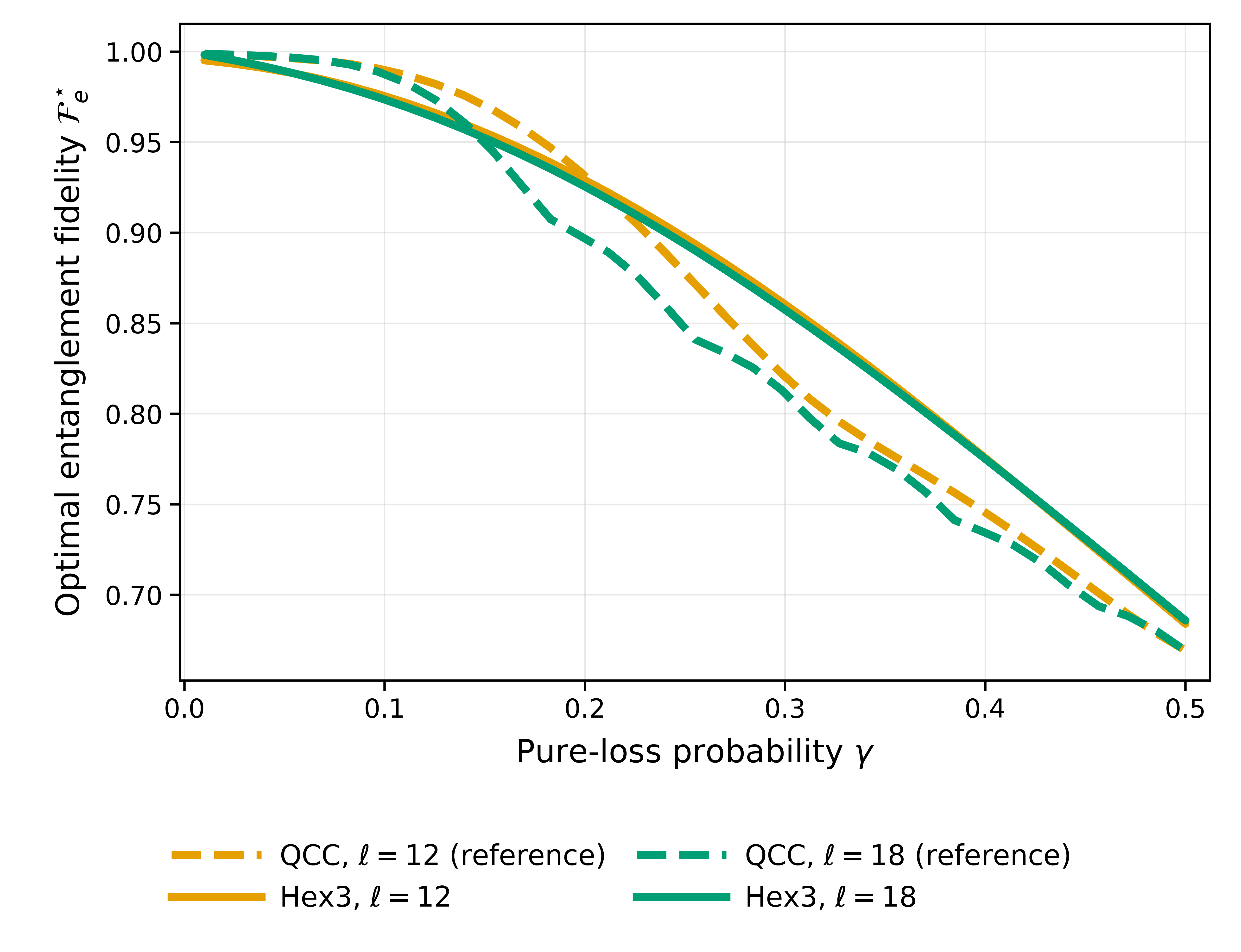}
\caption{Fixed-amplitude loss sweeps for Hex3. The candidate is compared with the equal-size QSC (left) and QCC (right). Candidate and reference amplitudes are selected independently at $\gamma_0=0.095$ and held fixed across the loss scan, while recovery is reoptimized at each loss probability. The Hex2 and Hex1 loss sweeps are provided in Section~\ref{sec:appendix-hex-loss-sweeps-method1}.}
\label{fig:method1-hex3-loss}
\end{figure}

Hex3-$6$ outperforms QSC for $\gamma\gtrsim0.14$ and remains better through the end of the sampled range. Hex3-$12$ and Hex3-$18$ outperform QSC only over narrow intermediate-loss regions (Fig.\ \ref{fig:method1-hex3-loss}). Hex2-$6$ performs better than QSC only at high loss, while QSC outperforms the larger Hex2 codes (Fig.\ \ref{fig:method1-hex2-loss}) and all Hex1 codes (Fig.\ \ref{fig:method1-hex1-loss}). At their independently optimized amplitudes, all Hex codes use fewer photons than the corresponding QSCs (Table~\ref{tab:equalell-parameters}).

Hex3-12 and -18 both outperform QCC for $\gamma\gtrsim0.21$, while using substantially fewer photons (Fig.\ \ref{fig:method1-hex3-loss}; Table~\ref{tab:equalell-parameters}). Hex2 outperforms QCC at intermediate loss, particularly for $\ell=18$, where the two codes use similar photon numbers (Fig.\ \ref{fig:method1-hex2-loss}; Table~\ref{tab:equalell-parameters}). Hex1 also outperforms QCC at intermediate loss. The two codes use nearly the same number of photons (Fig.\ \ref{fig:method1-hex1-loss}; Table~\ref{tab:equalell-parameters}).
\FloatBarrier

\subsubsection{Ellipse family}
\label{subsec:method1-ellipse}

The Ellipse family is obtained using $\left(\tau_x,\tau_y, \tau_{xy}\right)=\left(2,1,0.5\right)$ as described in section \ref{sec:ellipsoid_constellations}.

\begin{figure}[H]
\centering
\includegraphics[width=0.495\linewidth]{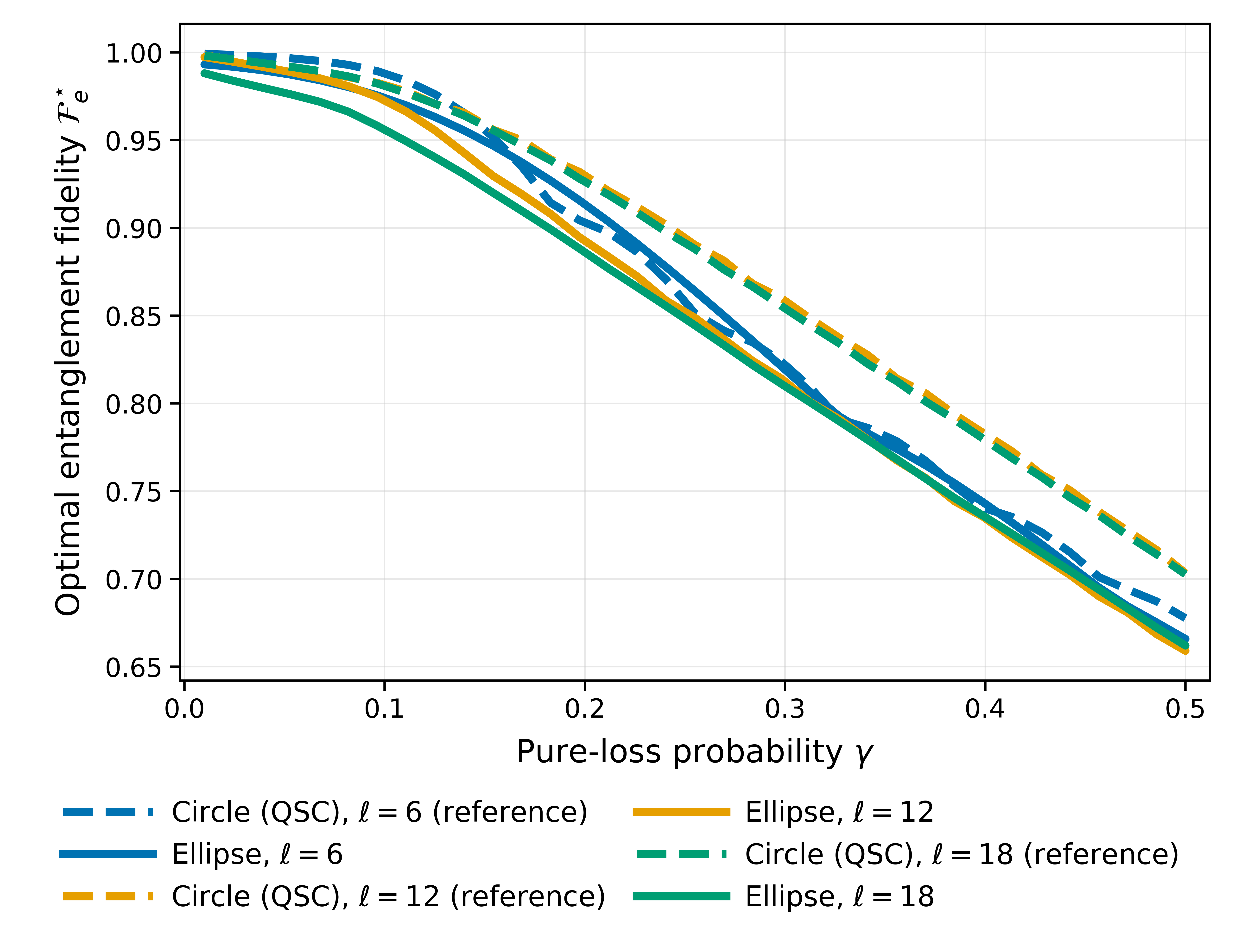}\hfill%
\includegraphics[width=0.495\linewidth]{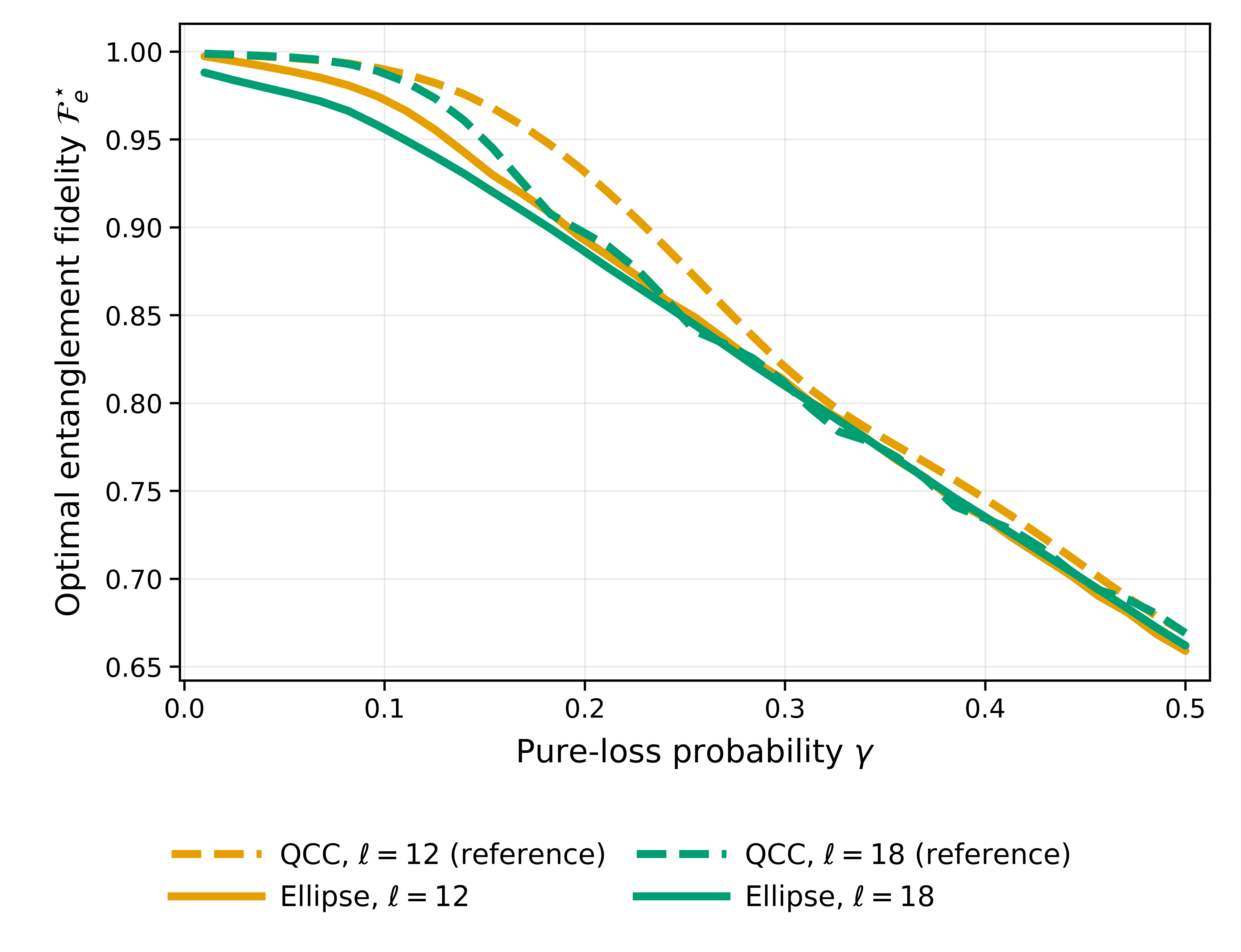}
\caption{Fixed-amplitude loss sweeps for the Ellipse family at equal point count. Left: Ellipse versus QSC for $\ell=6,12,18$. Right: Ellipse versus QCC for $\ell=12,18$. Each code uses its own best sampled amplitude at $\gamma_0=0.095$, held fixed as loss varies, while recovery is reoptimized at each sampled loss probability. Solid curves denote Ellipse codes and dashed curves the corresponding references.}
\label{fig:method1-ellipse-loss}
\end{figure}

Ellipse-$6$ exhibits a modest advantage over QSC-$6$ within limited intermediate-loss intervals. The larger Ellipse constellations do not outperform their corresponding QSC references over the sampled loss range, despite operating at broadly comparable mean photon numbers (Figs.\ \ref{fig:method1-ellipse-loss} and~\ref{fig:appendix-method1-ellipse-qsc}).

Against the QCC references, the Ellipse family exhibits no robust performance advantage over the sampled loss range.

\FloatBarrier

\subsection{Equal-size comparisons under the candidate photon budget}
\label{sec:method2-results}

Comparisons in this section retain each candidate at its independently optimized amplitude and select the reference under the candidate photon budget (\ref{eq:energy-match}) following the candidate photon-budget method defined in Section~\ref{sec:comparison-method-2-desc}. We present the Elliptical Hex3 and Hex3 sweeps here; the corresponding Elliptical Hex1, Elliptical Hex2, Hex1, and Hex2 sweeps are presented in Section~\ref{sec:appendix-hex-loss-sweeps-method2}. Hex3 and Elliptical Hex3 outperform the corresponding equal-size QSC and QCC reference throughout $0.01\leq\gamma\leq0.5$, across all sampled constellation sizes. 

\FloatBarrier

\subsubsection{Elliptical Hex families}
\label{subsec:method2-elliptical-hex}

\begin{figure}[H]
\centering
\includegraphics[width=0.495\linewidth]{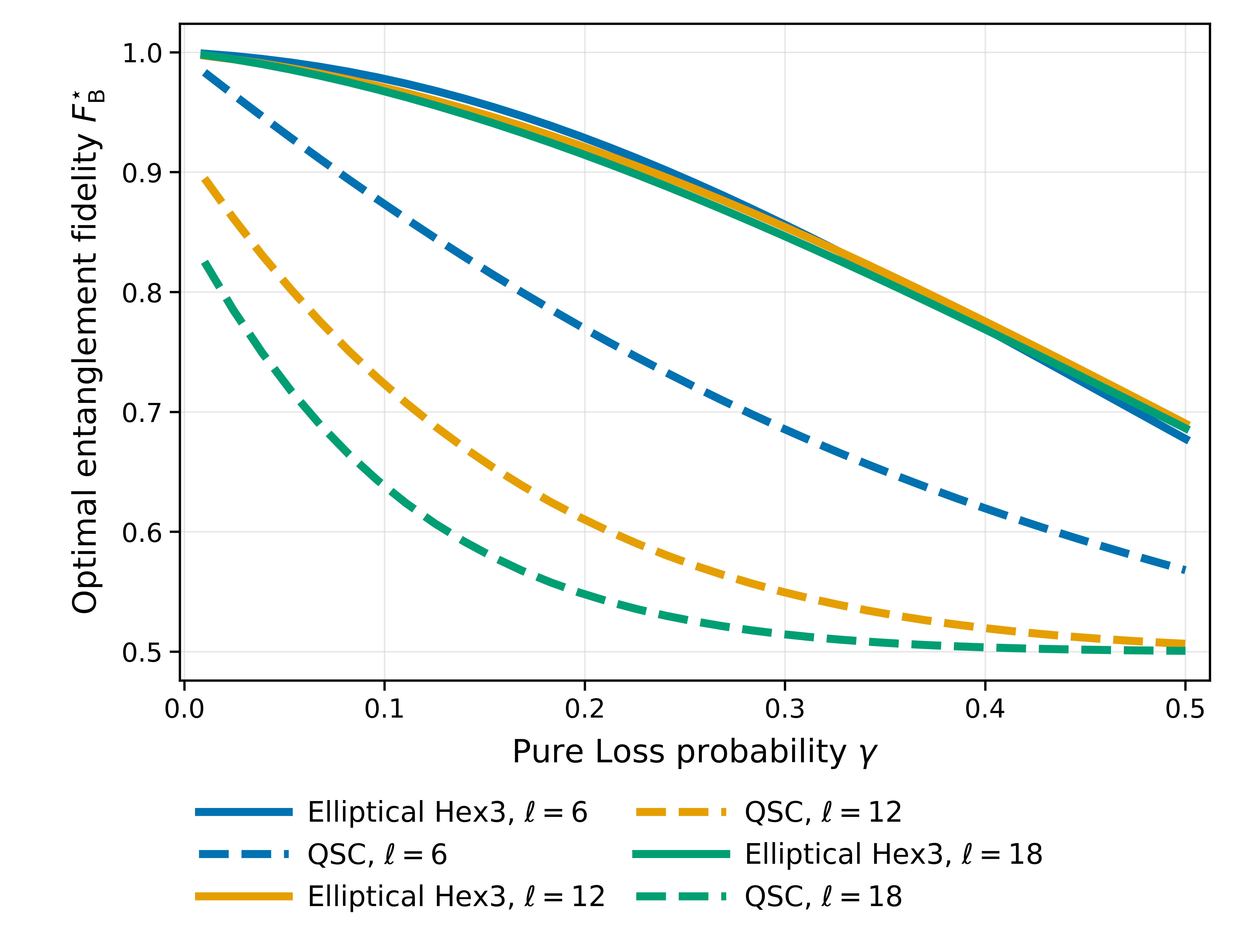}\hfill%
\includegraphics[width=0.495\linewidth]{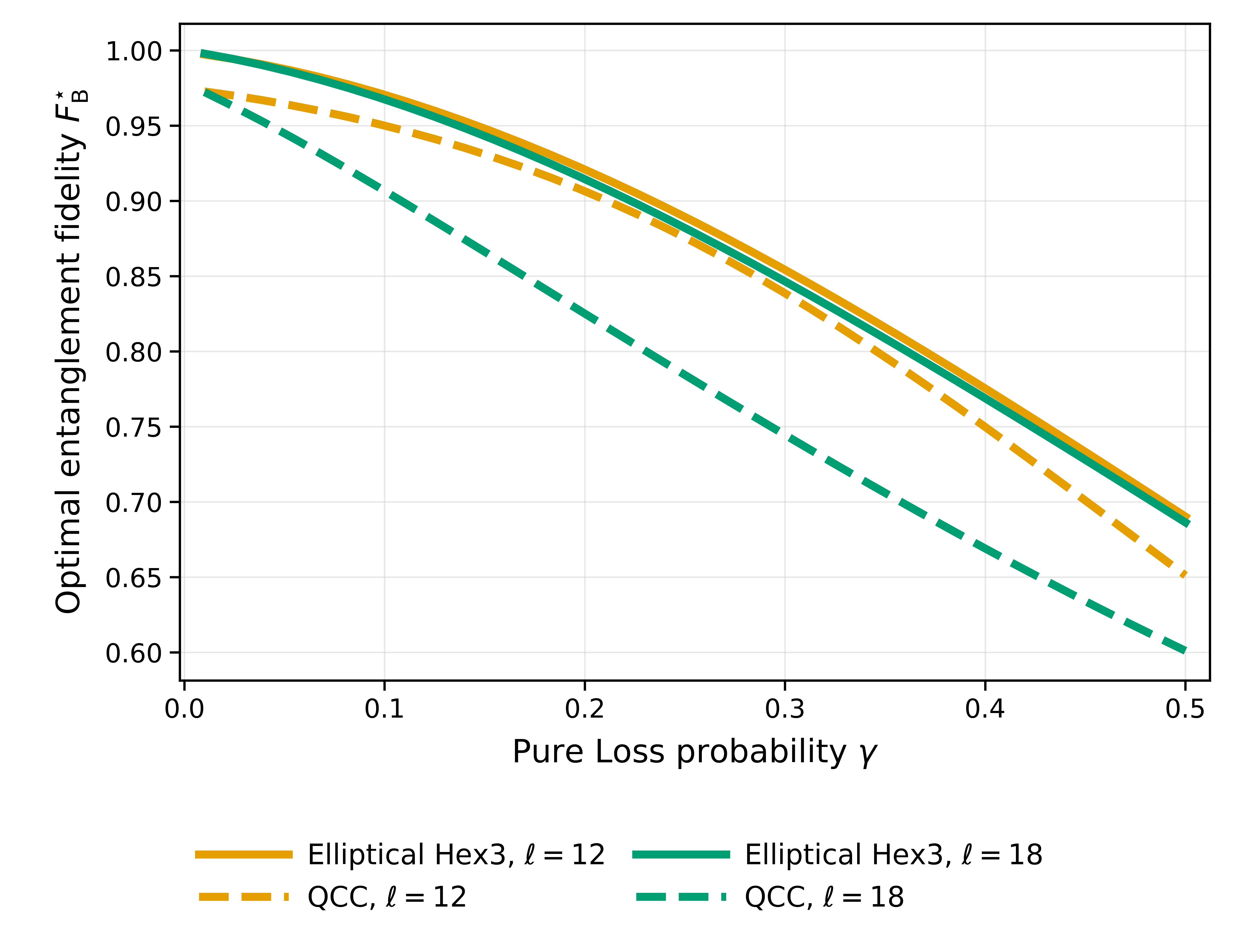}
\caption{Photon-budget fixed-amplitude loss sweeps for Elliptical Hex3. The candidate is compared with QSC (left) and QCC (right) at equal constellation size, with the reference selected under the candidate photon budget. Solid curves denote the candidate and dashed curves the selected reference. The Elliptical Hex2 and Elliptical Hex1 loss sweeps are provided in Section~\ref{sec:appendix-hex-loss-sweeps-method2}.}
\label{fig:method2-elliptical-hex3-loss}
\end{figure}

For Elliptical Hex3, following the comprison method in \ref{sec:comparison-method-2-desc} lowers the QSC mean photon number considerably in order to match that of the Elliptical Hex3 codes (Table~\ref{tab:equalell-parameters}). At these matched photon numbers, Elliptical Hex3 outperforms QSC at every size throughout $0.01\leq\gamma\leq0.5$ by a considerable margin (Fig.\ \ref{fig:method2-elliptical-hex3-loss}). Elliptical Hex2-$12$ and Elliptical Hex2-$18$ have higher fidelity  throughout the sweep, the six-point variant showcases advantage only at low loss (Fig.\ \ref{fig:method2-elliptical-hex2-loss}), and Elliptical Hex1 has higher fidelity throughout (Fig.\ \ref{fig:method2-elliptical-hex1-loss}). Thus, constraining the photon number for QSCs to match those of the new Elliptical Hex3 codes, highlights the performance advantage of the new codes.

For Elliptical Hex3, the QCC mean photon number was more than that of Elliptical Hex3 and was rescaled to match that of Elliptical Hex3 (Table~\ref{tab:equalell-parameters}). At this energy matched comparison Elliptical Hex3 outperforms QCC throughout the sampled range of $\gamma$ at both sizes (Fig.\ \ref{fig:method2-elliptical-hex3-loss}). The photon numbers for Elliptical Hex2 and Elliptical Hex1 are comparable to those of corresponding QCCs,  These two code families show regime-dependent advantage over QCC in some cases (Figs.\ \ref{fig:method2-elliptical-hex2-loss} and~\ref{fig:method2-elliptical-hex1-loss}). While Elliptical Hex2-$12$ has higher fidelity over an intermediate-loss region. 

\FloatBarrier

\subsubsection{Hex families}
\label{subsec:method2-hex}

\begin{figure}[H]
\centering
\includegraphics[width=0.495\linewidth]{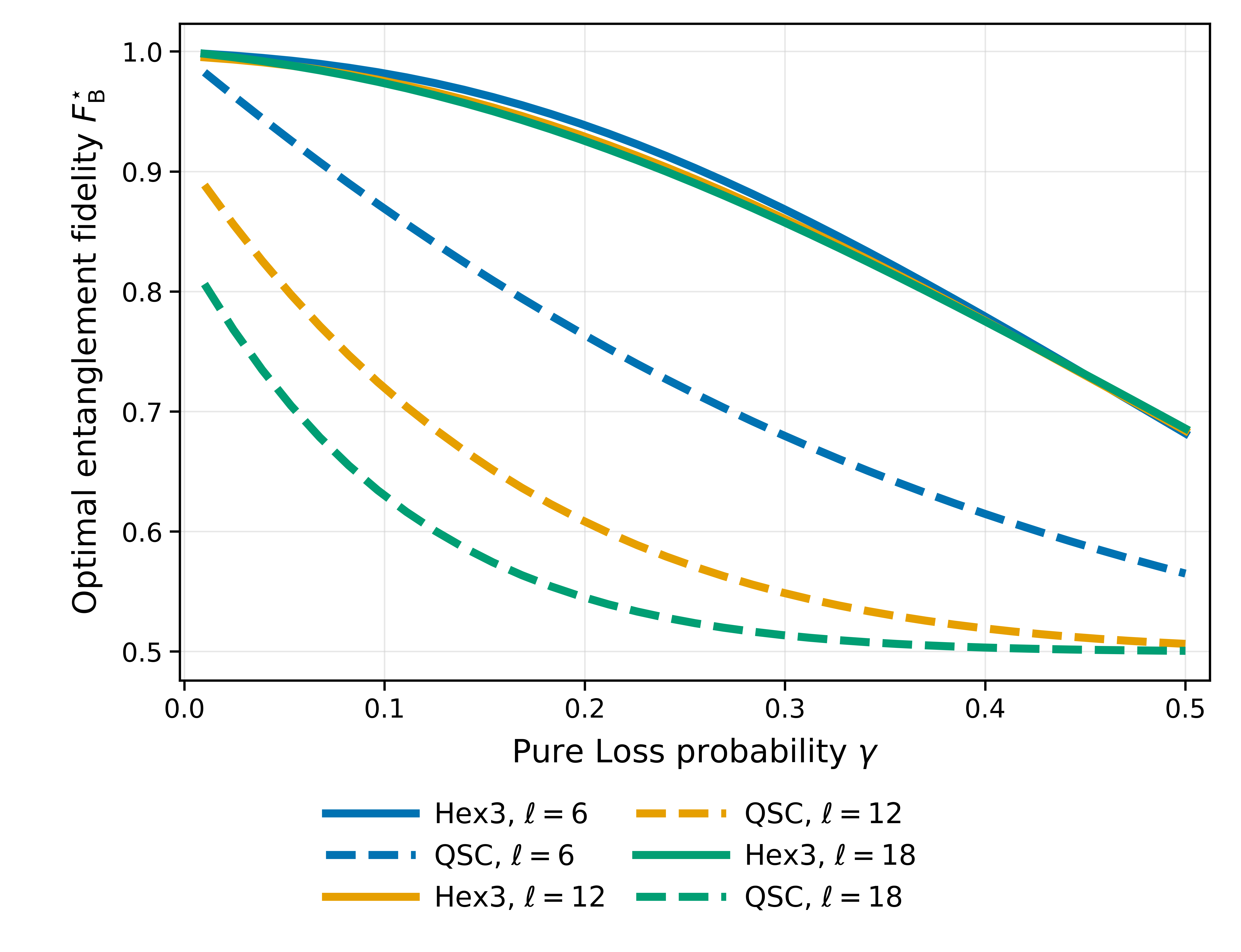}\hfill%
\includegraphics[width=0.495\linewidth]{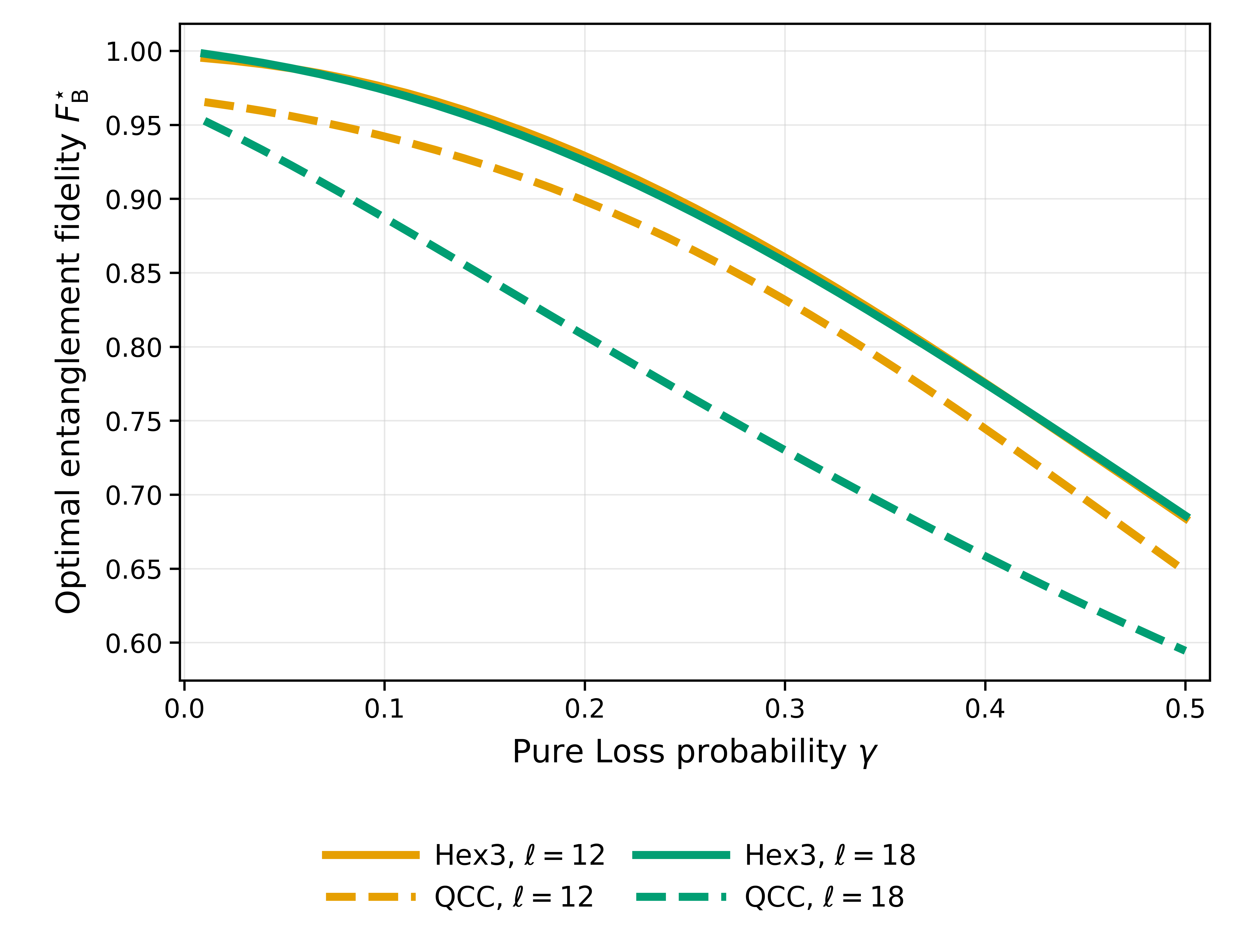}
\caption{Photon-budget fixed-amplitude entanglment fidelity under pure loss for Hex3. The candidate is compared with the equal-size QSC (left) and QCC (right); for these Hex3 pairs, the selected references are energy matched to the candidates. Each amplitude is held fixed throughout the loss sweep while recovery is reoptimized at every sampled loss probability. The Hex2 and Hex1 loss sweeps are provided in Section~\ref{sec:appendix-hex-loss-sweeps-method2}.}
\label{fig:method2-hex3-loss}
\end{figure}

All Hex family codes have lower mean photon number compared to QSCs, and on comparison following reduction of mean photon number of QSC: At these matched photon numbers, Hex3 outperforms QSC at all values of $\ell$ throughout the sampled loss range (Fig.\ \ref{fig:method2-hex3-loss}). Hex2 (Fig.\ \ref{fig:method2-hex2-loss}) and Hex1 (Fig.\ \ref{fig:method2-hex1-loss}) also outperform QSC throughout the sampled range at every available size. Constraining the QSC photon numbers to match those of the new Hex codes highlights their performance advantage across all three constructions.

For Hex3, the photon number is lower than that of QCC and the QCC is rescaled to match the photon number of Hex3 (Table~\ref{tab:equalell-parameters}). At this energy-matched comparison, Hex3 outperforms QCC throughout $0.01\leq\gamma\leq0.5$ at both sizes ($\ell=12,18$)(Fig.\ \ref{fig:method2-hex3-loss}). The photon numbers of Hex1 and QCC are already comparable; with Hex1 having comparably higher fidelity for a narrow intermediate noise regime (Fig.\ \ref{fig:method2-hex1-loss}). Hex2-$12$ and Hex2-$18$ have a higher average photon number than their corresponding QCC references: Hex2-$12$ shows advantage at high noise, while Hex2-$18$ has advantage at intermediate loss(Fig.\ \ref{fig:method2-hex2-loss}). Thus, energy matching reveals a full-range advantage for Hex3, while Hex1 and Hex2 retain regime-dependent advantages over QCC.
\FloatBarrier

\subsubsection{Ellipse family}
\label{subsec:method2-ellipse}

\begin{figure}[H]
\centering
\includegraphics[width=0.495\linewidth]{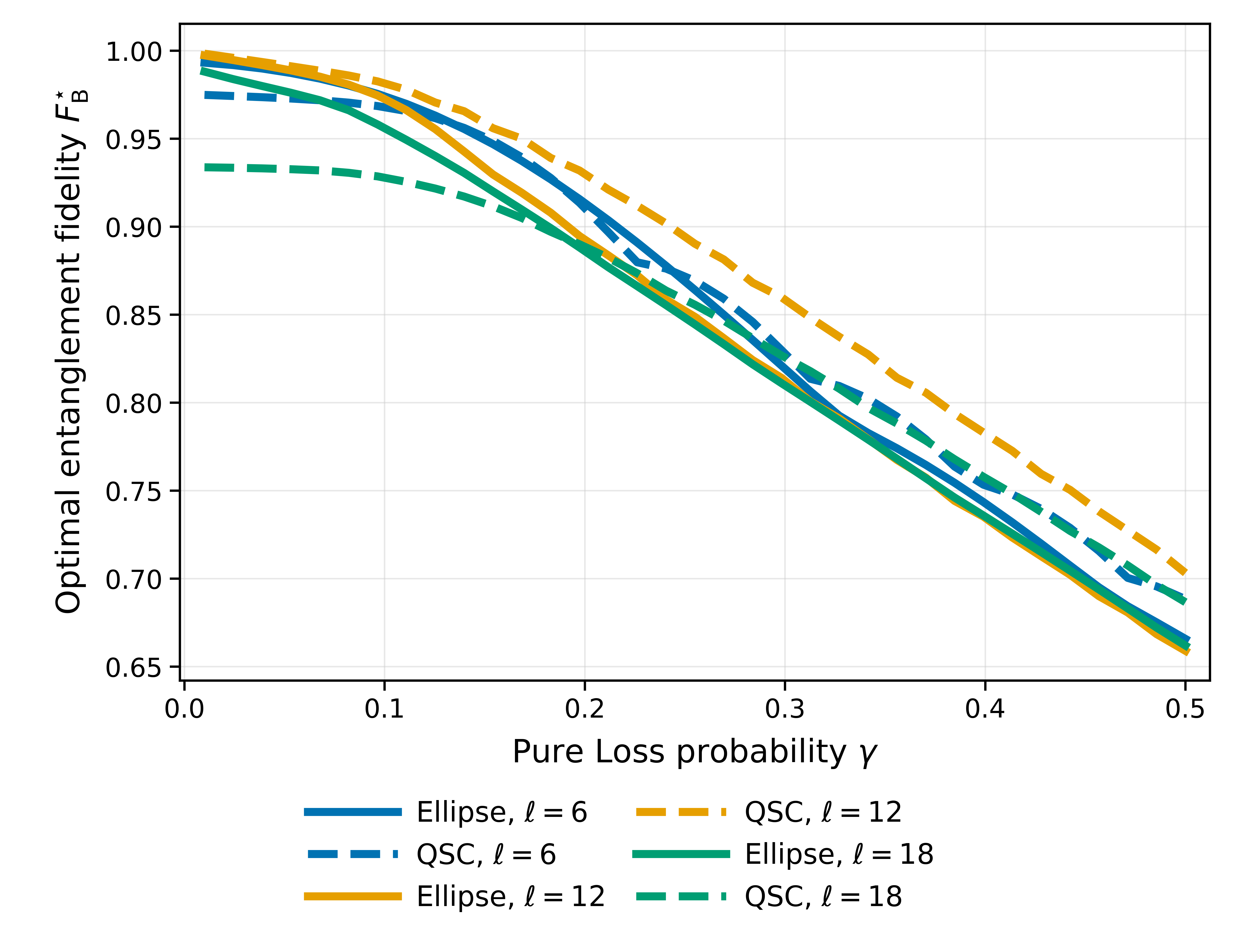}\hfill%
\includegraphics[width=0.495\linewidth]{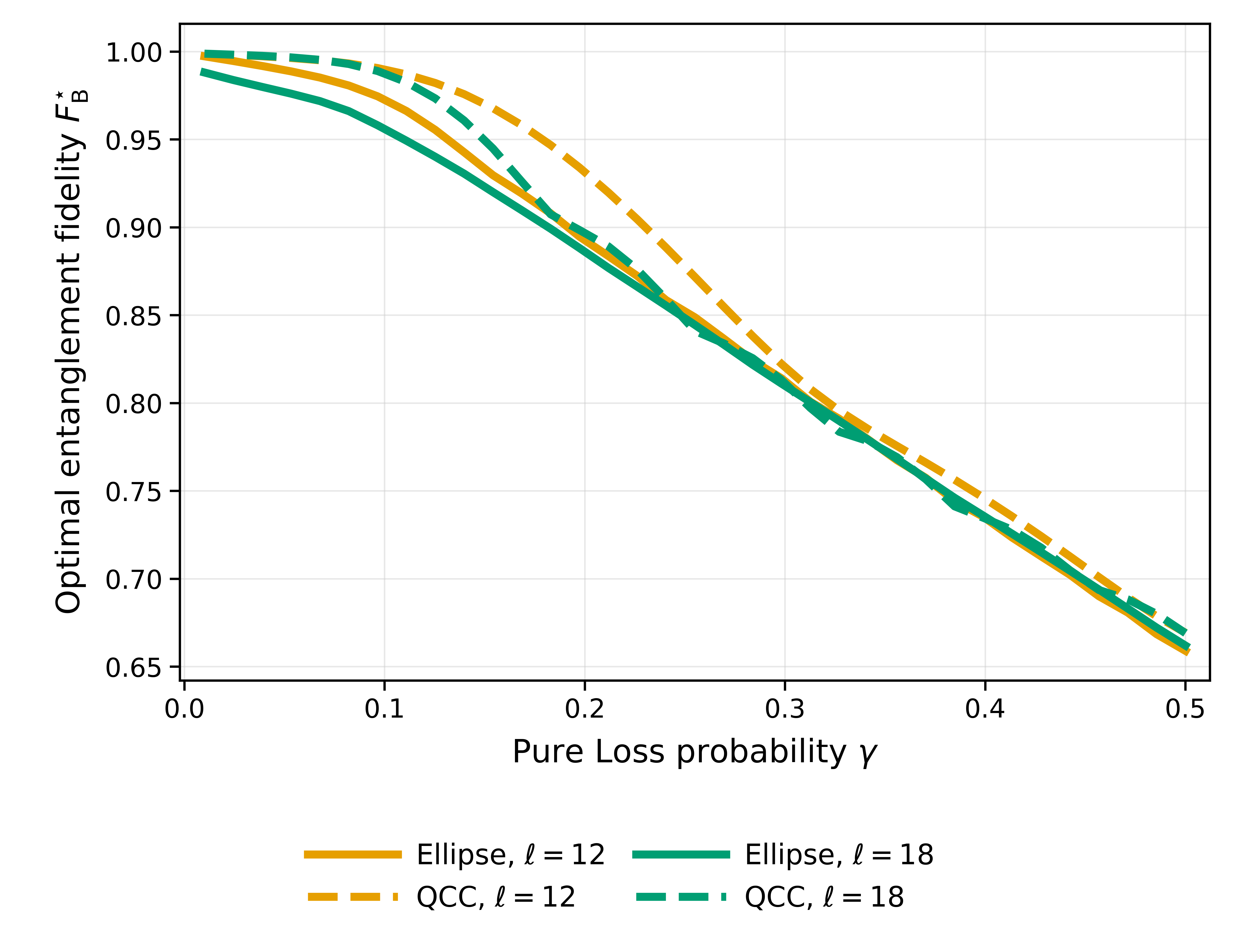}
\caption{Fixed-amplitude entanglement fidelity under pure loss for the Ellipse family under the candidate photon-budget rule. Left: comparison with QSC for $\ell=6,12,18$. Right: comparison with QCC for $\ell=12,18$. The QSC references at $\ell=6,18$ are energy matched, whereas QSC-$12$ and both QCC references retain their lower energy  independently optimized amplitudes. Candidate and reference amplitudes remain fixed across the loss scan while recovery is re-optimized at each sampled $\gamma$.}
\label{fig:method2-ellipse-loss}
\end{figure}

Ellipse codes have a slightly lower mean photon number compared to QSCs (Table~\ref{tab:equalell-parameters}). Ellipse-$6,18$ outperforms QSC only in the low-loss regime up to approximately $\gamma=0.18$. But have a lower fidelity compared to QSC for the rest of the sampled noise range. (Fig.\ \ref{fig:method2-ellipse-loss} and Fig.\ \ref{appendix-method2-ellipse-qsc}).

Both QCC references retain their independently optimized amplitudes as their mean photon number was less than that of the elliptical codes. (Fig.\ \ref{fig:method2-ellipse-loss} and Fig.\ \ref{appendix-method2-ellipse-qcc}). Because the photon-budget constraint does not alter either reference, these results preserve the independently optimized ordering as in \ref{fig:method1-ellipse-loss}
\FloatBarrier

\subsection{Key takeaways from the numerical simulations}

A concise overview of the results section is presented below as tables. 

\begin{table}[H]
\centering
\small
\setlength{\tabcolsep}{4pt}
\begin{tabular}{@{}l ccc ccc@{}}
\hline
\textbf{Code family} & \multicolumn{3}{c}{\textbf{(A) Independent optimization}} & \multicolumn{3}{c}{\textbf{(B) Mean photon budget}} \\
\cline{2-4} \cline{5-7}
 & \textbf{vs.\ QSC} & \textbf{vs.\ QCC} & \textbf{Fig.} & \textbf{vs.\ QSC} & \textbf{vs.\ QCC} & \textbf{Fig.} \\
\hline
Hex3
  & Partial adv. & Partial adv. & \ref{fig:method1-hex3-loss}
  & Strong adv. & Strong adv. & \ref{fig:method2-hex3-loss} \\
Hex2
  & No evident adv. & Partial adv. & \ref{fig:method1-hex2-loss}
  & Strong adv. & Partial adv. & \ref{fig:method2-hex2-loss} \\
Hex1
  & No evident adv. & Partial adv. & \ref{fig:method1-hex1-loss}
  & Strong adv. & Partial adv. & \ref{fig:method2-hex1-loss} \\
Elliptical-Hex3
  & Partial adv. & Partial adv. & \ref{fig:method1-elliptical-hex3-loss}
  & Strong adv. & Strong adv. & \ref{fig:method2-elliptical-hex3-loss} \\
Elliptical-Hex2
  & Partial adv. & No evident adv. & \ref{fig:method1-elliptical-hex2-loss}
  & Partial adv. & No evident adv. & \ref{fig:method2-elliptical-hex2-loss} \\
Elliptical-Hex1
  & No evident adv. & Partial adv. & \ref{fig:method1-elliptical-hex1-loss}
  & Strong adv. & No evident adv. & \ref{fig:method2-elliptical-hex1-loss} \\
Ellipse
  & No evident adv. & No evident adv. & \ref{fig:method1-ellipse-loss}
  & Partial adv. & No evident adv. & \ref{fig:method2-ellipse-loss} \\
\hline
\end{tabular}
\caption{Summary table for comparison of the proposed code families with QSC and QCC reference codes. (A) Independent optimization: code performance is compared after optimizing the amplitude of each code independently, as per \ref{sec:method1-results}. (B) Mean photon budget: code performance is compared at or below the mean photon number of the candidate code, as per \ref{sec:method2-results}. Where, `Strong adv.' means the code consistently has fidelity higher than the reference code throughout the sampled noise range, `Partial adv.' means either some codes in the family show advantage or the family has slightly- to moderately-higher fidelity compared to their reference counterparts in a narrow regime of loss probability $\gamma$, and `No evident adv.' means the family either performs comparable to or worse than the reference codes}
\label{tab:method_results_summary_table}
\end{table}

\section{Conclusions and future directions}\label{sec:conclusions}

In this work we developed a general algebraic framework for constructing bosonic algebraically-restricted constellation (BARC) codes whose logical code states are superpositions of coherent-state constellations over an arbitrary number of modes.
{By expressing the common constellation moments associated with photon loss and gain as coupled polynomial equations, we used equivalent polynomial ideals and moment symmetries to characterize their solution sets and construct the resulting BARC code families (Sections~\ref{sec:coherent_state_codes} and~\ref{sec:constellations}); exact finite-separation correctability additionally requires the coherent-state overlap terms in eqn.\ \eqref{eqn:weighted_KL_exact}.
}
Explicit code constructions, symmetry classifications, and low-order solutions were carried out in the single-mode setting, yielding the hexagonally-packed Hex1, Hex2, and
Hex3 families and the Ellipse family (Sections~\ref{sec:degree_two},~\ref{sec:ellipsoid_constellations},
and~\ref{sec:nonellipsoid_constellations}); the multi-mode formulation was developed
in parallel, with concrete multi-mode constellations left for future work.

A central structural result of this framework is a pair of universal resolution bounds (Section~\ref{sec:resolution_bounds}). 
Using a sphere-packing argument, we showed that any single-mode constellation with common atomic moments through degree $t$ has its resolution $d_E$ --- the minimum squared Euclidean separation between constellation points --- bounded above by an explicit function of the number of constellation points $N=Kl$ and the moments $c_{p,q}$ fixed by the polynomial constraints ({Remark}~\ref{thm:single_mode_resolution_bound}).
We extended this argument to the multi-mode setting by comparing volumes of higher-dimensional balls, giving an analogous universal bound (Theorem~\ref{thm:multimode_resolution_bound}). 
These bounds relate the number of coherent states, centered radial moments, and achievable resolution, independently of the particular geometric family. 
The relation between these moments and the exact constellation-state energy includes the overlap terms discussed in Section~\ref{sec:degree_two}.

Taken together, the sampled comparisons show that
noncircular lattice geometries can provide lower infidelity than the selected
references in well-defined loss regimes. Hex3 and Elliptical Hex3 give the
most consistent advantages across constellation sizes, loss probabilities,
and reference families, whereas the other constructions offer narrower,
regime-dependent improvements in some cases and no particular advantage in others.

Several directions remain open for future work. First, we have not yet constructed the logical Pauli operators for the general BARC code family.
Deriving closed-form logical operators, together with logical gates acting across code blocks, is a natural next step toward fault-tolerant quantum computation. 
{Second, the noise benchmarking in Section~\ref{sec:noise-benchmarking} considers only pure photon loss; extending it to finite-temperature reservoir noise and dephasing remains future work, with the thermal channel to be derived and validated independently of the retained pure-loss treatment in Appendices~\ref{sec:framework} and~\ref{sec:beam-splitter-derivation}.}
Third, Theorem~\ref{thm:multimode_resolution_bound} motivates a systematic search for explicit multi-mode BARC code families and their noise benchmarking, extending the single-mode classification of Section~\ref{sec:degree_two} and the weighted formulation of Section~\ref{sec:weighted_constellation_polynomials}.

Finally, an interesting long-term direction is concatenation with an outer qubit or qudit code, such as a surface code. 
This requires identifying logical operations on the concatenated code and developing fault-tolerant syndrome extraction, state preparation, measurement, and logical gate gadgets.
Their performance should be assessed under errors in the operations and measurements as well as residual errors after bosonic recovery, to determine the conditions under which concatenated BARC codes can support fault-tolerant quantum computation.

\section*{Acknowledgements}
AT is supported by the CQT PhD scholarship, the Google PhD fellowship program, and the CQT Young Researcher Career Development Grant. This project is supported by the National Research Foundation, Singapore, through the National Quantum Office, hosted in A*STAR, under the Quantum Engineering Programme 3.0 (W24Q3D0002), Hybrid Quantum-Classical Computing (HQCC) 1.0 (S24Q7D7001), and Advanced Quantum Algorithms and Solutions (AQAS) (S25Q9DA001) funding initiatives.

\section*{Code availability}

The code for the numerical simulations in this manuscript can be found in the following GitHub repository: \url{https://github.com/IAmPraneel/Noise_Simulations_for_BosonicQEC}

\onecolumngrid
\appendix

\maketitle
\bigskip

\section{Weighted multi-mode constellation polynomials}\label{sec:weighted_constellation_polynomials}

This appendix extends the {BARC code construction} of Section~\ref{sec:constellations} to the weighted multimode states of Section~\ref{sec:coherent_state_codes}.
We first separate the real and imaginary parts of the weighted ladder moments, retaining their explicit binomial expansions.
We then give an invertible transformation to weighted real monomial moments and prove equality of the resulting real polynomial ideals.
The constants $c_{p,q}$ here are normalized weighted moments, whereas the single-mode uniform-weight sums in Section~\ref{sec:constellations} are unnormalized.

Let $p_i=\left(p_{i,1},\dots,p_{i,n}\right)\in\Z_{\geq0}^n$ and $q_i=\left(q_{i,1},\dots,q_{i,n}\right)\in\Z_{\geq0}^n$, and a complex number $c_i\in\C$.
Fix $k\in\left[K\right]$ and $W_k>0$, impose $\sum_{j=1}^l\abs{w_j}^2=W_k$, and suppress the superscript $\left(k\right)$ on the weights and constellation points.
{The exact coherent-state normalization also includes overlap terms and is evaluated after the constellation is chosen.}
{For a weighted constellation $\mathcal{C}_k=\left\{\left(w_j,\alpha_j\right)\right\}_{j=1}^l$ as in Section~\ref{sec:weighted_KL}, write $w_j=r_j+\iota s_j$ and $\alpha_{j,m}=u_{j,m}+\iota v_{j,m}$ and define the polynomial $f_i\in\C\left[r,s,u,v\right]$ by}
\begin{equation}\label{eqn:multimode_ladder_polynomial}
    f_i\left(\mathcal{C}_k\right) = -c_i + \frac{1}{W_k} \sum_{j=1}^l \abs{w_j}^2 \prod_{m=1}^n \left(\alpha_{j,m}^*\right)^{p_{i,m}} \alpha_{j,m}^{q_{i,m}} .
\end{equation}
We will now show how each complex polynomials $f_i$ can be written in terms of its real part $g_i^\Re=\Re\left(f_i\right)$ and imaginary part $g_i^\Im=\Im\left(f_i\right)$.

\begin{proposition}\label{prop:real_polynomial_equations_multimode}
    Write $w_j=r_j+\iota s_j$, $\alpha_{j,m}=u_{j,m}+\iota v_{j,m}$, and $c_i=c_i^\Re+\iota c_i^\Im$ and consider multi-indices $a,b\in\Z_{\geq0}^n$.
    The complex polynomial in eqn.\ \eqref{eqn:multimode_ladder_polynomial} has a real part $g_i^\Re=\Re\left(f_i\right)$ polynomial and an imaginary part $g_i^\Im=\Im\left(f_i\right)$ polynomial of the form:
    \begin{equation}
    \begin{aligned}
        g_i^\Re&= -c_i^{\Re} + \frac{1}{W_k}\sum_{j=1}^l \left(r_j^2+s_j^2\right) \sum_{a,b : \nrm{a+b}_1\text{ even}}
            \left(\prod_{m=1}^n B_{m,a,b}^{\left(i\right)}\right) \left(-1\right)^{\nrm{a}_1+\frac{\nrm{a+b}_1}{2}} \\ &\qquad \left(\prod_{m=1}^n u_{j,m}^{p_{i,m}+q_{i,m}-a_m-b_m}v_{j,m}^{a_m+b_m}\right) \\
        &\text{and}\\
        g_i^\Im &= -c_i^{\Im} + \frac{1}{W_k}\sum_{j=1}^l \left(r_j^2+s_j^2\right) \sum_{a,b : \nrm{a+b}_1\text{ odd}} 
            \left(\prod_{m=1}^n B_{m,a,b}^{\left(i\right)}\right) \left(-1\right)^{\nrm{a}_1+\frac{\nrm{a+b}_1-1}{2}} \\ &\qquad \left(\prod_{m=1}^n u_{j,m}^{p_{i,m}+q_{i,m}-a_m-b_m}v_{j,m}^{a_m+b_m}\right) 
    \end{aligned}
    \end{equation}
    for $B_{m,a,b}^{\left(i\right)} = \binom{p_{i,m}}{a_m}\binom{q_{i,m}}{b_m}$ and where the sum over non-negative integer vectors $a=\left(a_1,\dots,a_n\right)$ and $b=\left(b_1,\dots,b_n\right)$ is from $0^n$ to $\left(p_{i,1},\dots,p_{i,n}\right)$ for $a$ and $0^n$ to $\left(q_{i,1},\dots,q_{i,n}\right)$ for $b$.
\end{proposition}
\begin{proof}
    Let $\alpha=\left(\alpha_1,\dots,\alpha_n\right)$ and $\alpha_m = u_m+\iota v_m$ for some $u_m,v_m\in\R$.
    Then we have
    \begin{equation}
    \begin{aligned}
        \prod_{m=1}^n \left(\alpha_m^*\right)^{p_m} \alpha_m^{q_m} 
        &= \prod_{m=1}^n \left(u_m-\iota v_m\right)^{p_m} \left(u_m+\iota v_m\right)^{q_m} \\
        &= \prod_{m=1}^n \left(\sum_{a_m=0}^{p_m} \binom{p_m}{a_m} u_m^{p_m-a_m}\left(-\iota v_m\right)^{a_m} \right) \left(\sum_{b_m=0}^{q_m} \binom{q_m}{b_m} u_m^{q_m-b_m}\left(\iota v_m\right)^{b_m} \right) \\
        &= \prod_{m=1}^n \sum_{a_m=0}^{p_m} \sum_{b_m=0}^{q_m} \left(-1\right)^{a_m} \iota^{a_m+b_m} 
            \underbrace{ \binom{p_m}{a_m}\binom{q_m}{b_m} }_{=B_{m,a,b}} u_m^{p_m+q_m-a_m-b_m} v_m^{a_m+b_m} .
    \end{aligned}
    \end{equation}
    Now let us denote $a=\left(a_1,\dots,a_n\right)$ and $b=\left(b_1,\dots,b_n\right)$ so that we have
    \begin{equation}
    \begin{aligned}
        \prod_{m=1}^n \left(\alpha_m^*\right)^{p_m} \alpha_m^{q_m} 
        &= \sum_{\substack{a,b\in\Z_{\geq0}^n\\
        0\leq a\leq p,\;0\leq b\leq q}} \left(\prod_{m=1}^n B_{m,a,b}\right) \left(-1\right)^{\nrm{a}_1} \iota^{\nrm{a+b}_1} \left(\prod_{m=1}^n u_m^{p_m+q_m-a_m-b_m} v_m^{a_m+b_m}\right) .
    \end{aligned}
    \end{equation}
    Similar to the proof of Proposition~\ref{prop:real_polynomial_equations_1mode}, consider the fact that $\Re\left(\iota^t\right) = \left(-1\right)^{t/2}$ whenever $t$ is even and $\Re\left(\iota^t\right) = 0$ otherwise, whereas $\Im\left(\iota^t\right) = \left(-1\right)^{\left(t-1\right)/2}$ whenever $t$ is odd and $\Im\left(\iota^t\right) = 0$ otherwise.
    Also, since $w_j=r_j+\iota s_j$, we have $\abs{w_j}^2=r_j^2+s_j^2$.
    Therefore the real part $g_i^\Re=\Re\left(f_i\right)$ and imaginary part $g_i^\Im=\Im\left(f_i\right)$ of $f_i$ are given by
    \begin{equation}
    \begin{aligned}
        g_i^\Re&=
            -c_i^{\Re} + \frac{1}{W_k}\sum_{j=1}^l \left(r_j^2+s_j^2\right) \sum_{a,b : \nrm{a+b}_1\text{ even}}
            \left(\prod_{m=1}^n B_{m,a,b}^{\left(i\right)}\right) \left(-1\right)^{\nrm{a}_1+\frac{\nrm{a+b}_1}{2}} \\ &\qquad \left(\prod_{m=1}^n u_{j,m}^{p_{i,m}+q_{i,m}-a_m-b_m}v_{j,m}^{a_m+b_m}\right) \\
        &\text{and}\\
        g_i^\Im &=
            -c_i^{\Im} + \frac{1}{W_k}\sum_{j=1}^l \left(r_j^2+s_j^2\right) \sum_{a,b : \nrm{a+b}_1\text{ odd}} 
            \left(\prod_{m=1}^n B_{m,a,b}^{\left(i\right)}\right) \left(-1\right)^{\nrm{a}_1+\frac{\nrm{a+b}_1-1}{2}} \\ &\qquad \left(\prod_{m=1}^n u_{j,m}^{p_{i,m}+q_{i,m}-a_m-b_m}v_{j,m}^{a_m+b_m}\right) \\
        &\text{for} \\
        &B_{m,a,b}^{\left(i\right)} = \binom{p_{i,m}}{a_m}\binom{q_{i,m}}{b_m} 
    \end{aligned}
    \end{equation}
    and where the sum over non-negative integer vectors $a,b$ is from $0^n$ to $\left(p_{i,1},\dots,p_{i,n}\right)$ for $a$ and $0^n$ to $\left(q_{i,1},\dots,q_{i,n}\right)$ for $b$.
    Here, the normalization condition is $\sum_{j=1}^l\left(r_j^2+s_j^2\right)-W_k=0$.
\end{proof}

We next extend Lemma~\ref{lem:real_complex_polynomial_coordinate_tranform} from one mode to $n$ modes.
For this purpose, take all pairs $p,q\in\Z_{\geq0}^n$ with $\abs{p}+\abs{q}\leq t$ and write $f_{p,q}$ for eqn.\ \eqref{eqn:multimode_ladder_polynomial} with the corresponding prescribed constant $c_{p,q}$.
All degree bounds in this construction refer to the amplitude variables; when the weights are also variables, they contribute the quadratic factor $r_j^2+s_j^2$.
For multi-indices $a,b\in\Z_{\geq0}^n$, define the normalized weighted real moments and their constraint polynomials by
\begin{equation}\label{eqn:weighted_real_moment_polynomials}
\begin{gathered}
    h_{a,b}^{\left(W\right)}=S_{a,b}^{\left(W\right)}-\tau_{a,b}
    \quad\text{for}\quad
    S_{a,b}^{\left(W\right)}=\frac{1}{W_k}\sum_{j=1}^l\left(r_j^2+s_j^2\right)\prod_{m=1}^n u_{j,m}^{a_m}v_{j,m}^{b_m}.
\end{gathered}
\end{equation}
Since $W_k>0$ is fixed, these are polynomials in the real variables $r,s,u,v$.

For a fixed multi-index $d\in\Z_{\geq0}^n$, let $D_d=\prod_{m=1}^n\left(d_m+1\right)$ and define the $D_d\times D_d$ matrix $\mathsf A^d$ with its entries given by
\begin{equation}\label{eqn:multimode_moment_matrix}
    \mathsf A^d_{p,q;b}=\prod_{m=1}^n A^{d_m}_{p_m,q_m;b_m},
    \qquad p+q=d,\quad 0\leq b\leq d,
\end{equation}
where the coefficients $A^{d_m}_{p_m,q_m;b_m}$ are defined in eqn.\ \eqref{eqn:monomial_moment_constant_complex}.
Throughout the rest of the section, multi-index equalities and inequalities are componentwise.
Let $c^d=\left(c_{p,q}\right)_{p+q=d}$ and $\tau^d=\left(\tau_{d-b,b}\right)_{0\leq b\leq d}$ and order the rows and columns by the product of the single-mode orderings described as follows.
For each mode $m$, order the row indices $\left(p_m,q_m\right)$ with $p_m+q_m=d_m$ as in the single-mode matrix $A^{d_m}$, and order the column indices by $b_m=0,\dots,d_m$.
We order the multimode rows and columns by the corresponding Cartesian-product orderings.

\begin{lemma}[Weighted multimode moment-coordinate transformation]\label{lem:weighted_multimode_moment_transform}
    For every $d\in\Z_{\geq0}^n$, the matrix $\mathsf A^d$ is invertible and
    \begin{equation}\label{eqn:moment_transform_matrix_inverse}
        \mathsf A^d=\bigotimes_{m=1}^n A^{d_m},
        \qquad
        \left(\mathsf A^d\right)^{-1}=\bigotimes_{m=1}^n\left(A^{d_m}\right)^{-1}.
    \end{equation}
    Moreover, $\tau^d=\left(\mathsf A^d\right)^{-1}c^d$ is real if and only if $c_{p,q}=c_{q,p}^*$ for every $p+q=d$.
    With these target moments, the weighted constraint polynomials satisfy
    \begin{equation}\label{eqn:weighted_polynomial_coordinate_change}
        f_{p,q}=\sum_{0\leq b\leq d}\mathsf A^d_{p,q;b}h_{d-b,b}^{\left(W\right)},
        \qquad p+q=d.
    \end{equation}
\end{lemma}
\begin{proof}
    For a fixed point $\alpha=\left(\alpha_m\right)_m$ with $\alpha_m=u_m+\iota v_m$, let $X_m=u_m-\iota v_m$ and $Y_m=u_m+\iota v_m$.
    The coordinate change is invertible in each mode, with
    \begin{equation}
        u_m=\frac{X_m+Y_m}{2},\qquad v_m=\frac{Y_m-X_m}{2\iota}.
    \end{equation}
    For $p+q=d$, apply the single-mode expansion in every mode:
    \begin{equation}
    \begin{aligned}
        \prod_{m=1}^n X_m^{p_m}Y_m^{q_m}
        &=\prod_{m=1}^n\left(\sum_{b_m=0}^{d_m}A^{d_m}_{p_m,q_m;b_m}u_m^{d_m-b_m}v_m^{b_m}\right)\\
        &=\sum_{0\leq b\leq d}\left(\prod_{m=1}^n A^{d_m}_{p_m,q_m;b_m}\right)\prod_{m=1}^n u_m^{d_m-b_m}v_m^{b_m}\\
        &=\sum_{0\leq b\leq d}\mathsf A^d_{p,q;b}\prod_{m=1}^n u_m^{d_m-b_m}v_m^{b_m}.
    \end{aligned}
    \end{equation}
    
    By the definition in eqn.~\eqref{eqn:multimode_moment_matrix}, the matrix entry with row index $\left(\left(p_1,q_1\right),\dots,\left(p_n,q_n\right)\right)$ and column index $\left(b_1,\dots,b_n\right)$ is
    \begin{equation}
        \mathsf A^d_{p,q;b}=\prod_{m=1}^n A^{d_m}_{p_m,q_m;b_m}.
    \end{equation}
    By the entrywise definition of the Kronecker product under these product orderings,
    \begin{equation}
        \left(\bigotimes_{m=1}^n A^{d_m}\right)_{p,q;b}=\prod_{m=1}^n A^{d_m}_{p_m,q_m;b_m}=\mathsf A^d_{p,q;b},
    \end{equation}
    and hence
    \begin{equation}
        \mathsf A^d=\bigotimes_{m=1}^n A^{d_m}.
    \end{equation}
    Each $A^{d_m}$ factor is invertible by Lemma~\ref{lem:real_complex_polynomial_coordinate_tranform}, so their tensor product is invertible with the stated inverse.

    Complex conjugation of the coefficients exchanges $p_m$ and $q_m$ in every mode, giving
    \begin{equation}
        \left(\mathsf A^d_{p,q;b}\right)^*=\mathsf A^d_{q,p;b}.
    \end{equation}
    If $\tau^d$ is real and $c^d=\mathsf A^d\tau^d$, then
    \begin{equation}
        c_{p,q}^*=\sum_{0\leq b\leq d}\left(\mathsf A^d_{p,q;b}\right)^*\tau_{d-b,b}
        =\sum_{0\leq b\leq d}\mathsf A^d_{q,p;b}\tau_{d-b,b}
        =c_{q,p}.
    \end{equation}
    Conversely, suppose $c_{p,q}=c_{q,p}^*$ and set $\tau^d=\left(\mathsf A^d\right)^{-1}c^d$.
    Then
    \begin{equation}
    \begin{aligned}
        \left(\mathsf A^d\left(\tau^d\right)^*\right)_{p,q}
        &=\sum_{0\leq b\leq d}\mathsf A^d_{p,q;b}\tau_{d-b,b}^*
        =\left(\sum_{0\leq b\leq d}\mathsf A^d_{q,p;b}\tau_{d-b,b}\right)^*
        =c_{q,p}^*=c_{p,q}.
    \end{aligned}
    \end{equation}
    Invertibility therefore gives $\left(\tau^d\right)^*=\tau^d$.

    Finally, multiply the monomial expansion for point $j$ by $\left(r_j^2+s_j^2\right)/W_k$ and sum over $j$:
    \begin{equation}
    \begin{aligned}
        f_{p,q}+c_{p,q}
        &=\frac{1}{W_k}\sum_{j=1}^l\left(r_j^2+s_j^2\right)\prod_{m=1}^n\left(u_{j,m}-\iota v_{j,m}\right)^{p_m}\left(u_{j,m}+\iota v_{j,m}\right)^{q_m}\\
        &=\sum_{0\leq b\leq d}\mathsf A^d_{p,q;b}S_{d-b,b}^{\left(W\right)}.
    \end{aligned}
    \end{equation}
    Subtracting $c_{p,q}=\sum_{0\leq b\leq d}\mathsf A^d_{p,q;b}\tau_{d-b,b}$ gives eqn.\ \eqref{eqn:weighted_polynomial_coordinate_change}.
    Thus the weights do not change the coordinate-transformation matrix.
\end{proof}
\begin{proposition}[Equivalent weighted real moment ideals]\label{prop:weighted_multimode_moment_ideal}
    Assume $c_{p,q}=c_{q,p}^*$ for $\abs{p}+\abs{q}\leq t$ and $c_{0,0}=1$, and obtain the real target moments from $\tau^d=\left(\mathsf A^d\right)^{-1}c^d$ for every $\abs{d}\leq t$.
    In $\R\left[r,s,u,v\right]$, the ideals
    \begin{equation}\label{eqn:weighted_multimode_ideals}
    \begin{gathered}
        I_{t,W}^{\left(n\right)}=\left\langle\Re\left(f_{p,q}\right),\Im\left(f_{p,q}\right):\abs{p}+\abs{q}\leq t\right\rangle,\\
        J_{t,W}^{\left(n\right)}=\left\langle h_{a,b}^{\left(W\right)}:\abs{a}+\abs{b}\leq t\right\rangle
    \end{gathered}
    \end{equation}
    are equal and hence have the same real solution set.
    Their degree-zero generator imposes the weight normalization.
\end{proposition}
\begin{proof}
    For $p+q=d$, eqn.\ \eqref{eqn:weighted_polynomial_coordinate_change} and reality of the $h_{a,b}^{\left(W\right)}$ give
    \begin{equation}
    \begin{aligned}
        \Re\left(f_{p,q}\right)&=\sum_{0\leq b\leq d}\Re\left(\mathsf A^d_{p,q;b}\right)h_{d-b,b}^{\left(W\right)},\\
        \Im\left(f_{p,q}\right)&=\sum_{0\leq b\leq d}\Im\left(\mathsf A^d_{p,q;b}\right)h_{d-b,b}^{\left(W\right)}.
    \end{aligned}
    \end{equation}
    These are real linear combinations, so every generator of $I_{t,W}^{\left(n\right)}$ belongs to $J_{t,W}^{\left(n\right)}$.

    Conversely, write $C^d=\left(\mathsf A^d\right)^{-1}$.
    Inverting the same polynomial identity and taking real parts yields
    \begin{equation}
    \begin{aligned}
        h_{d-b,b}^{\left(W\right)}
        &=\sum_{p+q=d}C^d_{b;p,q}f_{p,q}\\
        &=\sum_{p+q=d}\left[\Re\left(C^d_{b;p,q}\right)\Re\left(f_{p,q}\right)-\Im\left(C^d_{b;p,q}\right)\Im\left(f_{p,q}\right)\right].
    \end{aligned}
    \end{equation}
    Hence every generator of $J_{t,W}^{\left(n\right)}$ belongs to $I_{t,W}^{\left(n\right)}$.
    Applying the two inclusions for all $\abs{d}\leq t$ proves equality of the ideals.

    For $d=0$, the transformation matrix is $\left(1\right)$, so $\tau_{0,0}=c_{0,0}=1$ and
    \begin{equation}
        h_{0,0}^{\left(W\right)}=\Re\left(f_{0,0}\right)=\frac{1}{W_k}\left(\sum_{j=1}^l\left(r_j^2+s_j^2\right)-W_k\right),\qquad \Im\left(f_{0,0}\right)=0.
    \end{equation}
    Since $W_k>0$ is fixed, this generator is equivalent to the stated weight-normalization equation.
\end{proof}

The polynomial system depends on the weights through $\abs{w_j}^2$; their phases also enter the exact overlap sums in eqn.\ \eqref{eqn:weighted_KL_exact}.

\section{Spherical-design and cubature constructions of coherent-state codes}\label{sec:qsc_qcc_review}

We review the moment conditions underlying QSCs and QCCs~\cite{Jain2023,Yang2025}, keeping the coherent-state overlap corrections explicit.

\subsection{Spherical designs and QSCs}

A single-mode QSC is specified by pairwise disjoint constellations
$\left\{\mathcal C_j\right\}_j$ on the unit circle and an amplitude scale
$\sqrt N$. Its normalized constellation states are
\begin{equation}
\begin{gathered}
    \ket{\mathcal C_j}
    =\frac{1}{\mathcal N_j\sqrt{\abs{\mathcal C_j}}}
      \sum_{\alpha\in\mathcal C_j}\ket{\sqrt N\alpha} ,\qquad
    \mathcal N_j^2
    =\frac{1}{\abs{\mathcal C_j}}
      \sum_{\alpha,\beta\in\mathcal C_j}
      \braket{\sqrt N\beta}{\sqrt N\alpha} .
\end{gathered}
\end{equation}
For $p,q\in\Z_{\geq0}$, define
$L_{p,q}\left(a^\dag,a\right)=\left(a^\dag\right)^pa^q$.
By eqn.\ \eqref{eqn:KL_Gram}, exact detectability requires
\begin{equation}\label{eqn:knill_laflamme}
    \brakets{\mathcal C_i}{L_{p,q}\left(a^\dag,a\right)}{\mathcal C_j}
    =\lambda_{p,q}\braket{\mathcal C_i}{\mathcal C_j}
    \quad\forall i,j .
\end{equation}
The coherent-state eigenvalue relation gives the exact matrix element
\begin{equation}
    \brakets{\mathcal C_i}{L_{p,q}\left(a^\dag,a\right)}{\mathcal C_j}
    =\frac{N^{\left(p+q\right)/2}}
      {\mathcal N_i\mathcal N_j\sqrt{\abs{\mathcal C_i}\abs{\mathcal C_j}}}
      \sum_{\substack{\beta\in\mathcal C_i\\\alpha\in\mathcal C_j}}
      \left(\beta^*\right)^p\alpha^q\braket{\sqrt N\beta}{\sqrt N\alpha} .
\end{equation}
Let $\mathcal C=\bigcup_j\mathcal C_j$ and $\epsilon_N=\exp\left(-\frac N2 \min_{\alpha,\beta\in\mathcal C : \alpha\neq\beta} \abs{\alpha-\beta}^2\right)$.
For fixed constellations, $\mathcal N_j^2=1+O\left(\epsilon_N\right)$.
Separating equal-amplitude terms from the other terms therefore gives
\begin{equation}
    \brakets{\mathcal C_i}{L_{p,q}\left(a^\dag,a\right)}{\mathcal C_j}
    =\frac{N^{\left(p+q\right)/2}}{\abs{\mathcal C_j}}
      \sum_{\alpha\in\mathcal C_j}\left(\alpha^*\right)^p\alpha^q\,
      \delta_{i,j}
      +O\!\left(N^{\left(p+q\right)/2}\epsilon_N\right) .
\end{equation}
A spherical $t$-design satisfies
\begin{equation}\label{eqn:spherical_t_design}
    \frac{1}{\abs{\mathcal C_j}}
    \sum_{\alpha\in\mathcal C_j}f\left(\alpha^*,\alpha\right)
    =\int_{\S_1}f\left(\alpha^*,\alpha\right)\,d\nu\left(\alpha\right)
\end{equation}
for every polynomial of total degree at most $t$, where $\nu$ is
the uniform probability measure on the unit circle $\S_1$.
Thus spherical designs give the same diagonal ladder moments
for all logical constellations~\cite{Jain2023}. The residual in
eqn.\ \eqref{eqn:knill_laflamme} is
$O\left(N^{\left(p+q\right)/2}\epsilon_N\right)$; exact finite-separation detectability
requires the remaining overlap terms to satisfy that equation.
\subsection{Positive cubature formulas and QCCs}

The QCC construction uses positive cubature formulas~\cite{Yang2025}.
Let $\sigma$ be a probability measure on a common domain $\Omega\subseteq\C^n$, identified with a domain in $\R^{2n}$, and suppose the moments through degree $t$ are finite.
For each logical constellation, choose positive weights $v_{j,\alpha}$ satisfying
\begin{equation}\label{eqn:positive_cubature}
    \sum_{\alpha\in\mathcal C_j}v_{j,\alpha}f\left(\alpha^*,\alpha\right)=\int_\Omega f\left(\alpha^*,\alpha\right)\,d\sigma\left(\alpha\right)
\end{equation}
for every polynomial of total degree at most $t$.
Taking $f=1$ gives $\sum_{\alpha\in\mathcal C_j}v_{j,\alpha}=1$.
The corresponding normalized coherent-state superposition is
\begin{equation}
\begin{gathered}
    \ket{\mathcal C_j}=\frac{1}{\mathcal N_j}\sum_{\alpha\in\mathcal C_j}\sqrt{v_{j,\alpha}}\ket{\alpha},\\
    \mathcal N_j^2=\sum_{\alpha,\beta\in\mathcal C_j}\sqrt{v_{j,\beta}v_{j,\alpha}}\braket{\beta}{\alpha} .
\end{gathered}
\end{equation}
These cubature weights correspond to $v_{j,\alpha}=\abs{w_\alpha}^2/W_j$ in Section~\ref{sec:coherent_state_codes}, with real positive superposition coefficients chosen here.
Applying eqn.\ \eqref{eqn:positive_cubature} to $\prod_{m=1}^n\left(\alpha_m^*\right)^{p_m}\alpha_m^{q_m}$ gives the same weighted moments for every logical state whenever $\abs{p}+\abs{q}\leq t$.
For disjoint supports in the large-separation approximation, this implies $d_\updownarrow\geq t+1$ and $t_\downarrow\geq\left\lfloor t/2\right\rfloor+1$ under the threshold convention of Section~\ref{sec:coherent_state_codes}.
Exact finite-separation protection still requires eqn.\ \eqref{eqn:weighted_KL_exact}.
The spherical-design construction above is recovered with the uniform measure on one sphere and equal weights; QCCs also permit nonuniform weights and multiple energy shells~\cite{Yang2025}.

\section{Quantum-jump framework and Dyson expansion for deriving Kraus operators for bosonic noise channels}
\label{sec:framework}

For the pure-loss channel benchmarked in
Section~\ref{sec:noise-benchmarking}, we specialize the general
quantum-jump framework below to $\bar{n}=0$, so that only the
emission jump contributes.

\subsection{Lindblad master equation and its splitting}

The quantum-jump (or Monte Carlo wavefunction) method was introduced independently by
Dalibard, Castin and M\o{}lmer~\cite{dalibard1992} and by Dum, Zoller and
Ritsch~\cite{dum1992}, and given a systematic operator-theoretic foundation by
Carmichael~\cite{carmichael1993}. The approach unravels the Lindblad master equation into
individual quantum trajectories, each realising a stochastic sequence of instantaneous
``jumps'' separated by smooth non-Hermitian evolution. Summing over all trajectories
recovers the exact density-matrix evolution. A comprehensive treatment of the formalism and
its physical consequences is given in Carmichael's two-volume text~\cite{carmichael_book1,carmichael_book2}
and in Wiseman and Milburn~\cite{wiseman_milburn}; the present note follows those conventions.

The damped harmonic oscillator coupled to a thermal reservoir of mean occupation $\bar{n}$
obeys the Lindblad master equation
\begin{equation}
\dot{\hat{\rho}} = \gamma\left(\bar{n}+1\right)\!\left(\hat{a}\hat{\rho}\hat{a}^{\dagger} - \tfrac{1}{2}\hat{a}^{\dagger}\hat{a}\hat{\rho} - \tfrac{1}{2}\hat{\rho}\hat{a}^{\dagger}\hat{a}\right)
+\gamma\bar{n}\!\left(\hat{a}^{\dagger}\hat{\rho}\hat{a} - \tfrac{1}{2}\hat{a}\hat{a}^{\dagger}\hat{\rho} - \tfrac{1}{2}\hat{\rho}\hat{a}\hat{a}^{\dagger}\right)
\equiv \mathcal{L}\hat{\rho}.
\label{eq:Lindblad}
\end{equation}
The two jump operators are
\begin{equation}
\hat{L}_{-}=\sqrt{\gamma\left(\bar{n}+1\right)}\,\hat{a} \quad\left(\text{emission, energy loss}\right),\qquad
\hat{L}_{+}=\sqrt{\gamma\bar{n}}\,\hat{a}^{\dagger} \quad\left(\text{absorption, energy gain}\right).
\end{equation}
The key idea of the quantum-jump (or Monte Carlo wavefunction) approach is to split the
Lindbladian into two parts:
\begin{equation}
\mathcal{L}\hat{\rho} = \underbrace{-i\hat{H}_{\mathrm{eff}}\hat{\rho} + i\hat{\rho}\hat{H}_{\mathrm{eff}}^{\dagger}}_{\text{non-Hermitian free evolution}}
+\underbrace{\sum_{s\in\left\{+,-\right\}}\hat{L}_s\hat{\rho}\hat{L}_s^{\dagger}}_{\text{quantum jumps}},
\label{eq:split}
\end{equation}
where the effective non-Hermitian Hamiltonian is
\begin{equation}
\hat{H}_{\mathrm{eff}} = -\frac{i}{2}\sum_s\hat{L}_s^{\dagger}\hat{L}_s
= -\frac{i\gamma}{2}\left[\left(2\bar{n}+1\right)\hat{a}^{\dagger}\hat{a} + \bar{n}\right].
\label{eq:Heff}
\end{equation}
The splitting \eqref{eq:split} has a transparent physical interpretation: in any infinitesimal
interval $\left[t,t+dt\right]$, the system either undergoes \emph{no} jump (probability $1 - \nrm{\hat{L}_+\psi}^2 dt
- \nrm{\hat{L}_-\psi}^2 dt$) and evolves under the non-Hermitian $\hat{H}_{\mathrm{eff}}$, or it undergoes
\emph{exactly one} jump of type $s$ (probability $\nrm{\hat{L}_s\psi}^2 dt$) and the state is
updated by $\psi\mapsto\hat{L}_s\psi/\nrm{\hat{L}_s\psi}$.

\subsection{No-jump propagator}

The non-Hermitian part of the evolution generates the \emph{no-jump propagator}
\begin{equation}
\hat{G}\left(t\right) \equiv e^{-i\hat{H}_{\mathrm{eff}} t} = e^{-\gamma\bar{n}t/2}\,e^{-\Gamma t\,\hat{a}^{\dagger}\hat{a}},
\qquad \Gamma\equiv\frac{\gamma\left(2\bar{n}+1\right)}{2}.
\label{eq:Gprop}
\end{equation}
The two factors arise from the two terms in \eqref{eq:Heff}: the $\bar{n}$ piece contributes
the overall scalar decay $e^{-\gamma\bar{n}t/2}$, while the $\hat{a}^{\dagger}\hat{a}$ piece gives the
number-dependent damping $e^{-\Gamma t\hat{a}^{\dagger}\hat{a}}$.
Note that $\hat{G}\left(t\right)$ is \emph{not} unitary but contractive,
$\nrm{\hat{G}\left(t\right)\psi}^2\leq\nrm{\psi}^2$; for normalised $\psi$ the left-hand side is the probability of remaining in the
no-jump sector. Equality can hold: at $\bar n=0$ the vacuum is a dark state, $\hat{G}\left(t\right)\ket{0}=\ket{0}$.

The fundamental BCH commutation identity that will be used throughout is
\begin{equation}
e^{-s\hat{a}^{\dagger}\hat{a}}\,\hat{a} = e^{+s}\,\hat{a}\,e^{-s\hat{a}^{\dagger}\hat{a}},\qquad
e^{-s\hat{a}^{\dagger}\hat{a}}\,\hat{a}^{\dagger} = e^{-s}\,\hat{a}^{\dagger}\,e^{-s\hat{a}^{\dagger}\hat{a}},
\label{eq:BCH}
\end{equation}
which follows from $\left[\hat{a}^{\dagger}\hat{a},\hat{a}\right]=-\hat{a}$ and $\left[\hat{a}^{\dagger}\hat{a},\hat{a}^{\dagger}\right]=+\hat{a}^{\dagger}$. Applied $l$ times to move
$\hat{a}^l$ through a propagator:
\begin{equation}
\hat{a}^l\,\hat{G}\left(t\right) = e^{-\Gamma lt}\,\hat{G}\left(t\right)\,\hat{a}^l,\qquad
\hat{G}\left(t\right)\,\left(\hat{a}^{\dagger}\right)^k = e^{-\Gamma kt}\,\left(\hat{a}^{\dagger}\right)^k\,\hat{G}\left(t\right). \label{eq:BCHiter}
\end{equation}

\subsection{Dyson expansion over ordered trajectories}

Treating the jump part of \eqref{eq:split} as a perturbation on top of no-jump evolution,
the time-evolution superoperator $e^{\mathcal{L} t}$ admits the Dyson series
\begin{equation}
\begin{aligned}
\hat{\rho}\left(t\right) &= \sum_{j=0}^{\infty}\;\sum_{s_1,\dots,s_j\in\left\{+,-\right\}}
\int_0^t\!d\tau_j\int_0^{\tau_j}\!d\tau_{j-1}\cdots\int_0^{\tau_2}\!d\tau_1\;
\hat{A}_{\boldsymbol{s},\boldsymbol{\tau}}\left(t\right)\,\hat{\rho}\left(0\right)\,\hat{A}_{\boldsymbol{s},\boldsymbol{\tau}}^{\dagger}\left(t\right),\\
\hat{A}_{\boldsymbol{s},\boldsymbol{\tau}}\left(t\right)&=\hat{G}\left(t-\tau_j\right)\hat{L}_{s_j}\hat{G}\left(\tau_j-\tau_{j-1}\right)\hat{L}_{s_{j-1}}\cdots\hat{G}\left(\tau_2-\tau_1\right)\hat{L}_{s_1}\hat{G}\left(\tau_1\right).
\end{aligned}
\label{eq:Dyson}
\end{equation}
For $j=0$ the product is $\hat{G}\left(t\right)$, with no integrals and no jump-type sum.

This has a direct physical reading: the $j$-th term of the outer sum collects all histories
in which exactly $j$ quantum jumps occurred, with types $s_1,\dots,s_j$ at times
$0<\tau_1<\cdots<\tau_j<t$. Between consecutive jumps the system evolves freely under
$\hat{G}$. The $j=0$ term is the \emph{no-jump} contribution $\hat{G}\left(t\right)\hat{\rho}\left(0\right)\hat{G}^{\dagger}\left(t\right)$.

Crucially, summing \eqref{eq:Dyson} over all $j$ and all jump types recovers the \emph{exact}
solution of \eqref{eq:Lindblad}: no approximation is made. This is not a perturbative expansion
in jump probability; it is a reorganisation of the exact propagator into a sum over trajectory
sectors.

\paragraph{Goal.} We aim to rewrite \eqref{eq:Dyson} as
\begin{equation}
\hat{\rho}\left(t\right) = \sum_{k,l\ge0}\hat{K}_{kl}\left(t\right)\,\hat{\rho}\left(0\right)\,\hat{K}_{kl}^{\dagger}\left(t\right),
\label{eq:KrausDecomp}
\end{equation}
where $k$ labels the photons removed and $l$ the photons added, so that $\hat K_{kl}$ changes the photon number by $l-k$.
At $\bar n=0$ (Sec.~\ref{sec:T0}) all histories with the same number of emissions have \emph{proportional} ordered products, so their
integrals collapse to a scalar and each jump-number sector yields exactly one Kraus operator.
For $\bar n>0$ this is no longer true: already for one absorption and one emission the ordered words
$\hat a\hat a^{\dagger}=\hat n+1$ and $\hat a^{\dagger}\hat a=\hat n$ are not proportional, and the incoherent sum
$\sum_w\hat A_w\hat\rho\hat A_w^{\dagger}$ over such histories cannot be replaced by the outer product
$\left(\sum_w\hat A_w\right)\hat\rho\left(\sum_w\hat A_w\right)^{\dagger}$ of their coherent sum, which contains spurious cross terms.
The finite-temperature family in Sec.~\ref{sec:finiteT} is therefore obtained by a different route, and its index $\left(k,l\right)$ is not a
jump-number sector of the unravelling \eqref{eq:Dyson}.
Completeness of the Kraus family,
$\sum_{k,l}\hat{K}_{kl}^{\dagger}\hat{K}_{kl}=\mathbb{1}$, is equivalent to the trace-preservation
of the channel.

\subsection{Zero-temperature amplitude damping (\texorpdfstring{$T=0$}{T=0})}
\label{sec:T0}

At $T=0$ one has $\bar{n}=0$, so $\hat{L}_+=0$ and only the emission jump $\hat{L}_-=\sqrt{\gamma}\,\hat{a}$
survives. The propagator simplifies to $\hat{G}_0\left(t\right)=e^{-\gamma t\,\hat{a}^{\dagger}\hat{a}/2}$ and $\Gamma=\gamma/2$.

We write $\eta_t=e^{-\gamma t}$ for the transmissivity.
Here $\gamma$ is a decay rate; the loss probability used in
Section~\ref{sec:noise-benchmarking} is $1-\eta_t$.

\subsubsection{Derivation of the Kraus operators}

For a trajectory with exactly $l$ emission jumps at ordered times $0<\tau_1<\cdots<\tau_l<t$,
the integrand in \eqref{eq:Dyson} on the ket side is
\begin{equation}
\hat{G}_0\left(t-\tau_l\right)\,\hat{a}\,\hat{G}_0\left(\tau_l-\tau_{l-1}\right)\,\hat{a}\cdots\hat{a}\,\hat{G}_0\left(\tau_1\right).
\end{equation}
Using \eqref{eq:BCHiter} in the form $\hat a\,\hat G_0\left(\Delta\right)=e^{-\gamma\Delta/2}\hat G_0\left(\Delta\right)\,\hat a$, we move each $\hat a$ to the
right through all propagators standing to its right. The $r$-th jump passes propagators of total duration $\tau_r$ and therefore picks up
$e^{-\gamma\tau_r/2}$, while the propagators merge into $\hat G_0\left(t\right)$. The ket-side product is thus
$\left[\prod_{r}e^{-\gamma\tau_r/2}\right]\hat G_0\left(t\right)\hat a^l$, and the combined scalar from ket and bra sides is
\begin{equation}
\prod_{r=1}^l e^{-\gamma\tau_r}.
\end{equation}
The resulting contribution to $\hat{\rho}\left(t\right)$ is
\begin{equation}
\hat{\rho}^{\left(l\right)}\left(t\right) = \gamma^l\,\left[
\int_0^t\!d\tau_l\cdots\int_0^{\tau_2}\!d\tau_1\,\prod_{r=1}^l e^{-\gamma\tau_r}
\right]\hat{G}_0\left(t\right)\,\hat{a}^l\,\hat{\rho}\left(0\right)\,\left(\hat{a}^{\dagger}\right)^l\,\hat{G}_0^{\dagger}\left(t\right).
\end{equation}

The nested scalar integral is evaluated by induction. Define
$I_l\left(t\right)=\int_0^t\!d\tau_l\cdots\int_0^{\tau_2}\!d\tau_1\prod_r e^{-\gamma\tau_r}$.
For $l=1$: $I_1=\int_0^t e^{-\gamma\tau}d\tau=\left(1-e^{-\gamma t}\right)/\gamma$.
Assuming the result for $l-1$, the substitution $u=1-e^{-\gamma\tau_l}$ gives
\begin{equation}
I_l\left(t\right) = \frac{1}{\left(l-1\right)!\,\gamma^{l-1}}\int_0^t e^{-\gamma\tau_l}\left(1-e^{-\gamma\tau_l}\right)^{l-1}d\tau_l
= \frac{\left(1-e^{-\gamma t}\right)^l}{l!\,\gamma^l}.
\label{eq:Il}
\end{equation}
Substituting \eqref{eq:Il}:
\begin{equation}
\hat{\rho}^{\left(l\right)}\left(t\right) = \frac{\left(1-e^{-\gamma t}\right)^l}{l!}\;\hat{G}_0\left(t\right)\,\hat{a}^l\,\hat{\rho}\left(0\right)\,\left(\hat{a}^{\dagger}\right)^l\,\hat{G}_0^{\dagger}\left(t\right).
\end{equation}

Reading off the Kraus operators gives
\begin{equation}
\tcboxmath{
K_l^{\left(0\right)}\left(t\right) = \sqrt{\frac{\left(1-e^{-\gamma t}\right)^l}{l!}}\;e^{-\gamma t\,\hat{a}^{\dagger}\hat{a}/2}\,\hat{a}^l,
\qquad l=0,1,2,\dots
}
\label{eq:KrausT0}
\end{equation}

\subsubsection{Completeness}

Acting on a Fock state $\ket{n}$: $\hat{a}^l\ket{n}=\sqrt{n!/\left(n-l\right)!}\,\ket{n-l}$ for $l\le n$
(zero otherwise), then $e^{-\gamma t\hat{a}^{\dagger}\hat{a}/2}\ket{n-l}=e^{-\gamma t\left(n-l\right)/2}\ket{n-l}$. Hence
\begin{equation}
K_l^{\left(0\right)}\left(t\right)\ket{n} = \sqrt{\binom{n}{l}}\left(1-e^{-\gamma t}\right)^{l/2}e^{-\gamma t\left(n-l\right)/2}\ket{n-l},\qquad 0\le l\le n.
\end{equation}
The completeness sum over all $l$ then gives
\begin{equation}
\sum_{l=0}^{\infty}\brakets{n}{K_l^{\left(0\right)\dagger}K_l^{\left(0\right)}}{n}
= \sum_{l=0}^{n}\binom{n}{l}\left(1-e^{-\gamma t}\right)^l e^{-\gamma t\left(n-l\right)}
= \left(1-e^{-\gamma t}+e^{-\gamma t}\right)^n = 1,
\end{equation}
so $\sum_l K_l^{\left(0\right)\dagger}K_l^{\left(0\right)}=\mathbb{1}$.

\subsection{Phase-damping channel}
\label{sec:phase}

\subsubsection{Setup and no-jump propagator}

For a harmonic oscillator coupled to a dephasing reservoir the single jump operator is
$\hat{L}=\sqrt{\kappa}\,\hat{n}$, giving
\begin{equation}
\dot{\hat{\rho}} = \kappa\!\left(\hat{n}\hat{\rho}\hat{n} - \tfrac{1}{2}\hat{n}^2\hat{\rho} - \tfrac{1}{2}\hat{\rho}\hat{n}^2\right).
\end{equation}
The effective Hamiltonian is $\hat{H}_{\mathrm{eff}}^{\left(d\right)}=-i\left(\kappa/2\right)\hat{n}^2$, so the no-jump propagator is
\begin{equation}
S\left(\tau\right) = e^{-i\hat{H}_{\mathrm{eff}}^{\left(d\right)}\tau} = e^{-\kappa\tau\hat{n}^2/2},
\end{equation}
and the $k$-jump sector of the Dyson series \eqref{eq:Dyson} (with $j$ replaced by $k$ for
dephasing) is
\begin{equation}
\hat{\rho}^{\left(k\right)}\left(t\right)=\kappa^k\!\int_0^t\!dt_k\cdots\int_0^{t_2}\!dt_1\;
e^{-\kappa\left(t-t_k\right)\hat{n}^2/2}\,\hat{n}\,e^{-\kappa\left(t_k-t_{k-1}\right)\hat{n}^2/2}\,\hat{n}\cdots\hat{n}\,e^{-\kappa t_1\hat{n}^2/2}
\,\hat{\rho}\left(0\right)\,\left[\cdots\right]^{\dagger}.
\end{equation}

\subsubsection{Collapse of the time integral}

The key algebraic fact is
\begin{equation}
\left[\hat{n},\,\hat{n}^2\right]=0,
\label{eq:ncommute}
\end{equation}
which means the jump operator $\hat{n}$ commutes through every no-jump factor
$e^{-\kappa\tau\hat{n}^2/2}$ freely. As a result:
\begin{enumerate}[leftmargin=1.5em]
\item All $k+1$ no-jump factors of total duration $t$ merge into a single $e^{-\kappa t\hat{n}^2/2}$
on each side of $\hat{\rho}\left(0\right)$.
\item The $k$ jump operators $\hat{n}$ simply pile up adjacently.
\item The time-ordered integral collapses to the volume of the $k$-simplex,
$\int_0^t\!dt_k\cdots\int_0^{t_2}\!dt_1=t^k/k!$.
\end{enumerate}
Therefore
\begin{equation}
\hat{\rho}^{\left(k\right)}\left(t\right) = \frac{\left(\kappa t\right)^k}{k!}\;e^{-\kappa t\hat{n}^2/2}\,\hat{n}^k\,\hat{\rho}\left(0\right)\,\hat{n}^k\,e^{-\kappa t\hat{n}^2/2}.
\end{equation}
Summing over $k=0,1,2,\dots$ gives the Kraus operators
\begin{equation}
\tcboxmath{
\hat{D}_k\left(t\right) = \sqrt{\frac{\left(\kappa t\right)^k}{k!}}\;e^{-\kappa t\hat{n}^2/2}\,\hat{n}^k,\qquad k=0,1,2,\dots
}
\label{eq:KrausPhase}
\end{equation}

For an input Fock state $\ket{n}$ the jump count $k$ is \emph{Poisson} distributed with mean $\kappa t\,n^2$; for a general input it is the corresponding mixture of Poisson distributions, with mean $\kappa t\left\langle\hat n^2\right\rangle$. Completeness follows immediately because $\left[\hat{n},e^{-\kappa t\hat{n}^2}\right]=0$:
\begin{equation}
\sum_{k\ge0}\hat{D}_k^{\dagger}\hat{D}_k = e^{-\kappa t\hat{n}^2}\sum_{k\ge0}\frac{\left(\kappa t\hat{n}^2\right)^k}{k!}
= e^{-\kappa t\hat{n}^2}\cdot e^{\kappa t\hat{n}^2}=\mathbb{1}.
\end{equation}

\paragraph{Contrast with amplitude damping.} The reason for the simplification is the
commutativity \eqref{eq:ncommute}: the jump operator and the non-Hermitian Hamiltonian are
functions of the same operator $\hat{n}$, so they share eigenstates and commute. For amplitude
damping, $\left[\hat{a},\hat{a}^{\dagger}\hat{a}\right]=\hat{a}\ne0$, so $\hat{a}$ does not commute through $e^{-\Gamma\tau\hat{a}^{\dagger}\hat{a}}$;
instead, commuting $\hat{a}$ through the propagator produces a time-dependent scalar factor
$e^{\Gamma\tau}$ (see \eqref{eq:BCH}), and the time-ordered integrals contribute non-trivially.
That is why the amplitude-damping calculation requires the induction \eqref{eq:Il} and yields
binomial rather than Poisson statistics.

\subsection{Finite-temperature amplitude damping}
\label{sec:finiteT}

At finite temperature both jump operators $\hat L_\pm$ are active, and the trajectory sectors of \eqref{eq:Dyson} with fixed numbers of
absorptions and emissions are not each generated by a single operator. We therefore obtain an exact Kraus family from a factorisation of the channel,
built from two one-jump-type channels whose Kraus operators follow from the quantum-jump calculation of Sec.~\ref{sec:T0}.

\subsubsection{Two single-jump-type channels}

\paragraph{Pure loss.} By Sec.~\ref{sec:T0}, evolution under $\gamma'\mathcal D\left[\hat a\right]$ for a time $s$, with transmissivity $\tau=e^{-\gamma' s}$, has Kraus operators
\begin{equation}
\hat B_k\left(\tau\right)=\sqrt{\frac{\left(1-\tau\right)^k}{k!}}\;\tau^{\hat n/2}\,\hat a^k,\qquad k\ge0 .
\label{eq:Bk}
\end{equation}

\paragraph{Pure gain (quantum-limited amplifier).} For $\dot{\hat\rho}=\gamma_g\mathcal D\left[\hat a^{\dagger}\right]\hat\rho$ the only jump operator is $\sqrt{\gamma_g}\,\hat a^{\dagger}$, and
$\hat H_{\rm eff}=-\tfrac{i\gamma_g}{2}\hat a\hat a^{\dagger}=-\tfrac{i\gamma_g}{2}\left(\hat n+1\right)$, so $\hat G_g\left(\Delta\right)=e^{-\gamma_g\Delta\left(\hat n+1\right)/2}$.
Since there is only one jump type, all histories with $l$ jumps again have proportional ordered products. Using
$\hat G_g\left(\Delta\right)\hat a^{\dagger}=e^{-\gamma_g\Delta/2}\,\hat a^{\dagger}\hat G_g\left(\Delta\right)$, moving each $\hat a^{\dagger}$ to the left gives the ket-side product
$\left[\prod_r e^{-\gamma_g\left(s-\tau_r\right)/2}\right]\left(\hat a^{\dagger}\right)^l\hat G_g\left(s\right)$, and the time-ordered integral evaluates as in \eqref{eq:Il}:
$\gamma_g^l\int_{0<\tau_1<\cdots<\tau_l<s}\prod_r e^{-\gamma_g\left(s-\tau_r\right)}=\left(1-e^{-\gamma_g s}\right)^l/l!$. With gain $G\equiv e^{\gamma_g s}$ this gives
\begin{equation}
\hat A_l\left(G\right)=\sqrt{\frac{\left(1-G^{-1}\right)^l}{l!\,G}}\;\left(\hat a^{\dagger}\right)^l\,G^{-\hat n/2},\qquad l\ge0 .
\label{eq:Al}
\end{equation}
Completeness follows from $\sum_l\binom{n+l}{l}x^l=\left(1-x\right)^{-\left(n+1\right)}$ with $x=1-G^{-1}$:
$\sum_l\bra{n}\hat A_l^{\dagger}\hat A_l\ket{n}=G^{-\left(n+1\right)}\sum_l\binom{n+l}{l}\left(1-G^{-1}\right)^l=1$.

\subsubsection{Factorisation of the thermal channel}

Let $\chi_\rho\left(\xi\right)=\mathrm{Tr}\left[\hat\rho\,e^{\xi\hat a^{\dagger}-\xi^*\hat a}\right]$ be the symmetric characteristic function, which determines $\hat\rho$ uniquely.
For $\dot{\hat\rho}=\kappa_-\mathcal D\left[\hat a\right]\hat\rho+\kappa_+\mathcal D\left[\hat a^{\dagger}\right]\hat\rho$, the master equation becomes the first-order equation
$\partial_t\chi=r\,\left(\xi\partial_\xi+\xi^*\partial_{\xi^*}\right)\chi-\tfrac{\kappa_++\kappa_-}{2}\abs{\xi}^2\chi$ with $r=\left(\kappa_+-\kappa_-\right)/2$~\cite{wiseman_milburn,carmichael_book1},
whose solution is
$\chi_t\left(\xi\right)=\chi_0\left(e^{rt}\xi\right)\exp\!\left[-\tfrac{\kappa_++\kappa_-}{4r}\left(e^{2rt}-1\right)\abs{\xi}^2\right]$. This gives:
\begin{itemize}[leftmargin=1.5em]
\item thermal damping \eqref{eq:Lindblad} ($\kappa_-=\gamma\left(\bar n+1\right)$, $\kappa_+=\gamma\bar n$):\quad
$\chi_0\left(\xi\right)\mapsto\chi_0\left(\sqrt{\eta_t}\,\xi\right)\,e^{-\left(1-\eta_t\right)\left(\bar n+\frac12\right)\abs{\xi}^2}$;
\item pure loss $\tau$:\quad $\chi_0\left(\xi\right)\mapsto\chi_0\left(\sqrt{\tau}\,\xi\right)\,e^{-\left(1-\tau\right)\abs{\xi}^2/2}$;
\item pure gain $G$:\quad $\chi_0\left(\xi\right)\mapsto\chi_0\left(\sqrt{G}\,\xi\right)\,e^{-\left(G-1\right)\abs{\xi}^2/2}$.
\end{itemize}
Composing loss followed by gain gives $\chi_0\left(\sqrt{G\tau}\,\xi\right)\exp\left\{-\left[G\left(1-\tau\right)+G-1\right]\abs{\xi}^2/2\right\}$. This equals the thermal map if and only if
$G\tau=\eta_t$ and $2G-1-\eta_t=\left(1-\eta_t\right)\left(2\bar n+1\right)$, that is,
\begin{equation}
G=G_t\equiv1+\left(1-\eta_t\right)\bar n,\qquad \tau=\eta_t/G_t,\qquad
e^{\mathcal L t}=\mathcal A_{G_t}\circ\mathcal B_{\eta_t/G_t},
\label{eq:factorise}
\end{equation}
where $\mathcal B$ and $\mathcal A$ denote the channels with Kraus operators \eqref{eq:Bk} and \eqref{eq:Al}.

\subsection{Exact finite-temperature Kraus family}

The Kraus operators of a composition of channels are the products of their Kraus operators: $\hat K_{kl}=\hat A_l\left(G_t\right)\hat B_k\left(\eta_t/G_t\right)$.
Using $G_t^{-\hat n/2}\left(\eta_t/G_t\right)^{\hat n/2}=\left(\eta_t/G_t^2\right)^{\hat n/2}$, $1-\eta_t/G_t=\left(1-\eta_t\right)\left(\bar n+1\right)/G_t$ and $1-G_t^{-1}=\left(1-\eta_t\right)\bar n/G_t$, we obtain
\begin{equation}
\tcboxmath{
\hat{K}_{kl}\left(t\right) = C_{kl}\;\left(\hat{a}^{\dagger}\right)^l\,e^{-A\left(t\right)\hat{a}^{\dagger}\hat{a}/2}\,\hat{a}^k,
\qquad k,l=0,1,2,\dots
}
\label{eq:KrausFiniteT}
\end{equation}
with
\begin{equation}
C_{kl}\left(t\right) = \sqrt{\frac{N_c^k\,N_s^l}{k!\;l!\;G_t}},
\qquad N_c\equiv\frac{\left(1-\eta_t\right)\left(\bar{n}+1\right)}{G_t},
\quad N_s\equiv\frac{\left(1-\eta_t\right)\bar{n}}{G_t},
\qquad
e^{-A\left(t\right)} = \frac{\eta_t}{G_t^2}.
\label{eq:Ckl}
\end{equation}
Here $k$ counts photons removed (weight $N_c$) and $l$ photons added (weight $N_s$); $\hat K_{kl}$ changes the photon number by $l-k$.
Equivalently,
\begin{equation}
A\left(t\right)=\gamma t+2\ln G_t,\qquad \dot A\left(0\right)=\gamma\left(2\bar n+1\right)=2\Gamma,
\label{eq:A}
\end{equation}
so that at short times $e^{-A\hat n/2}$ reproduces the number-dependent factor $e^{-\Gamma t\hat n}$ of the no-jump propagator \eqref{eq:Gprop}.

\paragraph{$T=0$ limit.} Setting $\bar{n}=0$ gives $G_t=1$, $N_s=0$, $N_c=1-\eta_t$,
and $e^{-A}=e^{-\gamma t}$. Only the $l=0$ operators survive, and \eqref{eq:KrausFiniteT}
reduces exactly to \eqref{eq:KrausT0} (with $k$ in place of $l$).

\paragraph{Long-time thermalisation.} As $t\to\infty$, $\eta_t\to0$, $G_t\to1+\bar{n}$, $N_c\to1$, $N_s\to\bar n/\left(1+\bar n\right)$,
and $e^{-A\hat n/2}\to\ket0\!\bra0$. Hence $\hat K_{kl}\to C_{kl}\sqrt{k!\,l!}\,\ket{l}\!\bra{k}$, and
\begin{equation}
\lim_{t\to\infty}\sum_{k,l}\hat{K}_{kl}\left(t\right)\,\hat{\rho}\left(0\right)\,\hat{K}_{kl}^{\dagger}\left(t\right)
= \sum_{n=0}^{\infty}\frac{\bar{n}^n}{\left(1+\bar{n}\right)^{n+1}}\,\ket{n}\bra{n}
\equiv\hat\rho_{\mathrm{th}},
\end{equation}
the Bose--Einstein thermal state, independently of $\hat{\rho}\left(0\right)$.

\paragraph{Completeness.} Since both factors are complete,
$\sum_{k,l}\hat K_{kl}^{\dagger}\hat K_{kl}=\sum_k\hat B_k^{\dagger}\left(\sum_l\hat A_l^{\dagger}\hat A_l\right)\hat B_k=\sum_k\hat B_k^{\dagger}\hat B_k=\mathbb 1$.

\subsection{Combined amplitude damping and dephasing}
\label{sec:combined}

\subsubsection{Commutativity of the two Lindbladians}

\begin{theorem}
\label{thm:commute}
For every $\bar{n},\gamma,\kappa\ge0$, the superoperators
$\mathcal{L}_a^{\left(T\right)}\hat{\rho}=\gamma\left(\bar{n}+1\right)\mathcal{D}\left[\hat{a}\right]\hat{\rho}+\gamma\bar{n}\mathcal{D}\left[\hat{a}^{\dagger}\right]\hat{\rho}$ and
$\mathcal{L}_d\hat{\rho}=\kappa\mathcal{D}\left[\hat{n}\right]\hat{\rho}$ satisfy $\left[\mathcal{L}_a^{\left(T\right)},\mathcal{L}_d\right]=0$ on operators with finite Fock support (in particular on all Fock dyads $\ket m\!\bra n$), on which both unbounded generators are defined.
\end{theorem}

\begin{proof}
It suffices to show $\left[\mathcal{L}_a^{\left(T\right)},\mathcal{L}_d\right]\left(\ket{m}\bra{n}\right)=0$ for all $m,n$. Direct computation on Fock dyads gives
\begin{align}
\mathcal{L}_d\left(\ket{m}\bra{n}\right) &= -\tfrac{\kappa}{2}\left(m-n\right)^2\,\ket{m}\bra{n},\label{eq:Ldeig}\\
\mathcal{L}_a^{\left(T\right)}\left(\ket{m}\bra{n}\right) &=
\gamma\left(\bar{n}+1\right)\!\sqrt{mn}\,\ket{m{-}1}\bra{n{-}1}
+\gamma\bar{n}\!\sqrt{\left(m{+}1\right)\left(n{+}1\right)}\,\ket{m{+}1}\bra{n{+}1}
+c_0\left(m,n\right)\,\ket{m}\bra{n}.\label{eq:Laact}
\end{align}
Equation~\eqref{eq:Ldeig} shows $\mathcal{L}_d$ is diagonal with eigenvalue
$\mu_d\left(q\right)=-\frac{\kappa}{2}q^2$ depending only on $q\equiv m-n$.
Equation~\eqref{eq:Laact} shows $\mathcal{L}_a^{\left(T\right)}$ maps $\ket{m}\bra{n}$ to three dyads
all sharing the same $q$: $\left(m-1\right)-\left(n-1\right)=m-n=q$ and $\left(m+1\right)-\left(n+1\right)=q$.
Applying $\mathcal{L}_d$ to \eqref{eq:Laact} therefore multiplies every output term by the same
factor $\mu_d\left(q\right)$. Hence $\mathcal{L}_d\mathcal{L}_a^{\left(T\right)}=\mu_d\left(q\right)\mathcal{L}_a^{\left(T\right)}=\mathcal{L}_a^{\left(T\right)}\mathcal{L}_d$.
\end{proof}

\subsubsection{Channel factorisation and product Kraus operators}

Theorem~\ref{thm:commute} with $e^{A+B}=e^Ae^B$ (commuting generators) gives
$\mathcal{E}_t=e^{\left(\mathcal{L}_a^{\left(T\right)}+\mathcal{L}_d\right)t}=\mathcal{E}_t^{\left(a,T\right)}\circ\mathcal{E}_t^{\left(d\right)}\,=\mathcal{E}_t^{\left(d\right)}\circ\mathcal{E}_t^{\left(a,T\right)}$.
Using the second ordering, the Kraus operators of the combined channel are
$\hat{M}_{kl,j}\equiv\hat{D}_j\hat{K}_{kl}$, and
$\sum_{k,l,j}\hat{M}_{kl,j}^\dagger\hat{M}_{kl,j}=\sum_{k,l}\hat{K}_{kl}^\dagger\left(\sum_{j}\hat{D}_j^\dagger\hat{D}_j\right)\hat{K}_{kl}=\mathbb{1}$.

\subsubsection{Complete table of Kraus operators}

Table~\ref{tab:Kraus} collects the Kraus families derived above. The triple $\left(k,l,j\right)$ counts photons removed ($k$), added ($l$;
zero at $T=0$), and phase kicks from the dephasing reservoir ($j$).
The scalar prefactors $C$ of the product Kraus operators factorise into three Poisson-like weights. The outcome probabilities
$\mathrm{Tr}\,\hat M_{kl,j}\hat\rho\hat M_{kl,j}^{\dagger}$ do not factorise in this way in general, since they depend on the photon-number distribution of the input.
As explained in Sec.~\ref{sec:finiteT}, $\left(k,l\right)$ labels the loss and gain stages of the factorisation \eqref{eq:factorise}, not counts of emission and absorption jumps along a trajectory.

\subsection{Zero-temperature amplitude damping as dominant noise source in circuit QED platforms}
\label{loss_channel}

Notation: in this subsection only, $\kappa$ denotes the amplitude-damping rate (called $\gamma$ elsewhere in this appendix) and $\Gamma_\varphi$ the pure-dephasing rate; in Sec.~\ref{sec:phase}, $\kappa$ is the dephasing rate.

In high-quality-factor superconducting cavities, the pure dephasing rate
$\Gamma_\varphi$ is typically orders of magnitude smaller than the
amplitude-damping rate $\kappa$, so $T_2$ approaches its loss-limited maximum,
$T_2 \approx 2T_1$ \cite{Reagor2016,Reagor2013}. This asymmetry arises because
dephasing requires a coupling that shifts the mode's own frequency, probing
the noise spectrum at zero frequency, whereas amplitude damping is an
energy-exchanging process that probes the bath at the mode's finite frequency
$\omega_0$. A bare linear cavity has no low-frequency ``handle'' for the
environment: unlike a Josephson qubit, whose frequency depends on a material
parameter (charge or flux) sensitive to $1/f$ noise \cite{Rigetti2012}, a
cavity's frequency is fixed purely by its geometry, so residual losses couple
only to the field operator itself, giving $\kappa = \omega_0/Q$. For example,
Milul \textit{et al.} measured a single-photon relaxation time
$T_1 = 25.6 \pm 0.2$~ms in a niobium 3D cavity, while bounding the cavity's
own intrinsic dephasing time to $T_\varphi > 0.5$~s --- over an order of
magnitude longer, confirming that essentially all residual decoherence
($T_2 = 34 \pm 1$~ms) came from amplitude damping and transmon-induced
heating rather than intrinsic dephasing \cite{Milul2023}.
Although circuit-QED experiments operate at dilution-refrigerator temperatures, bosonic modes can exhibit higher effective temperatures due to imperfect thermalization, residual radiation, and coupling to ancillary modes. Finite-temperature amplitude damping is therefore described by
\begin{equation}
\dot{\rho}
=
\kappa\left(\bar{n}_{\mathrm{th}}+1\right)\mathcal{D}\left[a\right]\rho
+
\kappa\bar{n}_{\mathrm{th}}\mathcal{D}\left[a^\dagger\right]\rho,
\end{equation}
with
$\Gamma_{\uparrow}/\Gamma_{\downarrow}
=
\bar{n}_{\mathrm{th}}/\left(1+\bar{n}_{\mathrm{th}}\right)
\simeq
\bar{n}_{\mathrm{th}}$.
Experimentally, residual thermal populations of
$10^{-3}$--$10^{-2}$ are realistic, with microwave resonators exhibiting
values ranging from $\sim 10^{-1}$ to $10^{-3}$~\cite{Wang2019}.
High-$Q$ storage cavities have reached
$\bar{n}_{\mathrm{th}}\sim10^{-4}$, whereas multimode cavity memories
show occupations of $0.2\%$--$2\%$, corresponding to
$\Gamma_{\uparrow}/\Gamma_{\downarrow}
\sim2\times10^{-3}$--$2\times10^{-2}$.
In the cat-code experiment of Ofek \textit{et al.}, thermal single-photon
gain events occurred approximately once every $6\,\mathrm{ms}$, while
ordinary photon loss remained the dominant error channel~\cite{Ofek2016}.
Thus, zero-temperature amplitude damping provides an excellent
leading-order description, although thermal photon gain can become
relevant for high-fidelity bosonic error correction.

\begin{table}[ht]
\centering
\renewcommand{\arraystretch}{2.3}
\caption{All Kraus operator families for the bosonic channels considered in this work.}
\label{tab:Kraus}
\small
\begin{tabular}{>{\raggedright\arraybackslash}p{2.6cm} p{7.0cm} p{2.9cm} p{1.5cm}}
\hline
\textbf{Channel} & \textbf{Kraus operator} & \textbf{Coefficient} $\abs{C}^2$ & \textbf{Indices} \\
\hline
$T=0$ amplitude damping
 & $\displaystyle K_l^{\left(0\right)}=\sqrt{\frac{\left(1-\eta_t\right)^l}{l!}}\;e^{-\gamma t\hat{a}^{\dagger}\hat{a}/2}\,\hat{a}^l$
 & $\displaystyle\frac{\left(1-\eta_t\right)^l}{l!}$ & $l\ge0$ \\
Pure dephasing
 & $\displaystyle\hat{D}_j=\sqrt{\frac{\left(\kappa t\right)^j}{j!}}\;e^{-\kappa t\hat{n}^2/2}\,\hat{n}^j$
 & $\displaystyle\frac{\left(\kappa t\right)^j}{j!}$ & $j\ge0$ \\
Finite-$T$ amp.\ damp.
 & $\displaystyle\hat{K}_{kl}=\sqrt{\frac{N_c^k N_s^l}{k!\,l!\,G_t}}\;\left(\hat{a}^{\dagger}\right)^l e^{-A\hat{a}^{\dagger}\hat{a}/2}\,\hat{a}^k$
 & $\displaystyle\frac{N_c^k N_s^l}{k!\,l!\,G_t}$ & $k,l\ge0$ \\
$T=0$ amp.\ + deph.
 & \small $\displaystyle\hat{M}_{l,j}^{\left(0\right)}=\sqrt{\frac{\left(1-\eta_t\right)^l\left(\kappa t\right)^j}{l!\,j!}}\;e^{-\kappa t\hat{n}^2/2}\hat{n}^j e^{-\gamma t\hat{n}/2}\hat{a}^l$
 & $\displaystyle\frac{\left(1-\eta_t\right)^l\left(\kappa t\right)^j}{l!\,j!}$ & $l,j\ge0$ \\
Fin.-$T$ amp.\ + deph.
 & \scriptsize $\displaystyle\hat{M}_{kl,j}=\sqrt{\frac{N_c^k N_s^l\left(\kappa t\right)^j}{k!\,l!\,j!\,G_t}}\;e^{-\kappa t\hat{n}^2/2}\hat{n}^j\left(\hat{a}^{\dagger}\right)^l e^{-A\hat{n}/2}\hat{a}^k$
 & $\displaystyle\frac{N_c^k N_s^l\left(\kappa t\right)^j}{k!\,l!\,j!\,G_t}$ & $k,l,j\ge0$ \\
\hline
\end{tabular}
\end{table}

\subsection{Monomial expansions}
\label{sec:monomial}

The closed-form Kraus operators in Table~\ref{tab:Kraus} all contain factors of
$e^{-\beta\hat{n}}$ (linear exponential of the number operator) or $e^{-\beta\hat{n}^2/2}$
(Gaussian in $\hat{n}$). In many applications\textemdash{}error-correction syndrome analysis, Wigner-function
computation, representation-theory decompositions\textemdash{}it is useful to have each Kraus operator
expressed as a convergent series of pure ladder monomials $\left(\hat{a}^{\dagger}\right)^p\hat{a}^q$. We derive these
expansions systematically.

\subsection{Key algebraic tools}

\begin{lemma}[Normal-ordered exponential of $\hat{n}$]
\label{lem:normalorder}
For every complex $\beta$,
\begin{equation}
e^{-\beta\hat{a}^{\dagger}\hat{a}} = \sum_{m=0}^{\infty}\frac{\left(e^{-\beta}-1\right)^m}{m!}\left(\hat{a}^{\dagger}\right)^m\hat{a}^m.
\label{eq:normalexp}
\end{equation}
\end{lemma}
\begin{proof}
Both sides are diagonal in the Fock basis and it suffices to compare them on $\ket{n}$.
The left side gives $e^{-\beta n}$. For the right side, $\left(\hat{a}^{\dagger}\right)^m\hat{a}^m\ket{n}
=\frac{n!}{\left(n-m\right)!}\ket{n}$ for $m\le n$ (zero for $m>n$), so
\begin{equation}
\mathrm{RHS}\,\ket{n} = \sum_{m=0}^n\binom{n}{m}\left(e^{-\beta}-1\right)^m\cdot1^{n-m}\ket{n}
= \left(1+\left(e^{-\beta}-1\right)\right)^n\ket{n} = e^{-\beta n}\ket{n}.
\qquad\square
\end{equation}
\end{proof}

\begin{remark}
Lemma~\ref{lem:normalorder} holds for all complex $\beta$. For our applications $\beta>0$
so $e^{-\beta}\in\left(0,1\right)$ and $\abs{e^{-\beta}-1}<1$. On any state of finite Fock support the series terminates, since $\left(\hat{a}^{\dagger}\right)^m\hat{a}^m\ket{n}=0$ for $m>n$. Each term $\left(\hat{a}^{\dagger}\right)^m\hat{a}^m$ preserves the photon number,
so at $T=0$ the expansion index $m$ can be interpreted as counting virtual self-contractions produced by
the no-jump propagator between the physical jump events.
\end{remark}

For the dephasing Gaussian factor we use two standard identities:
\begin{equation}
e^{-\kappa t\hat{n}^2/2} = \sum_{\rho=0}^{\infty}\frac{\left(-\kappa t/2\right)^\rho}{\rho!}\,\hat{n}^{2\rho},
\label{eq:dephasexp}
\end{equation}
and the normal-ordering of powers of $\hat{n}$,
\begin{equation}
\hat{n}^r = \sum_{p=0}^r S\left(r,p\right)\,\left(\hat{a}^{\dagger}\right)^p\hat{a}^p,
\label{eq:Stirling}
\end{equation}
where $S\left(r,p\right)$ are Stirling numbers of the second kind, defined by
$x^r=\sum_{p=0}^r S\left(r,p\right)\,x\left(x-1\right)\cdots\left(x-p+1\right)$ (falling-factorial expansion). Both
\eqref{eq:dephasexp} and \eqref{eq:Stirling} are verified by evaluating on $\ket{n}$.

We also use the \emph{shift identity}: for any formal power series $f$,
\begin{equation}
f\left(\hat{n}\right)\,\hat{a}^k = \hat{a}^k\,f\left(\hat{n}-k\right),\qquad \left(\hat{a}^{\dagger}\right)^k\,f\left(\hat{n}\right) = f\left(\hat{n}-k\right)\,\left(\hat{a}^{\dagger}\right)^k,
\label{eq:shift}
\end{equation}
following from $\left[\hat{n},\hat{a}\right]=-\hat{a}$ and $\left[\hat{n},\hat{a}^{\dagger}\right]=+\hat{a}^{\dagger}$ applied repeatedly.

Finally, the inclusion-exclusion (finite-difference) representation of falling-factorial
coefficients: if $f\left(n\right)=\sum_{p\ge0}d_p\,n^{\left(p\right)}$ where $n^{\left(p\right)}=n\left(n-1\right)\cdots\left(n-p+1\right)$, then
\begin{equation}
d_p = \frac{1}{p!}\Delta^p f\left(0\right) = \frac{1}{p!}\sum_{q=0}^p\left(-1\right)^{p-q}\binom{p}{q}f\left(q\right).
\label{eq:finDiff}
\end{equation}

\subsubsection{Amplitude damping at \texorpdfstring{$T=0$}{T=0}}

Set $\lambda_0=e^{-\gamma t/2}\in\left(0,1\right)$, so $e^{-\beta}-1=\lambda_0-1$. Applying
Lemma~\ref{lem:normalorder} with $\beta=\gamma t/2$ to the exponential in \eqref{eq:KrausT0}
and absorbing $\hat{a}^l$ into the series:
\begin{equation}
K_l^{\left(0\right)}\left(t\right) = C_l^{\left(0\right)}\sum_{m=0}^{\infty}\frac{\left(\lambda_0-1\right)^m}{m!}\,\left(\hat{a}^{\dagger}\right)^m\hat{a}^{m+l},
\label{eq:mono_T0}
\end{equation}
where $C_l^{\left(0\right)}=\sqrt{\left(1-\eta_t\right)^l/l!}$. Every monomial $\left(\hat{a}^{\dagger}\right)^m\hat{a}^{m+l}$ has net
photon-number shift $m-\left(m+l\right)=-l$, independent of the summation index $m$. This is the
expected result: the Kraus operator $K_l^{\left(0\right)}$ removes exactly $l$ photons on net, regardless
of the virtual self-contractions encoded by $m$.

On a Fock state, $\left(\hat a^{\dagger}\right)^m\hat a^{m+l}\ket n=\sqrt{n!/\left(n-l\right)!}\,\frac{\left(n-l\right)!}{\left(n-l-m\right)!}\ket{n-l}$, so the $1/m!$ is cancelled by a falling factorial and
\begin{equation}
K_l^{\left(0\right)}\ket n=C_l^{\left(0\right)}\sqrt{\tfrac{n!}{\left(n-l\right)!}}\;\sum_{m=0}^{n-l}\binom{n-l}{m}\mu_0^m\,\ket{n-l},\qquad \mu_0=\lambda_0-1 .
\end{equation}
The series therefore terminates on any state of finite Fock support, and it is exact once $m_{\max}\geq n_{\max}-l$. For smaller $m_{\max}$ the truncation error on $\ket n$
is $C_l^{\left(0\right)}\sqrt{n!/\left(n-l\right)!}\,\abs{\sum_{m>m_{\max}}\binom{n-l}{m}\mu_0^m}$, which grows with $n$. Consequently, no truncation bound that is independent of the Fock-space dimension exists; for instance, at $\gamma t=0.7$ and $m_{\max}=4$ the truncation error on $\ket{10}$ is $0.45$, whereas $\left(1-\lambda_0\right)^{m_{\max}+1}/\left(m_{\max}+1\right)!\approx2\times10^{-5}$.
In numerical work we therefore evaluate the closed form \eqref{eq:KrausT0} on the truncated Fock space, and use the series only with input-specific bounds.

\paragraph{Explicit low-order expansions.} Defining $\mu_0\equiv\lambda_0-1=e^{-\gamma t/2}-1\in\left(-1,0\right)$:
\begin{align}
K_0^{\left(0\right)}\left(t\right) &= \mathbb{1} + \mu_0\,\hat{a}^{\dagger}\hat{a} + \tfrac{\mu_0^2}{2}\,\left(\hat{a}^{\dagger}\right)^2\hat{a}^2
+ \tfrac{\mu_0^3}{6}\,\left(\hat{a}^{\dagger}\right)^3\hat{a}^3+\cdots \label{eq:K0exp}\\
K_1^{\left(0\right)}\left(t\right) &= \sqrt{1-\eta_t}\left[\hat{a} + \mu_0\,\hat{a}^{\dagger}\hat{a}^2 + \tfrac{\mu_0^2}{2}\,\left(\hat{a}^{\dagger}\right)^2\hat{a}^3
+\cdots\right]\label{eq:K1exp}\\
K_2^{\left(0\right)}\left(t\right) &= \sqrt{\tfrac{\left(1-\eta_t\right)^2}{2}}\left[\hat{a}^2 + \mu_0\,\hat{a}^{\dagger}\hat{a}^3
+ \tfrac{\mu_0^2}{2}\,\left(\hat{a}^{\dagger}\right)^2\hat{a}^4+\cdots\right]\label{eq:K2exp}\\
K_3^{\left(0\right)}\left(t\right) &= \sqrt{\tfrac{\left(1-\eta_t\right)^3}{6}}\left[\hat{a}^3 + \mu_0\,\hat{a}^{\dagger}\hat{a}^4+\cdots\right]
\label{eq:K3exp}
\end{align}
The $l=0$ operator \eqref{eq:K0exp} is a diagonal (photon-number-preserving) series that
converges to the identity at $\gamma t=0$ ($\mu_0=0$) and approaches the projector onto
$\ket{0}$ as $\gamma t\to\infty$ ($\mu_0\to-1$). Each $l\ge1$ operator consists of $l$ net
lowering steps dressed by virtual raising-and-lowering pairs.

\subsubsection{Finite-temperature amplitude damping}

Set $\lambda=e^{-A\left(t\right)/2}=\sqrt{\eta_t}/G_t\in\left(0,1\right)$ and $\mu\equiv\lambda-1$. Applying
Lemma~\ref{lem:normalorder} with $\beta=A\left(t\right)/2$ to \eqref{eq:KrausFiniteT} gives
\begin{equation}
\hat{K}_{kl}\left(t\right) = C_{kl}\sum_{m=0}^{\infty}\frac{\mu^m}{m!}\,\left(\hat{a}^{\dagger}\right)^{l+m}\hat{a}^{k+m},
\label{eq:mono_finT}
\end{equation}
where $C_{kl}=\sqrt{N_c^k N_s^l/\left(k!\,l!\,G_t\right)}$. Every monomial has net shift
$\left(l+m\right)-\left(k+m\right)=l-k$, constant for all $m$, consistent with $k$ photons removed and $l$ added
on net.

\paragraph{Explicit low-order expansions.}
\begin{align}
\hat{K}_{00}\left(t\right) &= \frac{1}{\sqrt{G_t}}\left[\mathbb{1}+\mu\,\hat{a}^{\dagger}\hat{a}+\tfrac{\mu^2}{2}\,\left(\hat{a}^{\dagger}\right)^2\hat{a}^2
+\tfrac{\mu^3}{6}\,\left(\hat{a}^{\dagger}\right)^3\hat{a}^3+\cdots\right]\label{eq:K00}\\
\hat{K}_{10}\left(t\right) &= \frac{\sqrt{N_c}}{\sqrt{G_t}}\left[\hat{a}+\mu\,\hat{a}^{\dagger}\hat{a}^2+\tfrac{\mu^2}{2}\,\left(\hat{a}^{\dagger}\right)^2\hat{a}^3+\cdots\right]\label{eq:K10}\\
\hat{K}_{01}\left(t\right) &= \frac{\sqrt{N_s}}{\sqrt{G_t}}\left[\hat{a}^{\dagger}+\mu\,\left(\hat{a}^{\dagger}\right)^2\hat{a}+\tfrac{\mu^2}{2}\,\left(\hat{a}^{\dagger}\right)^3\hat{a}^2+\cdots\right]\label{eq:K01}\\
\hat{K}_{11}\left(t\right) &= \frac{\sqrt{N_cN_s}}{\sqrt{G_t}}\left[\hat{a}^{\dagger}\hat{a}+\mu\,\left(\hat{a}^{\dagger}\right)^2\hat{a}^2+\tfrac{\mu^2}{2}\,\left(\hat{a}^{\dagger}\right)^3\hat{a}^3+\cdots\right]\label{eq:K11}\\
\hat{K}_{20}\left(t\right) &= \frac{\sqrt{N_c^2/2}}{\sqrt{G_t}}\left[\hat{a}^2+\mu\,\hat{a}^{\dagger}\hat{a}^3+\cdots\right]\label{eq:K20}\\
\hat{K}_{02}\left(t\right) &= \frac{\sqrt{N_s^2/2}}{\sqrt{G_t}}\left[\left(\hat{a}^{\dagger}\right)^2+\mu\,\left(\hat{a}^{\dagger}\right)^3\hat{a}+\cdots\right]\label{eq:K02}
\end{align}

At $\bar{n}=0$: $G_t=1$, $\mu\to\mu_0$, $N_s=0$, $N_c=1-\eta_t$, and only the $l=0$
operators survive, reproducing \eqref{eq:K0exp}--\eqref{eq:K3exp}.

\paragraph{Interpretation of $m$.} The index $m$ labels the normal-ordering corrections of $e^{-A\hat n/2}$, each a number-preserving pair $\left(\hat a^{\dagger}\right)^m\hat a^m$. At finite temperature $e^{-A\hat n/2}$ is not the no-jump propagator, so $m$ has no direct trajectory meaning.
As at $T=0$, the factor $1/m!$ is cancelled on Fock states by a falling factorial: on photon numbers $n$ the $m$-th correction scales as $\binom{n}{m}\abs{\mu}^m$.
The leading monomial $\left(\hat a^{\dagger}\right)^l\hat a^k$ therefore dominates only when $n\abs{\mu}$ is small on the photon numbers occupied by the input, which requires short times ($\mu\to0$ as $t\to0$); a low temperature alone does not make $\abs{\mu}$ small.

\subsubsection{Pure dephasing}

The Kraus operator \eqref{eq:KrausPhase} contains only $\hat{n}$, hence all monomials must be
diagonal. Expanding $e^{-\kappa t\hat{n}^2/2}$ via \eqref{eq:dephasexp} and normal-ordering each
$\hat{n}^{j+2\rho}$ via \eqref{eq:Stirling}:
\begin{equation}
\hat{D}_j\left(t\right) = \sqrt{\frac{\left(\kappa t\right)^j}{j!}}\sum_{p=0}^{\infty}
\underbrace{\left[\sum_{\rho=0}^{\infty}\frac{\left(-\kappa t/2\right)^\rho}{\rho!}S\left(j+2\rho,p\right)\right]}_{\displaystyle d_p^{\left(j\right)}}
\left(\hat{a}^{\dagger}\right)^p\hat{a}^p.
\label{eq:DhatExpand}
\end{equation}
Every monomial $\left(\hat{a}^{\dagger}\right)^p\hat{a}^p$ is diagonal (net shift zero), confirming that dephasing
preserves photon-number populations and destroys only off-diagonal coherences.

The double sum in \eqref{eq:DhatExpand} can be compressed using \eqref{eq:finDiff}. Define
$\vartheta_q\equiv e^{-q^2\kappa t/2}$ (the no-jump weight for photon number $q$). Then
\begin{equation}
d_p^{\left(j\right)} = \frac{1}{p!}\sum_{q=0}^p\left(-1\right)^{p-q}\binom{p}{q}q^j\,\vartheta_q.
\label{eq:dp}
\end{equation}
This finite sum terminates at $q=p$, making it convenient for numerical evaluation.

\paragraph{Explicit low-order expansions.} Set $\lambda_d=e^{-\kappa t/2}$, so that
$\vartheta_q=\lambda_d^{q^2}$ (note the quadratic exponent). For $j\geq1$ one has $d_0^{\left(j\right)}=0$, since $q^j$ vanishes at $q=0$. The first few operators are:
\begin{align}
\hat{D}_0\left(t\right) &= \mathbb{1} + \left(\lambda_d-1\right)\,\hat{a}^{\dagger}\hat{a}
+\tfrac{1}{2}\left(\lambda_d^4-2\lambda_d+1\right)\,\left(\hat{a}^{\dagger}\right)^2\hat{a}^2
+\tfrac{1}{6}\left(\lambda_d^9-3\lambda_d^4+3\lambda_d-1\right)\,\left(\hat{a}^{\dagger}\right)^3\hat{a}^3+\cdots
\label{eq:D0}\\
\hat{D}_1\left(t\right) &= \sqrt{\kappa t}\left[
\lambda_d\,\hat{a}^{\dagger}\hat{a} + \left(\lambda_d^4-\lambda_d\right)\,\left(\hat{a}^{\dagger}\right)^2\hat{a}^2
+\tfrac{1}{2}\left(\lambda_d^9-2\lambda_d^4+\lambda_d\right)\,\left(\hat{a}^{\dagger}\right)^3\hat{a}^3+\cdots\right]
\label{eq:D1}\\
\hat{D}_2\left(t\right) &= \frac{\kappa t}{\sqrt{2}}\left[
\lambda_d\,\hat{a}^{\dagger}\hat{a} + \left(2\lambda_d^4-\lambda_d\right)\,\left(\hat{a}^{\dagger}\right)^2\hat{a}^2
+\tfrac{1}{2}\left(3\lambda_d^9-4\lambda_d^4+\lambda_d\right)\,\left(\hat{a}^{\dagger}\right)^3\hat{a}^3+\cdots\right]
\label{eq:D2}\\
\hat{D}_3\left(t\right) &= \sqrt{\frac{\left(\kappa t\right)^3}{6}}\left[
\lambda_d\,\hat{a}^{\dagger}\hat{a} + \left(4\lambda_d^4-\lambda_d\right)\,\left(\hat{a}^{\dagger}\right)^2\hat{a}^2+\cdots\right]
\label{eq:D3}
\end{align}
where $\lambda_d^{q^2}\equiv e^{-q^2\kappa t/2}$ (for example $\lambda_d^4=e^{-2\kappa t}$,
$\lambda_d^9=e^{-9\kappa t/2}$). Note that $\hat{D}_0$ converges to $\mathbb{1}$ at $\kappa t=0$
($\lambda_d=1$, all coefficients vanish except $p=0$) and approaches the projector onto the
vacuum for large $\kappa t$.

\paragraph{Comparison with $T=0$ amplitude damping.} In \eqref{eq:K0exp} the expansion
parameter is $\mu_0=\lambda_0-1$ with a \emph{linear} exponent in $t$. In \eqref{eq:D0}
the effective expansion parameter at the $p$-th order is $\vartheta_p-1=\lambda_d^{p^2}-1$
with a \emph{quadratic} exponent $p^2\kappa t/2$. This Gaussian suppression in $p$ is the
direct signature of the quadratic $\hat{n}^2$ in the dephasing Hamiltonian, and means that
dephasing decorrelates higher Fock components faster than lower ones.

\subsubsection{Combined \texorpdfstring{$T=0$}{T=0} amplitude damping and dephasing}

The Kraus operator is $\hat{M}_{l,j}^{\left(0\right)}=C_{l,j}^{\left(0\right)}\,e^{-\kappa t\hat{n}^2/2}\hat{n}^j e^{-\gamma t\hat{n}/2}\hat{a}^l$
with $C_{l,j}^{\left(0\right)}=\sqrt{\left(1-\eta_t\right)^l\left(\kappa t\right)^j/\left(l!\,j!\right)}$.

\paragraph{Step 1: merge commuting exponentials.} Both $e^{-\gamma t\hat{n}/2}$ and $e^{-\kappa t\hat{n}^2/2}$
are functions of $\hat{n}$ and commute:
\begin{equation}
e^{-\gamma t\hat{n}/2}e^{-\kappa t\hat{n}^2/2} = e^{-\gamma t\hat{n}/2-\kappa t\hat{n}^2/2}\equiv F\left(\hat{n}\right),
\qquad F\left(n\right)=e^{-\gamma tn/2-\kappa tn^2/2}.
\end{equation}

\paragraph{Step 2: normal-order the diagonal factor.} Write $\hat M^{\left(0\right)}_{l,j}=C^{\left(0\right)}_{l,j}\,g_j\left(\hat n\right)\,\hat a^l$ with $g_j\left(n\right)=n^jF\left(n\right)$. By \eqref{eq:finDiff},
$g_j\left(\hat n\right)=\sum_p d_p^{\left(j\right)}\left(\gamma t,\kappa t\right)\left(\hat a^{\dagger}\right)^p\hat a^p$ with
\begin{equation}
d_p^{\left(j\right)}\left(\gamma t,\kappa t\right) = \frac{1}{p!}\sum_{q=0}^p\left(-1\right)^{p-q}\binom{p}{q}\,q^j
\,e^{-\gamma tq/2-\kappa tq^2/2}.
\label{eq:dpjl}
\end{equation}
Right-multiplying by $\hat a^l$ keeps each monomial normal-ordered:
\begin{equation}
\hat{M}_{l,j}^{\left(0\right)}\left(t\right) = C_{l,j}^{\left(0\right)}\sum_{p=0}^{\infty}d_p^{\left(j\right)}\left(\gamma t,\kappa t\right)\,\left(\hat{a}^{\dagger}\right)^p\hat{a}^{p+l},
\label{eq:mono_T0dep}
\end{equation}
with net shift $-l$. The coefficients do not depend on $l$.

\paragraph{Explicit low-order expansions.} With $\lambda_0=e^{-\gamma t/2}$ and $\lambda_d=e^{-\kappa t/2}$, so that $F\left(q\right)=\lambda_0^{q}\lambda_d^{q^2}$:
\begin{align}
\hat{M}_{0,0}^{\left(0\right)} &= \mathbb{1} + \left(\lambda_0\lambda_d-1\right)\,\hat{a}^{\dagger}\hat{a}
+\tfrac{1}{2}\left(\lambda_0^2\lambda_d^4-2\lambda_0\lambda_d+1\right)\,\left(\hat{a}^{\dagger}\right)^2\hat{a}^2+\cdots
\label{eq:M00}\\
\hat{M}_{1,0}^{\left(0\right)} &= \sqrt{1-\eta_t}\,\left[\hat{a}
+\left(\lambda_0\lambda_d-1\right)\,\hat{a}^{\dagger}\hat{a}^2+\cdots\right]\\
\hat{M}_{0,1}^{\left(0\right)} &= \sqrt{\kappa t}\,\left[\lambda_0\lambda_d\,\hat{a}^{\dagger}\hat{a}
+\left(\lambda_0^2\lambda_d^4-\lambda_0\lambda_d\right)\,\left(\hat{a}^{\dagger}\right)^2\hat{a}^2+\cdots\right]\\
\hat{M}_{1,1}^{\left(0\right)} &= \sqrt{\left(1-\eta_t\right)\kappa t}\,\left[\lambda_0\lambda_d\,\hat{a}^{\dagger}\hat{a}^2
+\left(\lambda_0^2\lambda_d^4-\lambda_0\lambda_d\right)\,\left(\hat{a}^{\dagger}\right)^2\hat{a}^3+\cdots\right]
\end{align}
Setting $\kappa=0$ and $j=0$ in \eqref{eq:dpjl} gives $d_p^{\left(0\right)}=\left(\lambda_0-1\right)^p/p!$, so \eqref{eq:mono_T0dep} reduces to \eqref{eq:mono_T0}.
All displayed expansions have been verified against direct evaluation of the operators on a truncated Fock space.

\subsubsection{Combined finite-\texorpdfstring{$T$}{T} amplitude damping and dephasing}

By \eqref{eq:KrausFiniteT} and \eqref{eq:KrausPhase}, $\hat{M}_{kl,j}=\hat D_j\hat K_{kl}=C_{kl,j}\,e^{-\kappa t\hat{n}^2/2}\hat{n}^j\left(\hat{a}^{\dagger}\right)^l e^{-A\hat{n}/2}\hat{a}^k$, with
$C_{kl,j}=\sqrt{N_c^kN_s^l\left(\kappa t\right)^j/\left(k!\,l!\,j!\,G_t\right)}$. Moving the diagonal factor to the right of $\left(\hat a^{\dagger}\right)^l$ with the shift identity \eqref{eq:shift},
$f\left(\hat n\right)\left(\hat a^{\dagger}\right)^l=\left(\hat a^{\dagger}\right)^lf\left(\hat n+l\right)$, gives
$\hat M_{kl,j}=C_{kl,j}\,\left(\hat a^{\dagger}\right)^l\,h_{l,j}\left(\hat n\right)\,\hat a^k$ with $h_{l,j}\left(n\right)=\left(n+l\right)^j\,\lambda_d^{\left(n+l\right)^2}\lambda^{n}$, where
$\lambda=e^{-A/2}=\sqrt{\eta_t}/G_t$ and $\lambda_d=e^{-\kappa t/2}$. Normal-ordering $h_{l,j}\left(\hat n\right)$ with \eqref{eq:finDiff} gives the single-sum expansion
\begin{equation}
\hat{M}_{kl,j}\left(t\right) = C_{kl,j}\sum_{p=0}^{\infty}c_p^{\left(l,j\right)}\,\left(\hat{a}^{\dagger}\right)^{l+p}\hat{a}^{k+p},
\qquad
c_p^{\left(l,j\right)}=\frac{1}{p!}\sum_{q=0}^{p}\left(-1\right)^{p-q}\binom{p}{q}\,\left(q+l\right)^j\,\lambda_d^{\left(q+l\right)^2}\lambda^{q},
\label{eq:Mfull}
\end{equation}
with net shift $l-k$ for every monomial. The coefficients depend on $l$ and $j$ but not on $k$.

\paragraph{Explicit low-order cases.}
\begin{align}
\hat{M}_{00,0}\left(t\right) &= \frac{1}{\sqrt{G_t}}\left[\mathbb{1}+\left(\lambda\lambda_d-1\right)\,\hat{a}^{\dagger}\hat{a}
+\tfrac{1}{2}\left(\lambda^2\lambda_d^4-2\lambda\lambda_d+1\right)\,\left(\hat{a}^{\dagger}\right)^2\hat{a}^2+\cdots\right]
\label{eq:M000}\\
\hat{M}_{10,0}\left(t\right) &= \frac{\sqrt{N_c}}{\sqrt{G_t}}\left[\hat{a}
+\left(\lambda\lambda_d-1\right)\,\hat{a}^{\dagger}\hat{a}^2+\cdots\right]
\label{eq:M100}\\
\hat{M}_{01,0}\left(t\right) &= \frac{\sqrt{N_s}}{\sqrt{G_t}}\left[\lambda_d\,\hat{a}^{\dagger}
+\left(\lambda\lambda_d^4-\lambda_d\right)\,\left(\hat{a}^{\dagger}\right)^2\hat{a}+\cdots\right]
\label{eq:M010}\\
\hat{M}_{11,0}\left(t\right) &= \frac{\sqrt{N_cN_s}}{\sqrt{G_t}}\left[\lambda_d\,\hat{a}^{\dagger}\hat{a}+\left(\lambda\lambda_d^4-\lambda_d\right)\left(\hat{a}^{\dagger}\right)^2\hat{a}^2+\cdots\right]
\label{eq:M110}\\
\hat{M}_{00,1}\left(t\right) &= \frac{\sqrt{\kappa t}}{\sqrt{G_t}}\left[\lambda\lambda_d\,\hat{a}^{\dagger}\hat{a}
+\left(\lambda^2\lambda_d^4-\lambda\lambda_d\right)\,\left(\hat{a}^{\dagger}\right)^2\hat{a}^2+\cdots\right]
\label{eq:M001}\\
\hat{M}_{10,1}\left(t\right) &= \frac{\sqrt{\kappa t N_c}}{\sqrt{G_t}}\left[\lambda\lambda_d\,\hat{a}^{\dagger}\hat{a}^2
+\left(\lambda^2\lambda_d^4-\lambda\lambda_d\right)\,\left(\hat{a}^{\dagger}\right)^2\hat{a}^3+\cdots\right]
\label{eq:M101}\\
\hat{M}_{01,1}\left(t\right) &= \frac{\sqrt{\kappa t N_s}}{\sqrt{G_t}}\left[\lambda_d\,\hat{a}^{\dagger}
+\left(2\lambda\lambda_d^4-\lambda_d\right)\,\left(\hat{a}^{\dagger}\right)^2\hat{a}+\cdots\right]
\end{align}
Setting $\kappa=0$ ($\lambda_d\to1$) in \eqref{eq:M000}--\eqref{eq:M010} recovers
\eqref{eq:K00}--\eqref{eq:K01} with $\mu=\lambda-1$. Setting $\bar{n}=0$ additionally
gives $G_t=1$, $\mu\to\mu_0$, and only the $l=0$ operators survive, recovering \eqref{eq:K0exp}--\eqref{eq:K1exp}.
All displayed expansions have been verified against direct evaluation of the operators on a truncated Fock space.

\begin{table}[ht]
\centering
\renewcommand{\arraystretch}{1.7}
\caption{Normal-ordered monomial structure. $\mu=e^{-A/2}-1$; $S\left(r,p\right)$: Stirling numbers of
the 2nd kind; $d_p^{\left(j\right)}$ and $c_p^{\left(l,j\right)}$ as in Eqs.~\eqref{eq:dpjl} and \eqref{eq:Mfull}.
Net shift = (total creation) $-$ (total annihilation) powers.}
\label{tab:monomial}
\small
\begin{tabular}{>{\raggedright\arraybackslash}p{3.0cm}p{5.5cm}cc}
\toprule
\textbf{Channel} & \textbf{Monomial expansion} & \textbf{Net shift} & \textbf{Type}\\
\midrule
$T=0$ amp. damp. & $C_l^{\left(0\right)}\sum_m\frac{\mu_0^m}{m!}\left(\hat{a}^{\dagger}\right)^m\hat{a}^{m+l}$ & $-l$ & off-diag.\\
Fin.-$T$ amp. damp. & $C_{kl}\sum_m\frac{\mu^m}{m!}\left(\hat{a}^{\dagger}\right)^{l+m}\hat{a}^{k+m}$ & $l-k$ & general\\
Pure dephasing & $C_j^{\left(d\right)}\sum_p d_p^{\left(j\right)}\left(\hat{a}^{\dagger}\right)^p\hat{a}^p$ & $0$ & diagonal\\
$T=0$ amp. + deph. & $C_{l,j}^{\left(0\right)}\sum_p d_p^{\left(j\right)}\left(\hat{a}^{\dagger}\right)^p\hat{a}^{p+l}$ & $-l$ & off-diag.\\
Fin.-$T$ amp. + deph. & $C_{kl,j}\sum_{p}c_p^{\left(l,j\right)}\left(\hat{a}^{\dagger}\right)^{l+p}\hat{a}^{k+p}$ & $l-k$ & general\\
\bottomrule
\end{tabular}
\end{table}

\noindent The key structural feature visible in Table~\ref{tab:monomial} is that the net photon-number
shift $l-k$ (or $-l$ at $T=0$) is the same for \emph{every} monomial in the expansion,
regardless of the internal summation indices. This is a consequence of the form $\left(\hat a^{\dagger}\right)^lf\left(\hat n\right)\hat a^k$ of every Kraus operator: the Kraus operator $\hat{K}_{kl}$ is an eigenoperator of the
photon-number-shift superoperator $\hat{\rho}\mapsto\left[\hat{n},\hat{\rho}\right]$ with eigenvalue $l-k$.

\subsection{Coefficient scaling, Stirling asymptotics, and 2-D maps}
\label{sec:scaling}

\subsubsection{Factorised Poisson structure}

The squared coefficient of \eqref{eq:KrausFiniteT} factorises as
\begin{equation}
\abs{C_{kl}}^2 = \frac{1}{G_t}\cdot\frac{N_c^k}{k!}\cdot\frac{N_s^l}{l!},
\label{eq:Cfact}
\end{equation}
i.e.\ a product of two independent Poisson-like kernels in $k$ (loss) and $l$ (gain).
For the combined channel with dephasing, an additional independent Poisson factor
$\left(\kappa t\right)^j/j!$ appears:
\begin{equation}
\abs{C_{kl,j}}^2 = \frac{1}{G_t}\cdot\frac{N_c^k}{k!}\cdot\frac{N_s^l}{l!}\cdot\frac{\left(\kappa t\right)^j}{j!}.
\label{eq:Cfact3}
\end{equation}
The three factors are exactly independent because $\left[\mathcal{L}_a^{\left(T\right)},\mathcal{L}_d\right]=0$
(Theorem~\ref{thm:commute}): the combined channel is a composition, so its Kraus operators are products.
These are coefficient magnitudes, not outcome probabilities. The probability of outcome $\left(k,l,j\right)$ is $\mathrm{Tr}\,\hat M_{kl,j}\hat\rho\hat M_{kl,j}^{\dagger}$, which depends on the input through falling factorials of $\hat n$. On $\ket n$, for example, the loss index is binomial, $k\sim\mathrm{Bin}\left(n,N_c\right)$.

\subsubsection{Stirling asymptotics}

Using $\ln p!\approx p\ln p-p+\frac{1}{2}\ln\left(2\pi p\right)$ for large $p$, the coefficients
$\abs{C_{kl}}$ and $\abs{C_{kl,j}}$ are approximated by:
\begin{align}
\abs{C_l^{\left(0\right)}}\left(t\right) &\approx \frac{1}{\left(2\pi l\right)^{1/4}}\left[\frac{e\left(1-\eta_t\right)}{l}\right]^{l/2}
\label{eq:StirT0}\\
\abs{C_{kl}\left(t\right)} &\approx \frac{1}{\sqrt{G_t}\left(2\pi k\right)^{1/4}\left(2\pi l\right)^{1/4}}
\left[\frac{eN_c}{k}\right]^{k/2}\!\!\left[\frac{eN_s}{l}\right]^{l/2}
\label{eq:StirFinT}\\
\abs{C_j^{\left(d\right)}}\left(t\right) &\approx \frac{1}{\left(2\pi j\right)^{1/4}}\left[\frac{e\kappa t}{j}\right]^{j/2}
\label{eq:StirDeph}\\
\abs{C_{kl,j}\left(t\right)} &\approx \frac{1}{\sqrt{G_t}\left(2\pi k\right)^{1/4}\left(2\pi l\right)^{1/4}\left(2\pi j\right)^{1/4}}
\left[\frac{eN_c}{k}\right]^{k/2}\!\!\left[\frac{eN_s}{l}\right]^{l/2}\!\!\left[\frac{e\kappa t}{j}\right]^{j/2}
\label{eq:StirComb}
\end{align}
Once an index exceeds $e$ times the corresponding parameter (e.g.\ $l>e\left(1-\eta_t\right)$), the factor
$\left[e\left(1-\eta_t\right)/l\right]^{l/2}$ decays \emph{super-exponentially} as $l^{-l/2}$\textemdash{}faster than any
exponential. For $\abs{C}\propto\left(p!\right)^{-1/2}$ the Stirling approximation is accurate to about $2\%$ at index $2$ and to better than $1\%$ from index $5$.

\subsection{Experimentally relevant regime (Ofek et al.\ 2016)}
\label{ssec:ofek}

Superconducting 3D storage cavities used in bosonic-code experiments operate at
$\omega_c/2\pi\approx5\,\mathrm{GHz}$, base temperature $\approx15\,\mathrm{mK}$.
The equilibrium Bose--Einstein population at these conditions is
$\bar{n}_\mathrm{BE}=\left[e^{\hbar\omega_c/k_BT}-1\right]^{-1}\approx1.1\times10^{-7}$, but
residual microwave heating and spurious photon sources raise the effective thermal
occupation to $\bar{n}_\mathrm{eff}\sim10^{-2}$--$10^{-1}$. We use $\bar{n}=0.05$ as a
representative baseline.
Table~\ref{tab:device} records the key parameters from Ofek et al.~\cite{Ofek2016}.

\begin{table}[ht]
\centering
\caption{Device parameters from Ofek et al., Nature (2016).}
\label{tab:device}
\small
\begin{tabular}{llll}
\toprule
\textbf{Quantity} & \textbf{Symbol} & \textbf{Value} & \textbf{Source}\\
\midrule
Resonator amplitude-damping time & $\tau_c$ & $250\,\mu$s & paper, p.~2\\
Resonator coherence time & $T_{2c}$ & $330\,\mu$s & paper, p.~2\\
Pure dephasing time (derived) & $T_\phi$ & $970.6\,\mu$s & $1/T_\phi=1/T_{2c}-1/\left(2\tau_c\right)$\\
Dephasing-to-loss ratio & $\kappa_\phi/\gamma$ & $0.258$ & derived\\
Effective thermal occupation & $\bar{n}_\mathrm{eff}$ & $\approx0.05$ & estimated\\
Optimal syndrome wait time & $t_w$ & $20\,\mu$s & Table~1 of \cite{Ofek2016}\\
Full QEC run time & $t_\mathrm{QEC}$ & $110\,\mu$s & Figure~3 of \cite{Ofek2016}\\
\bottomrule
\end{tabular}
\end{table}

The pure dephasing time $T_\phi=970.6\,\mu$s is derived from the standard relation
$T_{2c}^{-1}=\left(2\tau_c\right)^{-1}+T_\phi^{-1}$: dephasing in this resonator is approximately
four times \emph{slower} than amplitude damping. We write $\gamma=1/\tau_c$ and $\kappa_\phi\equiv1/T_\phi$.
For the dephasing Lindbladian $\kappa\mathcal D\left[\hat n\right]$ of Sec.~\ref{sec:phase}, Eq.~\eqref{eq:Ldeig} shows that the coherence between $\ket0$ and $\ket1$ decays as $e^{-\kappa t/2}$.
The rate entering the Kraus operators $\hat D_j$ is therefore $\kappa=2\kappa_\phi$, i.e.\ $\kappa t=2t/T_\phi$ and $\kappa/\gamma=2\tau_c/T_\phi\approx0.515$.
Table~\ref{tab:weights} evaluates the
weights at key experimental time points.

\begin{table}[ht]
\centering
\caption{Per-jump weights at experimental time points ($\bar{n}=0.05$, $\kappa_\phi/\gamma=0.258$), with $N_c$, $N_s$, $G_t$ and $e^{-A}=\eta_t/G_t^2$ from Eq.~\eqref{eq:Ckl} and the dephasing parameter $\kappa t=2\kappa_\phi t$.}
\label{tab:weights}
\small
\begin{tabular}{cccccccc}
\toprule
$t$ ($\mu$s) & $1-\eta_t$ & $N_c$ (loss) & $N_s$ (gain) & $\kappa_\phi t=t/T_\phi$ & $\kappa t$ (deph.) & $G_t$ & $e^{-A}$\\
\midrule
$20\ \left(t_w\right)$  & $0.077$ & $0.080$ & $3.83\times10^{-3}$ & $0.021$ & $0.041$ & $1.004$ & $0.916$\\
$50$     & $0.181$ & $0.189$ & $8.98\times10^{-3}$ & $0.052$ & $0.103$ & $1.009$ & $0.804$\\
$110\ \left(t_\mathrm{QEC}\right)$ & $0.356$ & $0.367$ & $1.75\times10^{-2}$ & $0.113$ & $0.227$ & $1.018$ & $0.622$\\
$250\ \left(\tau_c\right)$ & $0.632$ & $0.643$ & $3.06\times10^{-2}$ & $0.258$ & $0.515$ & $1.032$ & $0.346$\\
\bottomrule
\end{tabular}
\end{table}

\subsection{Support hierarchy and practical truncation}

From Table~\ref{tab:weights}, the three parameters obey a clear ordering throughout the
QEC window:
\begin{equation}
N_s \;<\; \kappa t \;\lesssim\; N_c
\quad\text{for all }t\le\tau_c .
\end{equation}
At $t_w=20\,\mu$s these are $N_s\approx 3.8$$\,\times10^{-3}$, $\kappa t\approx0.041$,
$N_c\approx 0.080$, spanning more than an order of magnitude. The resulting hierarchy of
single-jump coefficient magnitudes is
\begin{equation}
\abs{\mathcal{C}_{1,0,0}} = 0.283 \;>\; \abs{\mathcal{C}^{\left(d\right)}_{0,0,1}} = 0.203 \;>\;
\abs{\mathcal{C}_{0,1,0}} = 0.062 \;>\; \abs{\mathcal{C}_{2,0,0}} = 0.057.
\label{eq:hierarchy}
\end{equation}
The dominant correctable error is single-photon loss ($k=1$), as the cat code is designed
to correct. Single dephasing ($j=1$) is weaker: the ratio of squared coefficients is $\kappa t/N_c\approx0.51$ at $t_w$, which approaches $2\kappa_\phi/\gamma$ at short times;
thermal gain ($l=1$) is weaker still by $\sqrt{N_s/N_c}=\sqrt{\bar{n}/\left(\bar{n}+1\right)}
\approx0.22$. These are coefficient magnitudes; the corresponding error probabilities depend on the code state. For example, the probability associated with $\hat K_{10}$ is approximately $N_c\left\langle\hat n\right\rangle$ at short times.

Because the weight carried by the discarded Kraus operators depends on the input state, a truncation of the Kraus families must be justified for the states of interest. Input-specific bounds follow from the outcome statistics on Fock states, together with $\Pr\left(X\ge r\right)\le\mathbb E\binom{X}{r}$ for any non-negative integer random variable $X$.
On $\ket n$ the loss index is binomial, $k\sim\mathrm{Bin}\left(n,N_c\right)$. Given $m$ photons after the loss stage, the gain index has
$\Pr\left(l\right)=\binom{m+l}{l}N_s^lG_t^{-\left(m+1\right)}$, and the dephasing index is Poisson with mean $\kappa t\,n^2$. Averaging over the input $\hat\rho$ gives
\begin{equation}
\Pr\left(k\ge r\right)\le\frac{N_c^r}{r!}\left\langle\left(\hat a^{\dagger}\right)^r\hat a^r\right\rangle,\qquad
\Pr\left(l\ge r\right)\le\left\langle\tbinom{\hat n+r}{r}\right\rangle\left[\left(1-\eta_t\right)\bar n\right]^r,\qquad
\Pr\left(j\ge r\right)\le\frac{\left(\kappa t\right)^r}{r!}\left\langle\hat n^{2r}\right\rangle .
\label{eq:trunc}
\end{equation}
Choosing $k_{\max},l_{\max},j_{\max}$ so that the sum of these bounds (with $r=k_{\max}+1$, etc.) is below a target tolerance guarantees that the discarded Kraus operators carry at most that weight on the given input.
For pure loss this is the input-specific tail bound used with the closed-form operators \eqref{eq:KrausT0}.

\subsection{Two-dimensional contour maps}
\label{ssec:contours}

The function $\ln\abs{C_{kl,j}}$ is a sum of independent $k$-, $l$-, and $j$-parts,
$\ln\abs{C_{kl,j}}=f_c\left(k\right)+f_s\left(l\right)+f_d\left(j\right)-\tfrac12\ln G_t$ with $f_N\left(x\right)=\tfrac{x}{2}\ln N-\tfrac12\ln x!$.
Because of the factorial terms each part is concave rather than linear, so the level sets are curved, as Fig.~\ref{fig:contours} shows.

Figure~\ref{fig:contours} shows the full $\log_{10}\abs{C_{kl}}$ map in the $\left(k,l\right)$ plane
(top row, finite-$T$ amplitude damping) and the $\log_{10}\abs{C_{kl=0,j}}$ map in the $\left(k,j\right)$
plane (bottom row, combined channel with $l=0$), at three experimental time points. Several
features stand out:

\begin{itemize}[leftmargin=1.5em,itemsep=2pt]
\item \textbf{Separable contours.} The contours of the separable function $f_c\left(k\right)+f_s\left(l\right)$ are curved. Near the origin one step in $l$ costs a factor $\sqrt{N_s}$ and one step in $k$ a factor $\sqrt{N_c}$, so one gain step is worth about $\ln N_s/\ln N_c\approx2$ loss steps at $\bar n=0.05$. This asymmetry is visible in the top panels as a steeper colour
gradient in the $l$ direction.
\item \textbf{Growth with time.} Moving from $t=20\,\mu$s to $t=250\,\mu$s, the parameters
$N_c$ and $\kappa t$ increase and the high-weight region expands toward larger
$k$ and $j$. At $t=250\,\mu$s ($=\tau_c$), $N_c\approx 0.64$, so coefficients of order unity extend to
$k=1$ and appreciable weight extends to $k\approx4$--5.
\item \textbf{Dephasing axis.} In the $\left(k,j\right)$ bottom row, the dephasing parameter
$\kappa t=2\kappa_\phi t$ is intermediate: at $t=250\,\mu$s it reaches $0.515$, so the $j=1$ term
carries squared-coefficient weight $0.515$ relative to $j=0$.
\item \textbf{Parameter markers.} The white horizontal line marks $k=N_c$ and the cyan vertical line marks $l=N_s$ in the top
panels; the yellow line marks $j=\kappa t$ in the bottom panels. Beyond $e$ times these values the coefficients fall off super-exponentially.
\end{itemize}

\begin{figure}[ht]
\centering
\includegraphics[width=\textwidth]{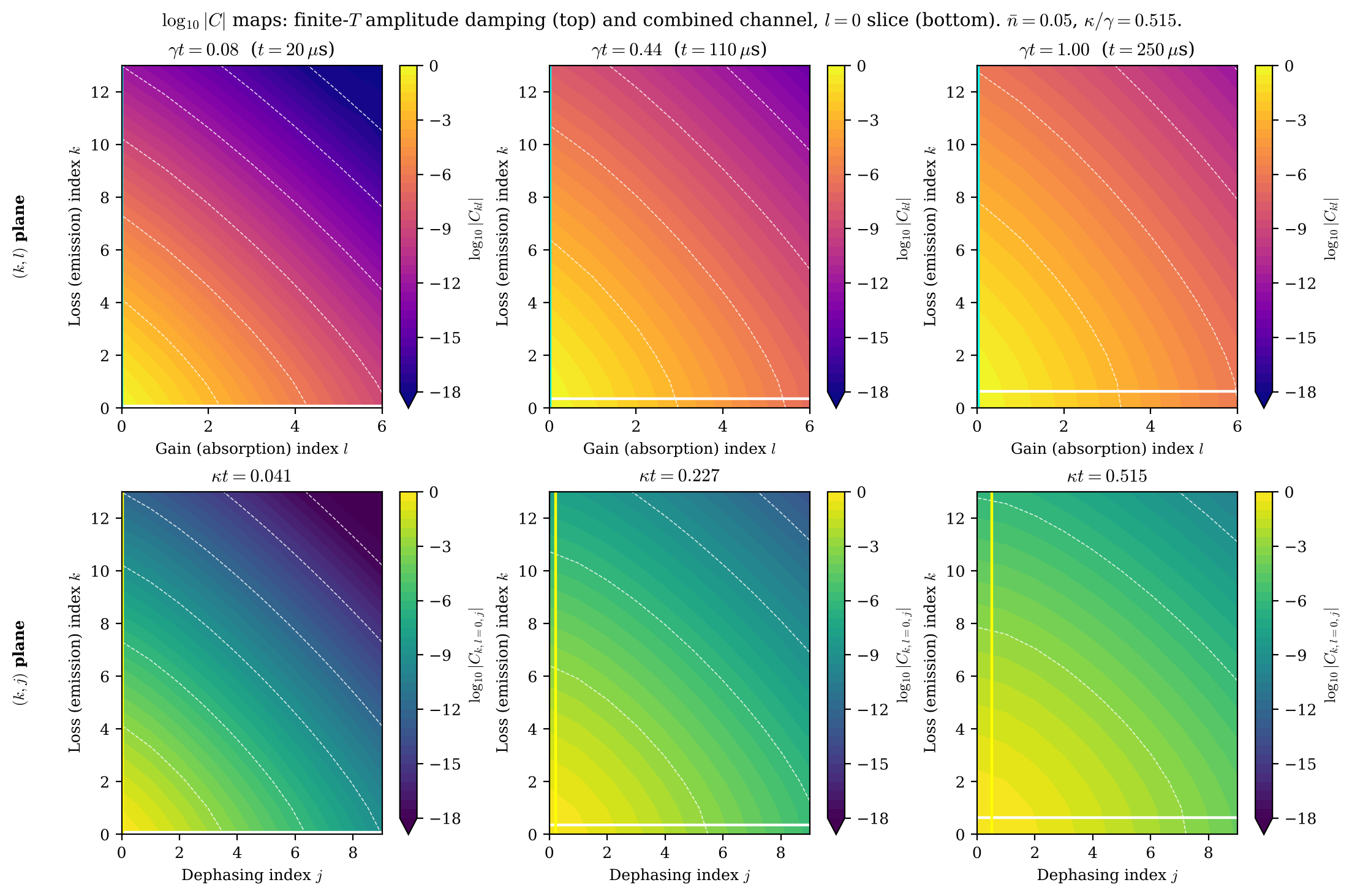}
\caption{$\log_{10}\abs{C}$ coefficient maps, computed from Eqs.~\eqref{eq:KrausPhase} and \eqref{eq:Ckl}, at $\bar{n}=0.05$, $\kappa/\gamma=2\kappa_\phi/\gamma=0.515$.
\textbf{Top row:} finite-$T$ amplitude damping in the $\left(k,l\right)$ plane.
The white horizontal line marks $k=N_c$;
the cyan vertical line marks $l=N_s$.
\textbf{Bottom row:} combined channel ($l=0$ slice) in the $\left(k,j\right)$ plane;
white marks $k=N_c$, yellow marks $j= \kappa t$.
Dashed contour lines are spaced by 3~decades. The contour
structure reflects the additive separability of $\ln\abs{C_{kl,j}}$.}
\label{fig:contours}
\end{figure}

\subsection{Coefficient slices and the loss--dephasing--gain hierarchy}

Figure~\ref{fig:slices} shows three one-dimensional slices through the coefficient array,
providing a complementary view of the super-exponential decay.

\begin{figure}[ht]
\centering
\includegraphics[width=\textwidth]{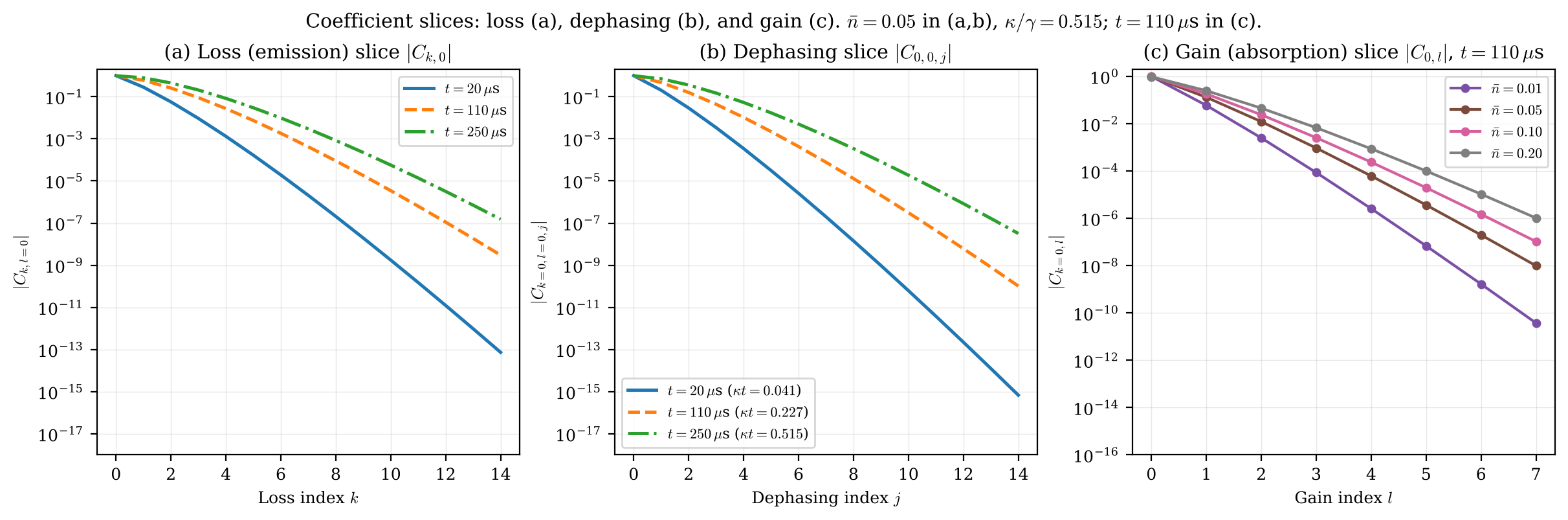}
\caption{Coefficient slices on a log scale ($\bar{n}=0.05$ unless noted, $\kappa/\gamma=0.515$).
\textbf{(a)} Loss slice $\abs{C_{k,l=0}}$ vs $k$ at three times: the curves flatten
(slower decay) as time increases and $N_c$ grows.
\textbf{(b)} Dephasing slice $\abs{C_{k=0,l=0,j}}$ vs $j$: since $\kappa t<N_c$, this slice decays faster than (a) at the same time.
\textbf{(c)} Gain slice $\abs{C_{k=0,l}}$ vs $l$ at $t=110\,\mu$s for four values of
$\bar{n}$: even at $\bar{n}=0.20$ the gain coefficients decay considerably faster
than the loss coefficients, so that thermal gain is a sub-leading error
channel at mK temperatures.}
\label{fig:slices}
\end{figure}

\paragraph{Panel (a): loss.} The parameter $N_c=0.080$ at $t_w$ grows to
$N_c=0.643$ at $t=\tau_c$. Each curve decays super-exponentially: the Stirling envelope
$\left[eN_c/k\right]^{k/2}$ predicts, for example, $\abs{C_{5,0}}/\abs{C_{0,0}}\approx10^{-3.8}$ at
$t_w$, consistent with the plot. The loss coefficients dominate throughout.

\paragraph{Panel (b): dephasing.} The smaller parameter $\kappa t\le0.515$
causes faster super-exponential decay: $\abs{C_{0,0,3}}/\abs{C_{0,0,0}}\approx10^{-0.8}$ at
$t=250\,\mu$s, compared to $\abs{C_{3,0,0}}/\abs{C_{0,0,0}}\approx10^{-0.7}$ for loss at the
same time. Whether dephasing indices $j\geq3$ may be neglected for a given code state follows from the bound \eqref{eq:trunc}, which involves $\left\langle\hat n^{2r}\right\rangle$.

\paragraph{Panel (c): gain.} At $t=110\,\mu$s ($t_\mathrm{QEC}$), the gain
coefficient for $l=1$ is $\abs{C_{0,1}}\approx0.25$ at $\bar{n}=0.20$ and $\approx0.13$
at $\bar{n}=0.05$ (at $t_w$ it is $0.062$, cf.\ \eqref{eq:hierarchy}). The strong $\bar{n}$ sensitivity
($N_s\propto\bar{n}$) makes the gain sector a sensitive diagnostic of residual thermal
occupation. We note that $l$ labels the gain stage of the factorisation \eqref{eq:factorise}, not a count of absorption jumps along a measured trajectory, so a trajectory experiment would need to be modelled with the unravelling \eqref{eq:Dyson}.

\section{Action of pure-loss channel on superposition of coherent states}
\label{sec:beam-splitter-derivation}

Here we outline how one can work out this action in a significantly smaller Hilbert space, compared to Fock subspaces with high truncation numbers. The core observation is that, under the loss channel, coherent states remain coherent. In what follows, we outline in greater detail how one can form an orthogonal basis out of these coherent states.

A beam-splitter channel with
transmissivity \(\eta=1-\gamma\), so that
\(\mathcal{N}_\gamma=\mathcal{L}_{1-\gamma}\). In the following derivation we
work with \(\mathcal{L}_\eta\), with the understanding that \(\eta=1-\gamma\).

Let \(\mathcal{L}_\eta\) denote the pure-loss channel with transmissivity
\(\eta\in\left[0,1\right]\). It can be modeled by a beam splitter with an environmental
vacuum mode,
\begin{equation}
    \mathcal{L}_\eta\left(\rho\right)
    =
    \operatorname{Tr}_{\mathrm{env}}
    \!\left[
        U_\eta\left(\rho\otimes\ket{0}\bra{0}\right)U_\eta^\dagger
    \right],
\end{equation}
where, up to an irrelevant phase,
\begin{equation}
    U_\eta \ket{\xi}\ket{0}
    =
    \ket{\sqrt{\eta}\xi}
    \ket{\sqrt{1-\eta}\xi}.
\end{equation}

The interpretation of the action of the pure-loss channel in phase space on a
superposition of coherent states is as follows:
\begin{equation}
    \begin{aligned}
        \ket{\psi}:=\left(a\ket{\alpha}+b\ket{\beta}\right)\ket{0}
        &\xmapsto{\mathcal{L}_{\eta}}
        \ket{\psi'}=
        a\ket{\sqrt{\eta}\alpha}\ket{\sqrt{1-\eta}\alpha}
        + b\ket{\sqrt{\eta}\beta}\ket{\sqrt{1-\eta}\beta}.
    \end{aligned}
\end{equation}

Calculating the density operators before and after the action of the pure-loss
channel:
\begin{equation}
    \begin{aligned}
        \rho= \ket{\psi}\bra{\psi}
        ={}& \left(\abs{a}^2\ket{\alpha}\bra{\alpha}
        + a\bar{b}\ket{\alpha}\bra{\beta}
        + \bar{a}b\ket{\beta}\bra{\alpha}
        + \abs{b}^2\ket{\beta}\bra{\beta}\right)\otimes\ket0\bra0.
    \end{aligned}
\end{equation}
\begin{equation}
    \begin{aligned}
        \rho' = \ket{\psi'}\bra{\psi'}
        ={}& \abs{a}^2
        \ket{\sqrt{\eta}\alpha}\bra{\sqrt{\eta}\alpha}
        \otimes
        \ket{\sqrt{1-\eta}\alpha}\bra{\sqrt{1-\eta}\alpha}
        \\
        &+ a\bar{b}
        \ket{\sqrt{\eta}\alpha}\bra{\sqrt{\eta}\beta}
        \otimes
        \ket{\sqrt{1-\eta}\alpha}\bra{\sqrt{1-\eta}\beta}
        \\
        &+ \bar{a}b
        \ket{\sqrt{\eta}\beta}\bra{\sqrt{\eta}\alpha}
        \otimes
        \ket{\sqrt{1-\eta}\beta}\bra{\sqrt{1-\eta}\alpha}
        \\
        &+ \abs{b}^2
        \ket{\sqrt{\eta}\beta}\bra{\sqrt{\eta}\beta}
        \otimes
        \ket{\sqrt{1-\eta}\beta}\bra{\sqrt{1-\eta}\beta}.
    \end{aligned}
\end{equation}

We now perform a Gram--Schmidt orthogonalization for system 2 (environment), with an orthonormal basis $\left\{\ket u, \ket v\right\}$ constructed so that
\begin{align}
    \ket{\sqrt{1-\eta}\alpha}
        &= \ket{u},\nonumber\\
     \ket{\sqrt{1-\eta}\beta}
        &= A_{\alpha\beta}\ket{u}
        + \sqrt{1-\abs{A_{\alpha\beta}}^2}\ket{v}\textnormal{ with}\nonumber\\
     A_{\alpha\beta}
        &= \braket{\sqrt{1-\eta}\alpha }{ \sqrt{1-\eta}\beta}
        = \exp\left[
        -\left(1-\eta\right)
        \left(
        \frac{\abs{\alpha}^2+\abs{\beta}^2}{2}
        - \bar{\alpha}\beta
        \right)
        \right].   
\end{align}
Substituting into \(\rho'\),
\begin{equation}
    \begin{aligned}
        \rho'
        = \ket{\psi'}\bra{\psi'}
        ={}&
        \abs{a}^2
        \ket{\sqrt{\eta}\alpha}
        \bra{\sqrt{\eta}\alpha}
        \otimes
        \ket{u}\bra{u}
        \\
        &+ a\bar{b}
        \ket{\sqrt{\eta}\alpha}
        \bra{\sqrt{\eta}\beta}
        \otimes
        \ket{u}
        \left(
            A_{\alpha\beta}\ket{u}
            + \sqrt{1-\abs{A_{\alpha\beta}}^2}\ket{v}
        \right)^\dagger
        \\
        &+ \bar{a}b
        \ket{\sqrt{\eta}\beta}
        \bra{\sqrt{\eta}\alpha}
        \otimes
        \left(
            A_{\alpha\beta}\ket{u}
            + \sqrt{1-\abs{A_{\alpha\beta}}^2}\ket{v}
        \right)
        \bra{u}
        \\
        &+ \abs{b}^2
        \ket{\sqrt{\eta}\beta}
        \bra{\sqrt{\eta}\beta}
        \otimes
        \left(
            A_{\alpha\beta}\ket{u}
            + \sqrt{1-\abs{A_{\alpha\beta}}^2}\ket{v}
        \right)
        \\
        &\hspace{4.8cm}\times
        \left(
            A_{\alpha\beta}\ket{u}
            + \sqrt{1-\abs{A_{\alpha\beta}}^2}\ket{v}
        \right)^\dagger.
    \end{aligned}
\end{equation}

Expanding the adjoints,
\begin{equation}
    \begin{aligned}
        ={}& \abs{a}^2
        \ket{\sqrt{\eta}\alpha}
        \bra{\sqrt{\eta}\alpha}
        \otimes
        \ket{u}\bra{u}
\\
        &+ a\bar{b}
        \ket{\sqrt{\eta}\alpha}
        \bra{\sqrt{\eta}\beta}
        \otimes
        \ket{u}
        \left(
            \bar{A}_{\alpha\beta}\bra{u}
            + \sqrt{1-\abs{A_{\alpha\beta}}^2}\bra{v}
        \right)
        \\
        &+ \bar{a}b
        \ket{\sqrt{\eta}\beta}
        \bra{\sqrt{\eta}\alpha}
        \otimes
        \left(
            A_{\alpha\beta}\ket{u}
            + \sqrt{1-\abs{A_{\alpha\beta}}^2}\ket{v}
        \right)
        \bra{u}
        \\
        &+ \abs{b}^2
        \ket{\sqrt{\eta}\beta}
        \bra{\sqrt{\eta}\beta}
        \otimes
        \left(
            A_{\alpha\beta}\ket{u}
            + \sqrt{1-\abs{A_{\alpha\beta}}^2}\ket{v}
        \right)
        \\
        &\hspace{4.8cm}\times
        \left(
            \bar{A}_{\alpha\beta}\bra{u}
            + \sqrt{1-\abs{A_{\alpha\beta}}^2}\bra{v}
        \right).
    \end{aligned}
\end{equation}

{Taking the partial trace over system 2 removes the terms
proportional to $\ket{u}\bra{v}$ and $\ket{v}\bra{u}$ in that
system because their traces vanish for orthogonal $u,v$, giving}
\begin{equation}
    \begin{aligned}
        ={}&
        \abs{a}^2
        \ket{\sqrt{\eta}\alpha}
        \bra{\sqrt{\eta}\alpha}
        + a\bar{b}\bar{A}_{\alpha\beta}
        \ket{\sqrt{\eta}\alpha}
        \bra{\sqrt{\eta}\beta}
        \\
        &+ \bar{a}bA_{\alpha\beta}
        \ket{\sqrt{\eta}\beta}
        \bra{\sqrt{\eta}\alpha}
        + \abs{b}^2
        \ket{\sqrt{\eta}\beta}
        \bra{\sqrt{\eta}\beta}.
    \end{aligned}
\end{equation}

Allowing us to model the action of the pure-loss channel on a superposition of
coherent states as
\begin{equation}
    C
    \xmapsto{\mathcal{L}_{\eta}}
    C'
    \quad\text{and basis change from}\quad
    \ket{\alpha_i}\bra{\alpha_j}
    \xmapsto{\mathcal{L}_{\eta}}
    \ket{\sqrt{\eta}\alpha_i}\bra{\sqrt{\eta}\alpha_j},
\end{equation}
where
\begin{equation}
    C =
    \begin{pmatrix}
        \abs{a}^2 & a\bar{b} \\
        \bar{a}b & \abs{b}^2
    \end{pmatrix},
    \qquad
    C' =
    \begin{pmatrix}
        \abs{a}^2 & a\bar{b}\bar{A}_{\alpha\beta} \\
        \bar{a}bA_{\alpha\beta} & \abs{b}^2
    \end{pmatrix}.
\end{equation}
Equivalently, elementwise,
\begin{equation}
    C'_{ij}
    =
    C_{ij}\braket{e_j}{e_i},
    \qquad
    \ket{e_i}=\ket{\sqrt{1-\eta}\alpha_i}.
\end{equation}

\section{Algebraic geometry}

\subsection{Polynomial rings and ideals}

\subsubsection{Polynomial ring}

The set of polynomials in $x_1,\dots,x_n$ with coefficients in a field $\F$
\begin{equation}
    \F\left[x_1,\dots,x_n\right] = \left\{ \sum_\alpha a_\alpha x^\alpha : x^\alpha = \prod_{i=1}^n x_i^{\alpha_i} ,\, a_\alpha\in\F ,\, \alpha\in\Z_{\geq0}^n \right\}
\end{equation}
is a commutative ring, as it satisfies:
\begin{enumerate}
    \item Associativity: $f+\left(g+h\right) = \left(f+g\right)+h$ for all $f,g,h\in\F\left[x_1,\dots,x_n\right]$.

    \item Commutativity: $f+g=g+f$ and $f\cdot g = g\cdot f$ for all $f,g\in\F\left[x_1,\dots,x_n\right]$.

    \item Distributivity: $f\cdot\left(g+h\right) = f\cdot g + f\cdot h$ for all $f,g,h\in\F\left[x_1,\dots,x_n\right]$.

    \item Identities: There exist $0,1\in\F\left[x_1,\dots,x_n\right]$ such that $f+0=f$ and $f\cdot1=f$ for all $f\in\F\left[x_1,\dots,x_n\right]$.
    
    Here, $0$ is the polynomial $\sum_\alpha a_\alpha x^\alpha$ with $\forall\alpha, a_\alpha=0$ and $1$ is the polynomial $a_{\left(0,\dots,0\right)} x^{\left(0,\dots,0\right)}$ with $a_{\left(0,\dots,0\right)}=1$.

    \item Additive inverses: For al $f\in\F\left[x_1,\dots,x_n\right]$, there exists $g\in\F\left[x_1,\dots,x_n\right]$ such that $f+g=0$.

    For $f = \sum_\alpha a_\alpha x^\alpha$, its additive inverse is $g = \sum_\alpha b_\alpha x^\alpha$ where $b_\alpha$ is the additive inverse of $a_\alpha$ in $\F$.
\end{enumerate}
We call $\F\left[x_1,\dots,x_n\right]$ a \textit{polynomial ring}.

\subsubsection{Varieties}

\begin{definition}
    For polynomials $f_1,\dots,f_s\in\F\left[x_1,\dots,x_n\right]$, the set
    \begin{equation}
        \mathbf{V}\left(f_1,\dots,f_s\right) = \left\{ \left(a_1,\dots,a_n\right)\in\F^n : \forall i , f_i\left(a_1,\dots,a_n\right) = 0 \right\}
    \end{equation}
    is an \textit{affine variety} defined by $f_1,\dots,f_s$.
\end{definition}

In words, affine variety $\mathbf{V}\left(f_1,\dots,f_s\right)$ is the set of solutions in $\F^n$ to polynomial system of equations $f_1,\dots,f_s\in\F\left[x_1,\dots,x_n\right]$.
It can be shown (see~[Lemma 2 of Chapter 1.3]\cite{cox2015ideals}) that for any two affine varieties $V,W$, both $V\cup W$ and $V\cap W$ are also affine varieties.

\begin{definition}[{Rational functions~\cite[Chapter 1.3]{cox2015ideals}}]
    For a field $\F$, the set of \textit{rational functions} in $t_1,\dots,t_m$ is defined as
    \begin{equation}
        \F\left(t_1,\dots,t_m\right) = \left\{ f/g : f,g\in\F\left[x_1,\dots,x_m\right] ,\, g\neq0 \right\} .
    \end{equation}
    We say that $f/g$ is equal to $f'/g'$ if $fg'=f'g$.
\end{definition}

\begin{lemma}
    A set of rational function $\F\left(t_1,\dots,t_m\right)$ is a field.
\end{lemma}

\begin{definition}[{Parametric representation of an affine variety~\cite[Chapter 1.3]{cox2015ideals}}]
    For an affine variety $V=\mathbf{V}\left(f_1,\dots,f_s\right) \subseteq \F^n$, a \textit{rational parametric representation} of $V$ is rational functions $r_1,\dots,r_n \in k\left(t_1,\dots,t_m\right)$ such that
    \begin{equation}
        x_1=r_1\left(t_1,\dots,t_m\right) ,\, \dots ,\, x_n=r_n\left(t_1,\dots,t_m\right)
    \end{equation}
    is in $V$.
    A variety $V$ with a rational parametric representation is called \textit{unirational}.
    If $r_1,\dots,r_n$ are polynomials, then it is a \textit{polynomial parametric representation} of $V$.
\end{definition}

Essentially, given a rational parametric representation $r_1,\dots,r_n$ of $V$, rational function $r_j$ describes the points in the ``$j$-th dimension'' of field $\F^n$ that belongs to $V$ in terms of parameters $t_1,\dots,t_m$.
The polynomials $f_1,\dots,f_s$ that defines $V$ is called an \textit{implicit representation} of $V$.
All $V$ has an implicit representation, but not all $V$ has a rational parametric representation. 

\begin{example}[{Exercise 3 of~\cite[Chapter 1.3]{cox2015ideals}}]
    For monomial $f\in\F\left[x\right]$, we can express it as $f\left(x\right)=a x^\alpha$.
    So $\mathbf{V}\left(y-f\left(x\right)\right)$ can be parameterized by $t$ as
    \begin{equation}
    \begin{gathered}
        y = r_1\left(t\right) = at^\alpha ,\;
        x = r_2\left(t\right) = t
    \end{gathered}
    \end{equation}
    since $r_1$ and $r_2$ are polynomials, then $r_1,r_2$ is a polynomial parametric representation of $V$.
\end{example}

\subsubsection{Ideals}

\begin{definition}
    An \textit{ideal} of polynomial ring $R=\F\left[x_1,\dots,x_n\right]$ is a subset $I\subseteq R$ such that
    \begin{enumerate}
        \item $0\in I$
        \item $f,g\in I \Rightarrow f+g\in I$
        \item $f\in R ,\, g\in I \Rightarrow fg\in I$
    \end{enumerate}
    are satisfied.
\end{definition}

For any $f_1,\dots,f_k \in R = \F\left[x_1,\dots,x_n\right]$,
\begin{equation}
    \left\langle f_1,\dots,f_k\right\rangle := \left\{ \sum_{i=1}^k h_if_i : h_i\in R \right\}
\end{equation}
is an ideal \textit{generated} by $f_1,\dots,f_k$.
The set $\left\langle f_1,\dots,f_k\right\rangle$ is indeed an ideal of $R$ (see~\cite[Chapter 1.4, Lemma 3]{cox2015ideals}) since:
\begin{enumerate}
    \item For $0\in R$, $\sum_{i=1}^k 0 f_i = 0$ is in $\left\langle f_1,\dots,f_k\right\rangle$ by definition.
    \item Consider 
    \begin{equation}
        g = \sum_{i=1}^k h_if_i + \sum_{j=1}^k h_j' f_j = \sum_{i=1}^k \left(h_i+h_i'\right)f_i
    \end{equation}
    for some arbitrary $\left\{h_i,h_i'\right\}_i \subseteq R$.
    Polynomial $g$ is in $\left\langle f_1,\dots,f_k\right\rangle$ since $h_i+h_i'\in R$ for all $i$.
    \item Consider $h\in R$ and $g\in \left\langle f_1,\dots,f_k\right\rangle$.
    It holds that $hg = \sum_{i=1}^k \left(hh_i\right)f_i$ belongs to $\left\langle f_1,\dots,f_k\right\rangle$ as $hh_i\in R$.
\end{enumerate}

Conversely, for any given ideal $I$ of $R$ there exists some $f_1,\dots,f_k \in R = \F\left[x_1,\dots,x_n\right]$ for some finite $k$ such that $I=\left\langle f_1,\dots,f_k\right\rangle$ by the Hilbert basis theorem (see~\cite[Chapter 2.5, Theorem 4]{cox2015ideals}).
In this case we call $f_1,\dots,f_k$ a \textit{basis} of ideal $I$.

Any two bases of the same ideal defines the same affine variety.
Namely for $I=\left\langle f_1,\dots,f_s\right\rangle=\left\langle g_1,\dots,g_t\right\rangle$ (bases of ideals $I$), it holds that $\mathbf{V}\left(f_1,\dots,f_s\right)=\mathbf{V}\left(g_1,\dots,g_t\right)$~\cite[Proposition 4, Chapter 1.4]{cox2015ideals}.
This is due to $g_j\left(a_1,\dots,a_n\right) = \sum_i h_if_i\left(a_1,\dots,a_n\right) = 0$ for $\left(a_1,\dots,a_n\right)\in\mathbf{V}\left(f_1,\dots,f_s\right)$ for any $j\in\left[1,t\right]$, which implies that $\mathbf{V}\left(f_1,\dots,f_s\right)\subseteq\mathbf{V}\left(g_1,\dots,g_t\right)$.
Conversely, $f_j\left(b_1,\dots,b_n\right) = \sum_i h_ig_i\left(b_1,\dots,b_n\right)=0$ for $\left(b_1,\dots,b_n\right)\in\mathbf{V}\left(f_1,\dots,f_s\right)$ for any $j\in\left[1,s\right]$, which implies $\mathbf{V}\left(f_1,\dots,f_s\right)\supseteq\mathbf{V}\left(g_1,\dots,g_t\right)$.

For an affine variety $V$, the set of all polynomials $f\in R=\F\left[x_1,\dots,x_n\right]$ such that $\forall \left(a_1,\dots,a_n\right)\in V ,\, f\left(a_1,\dots,a_n\right)=0$, denoted as
\begin{equation}
    \mathbf{I}\left(V\right) = \left\{f\in R : \forall \left(a_1,\dots,a_n\right)\in V ,\, f\left(a_1,\dots,a_n\right)=0 \right\}
\end{equation}
is the \textit{ideal of} $V$.
Indeed $\mathbf{I}\left(V\right)$ is an ideal~\cite[Lemma 6, Chapter 1.4]{cox2015ideals}.

\begin{proposition}[{\cite[Lemma 7, Chapter 1]{cox2015ideals}}]
    For polynomials $f_1,\dots,f_s\in\F\left[x_1,\dots,x_n\right]$ it holds that $\left\langle f_1,\dots,f_s\right\rangle \subseteq \mathbf{I}\left(\mathbf{V}\left(f_1,\dots,f_s\right)\right)$.
    Equality does not hold in general.
    Namely there exist $f_1,\dots,f_s\in\F\left[x_1,\dots,x_n\right]$ such that $\left\langle f_1,\dots,f_s\right\rangle \subsetneq \mathbf{I}\left(\mathbf{V}\left(f_1,\dots,f_s\right)\right)$.
\end{proposition}

The following proposition in useful in solving a system of polynomial equations, as it shows that the affine variety of a set of polynomials $f_1,\dots,f_s$ is equal to the affine variety of an ideal generated by them $\left\langle f_1,\dots,f_s\right\rangle$.

\begin{proposition}[{\cite[Proposition 9, Chapter 2.5]{cox2015ideals}}]
    For an ideal $I\subseteq\F\left[x_1,\dots,x_n\right]$ consider the affine variety of ideal $I$:
    \begin{equation}
        \mathbf{V}\left(I\right) = \left\{\left(a_1,\dots,a_n\right)\in\F^n : f\left(a_1,\dots,a_n\right)=0 ,\; \forall f\in I\right\} .
    \end{equation}
    If $I=\left\langle f_1,\dots,f_s\right\rangle$ for some polynomials $f_1,\dots,f_s\in\F\left[x_1,\dots,x_n\right]$ then $\mathbf{V}\left(I\right) = \mathbf{V}\left(f_1,\dots,f_s\right)$.
\end{proposition}

\subsection{Polynomial division and monomial ordering}

\subsubsection{Single-variable polynomials}

\begin{proposition}[{\cite[Chapter 1.5, Proposition 2]{cox2015ideals}}]
    For any $f,g\in\F\left[x\right]$ such that $g\neq 0$, there exists a unique pair $q,r\in\F\left[x\right]$ such that
    \begin{equation}
        f = gq + r
    \end{equation}
    and either $r=0$ or $\deg r < \deg g$.
    Whenever $r=0$, we say that $g$ \textit{divides} $f$.
\end{proposition}

\begin{definition}
    Polynomial $h\in\F\left[x\right]$ is the greatest common divisor of $f_1,\dots,f_s\in \F\left[x\right]$, denoted as $\mathrm{gcd}\left(f_1,\dots,f_s\right)$ if
    \begin{enumerate}
        \item $h$ divides $f_j$ for all $j$
        \item If $g\neq h$ divides $f_1,\dots,f_s$, then $g$ divides $h$.
    \end{enumerate}
\end{definition}

\begin{proposition}
    For polynomials $f_1,\dots,f_s\in R=\F\left[x\right]$, it holds that
    \begin{enumerate}
        \item There exists a unique $g\in R$ (up to multiplication by $a\in\F$) such that $\mathrm{gcd}\left(f_1,\dots,f_s\right)$
        \item $\left\langle f_1,\dots,f_s\right\rangle = \left\langle g\right\rangle$
        \item $\mathrm{gcd}\left(f_1,\dots,f_s\right) = \mathrm{gcd}\left(f_1,\mathrm{gcd}\left(f_2,\dots,f_s\right)\right)$
    \end{enumerate}
\end{proposition}

\begin{example}[{Checking membership in a principal ideal~\cite[Chapter 2.1, Exercise 1b]{cox2015ideals}}]
    For $g\in R:=\F\left[x\right]$ and principal ideal $\left\langle g\right\rangle$, we could check whether $f\in\left\langle g\right\rangle$ by dividing $f$ by $g$:
    \begin{equation}
        f = gq + r
    \end{equation}
    for $q,r\in R$ and then check whether $r=0$ or $\deg r <\deg g$ holds.
    If $r=0$ then $f\in\left\langle g\right\rangle$, otherwise $f\notin\left\langle g\right\rangle$.

    For $g\left(x\right)=x^3-x^2+x$ and $f\left(x\right)=x^5-4x+1$, we use the polynomial division algorithm:
    \begin{enumerate}
        \item Set $q_0:=0$ and $r_0:=f$.
        \item Subtract $\frac{x^5}{x^3}g\left(x\right) = x^2 g\left(x\right)$ from $r_0\left(x\right)=f\left(x\right)$, which gives $r_1\left(x\right) := r_0\left(x\right) - x^2 g\left(x\right) = x^4 - x^3 - 4x + 1$.
        Set $q_1 := q_0+ x^2 = x^2$.
        \item Consider $xg\left(x\right) = x^4-x^3+x^2$. 
        Set $r_2\left(x\right) := r_1\left(x\right) - xg\left(x\right) = -x^2 -4x+1$.
        Set $q_2 := q_1 + x = x^2+x$.
        \item We obtain $q:= q_2\left(x\right) = x^2+x$ and $r:= r_2\left(x\right) = -x^2 -4x+1$.
        We could check that $f=gq+r$ holds.
        Since $r\neq0$, then we conclude that $f\notin\left\langle g\right\rangle$.
    \end{enumerate}
\end{example}

\subsubsection{Monomial ordering}

\begin{definition}[{Monomial ordering~\cite[Chapter 2.2, Definition 1]{cox2015ideals}}]\label{def:monomial_ordering}
    A \textit{monomial ordering} $>$ on $R=\F\left[x_1,\dots,x_n\right]$ is an ordering $>$ on monomials $x^\alpha=x_1^{\alpha_1}\dots x_n^{\alpha_n}$ where $x^\alpha>x^\beta$ if $\alpha>\beta$ for $\alpha,\beta\in\Z_{\geq0}^n$ satisfying:
    \begin{enumerate}
        \item $>$ is a total ordering on $\Z_{\geq0}^n$, i.e. (1) $\alpha>\beta \Rightarrow \alpha\neq\beta$, (2) $\alpha>\beta \Rightarrow \neg\alpha<\beta$, (3) $\left(\alpha>\gamma\right) \wedge \left(\gamma>\beta\right) \Rightarrow \alpha>\beta$, and (4) $\alpha\neq\beta \Rightarrow \left(\alpha>\beta\right)\vee\left(\alpha<\beta\right)$.
        
        \item $\alpha>\beta$ implies $\alpha+\gamma>\beta+\gamma$.
        
        \item $>$ is well-ordering on $\Z_{\geq0}^n$: For any non-empty $A\subseteq\Z_{\geq0}^n$, there exists $\alpha\in A$ such that $\beta>\alpha$ for all $\beta\neq\alpha$.
        This is equivalent to decreasing sequence $\alpha_1>\alpha_2>\dots$ eventually terminates~\cite[Chapter 2.2, Lemma 2]{cox2015ideals}~\footnote{Proof: For $>$ that is not a well-ordering, $\exists S\subseteq\Z_{\geq0}^n$ with no least element, so one can construct $\alpha_1>\alpha_2>\dots$ that goes on forever.
        On the other hand, if there is a non-terminating $\alpha_1>\alpha_2>\dots$, then subset $\left\{\alpha_1,\alpha_2,\dots\right\}\subseteq\Z_{\geq0}^n$ has no least element.}.
    \end{enumerate}
\end{definition}

We could order monomials of polynomials $\F\left[x_1,\dots,x_n\right]$ by using lexicographic ordering on the powers $\alpha=\left(\alpha_1,\dots,\alpha_n\right)\in\Z_{\geq0}^n$, i.e. on the $x^\alpha=x_1^{\alpha_1}\dots x_n^{\alpha_n}$ terms of polynomial $a=\sum_\alpha a_\alpha x^\alpha \in\F\left[x_1,\dots,x_n\right]$.

\begin{definition}[Lexicographic ordering]
    For $\alpha,\beta\in\Z_{\geq0}^n$, $\alpha >_\lex \beta$ if for $j=\min\left\{i\in\left[n\right] : \left(\alpha-\beta\right)_i\neq0\right\}$ (left-most nonzero element of $\alpha-\beta$) we have $\left(\alpha-\beta\right)_j>0$.
    For monomials $x^\alpha,x^\beta$ we have $x^\alpha>_\lex x^\beta$ iff $\alpha>_\lex\beta$.
\end{definition}

It is shown in~\cite[Chapter 2.2, Proposition 4]{cox2015ideals} that lexicographic ordering $>_\lex$ satisfy all conditions in Definition~\ref{def:monomial_ordering}.

We could also define monomial ordering that considers the total degree of the monomial, i.e. $\abs{\alpha}=\sum_{i=1}^n\alpha_i$ for $x^\alpha$.

\begin{definition}
    Graded lex order: $\alpha>_\grlex\beta$ if either:
    \begin{enumerate}
        \item $\abs{\alpha}>\abs{\beta}$, or
        \item $\abs{\alpha}=\abs{\beta}$ and $\alpha>_\lex\beta$.
    \end{enumerate}
\end{definition}
\begin{definition}
    Graded reverse lex order: $\alpha>_\grevlex\beta$ if either:
    \begin{enumerate}
        \item $\abs{\alpha}>\abs{\beta}$, or
        \item $\abs{\alpha}=\abs{\beta}$ and for $j=\max\left\{i : \left(\alpha-\beta\right)_i\neq0\right\}$ (right-most nonzero entry of $\alpha-\beta$) it holds that $\left(\alpha-\beta\right)_j<0$.
    \end{enumerate}
\end{definition}

\begin{example}
\begin{enumerate}
    \item For any $\gamma\in\Z_{\geq0}$, $x_i^\gamma >_\lex x_j^\gamma$ whenever $i<j$.
    \item $x_1^2x_2^2x_3x_4^3 >_\grlex x_1^2x_2x_3^2x_4^3$ and $x_1^2x_2^2x_3x_4^3 >_\grevlex x_1^2x_2x_3^2x_4^3$ since left-most nonzero difference is $1$ and right-most nonzero difference is $-1$.
    \item $x_1^2x_2x_3^2 >_\grlex x_1x_2^3x_3$ and $x_1x_2^3x_3 >_\grevlex x_1^2x_2x_3^2$.
\end{enumerate}
\end{example}

\begin{definition}
    For a monomial ordering $>$ and for $f=\sum_\alpha a_\alpha x^\alpha \in \F\left[x_1,\dots,x_n\right]$:
    \begin{enumerate}
        \item Multidegree of $f$:
        \begin{equation}
            \muldeg f := \max\left\{\alpha : a_\alpha \neq 0 \right\}
        \end{equation}
        i.e. largest degree of monomial with nonzero coefficient in $f$ w.r.t monomial order $>$.

        \item Leading coefficient $\lc f := a_{\muldeg f}$.

        \item Leading monomial $\lm f := x^{\muldeg f}$.

        \item Leading term $\lt f := a_{\muldeg f} x^{\muldeg f}$.
    \end{enumerate}
\end{definition}

\begin{proposition}[Properties of multidegree]
    For $f,g\in\F\left[x_1,\dots,x_n\right]$ and monomial ordering $>$, it holds that:
    \begin{enumerate}
        \item $\muldeg\left(fg\right) = \left(\muldeg f\right)+\left(\muldeg g\right)$.
        \item $\muldeg\left(f+g\right) \leq \max\left\{\muldeg f, \muldeg g\right\}$ whenever $f+g\neq0$.
        Equality holds if $\muldeg f \neq \muldeg g$.
    \end{enumerate}
\end{proposition}

\begin{example}
    For $f = 3x_1^2x_2x_3^2 + 2x_1x_2^3x_3$
    \begin{enumerate}
        \item w.r.t $>_\grlex$ we have $\muldeg f = \left(2,1,2\right)$ and $\lt f = 3x_1^2x_2x_3^2$.
        
        \item w.r.t $>_\grevlex$ we have $\muldeg f = \left(1,3,1\right)$ and $\lt f = 2x_1x_2^3x_3$.
    \end{enumerate}    
\end{example}

\subsubsection{Multi-variable polynomials}

\begin{example}[{Example 2 and 4 of~\cite[Chapter 2.4]{cox2015ideals}}]
    Divide $f = x^2y+xy^2+y^2$ by $f_1 = y^2-1$ and $f_2 = xy-1$.
\end{example}

\subsubsection{Monomial ideals}

\begin{definition}\label{def:monomial_ideal}
    For $A\subseteq \Z_{\geq0}^n$, an ideal $I\subseteq R$ is a \textit{monomial ideal} if it contains all $f\in R$ with finite terms of the form $f=\sum_{\alpha\in A} h_\alpha x^\alpha$ for $h_\alpha\in R$.
\end{definition}

\begin{proposition}[{\cite[Chapter 2.4, Lemma 2]{cox2015ideals}}]\label{prop:monomial_ideal_element_divisibility}
    For monomial ideal $I=\left\langle x^\alpha:\alpha\in A\right\rangle$, it holds that $x^\beta\in I$ iff there exists $\alpha\in A$ such that $x^\beta$ is divisible by $x^\alpha$.
\end{proposition}

\begin{proposition}[{\cite[Lemma 3]{cox2015ideals}}]
    For monomial ideal $I=\left\langle x^\alpha:\alpha\in A\right\rangle$, the following statements are equivalent:
    \begin{enumerate}
        \item $f\in I$.
        \item For $f = \sum_\alpha f_\alpha x^\alpha$, each $x^\alpha$ is in $I$
        \item $f = \sum_\alpha f_\alpha x^\alpha$ for $x^\alpha\in I$.
    \end{enumerate}
    it holds that $f\in I$ implies that each term of $f$ is in $I$.
\end{proposition}
\begin{proof}
    ($1\Rightarrow2$) Since $f\in I$, we have $f=\sum_{\alpha\in A} h_\alpha x^\alpha$.
    Thus for each $\alpha$, we have
    \begin{equation}
        h_\alpha x^\alpha = \left( \sum_\beta a_{\alpha,\beta} x^\beta \right) x^\alpha = \sum_\beta a_{\alpha,\beta} x^{\beta+\alpha} .
    \end{equation}
    Hence we can write $f= \sum_{\alpha,\beta} a_{\alpha,\beta} x^{\beta+\alpha}$.
    Each $\alpha,\beta$ term $a_{\alpha,\beta} x^{\beta+\alpha}$ is in $I$ since it has the form $h_\alpha x^\alpha$.

    ($2\Rightarrow3$) Assume that $f = \sum_\alpha f_\alpha x^\alpha$ for $x^\alpha\in I$.
    Then monomials in such $f$ must be in $I$.

    ($3\Rightarrow1$) Let $f = \sum_\alpha f_\alpha x^\alpha$ for $x^\alpha\in I$.
    Then each $x^\alpha$ must be of the form $x^\alpha = x^{\gamma_\alpha} x^{\beta_\alpha}$ for some $\beta_\alpha\in A$ and some $x^{\gamma_\alpha} \in R$.
    So, 
    \begin{equation}
        f = \sum_{\beta\in A} \left( \sum_{\alpha:\alpha=\gamma_\alpha+\beta} f_\alpha x^{\gamma_\alpha} \right) x^\beta .
    \end{equation}
    Therefore $f = \sum_{\beta\in A} h_\beta x^\beta$ for some $h_\beta\in R$, which implies that $f\in I$.
\end{proof}

\begin{proposition}[{Dickson's lemma~\cite[Chapter 2.5, Theorem 5]{cox2015ideals}}]
    For a monomial ideal $I=\left\langle x^\alpha : \alpha\in A\right\rangle$, there exists $\alpha^{\left(1\right)},\dots,\alpha^{\left(s\right)}\in A$ such that $I = \left\langle x^{\alpha^{\left(1\right)}},\dots,x^{\alpha^{\left(s\right)}}\right\rangle$.
    Namely, any monomial ideal is finitely generated.
\end{proposition}

\subsection{Gr\"obner basis}

Throughout, we denote $R=\F\left[x_1,\dots,x_n\right]$.

\begin{definition}[{\cite[Chapter 2.5, Definition 1]{cox2015ideals}}]\label{def:nonzero_leading_term_ideal}
    For an ideal $I\subseteq R$ over polynomial ring $R$ with a monomial ordering $<$, the \textit{set of leading term of non-zero elements} of $I$ is defined as
    \begin{equation}
        \lt I = \left\{ cx^\alpha : \exists f\in I\setminus\left\{0\right\} ,\; \lt f =cx^\alpha \right\} .
    \end{equation}
\end{definition}

For ideal $I\neq\left\{0\right\}$, it is shown in~\cite[Chapter 2.5, Proposition 3]{cox2015ideals} that $\left\langle\lt I\right\rangle$ is a monomial ideal (c.f. Definition~\ref{def:monomial_ideal}) and there exists $g_1,\dots,g_t\in I$ such that 
\begin{equation}
    \left\langle\lt I\right\rangle = \left\langle\lt g_1,\dots, \lt g_t\right\rangle .
\end{equation}
Namely, an ideal generated by the leading term of non-zero elements $\lt I$ is finitely generated.

\begin{definition}[{Gr\"obner basis~\cite[Chapter 2.5, Definition 5]{cox2015ideals}}]\label{def:grobner_basis}
    For a nonzero ideal $I\subseteq R$ over polynomial ring $R$ with monomial ordering $<$, a subset $G=\left\{g_1,\dots,g_t\right\}\subseteq I$ is a \textit{Gr\"obner basis} of $I$ whenever
    \begin{equation}
        \left\langle\lt g_1 , \dots, \lt g_t\right\rangle = \left\langle\lt I\right\rangle .
    \end{equation}
\end{definition}

To show that $g_1,\dots,g_t\in I$ is a Gr\"obner basis of $I=\left\langle g_1,\dots,g_t\right\rangle$, we must check whether $\lt f\in \left\langle\lt g_1,\dots, \lt g_t\right\rangle$ for all nonzero $f\in I$.
By Proposition~\ref{prop:monomial_ideal_element_divisibility} we can do this by checking whether $\lt f$ is divisible by some $\lt g_i$.

\section{Supplementary loss sweeps for BARC Hex codes}
\label{sec:appendix-hex-loss-sweeps}

This section contains the plots for Hex2, Hex1, Elliptical Hex1 and Elliptical Hex2 families corresponding to results presented in \ref{sec:method1-results} and \ref{sec:method2-results}.  The main text retains the representative Hex3 and Elliptical Hex3
plots. Each figure contains one family-wise pair: the left panel
compares the candidate with QSC and the right panel compares it with QCC. 

\subsection{Independently optimized amplitudes}
\label{sec:appendix-hex-loss-sweeps-method1}

Candidate and reference amplitudes use their own sampled optima at $\gamma_0=0.095$ and remain fixed across the loss sweep, while recovery is reoptimized at each loss probability.

\begin{figure}[H]
\centering
\includegraphics[width=0.475\linewidth]{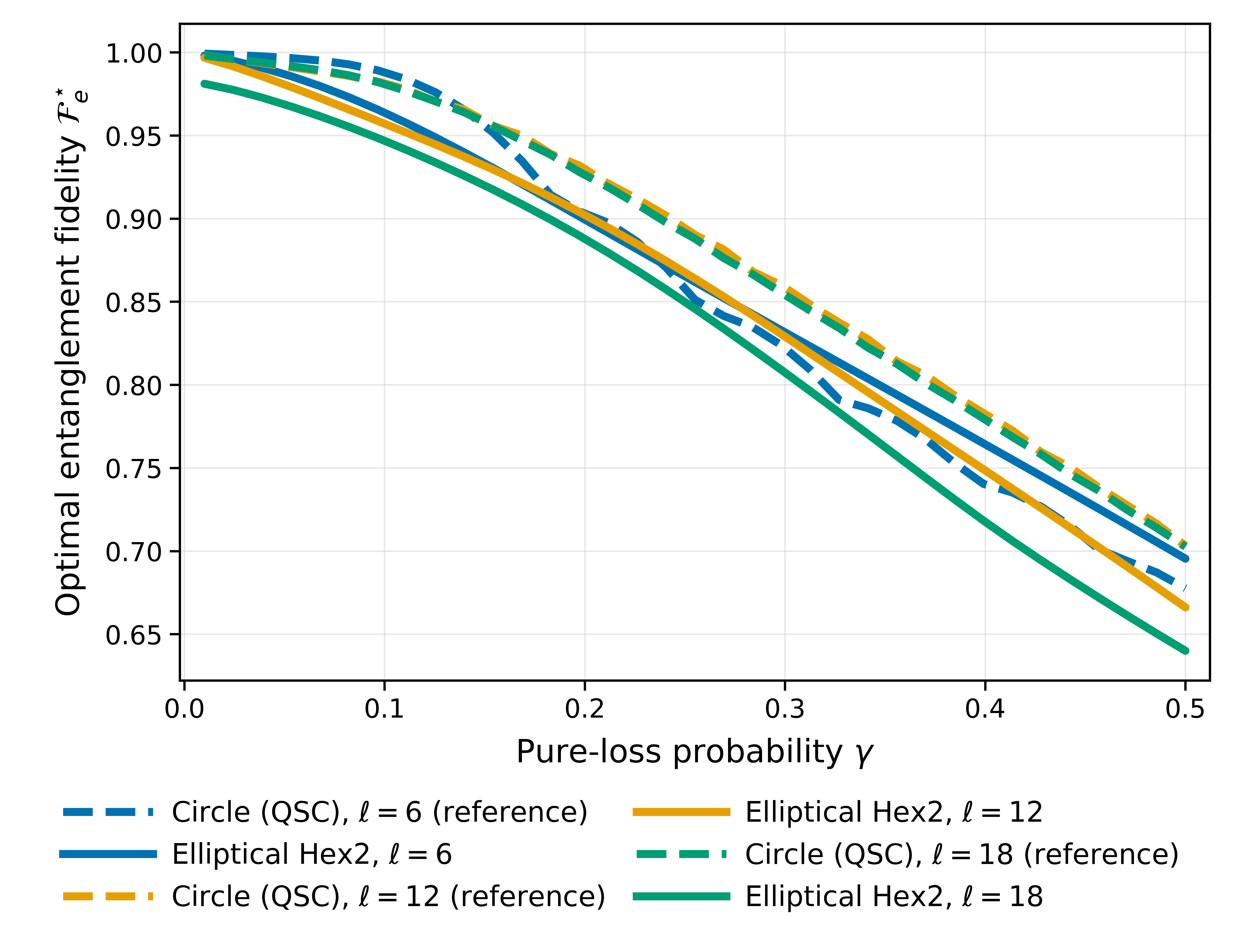}\hfill%
\includegraphics[width=0.475\linewidth]{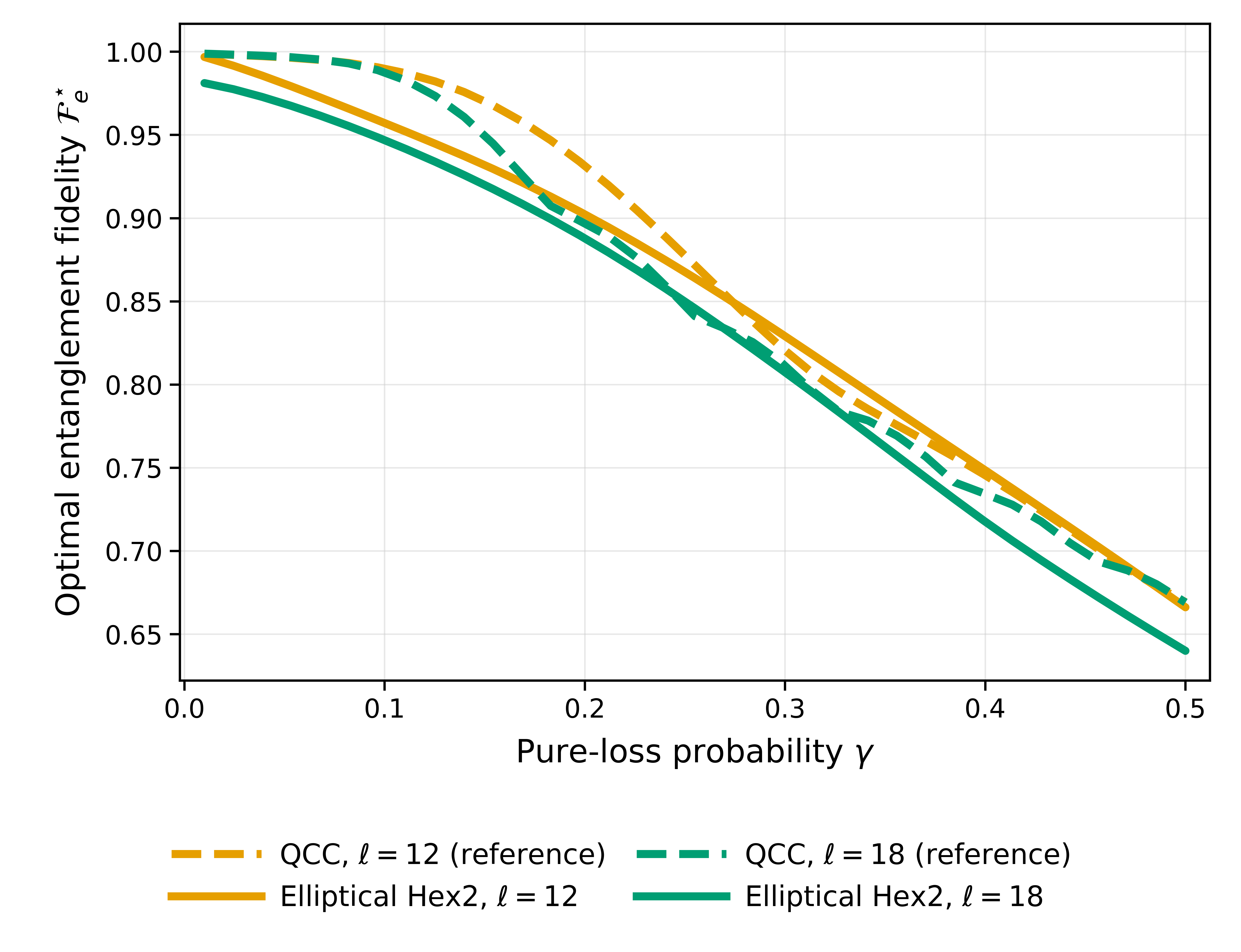}
\caption{Elliptical Hex2}
\label{fig:method1-elliptical-hex2-loss}
\end{figure}

\begin{figure}[H]
\centering
\includegraphics[width=0.475\linewidth]{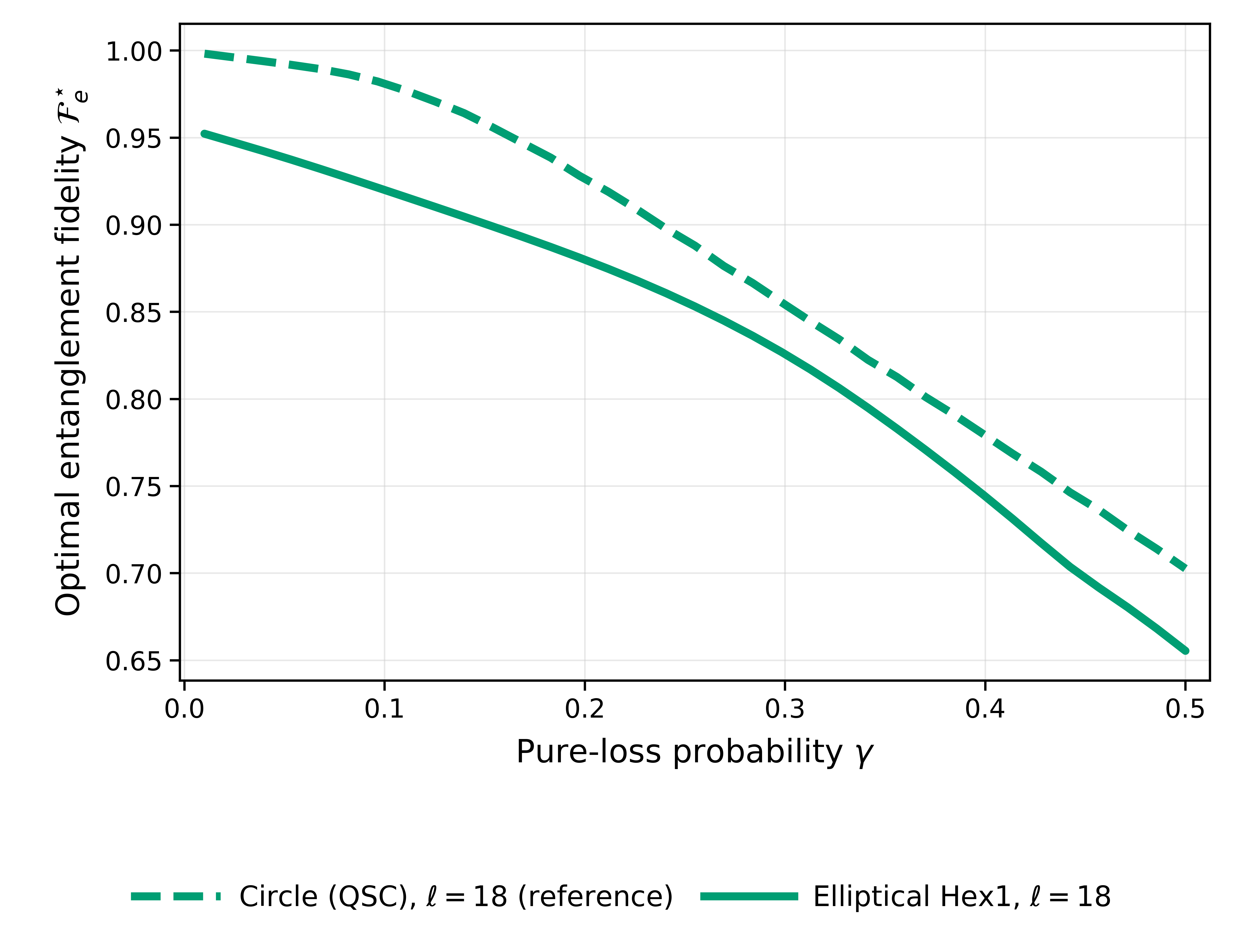}\hfill%
\includegraphics[width=0.475\linewidth]{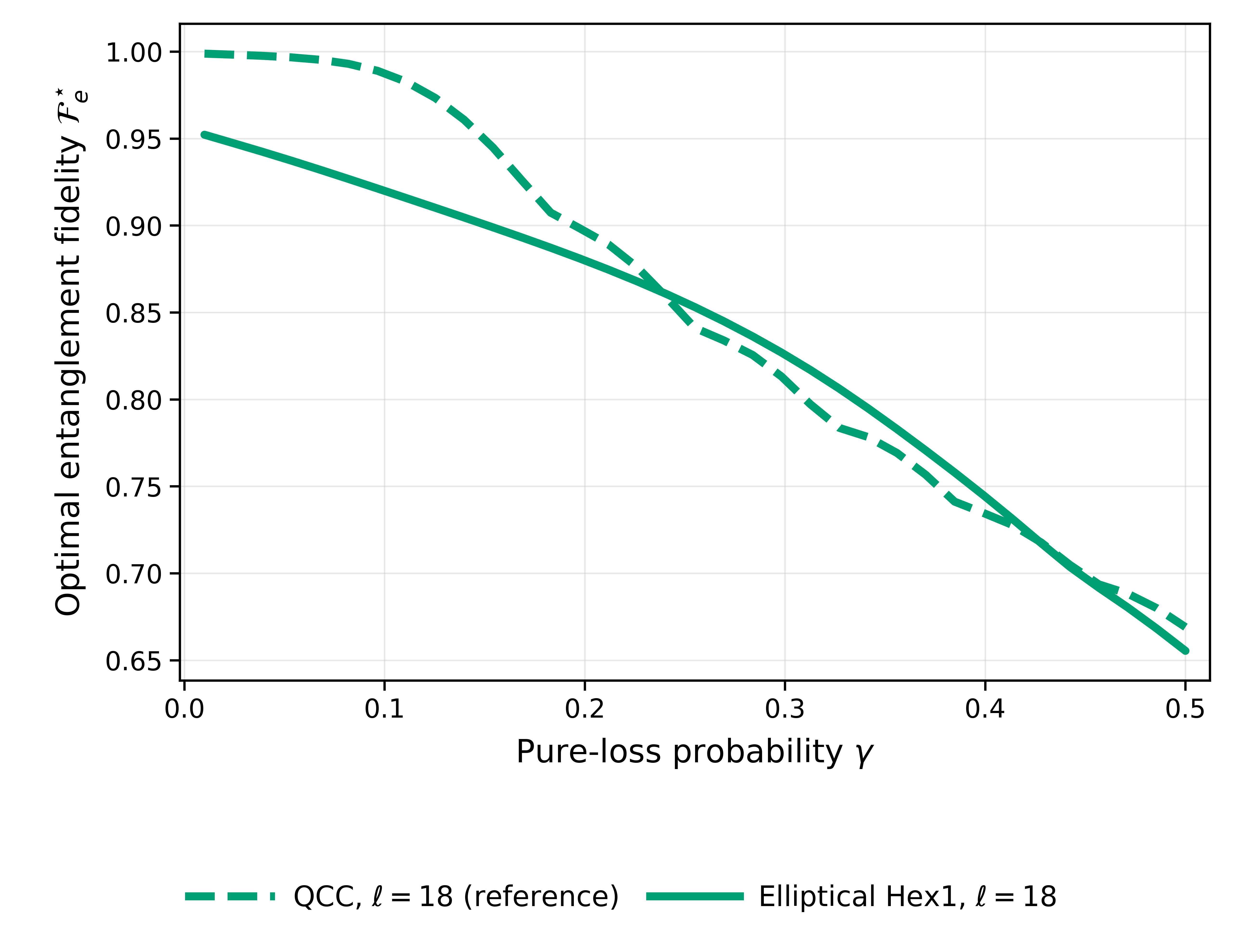}
\caption{Elliptical Hex1}
\label{fig:method1-elliptical-hex1-loss}
\end{figure}

\begin{figure}[H]
\centering
\includegraphics[width=0.475\linewidth]{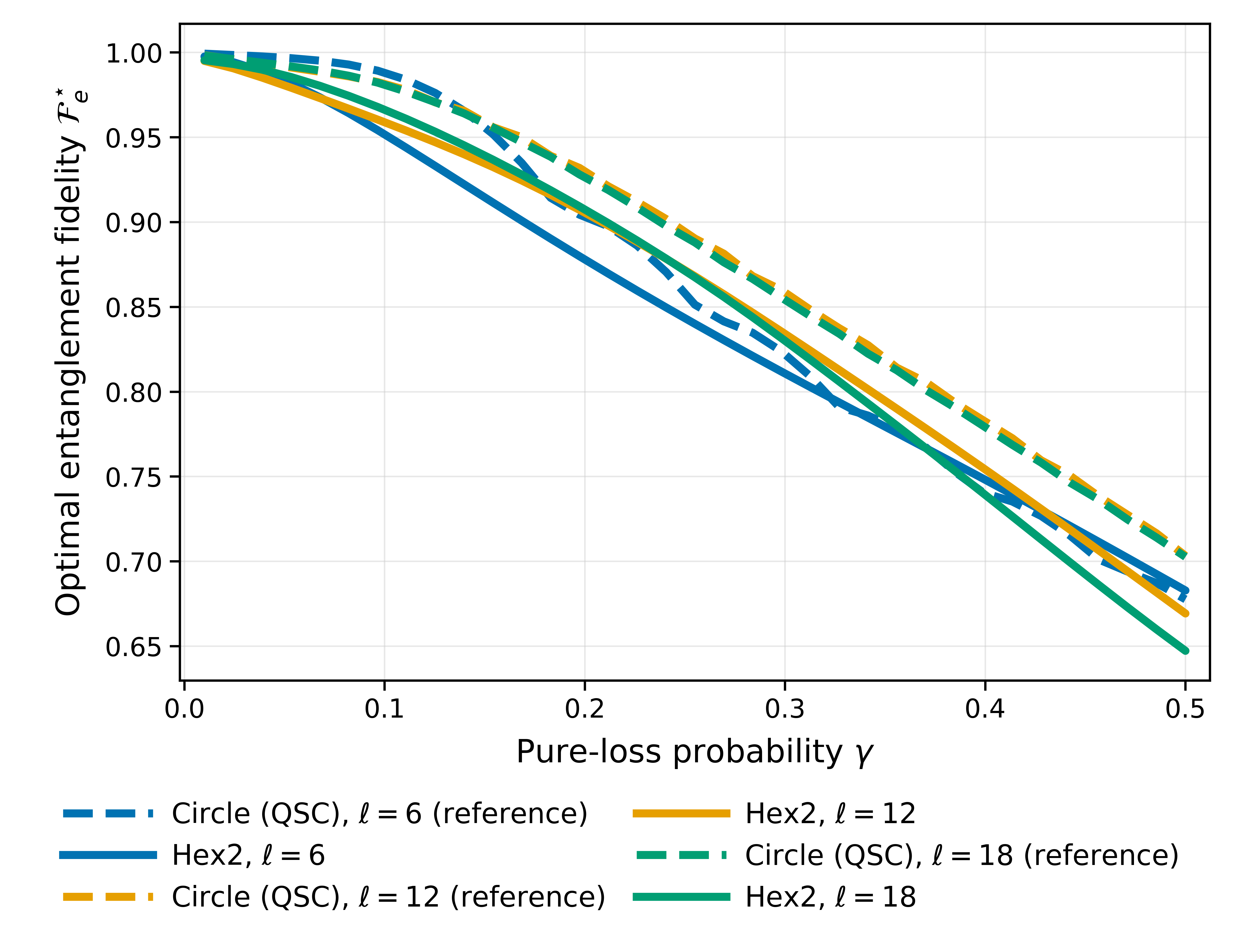}\hfill%
\includegraphics[width=0.475\linewidth]{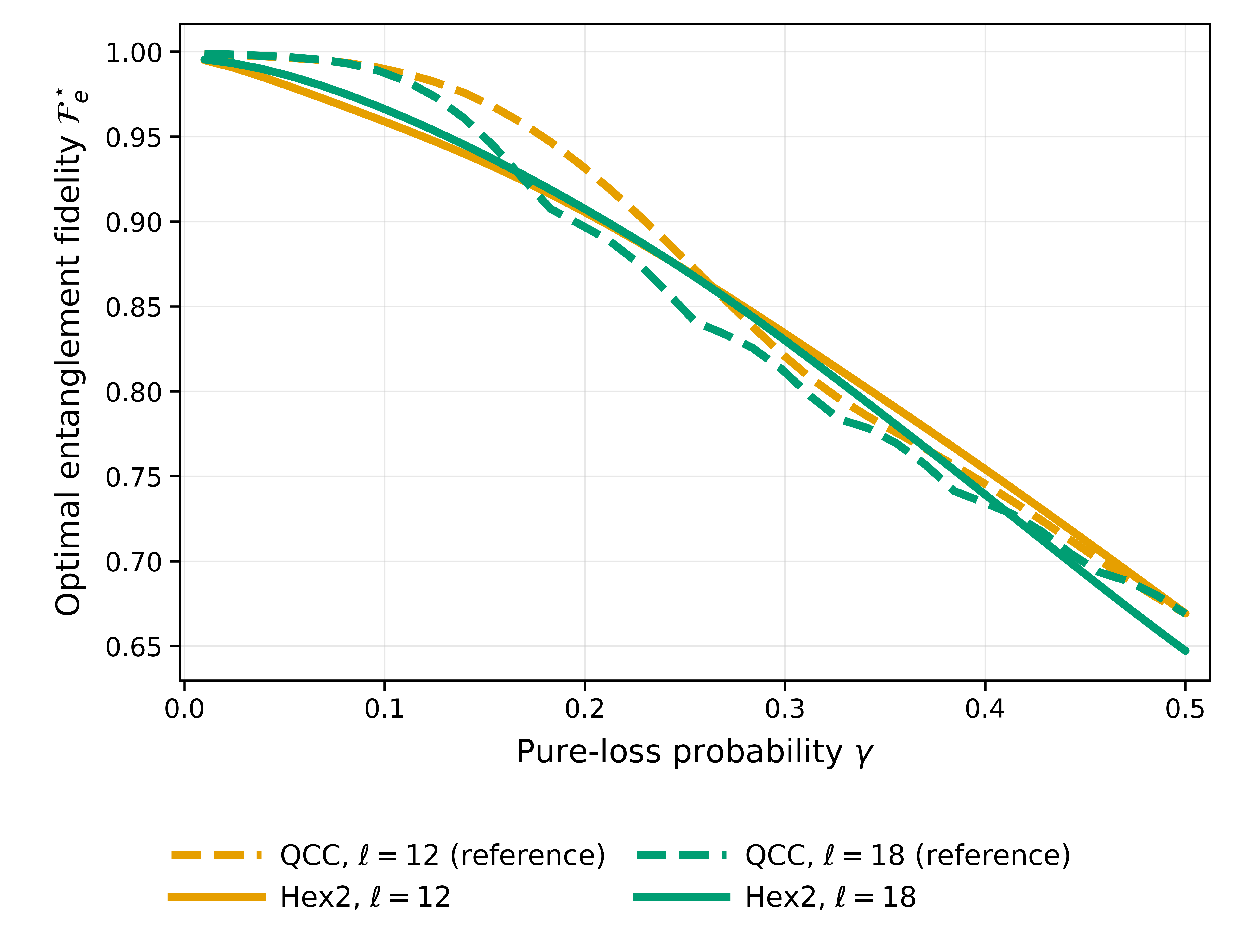}
\caption{Hex2}
\label{fig:method1-hex2-loss}
\end{figure}

\begin{figure}[H]
\centering
\includegraphics[width=0.475\linewidth]{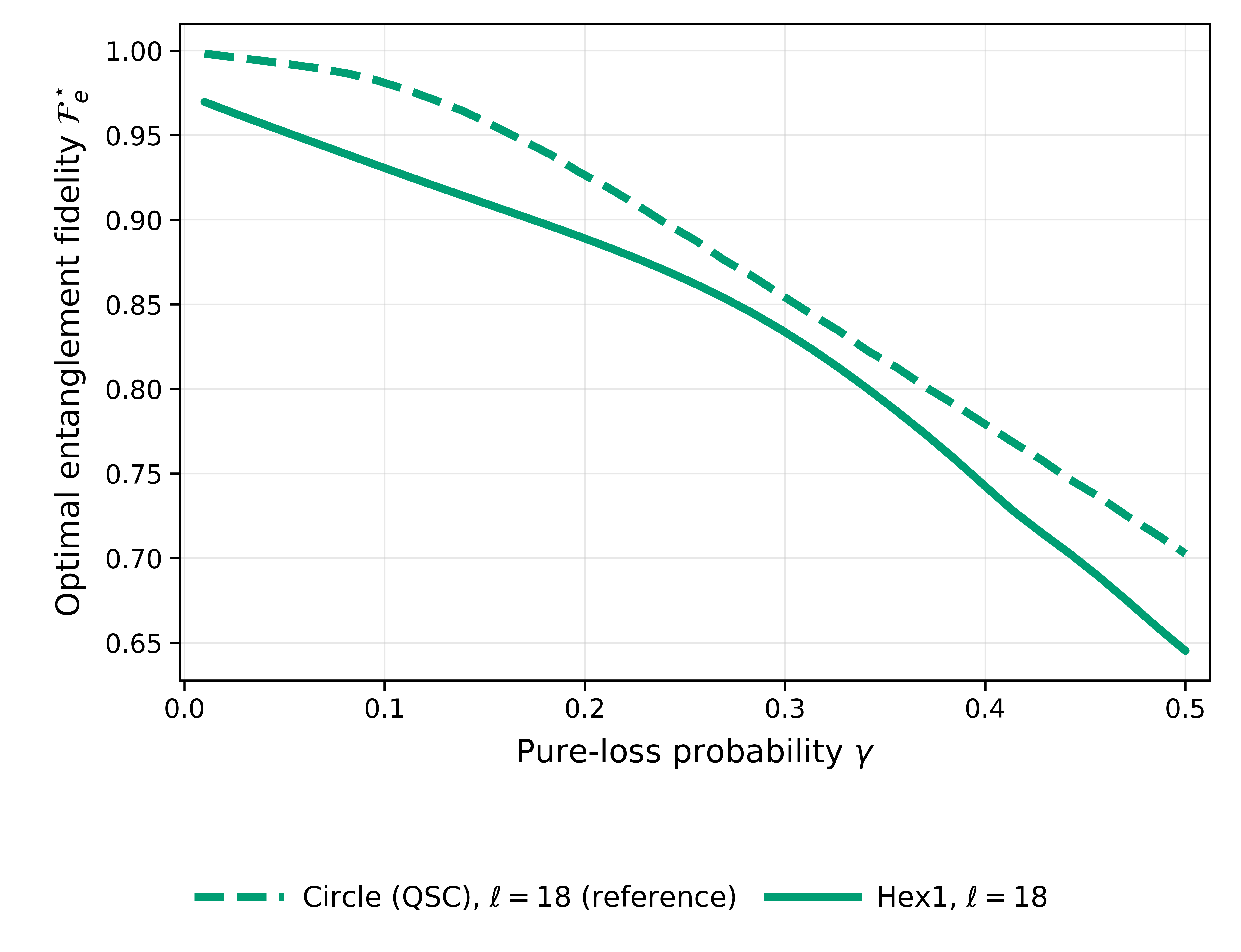}\hfill%
\includegraphics[width=0.475\linewidth]{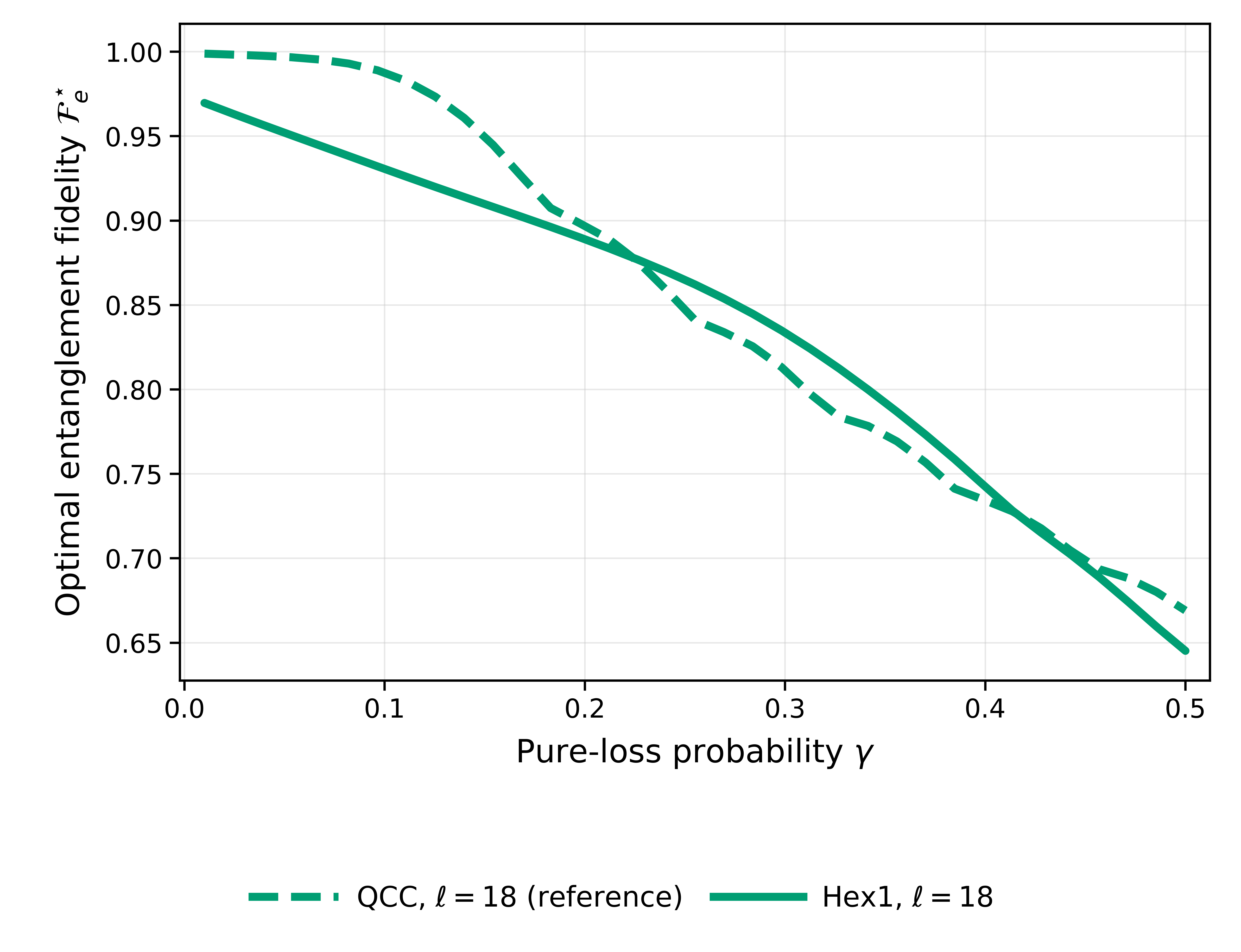}
\caption{Hex1}
\label{fig:method1-hex1-loss}
\end{figure}
\FloatBarrier

\subsection{Energy-matched amplitudes}
\label{sec:appendix-hex-loss-sweeps-method2}

The
candidate is compared with QSC (left) and QCC (right) at equal constellation
size. Each reference is selected under the candidate photon budget, and
recovery is reoptimized at each sampled loss probability.

\begin{figure}[H]
\centering
\includegraphics[width=0.475\linewidth]{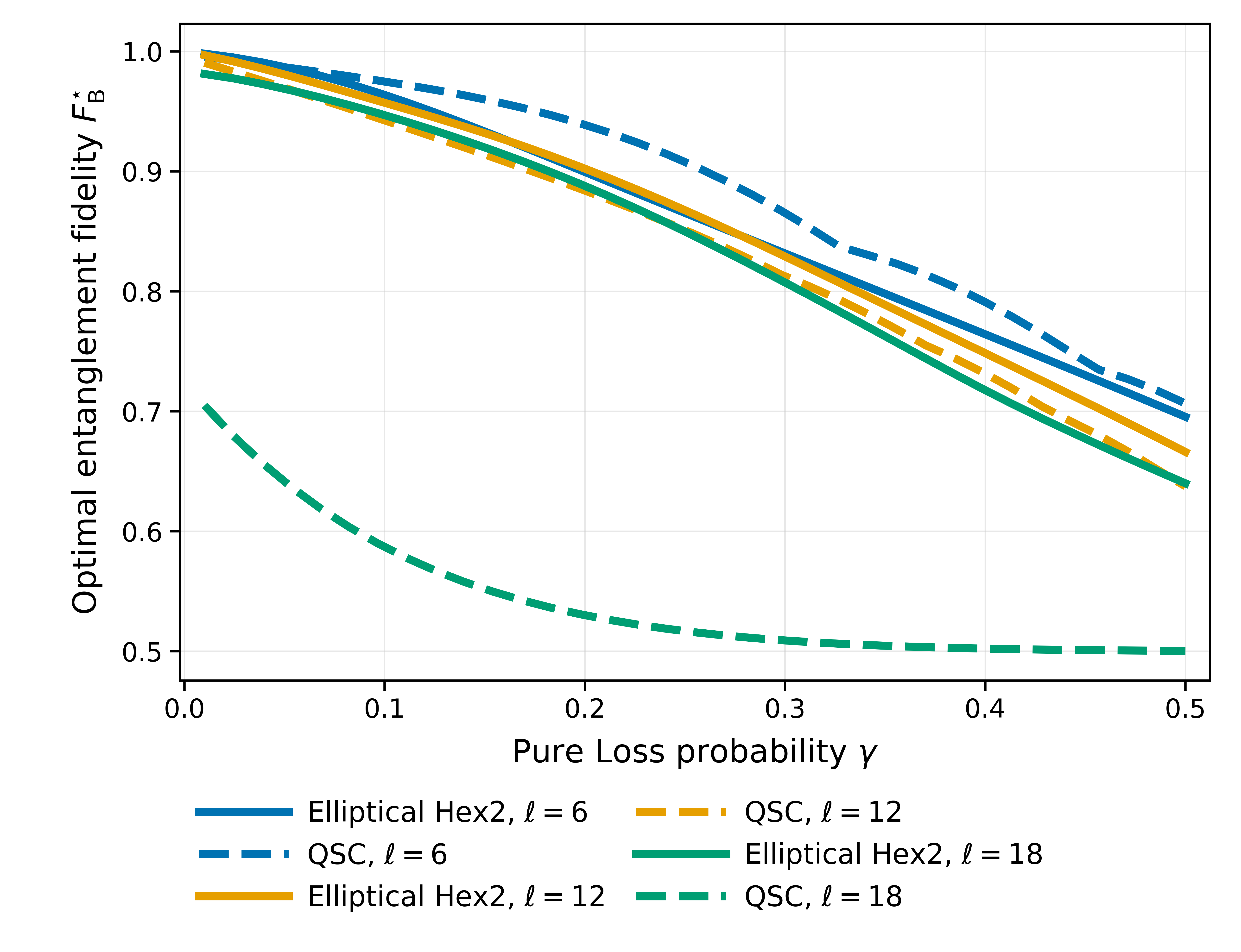}\hfill%
\includegraphics[width=0.475\linewidth]{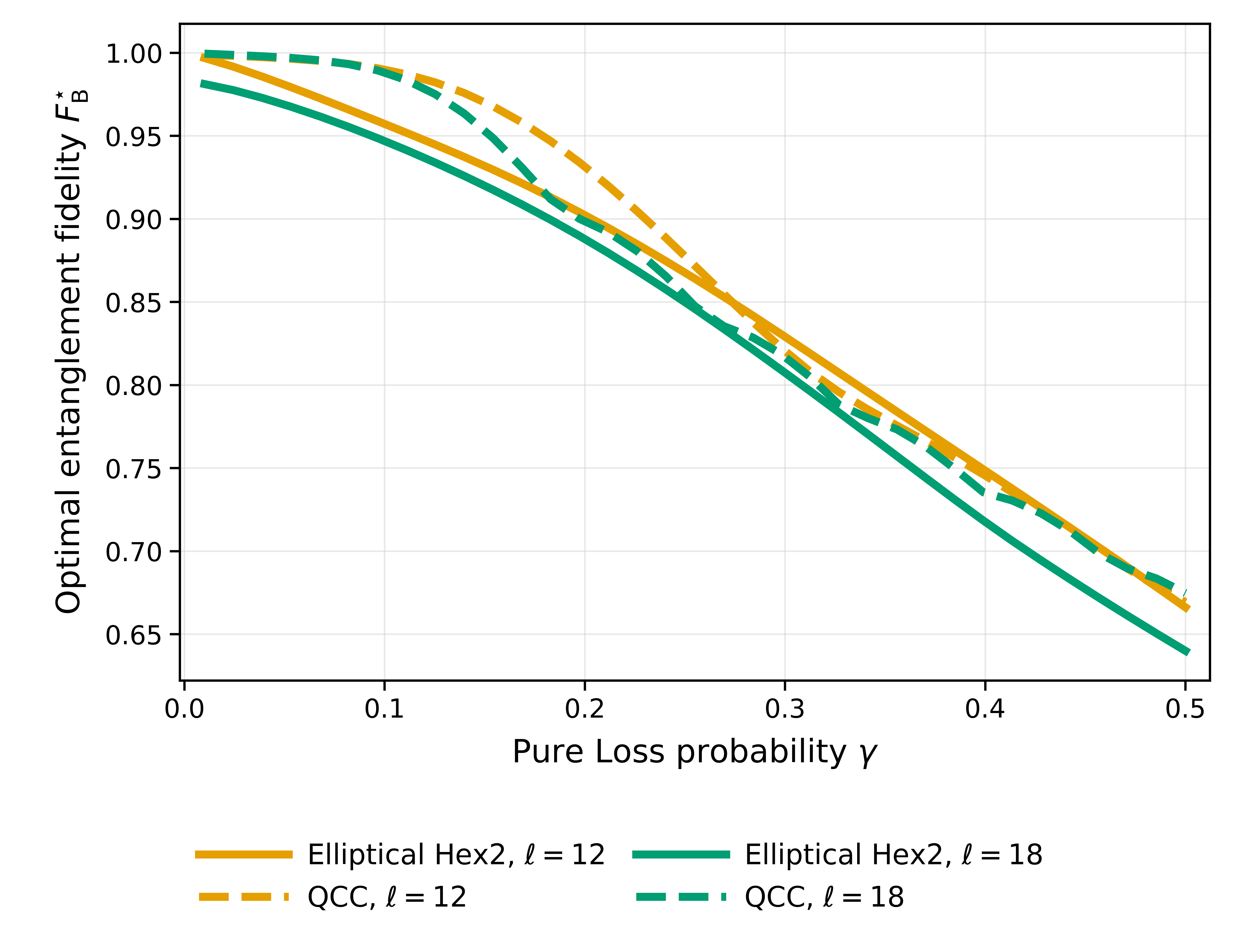}
\caption{ Elliptical Hex2. }
\label{fig:method2-elliptical-hex2-loss}
\end{figure}

\begin{figure}[H]
\centering
\includegraphics[width=0.475\linewidth]{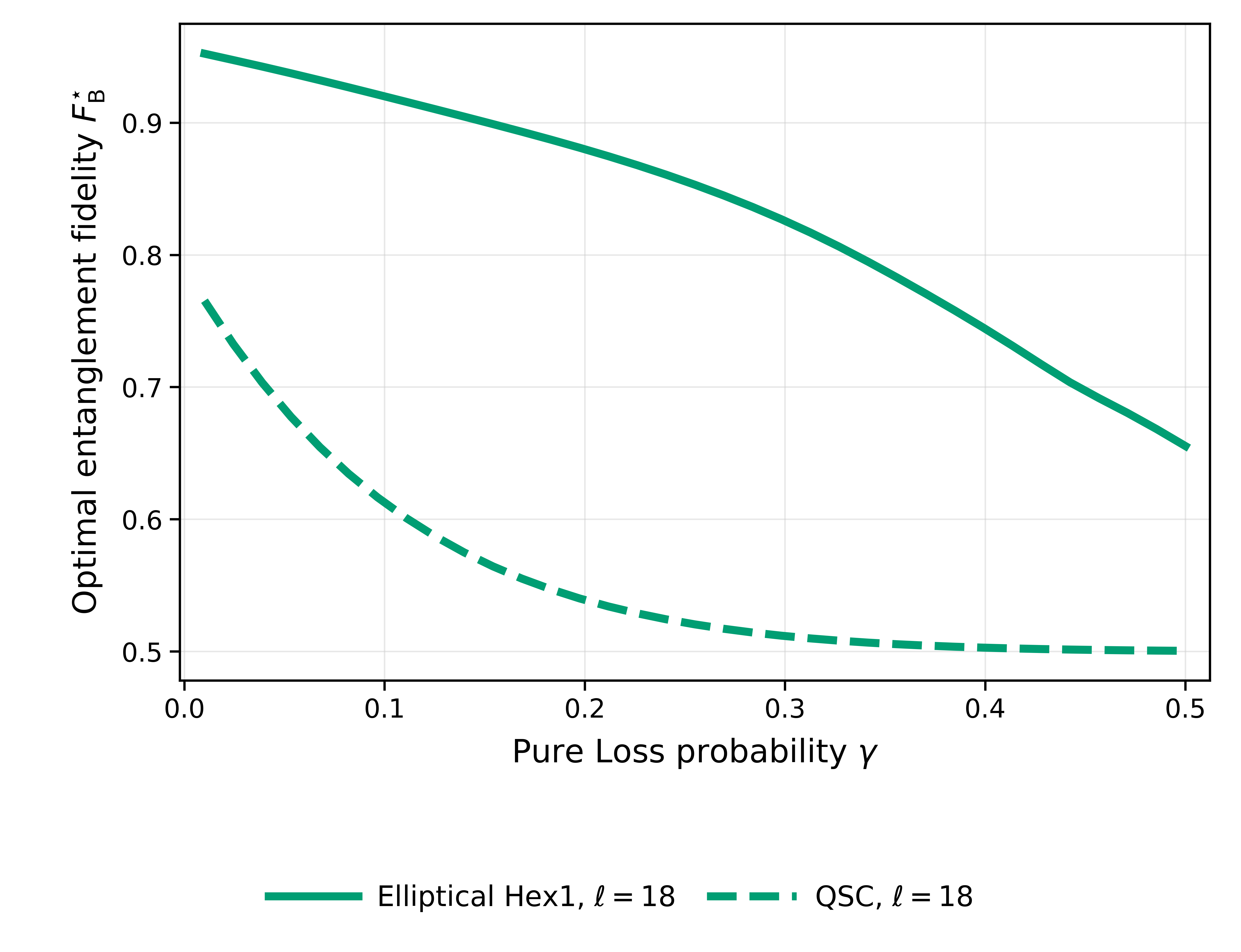}\hfill%
\includegraphics[width=0.475\linewidth]{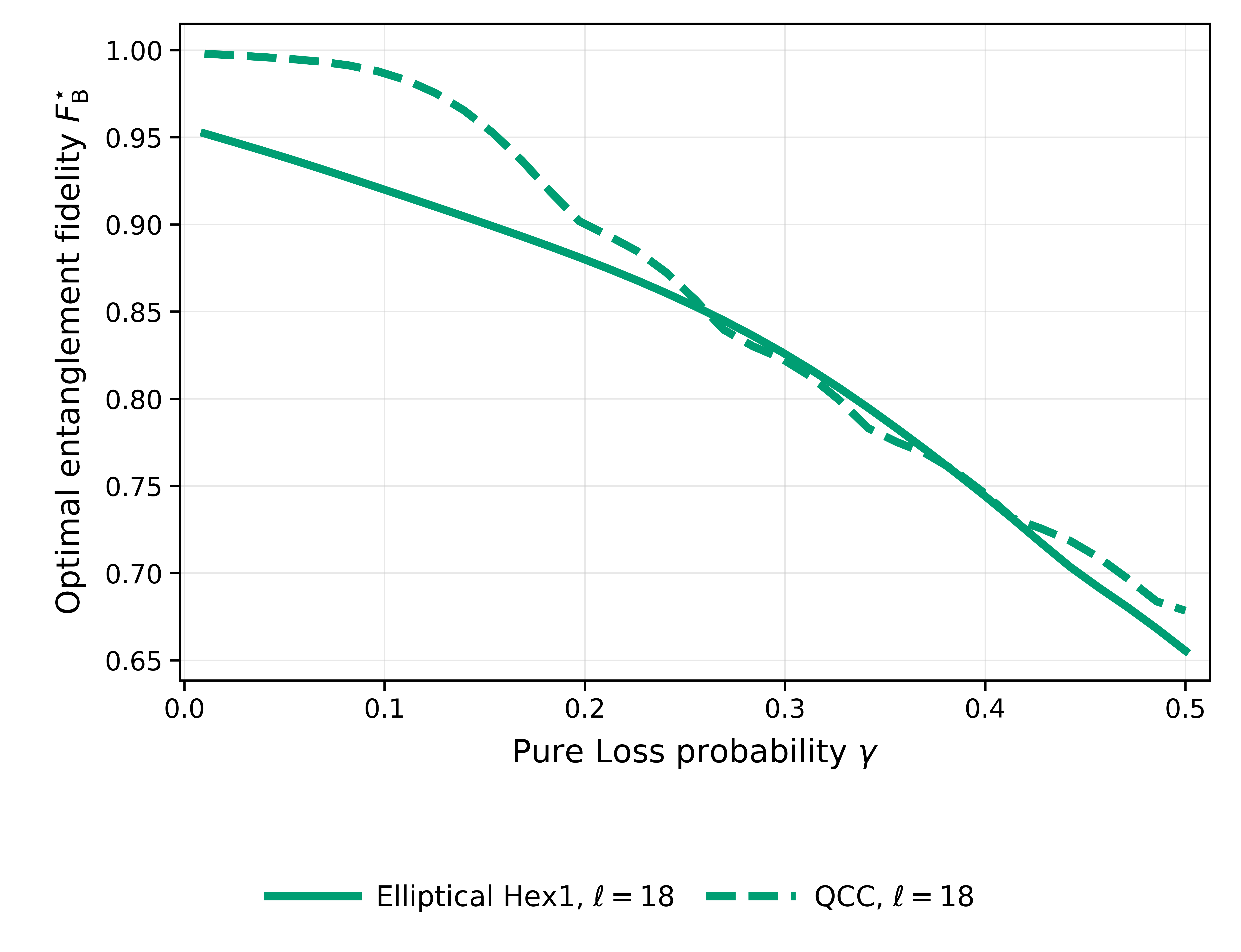}
\caption{ Elliptical Hex1. }
\label{fig:method2-elliptical-hex1-loss}
\end{figure}

\begin{figure}[H]
\centering
\includegraphics[width=0.475\linewidth]{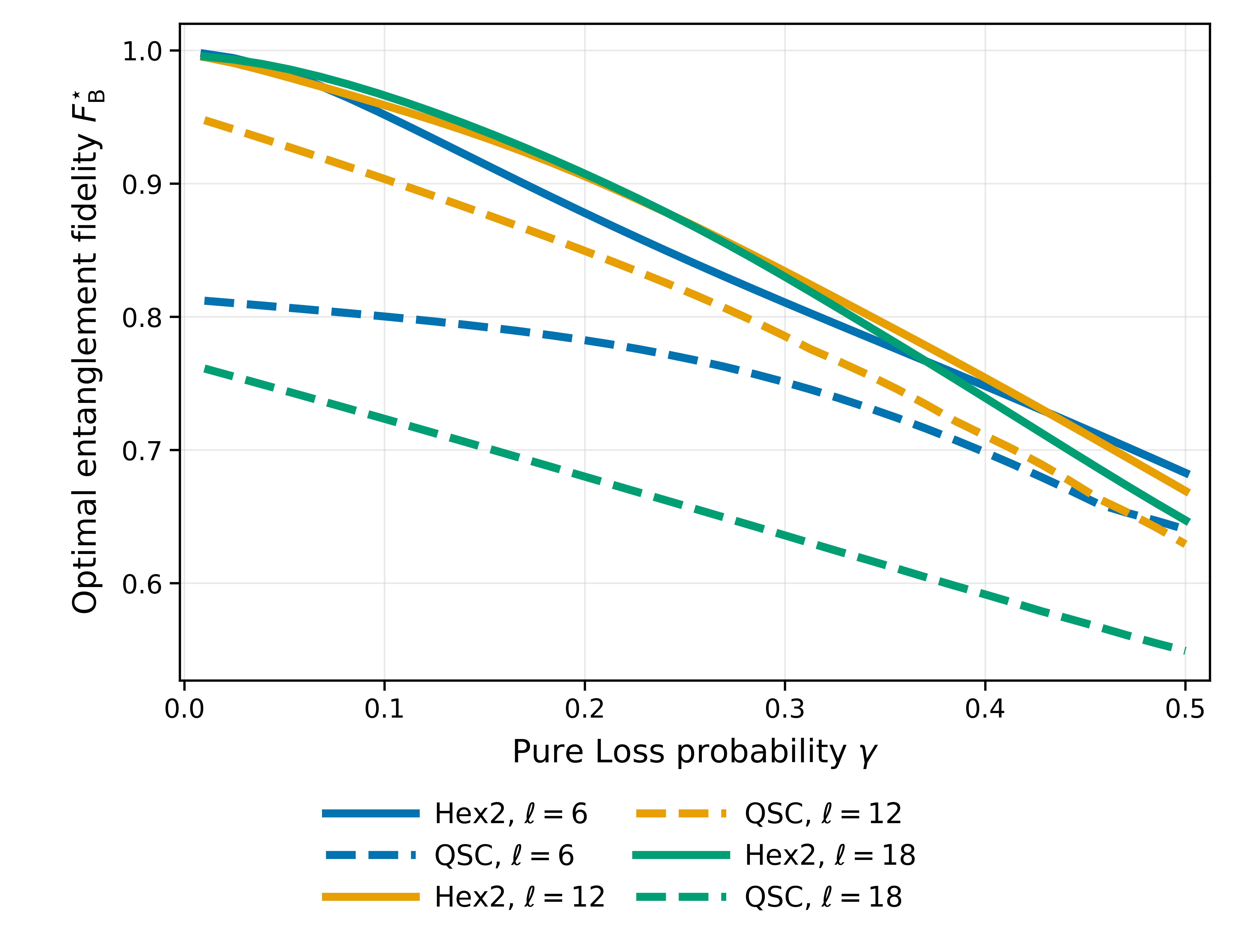}\hfill%
\includegraphics[width=0.475\linewidth]{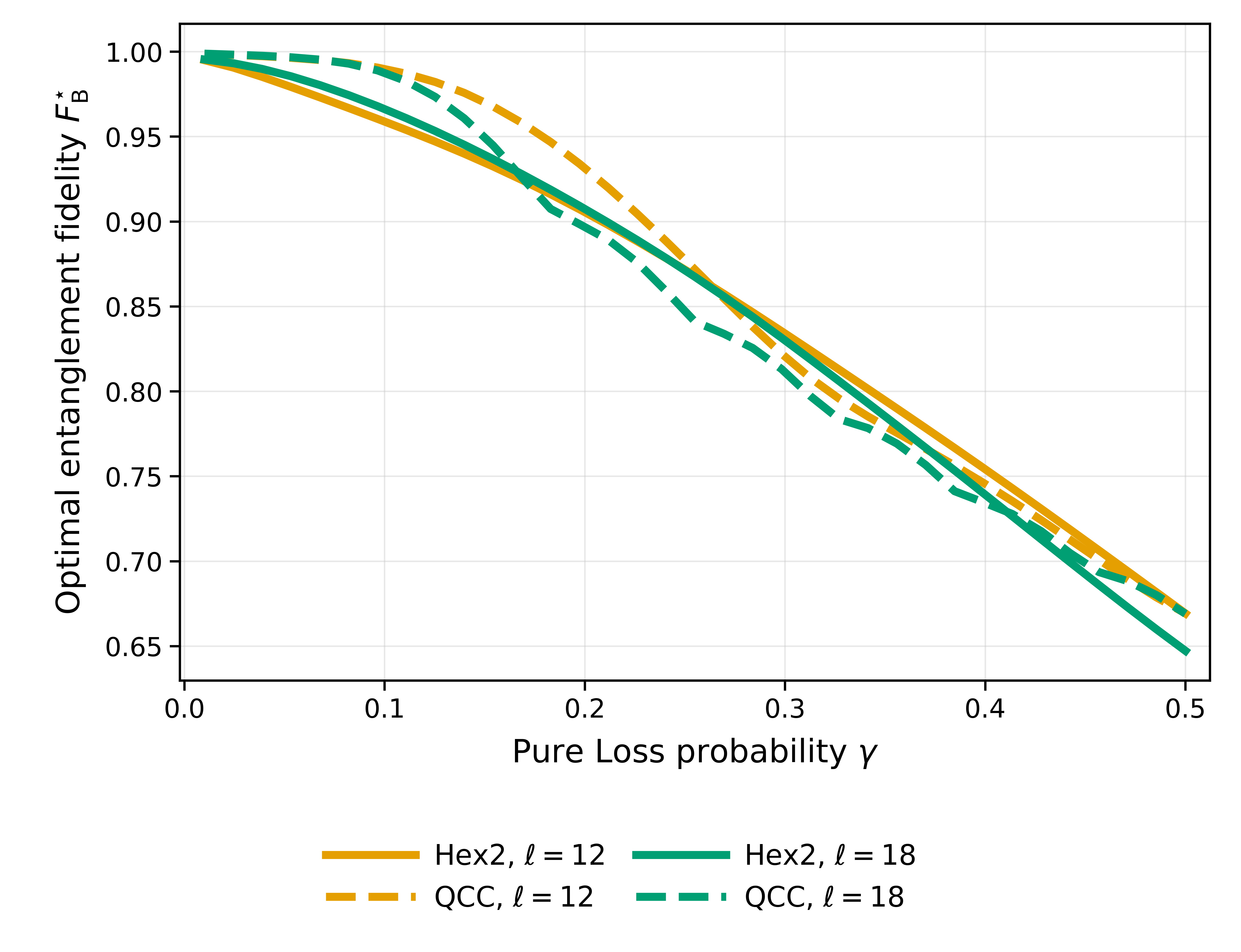}
\caption{ Hex2. }
\label{fig:method2-hex2-loss}
\end{figure}

\begin{figure}[H]
\centering
\includegraphics[width=0.475\linewidth]{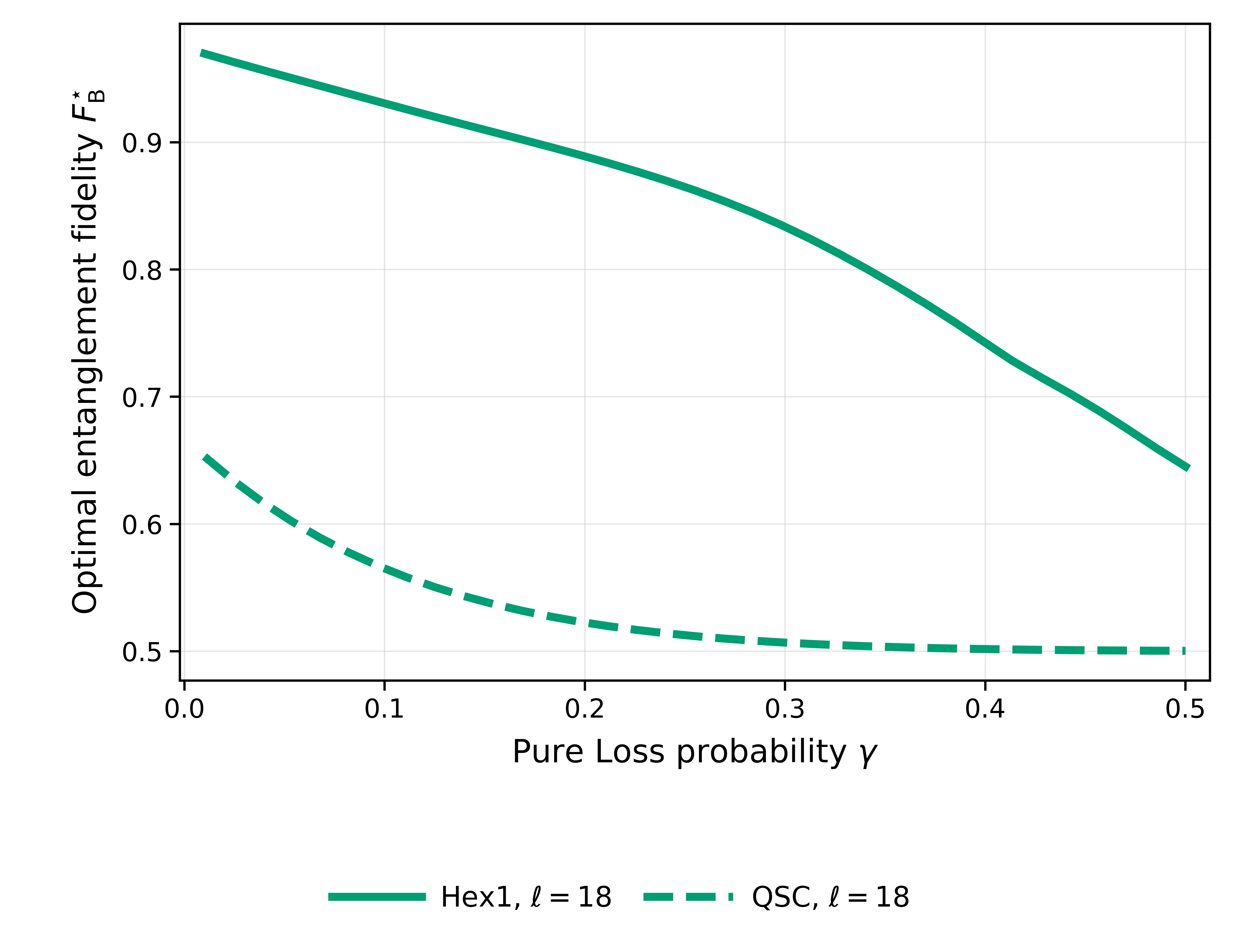}\hfill%
\includegraphics[width=0.475\linewidth]{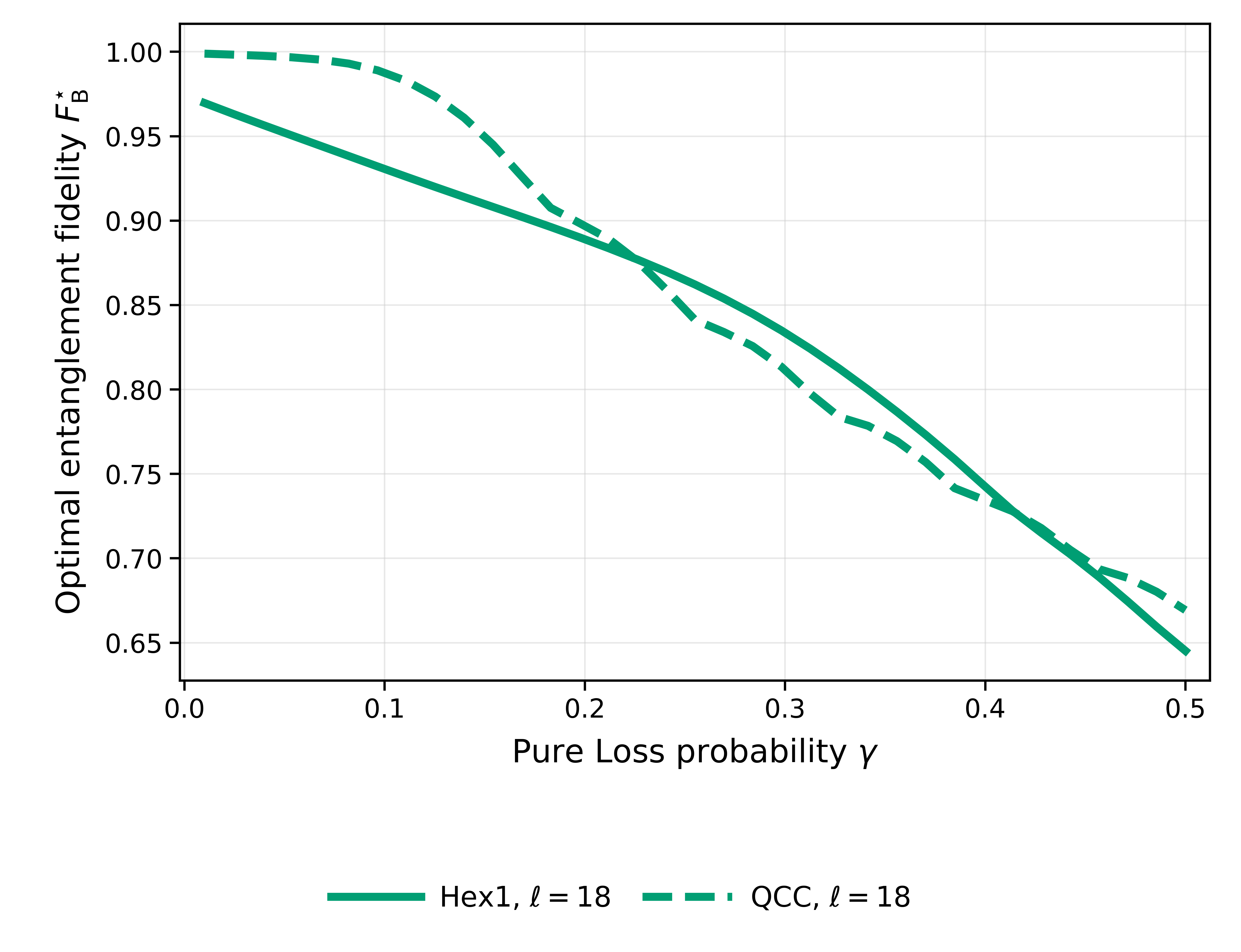}
\caption{ Hex1. }
\label{fig:method2-hex1-loss}
\end{figure}
\FloatBarrier

\section{Supplementary figures for size matched comparison}
\label{sec:appendix-method1-extra-figures}

This section collects the relative-infidelity curves referenced in
Section~\ref{sec:method1-results}. The selected amplitudes and corresponding
photon numbers are reported in Table~\ref{tab:equalell-parameters}. Each figure
contains one family-wise pair, with the QSC comparison on the left and the QCC
comparison on the right.

\begin{figure}[H]
\centering
\includegraphics[width=0.495\linewidth]{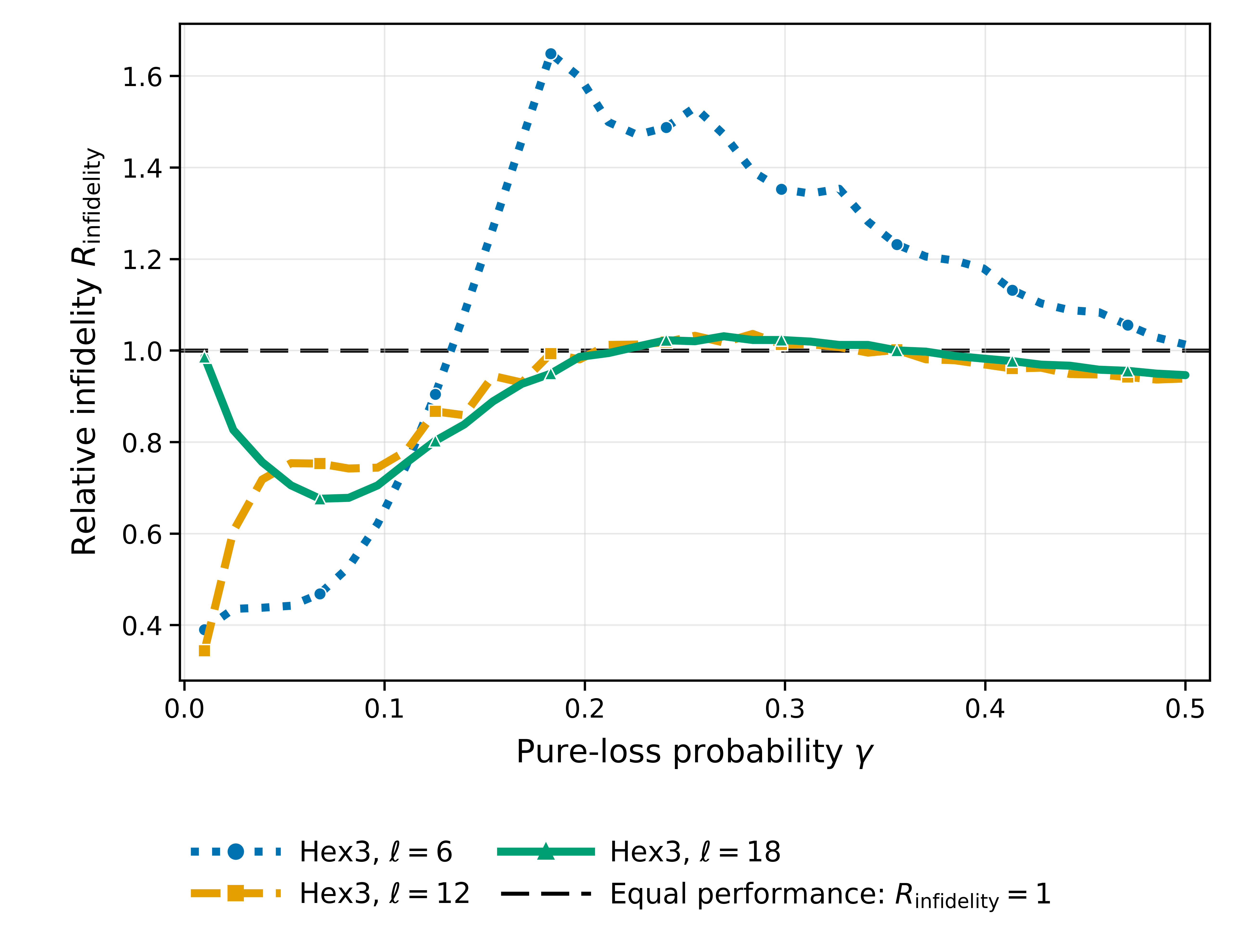}\hfill%
\includegraphics[width=0.495\linewidth]{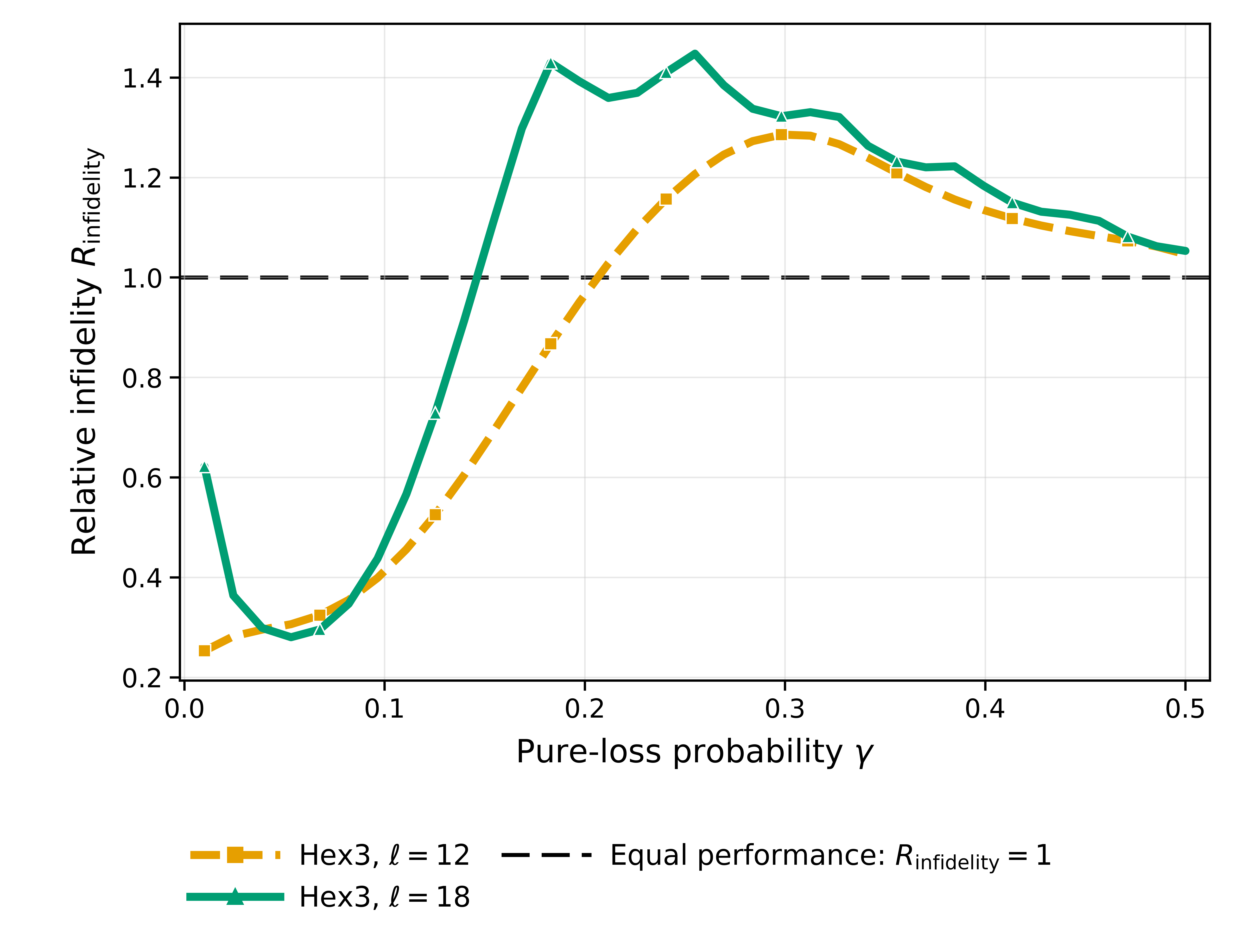}
\caption{ Hex3 }
\label{fig:appendix-method1-grid}
\label{fig:appendix-method1-hex3-qsc-c}
\label{fig:appendix-method1-hex3-qcc-c}
\label{fig:equalell-hex3-qsc}
\label{fig:equalell-hex3-qcc}
\end{figure}

\begin{figure}[H]
\centering
\includegraphics[width=0.495\linewidth]{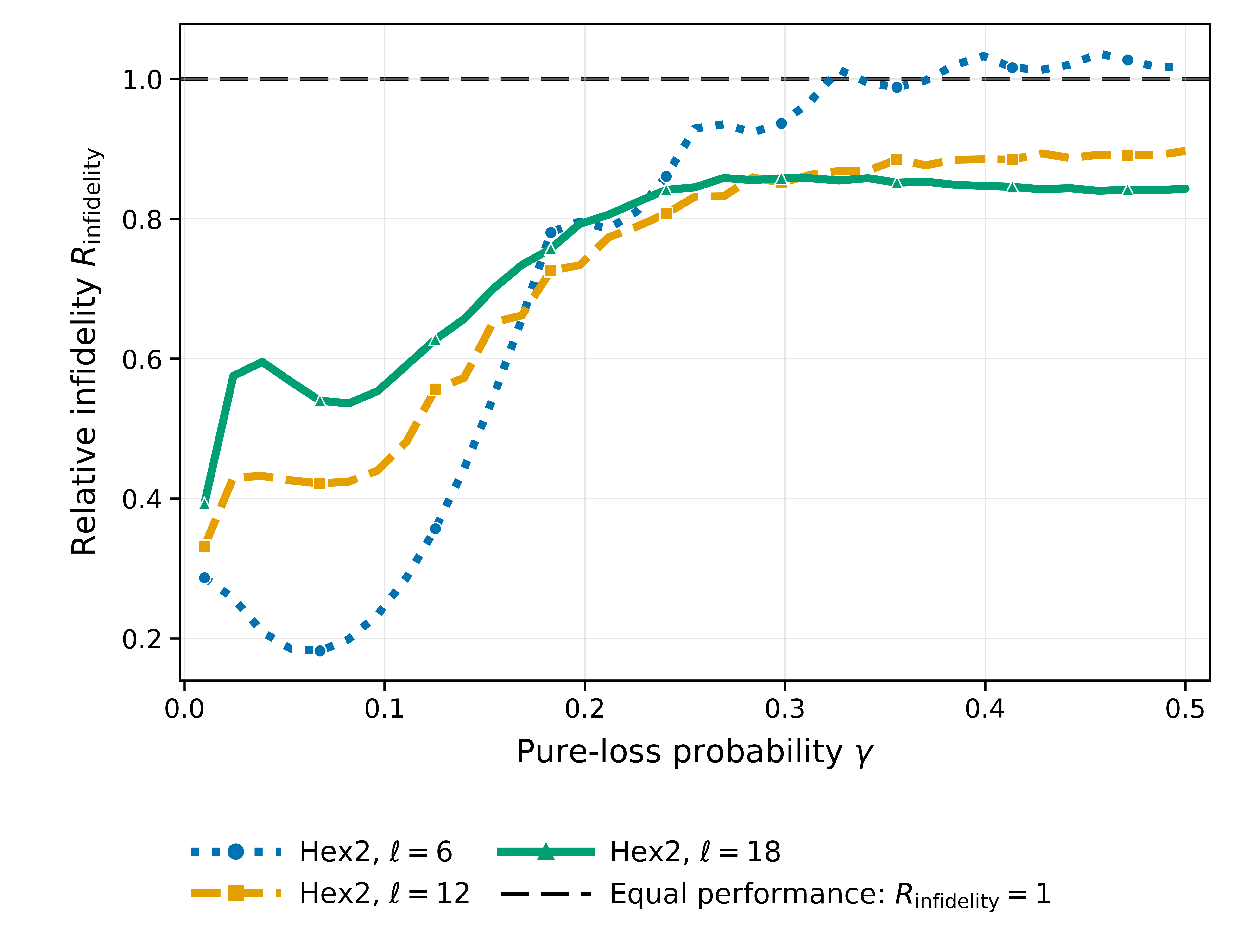}\hfill%
\includegraphics[width=0.495\linewidth]{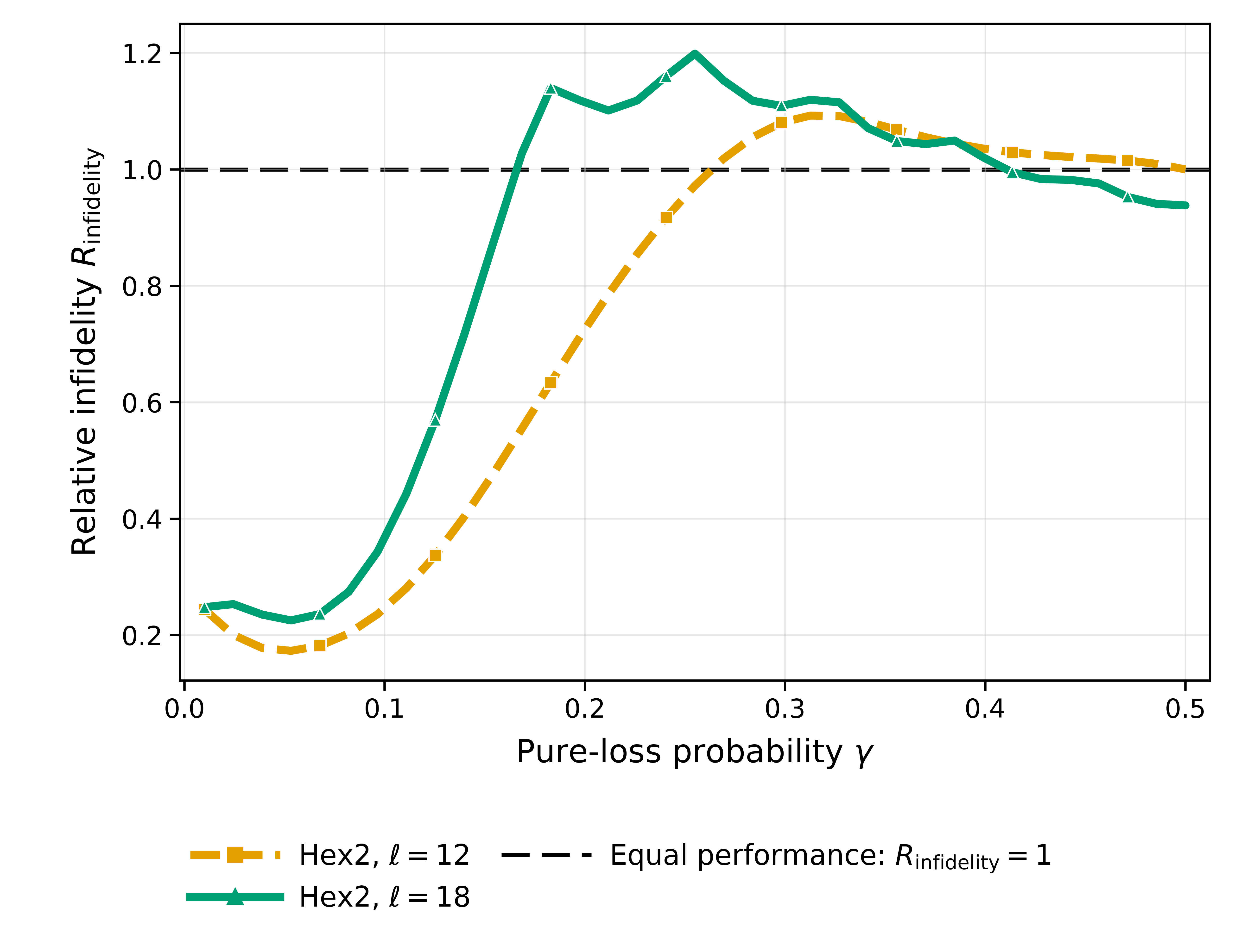}
\caption{ Hex2 }
\label{fig:appendix-method1-hex2-qsc-c}
\label{fig:appendix-method1-hex2-qcc-c}
\label{fig:appendix-method1-hex2-qsc}
\label{fig:appendix-method1-hex2-qcc}
\end{figure}

\begin{figure}[H]
\centering
\includegraphics[width=0.495\linewidth]{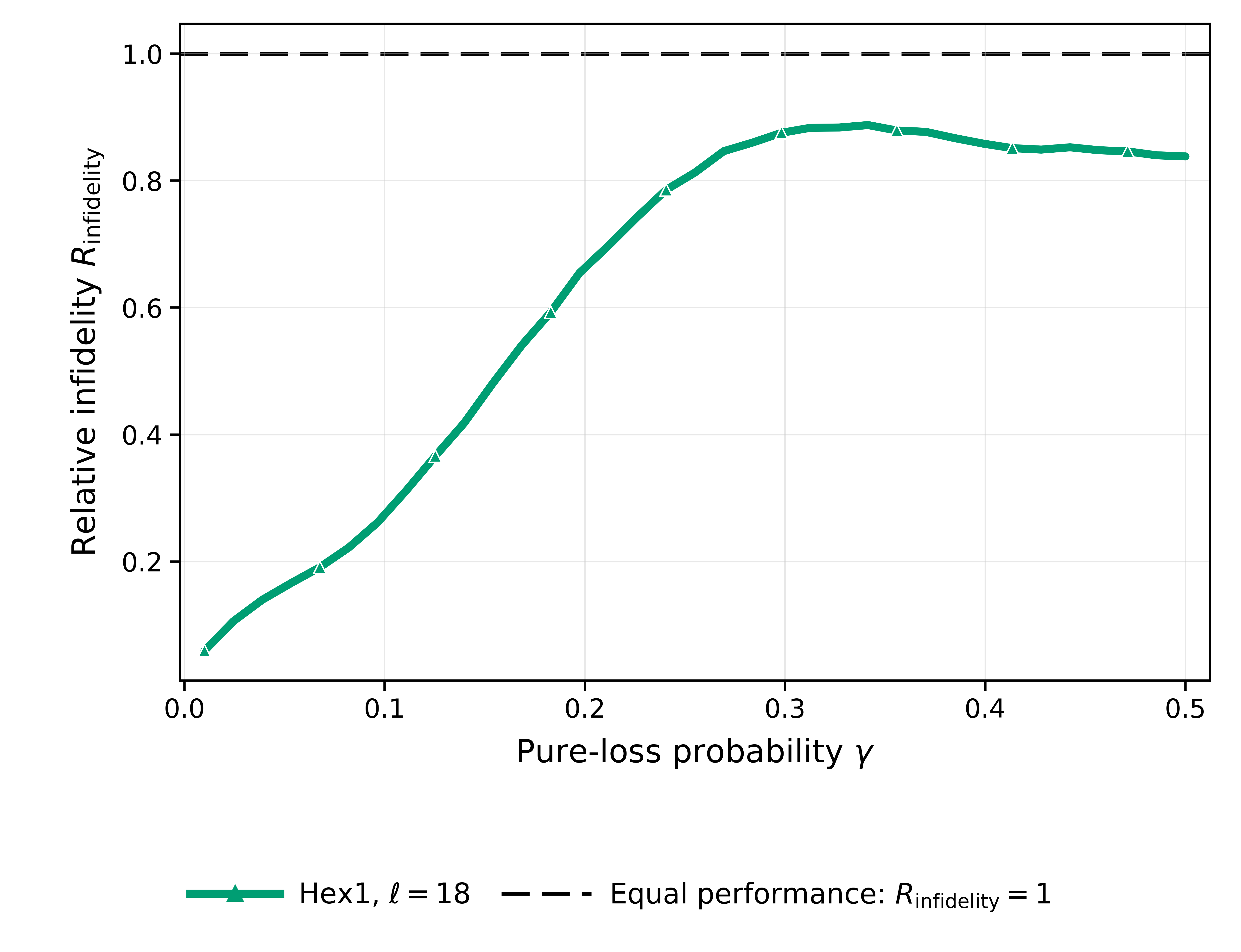}\hfill%
\includegraphics[width=0.495\linewidth]{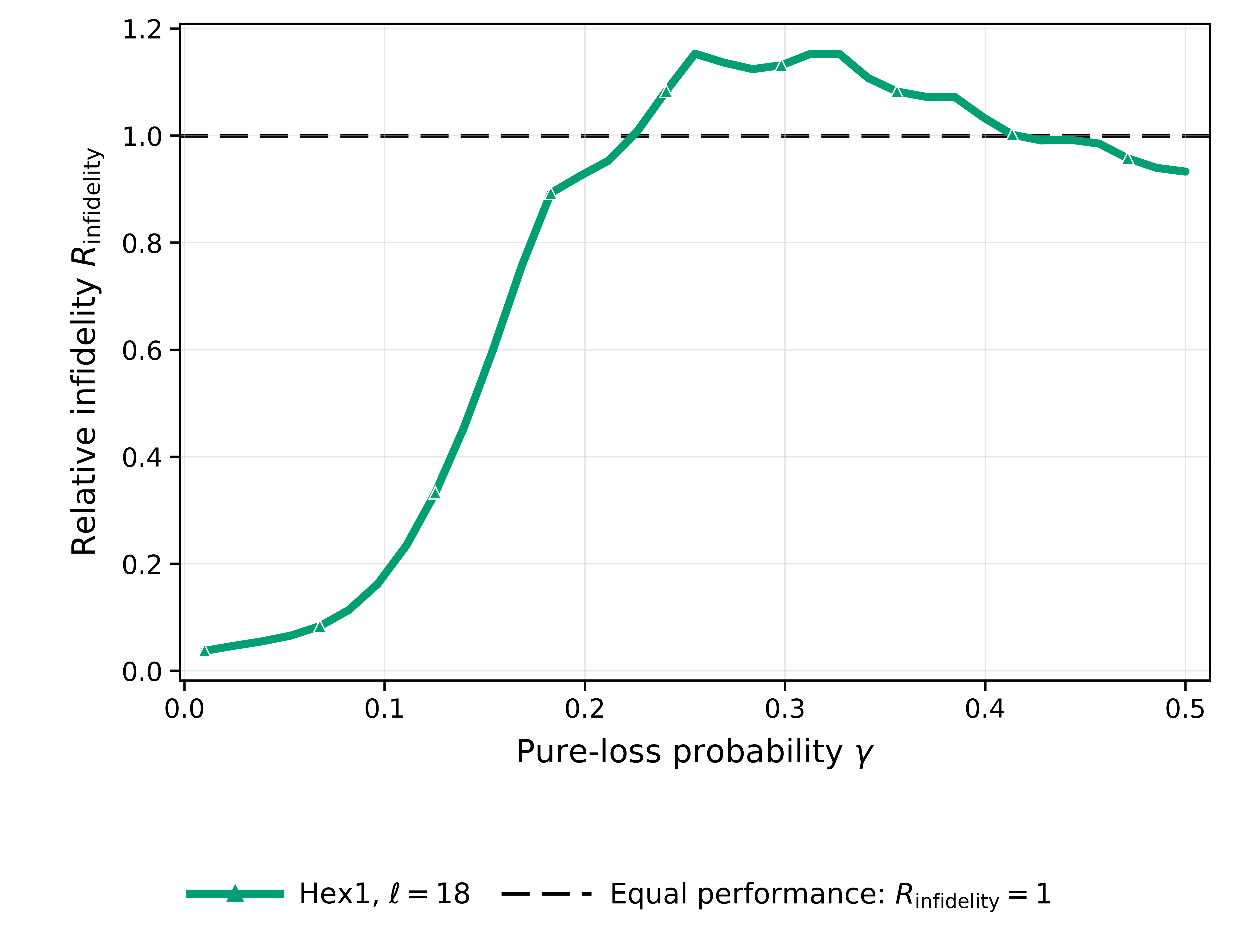}
\caption{ Hex1 }
\label{fig:appendix-method1-hex1-qsc-c}
\label{fig:appendix-method1-hex1-qcc-c}
\label{fig:appendix-method1-hex1-qsc}
\label{fig:appendix-method1-hex1-qcc}
\end{figure}

\begin{figure}[H]
\centering
\includegraphics[width=0.495\linewidth]{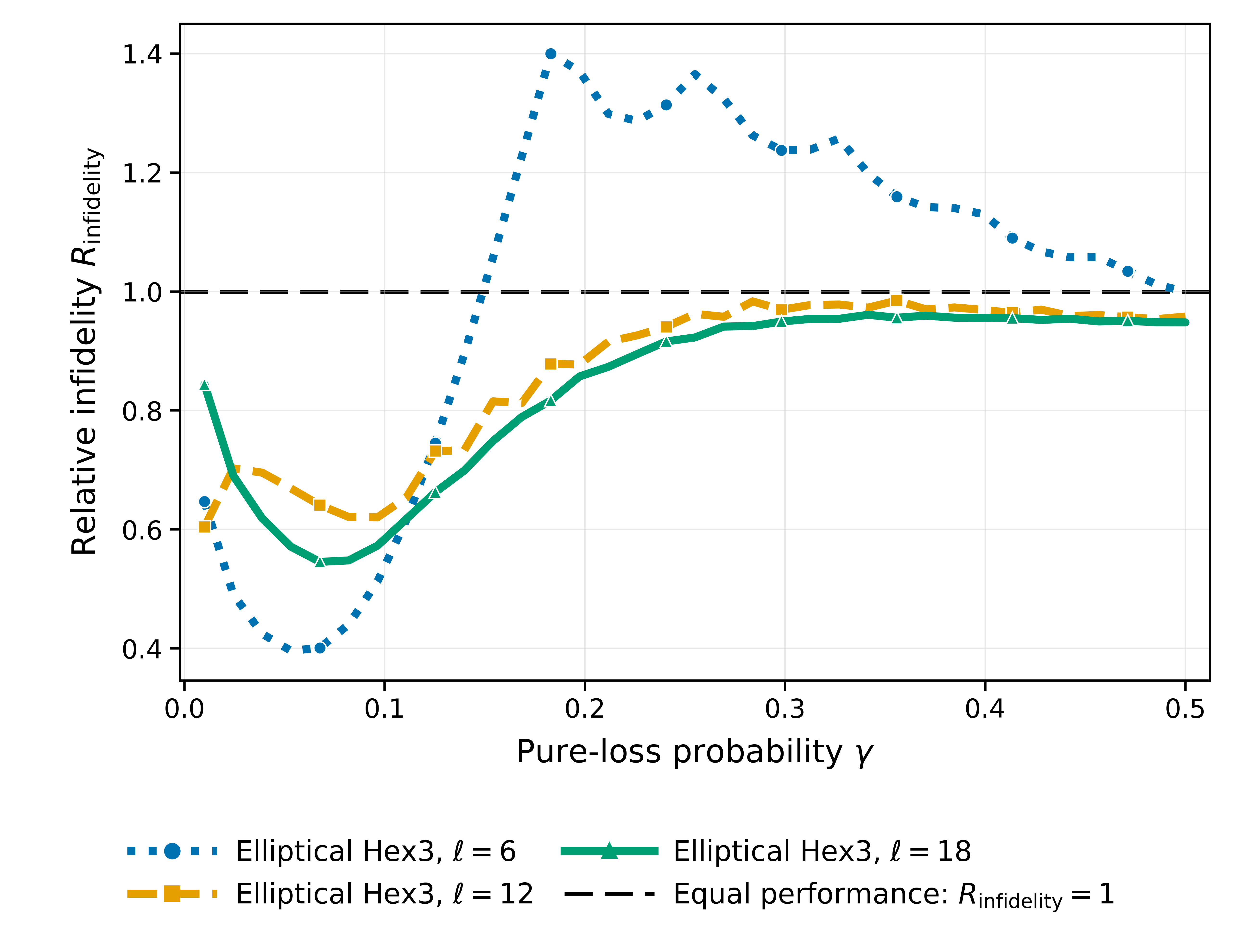}\hfill%
\includegraphics[width=0.495\linewidth]{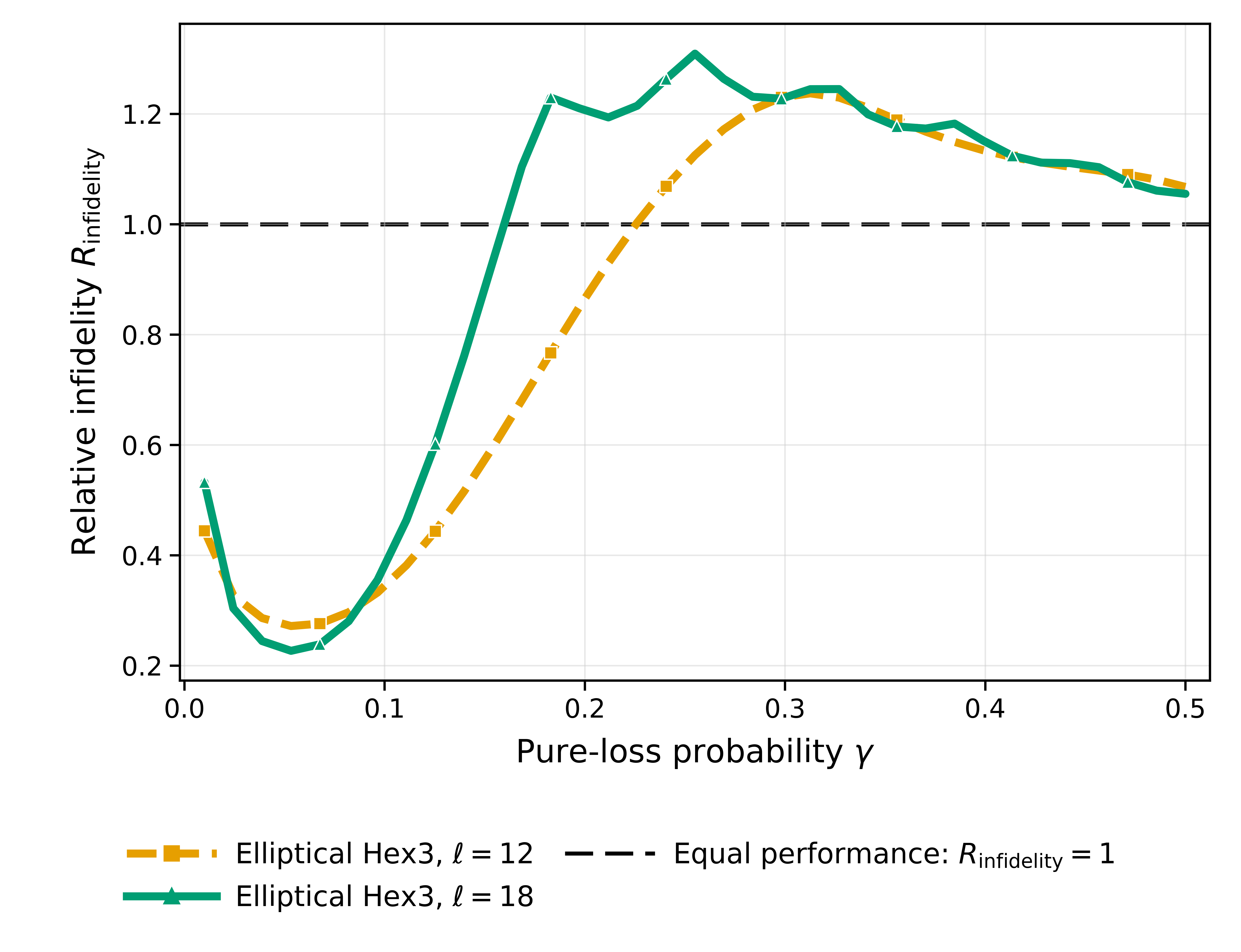}
\caption{ Elliptical Hex3 }
\label{fig:appendix-method1-elliptical-hex3-qsc-c}
\label{fig:appendix-method1-elliptical-hex3-qcc-c}
\end{figure}

\begin{figure}[H]
\centering
\includegraphics[width=0.495\linewidth]{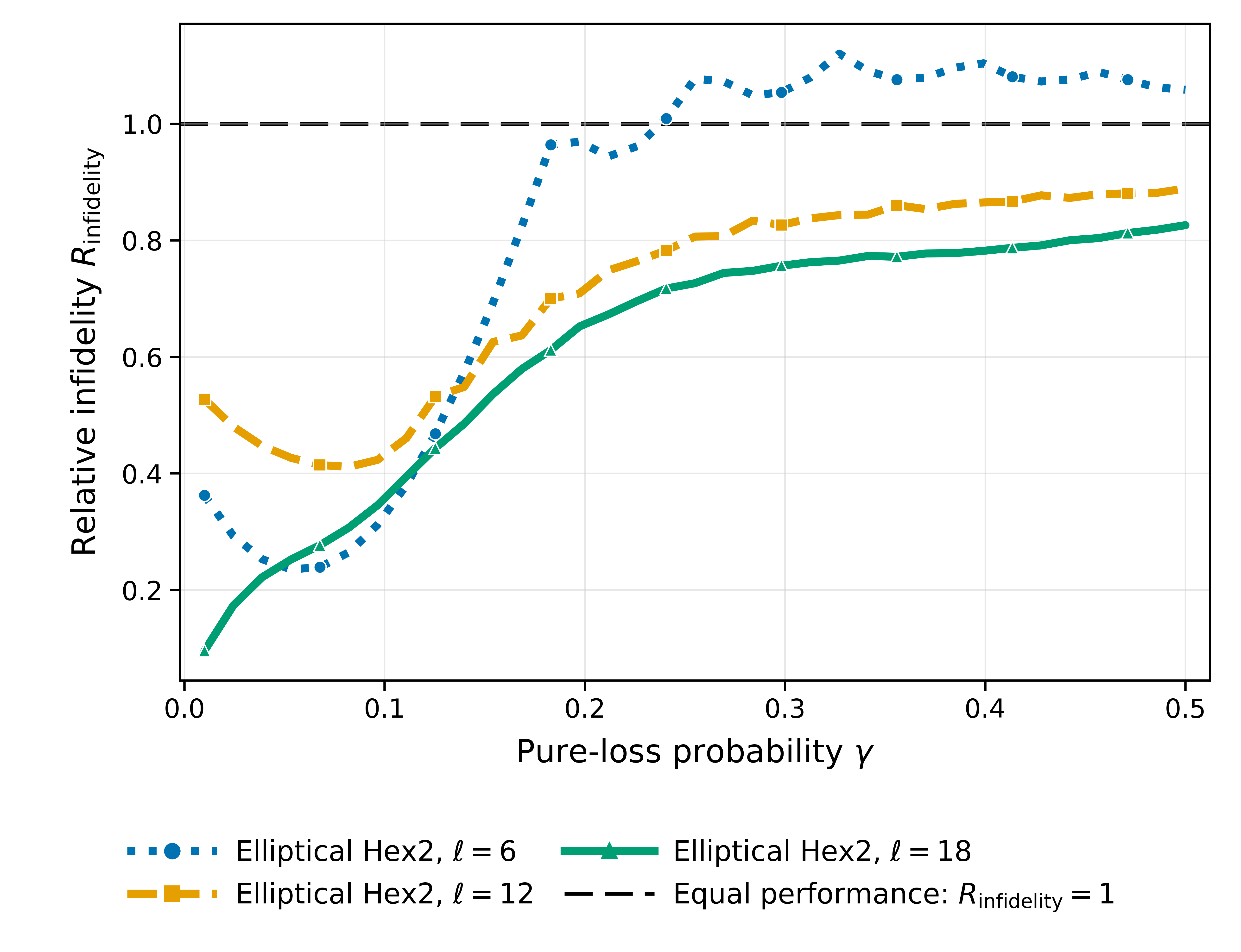}\hfill%
\includegraphics[width=0.495\linewidth]{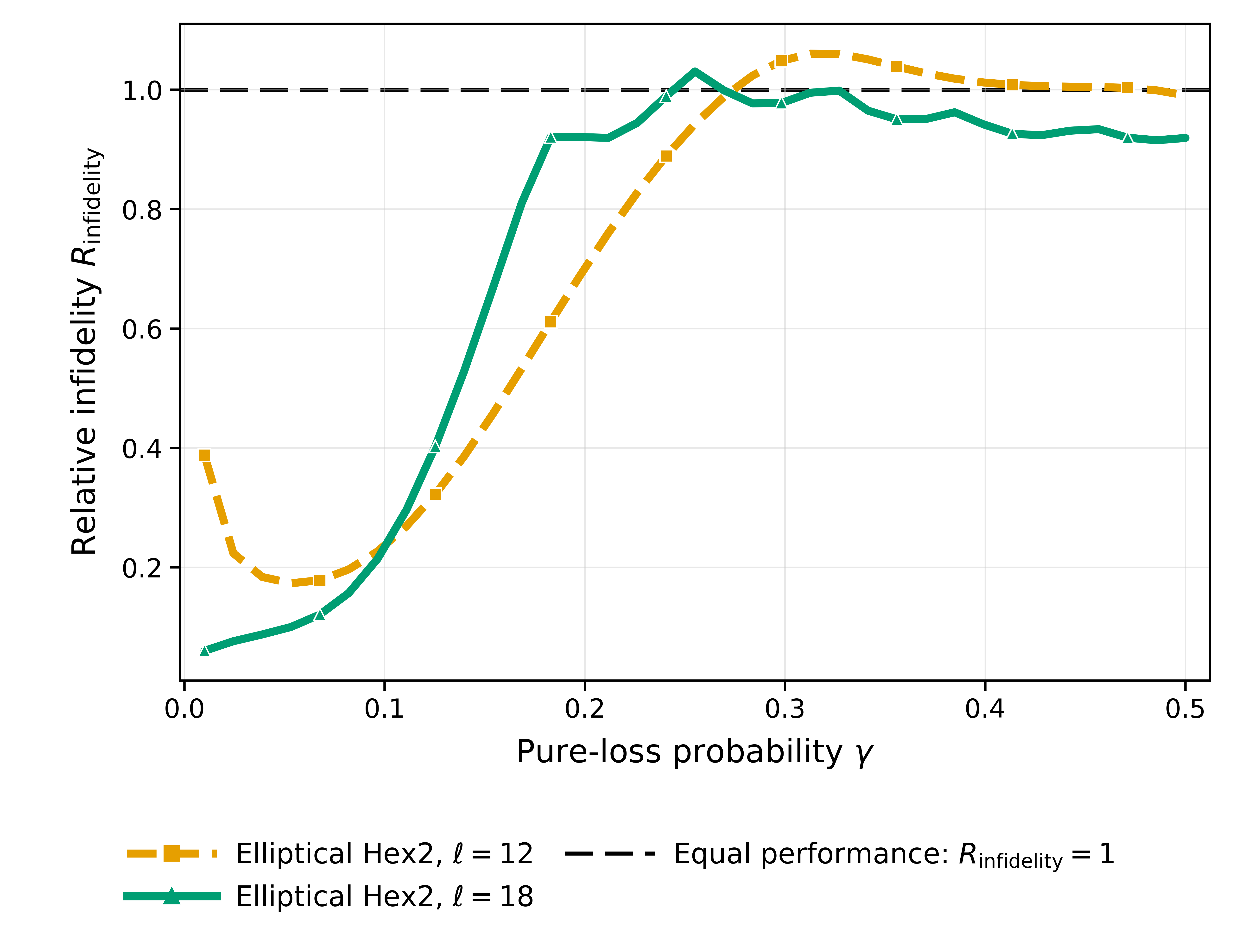}
\caption{Elliptical Hex2 }
\label{fig:appendix-method1-elliptical-hex2-qsc-c}
\label{fig:appendix-method1-elliptical-hex2-qcc-c}
\end{figure}

\begin{figure}[H]
\centering
\includegraphics[width=0.495\linewidth]{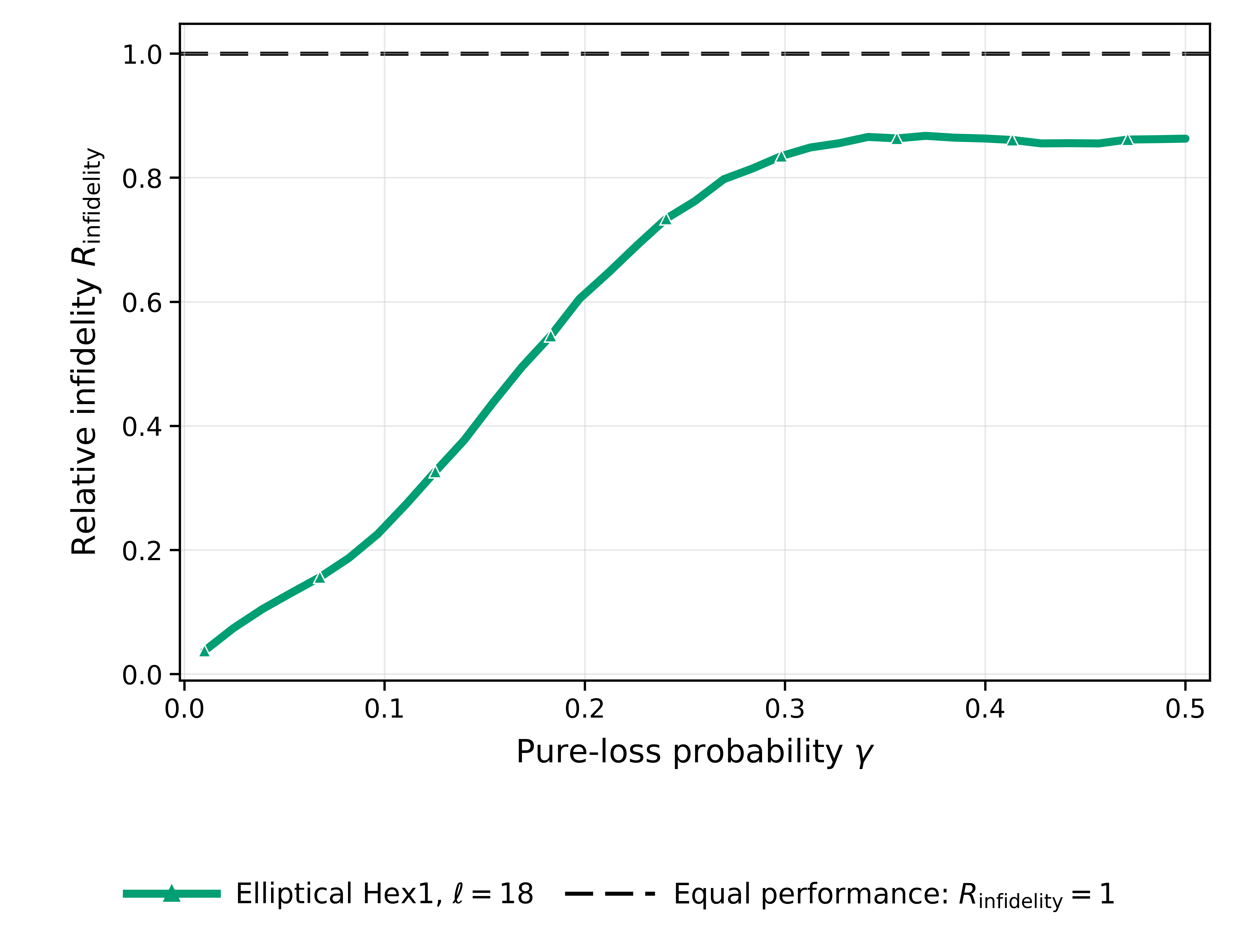}\hfill%
\includegraphics[width=0.495\linewidth]{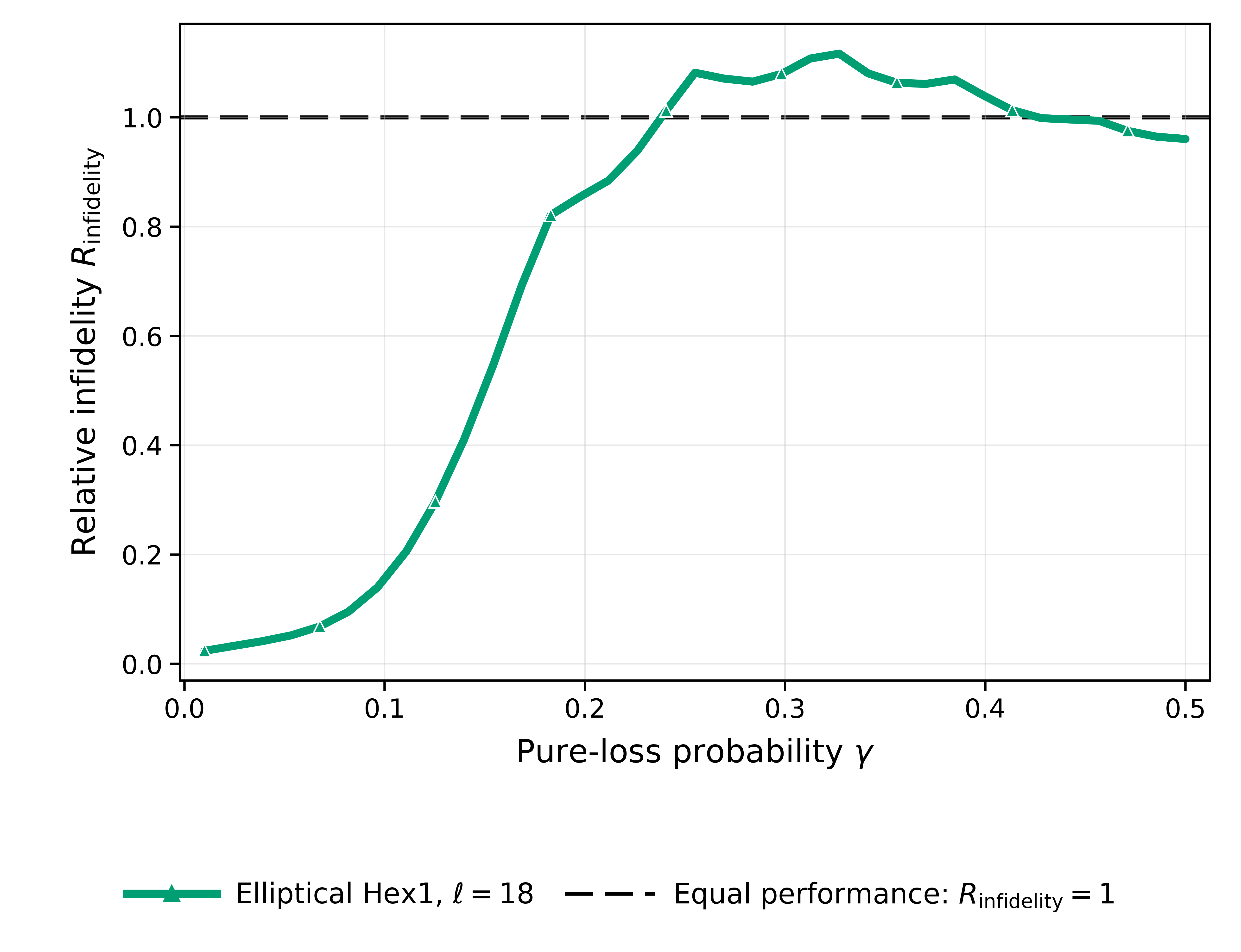}
\caption{ Elliptical Hex1}
\label{fig:appendix-method1-elliptical-hex1-qsc-c}
\label{fig:appendix-method1-elliptical-hex1-qcc-c}
\end{figure}

\begin{figure}[H]
\centering
\includegraphics[width=0.495\linewidth]{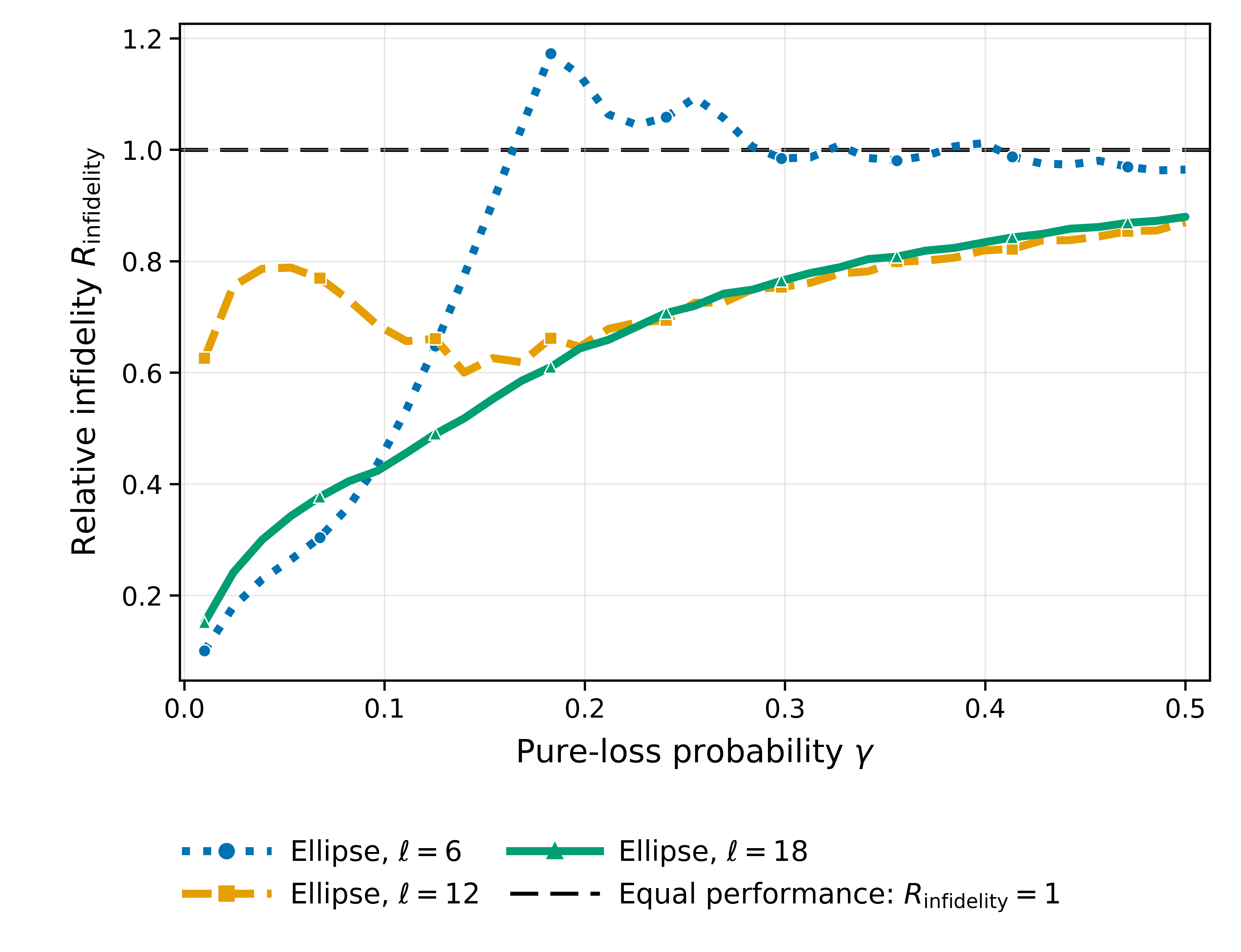}\hfill%
\includegraphics[width=0.495\linewidth]{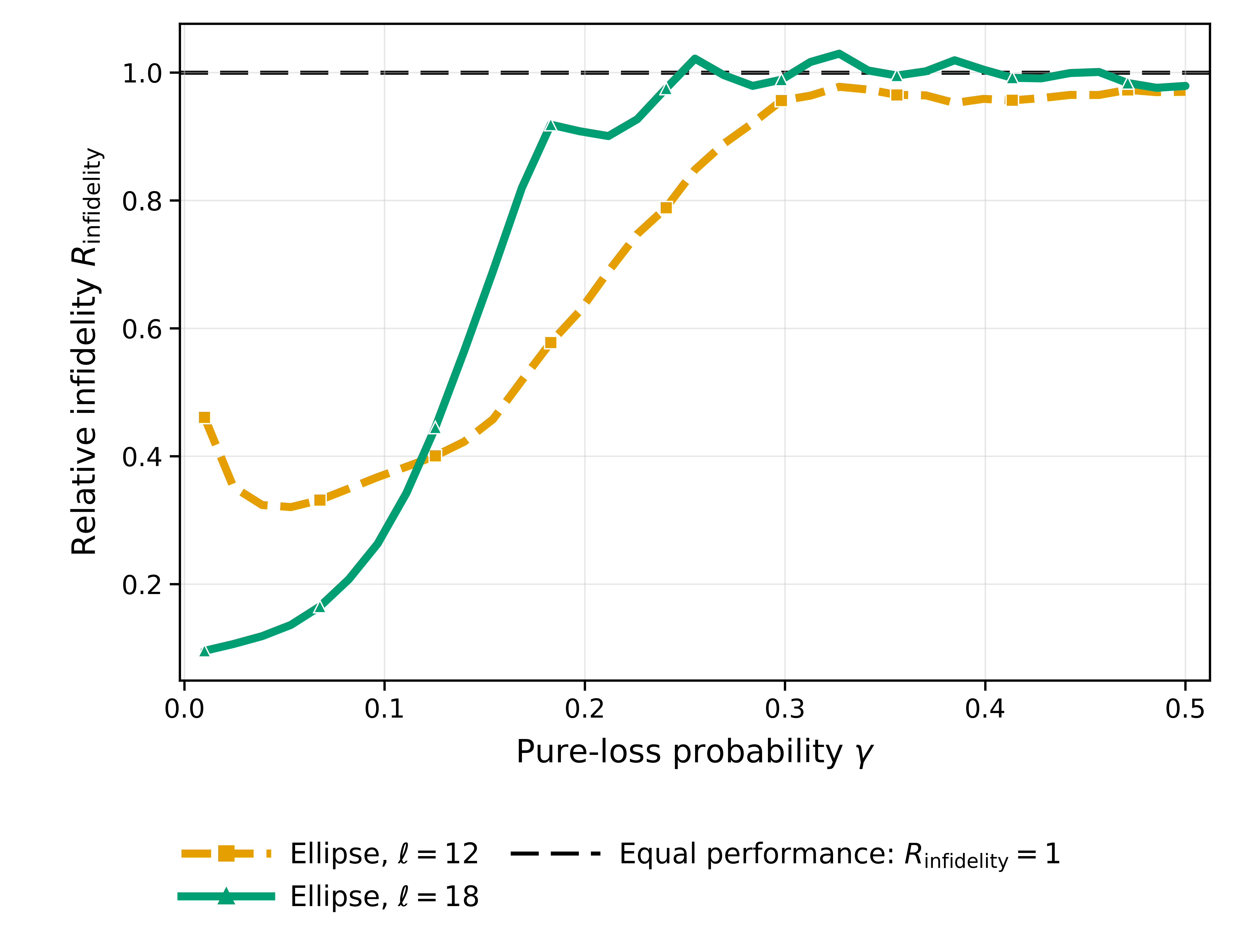}
\caption{ Ellipse }
\label{fig:appendix-method1-ellipse-qsc-c}
\label{fig:appendix-method1-ellipse-qcc-c}
\label{fig:appendix-method1-ellipse-qsc}
\label{fig:appendix-method1-ellipse-qcc}
\end{figure}
\FloatBarrier

\section{Supplementary figures for comparison under candidates mean photon budget}
\label{sec:appendix-method2-extra-figures}

\begin{figure}[H]
\centering
\includegraphics[width=0.495\linewidth]{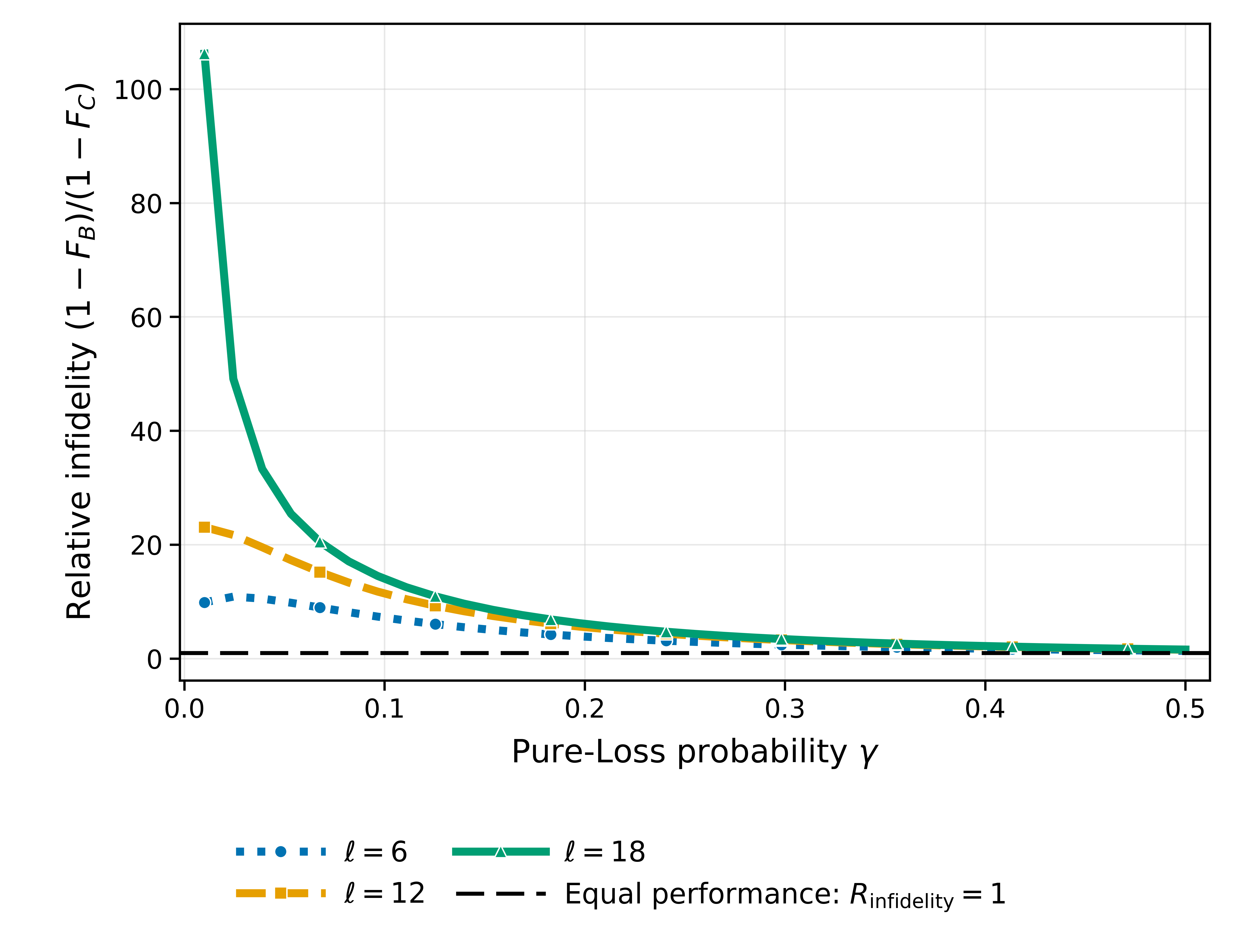}\hfill%
\includegraphics[width=0.495\linewidth]{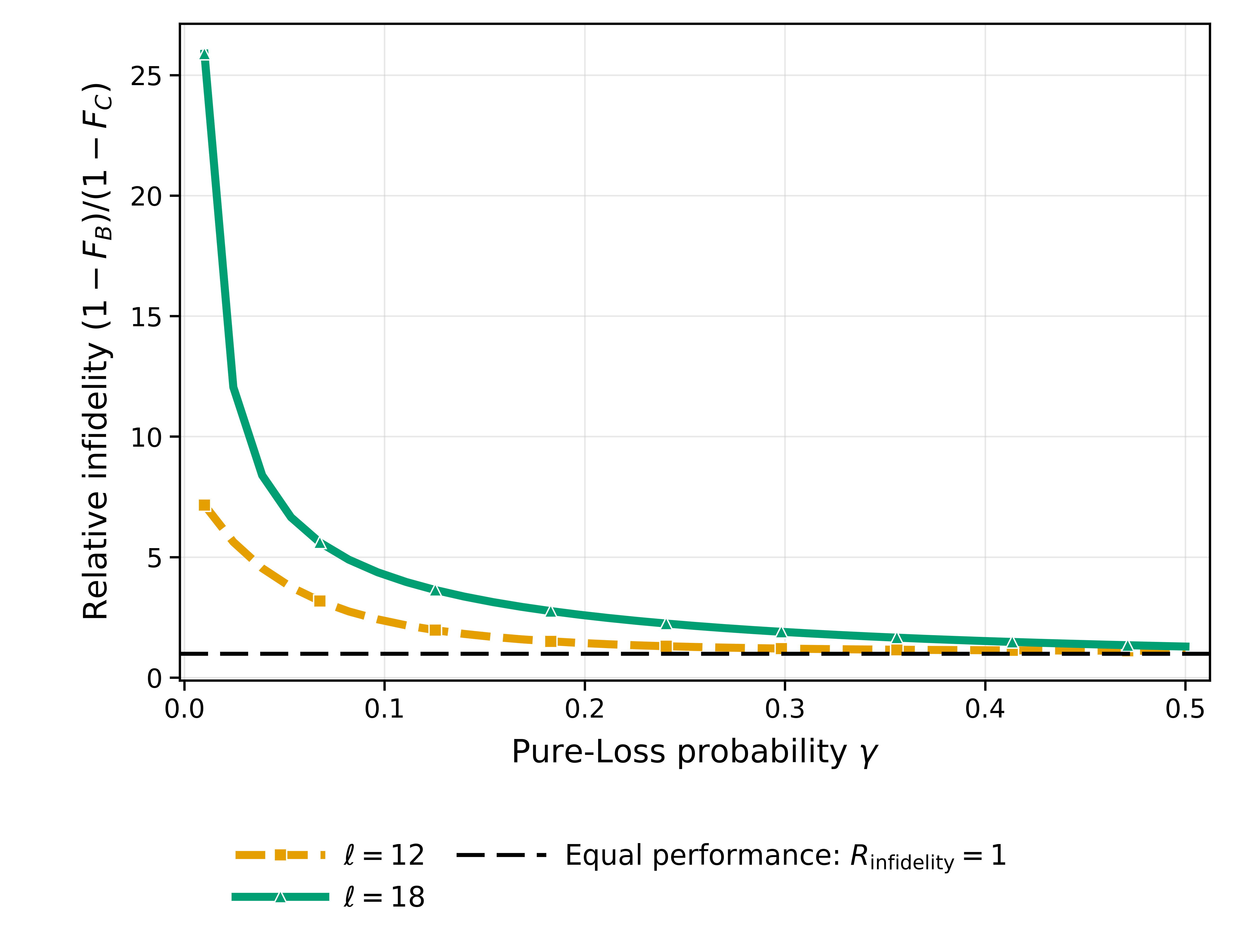}
\caption{ Hex3 }
\label{fig:appendix-method2-grid}
\label{appendix-method2-hex3-qsc}
\label{appendix-method2-hex3-qcc}
\end{figure}

\begin{figure}[H]
\centering
\includegraphics[width=0.495\linewidth]{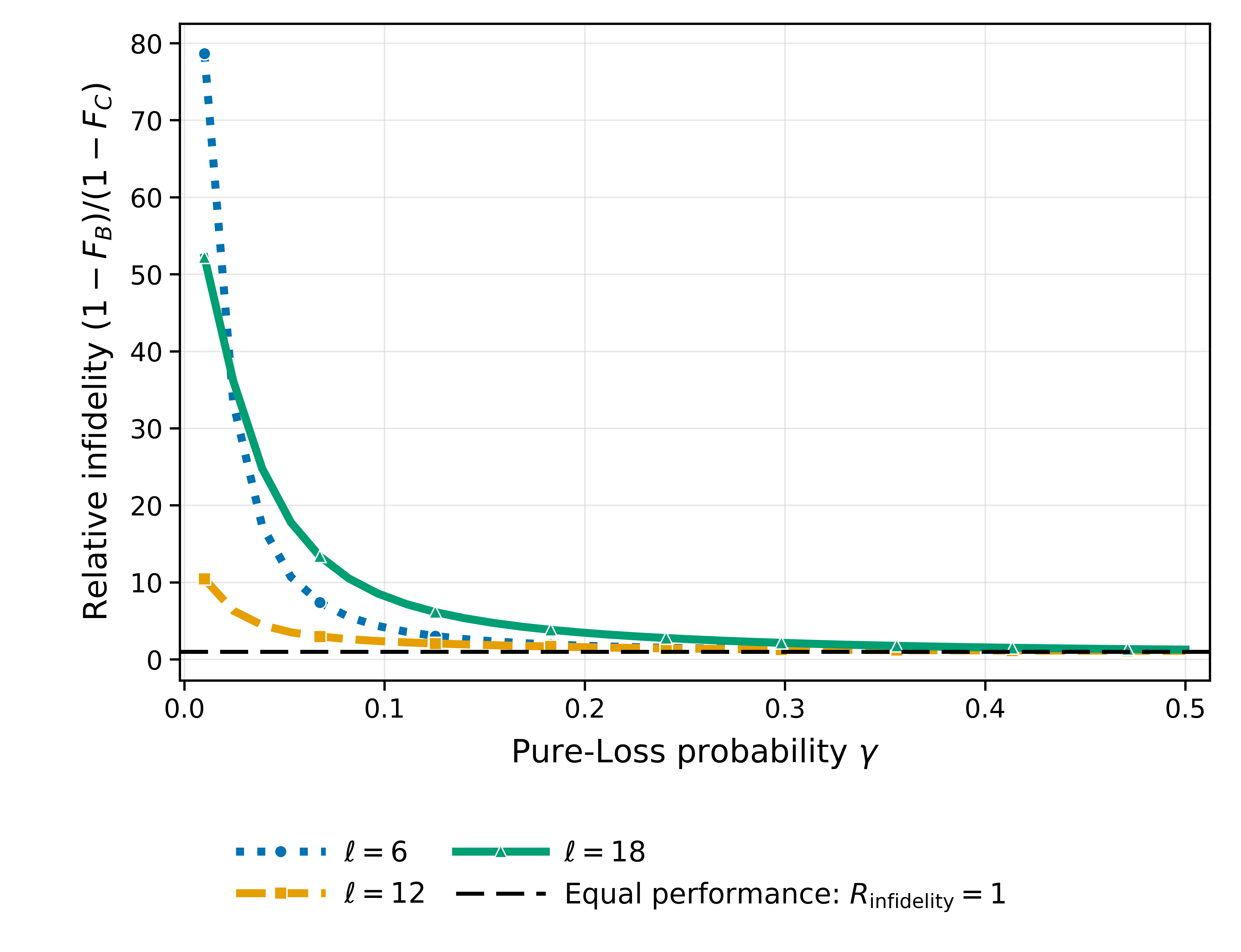}\hfill%
\includegraphics[width=0.495\linewidth]{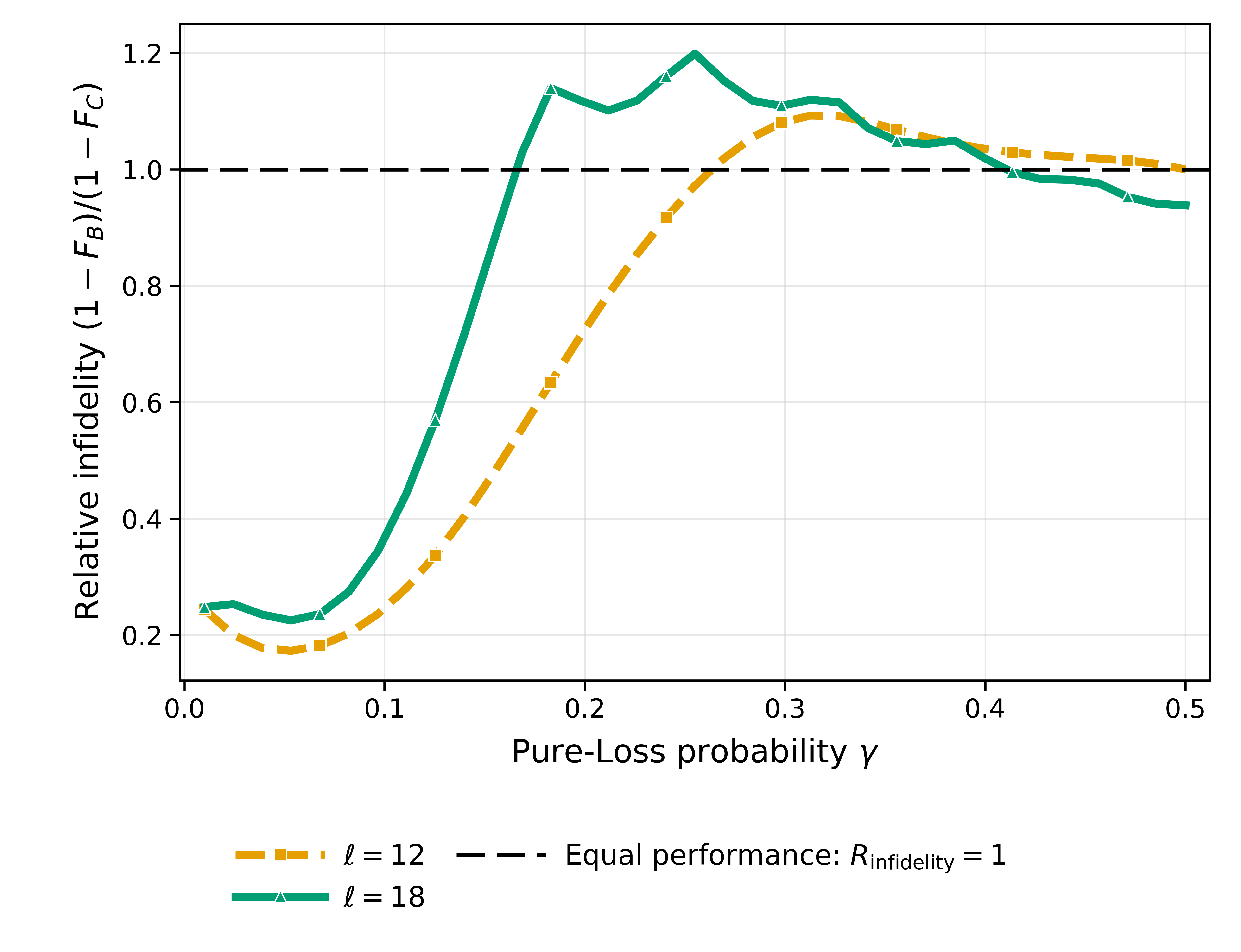}
\caption{ Hex2 }
\label{appendix-method2-hex2-qsc}
\label{appendix-method2-hex2-qcc}
\end{figure}

\begin{figure}[H]
\centering
\includegraphics[width=0.495\linewidth]{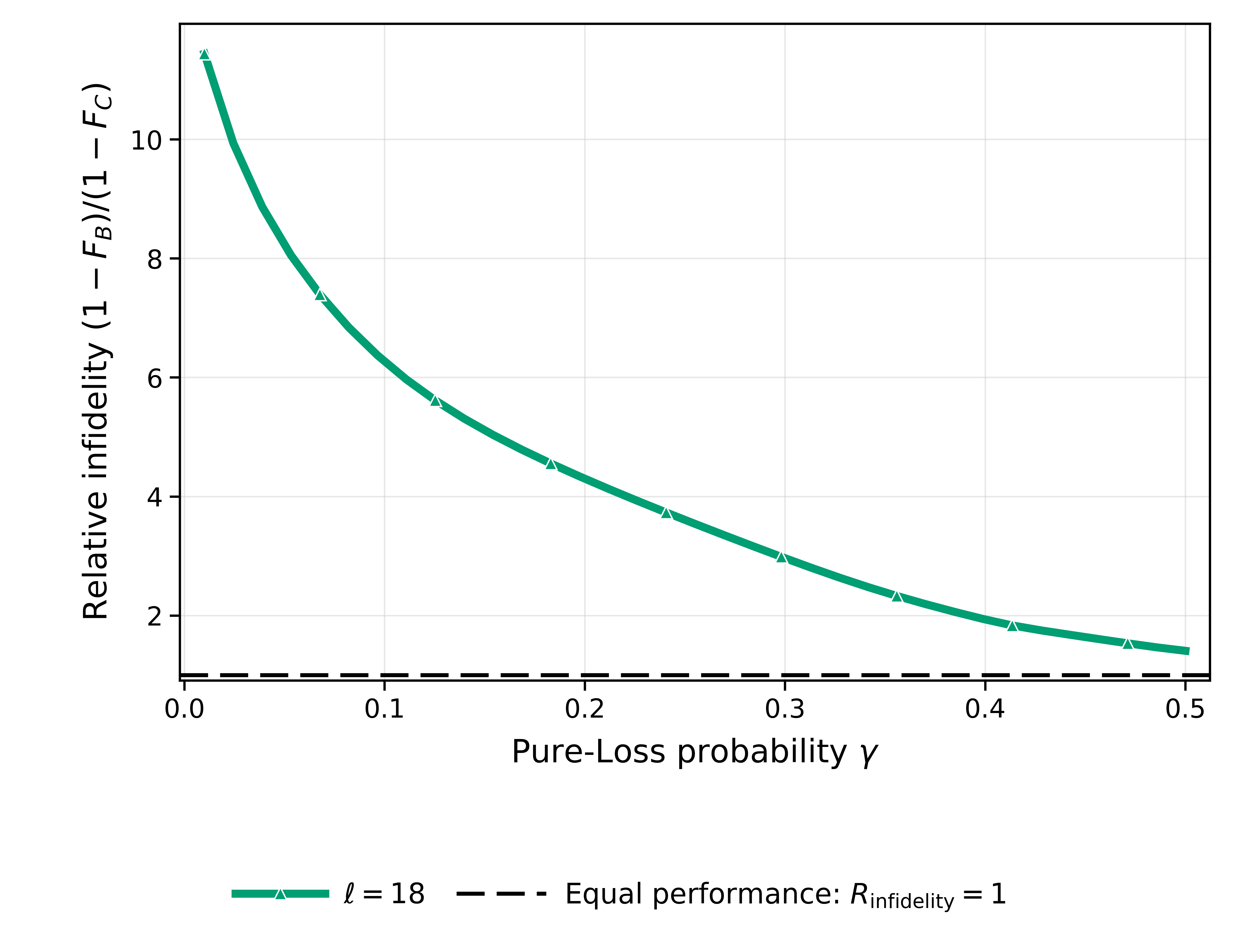}\hfill%
\includegraphics[width=0.495\linewidth]{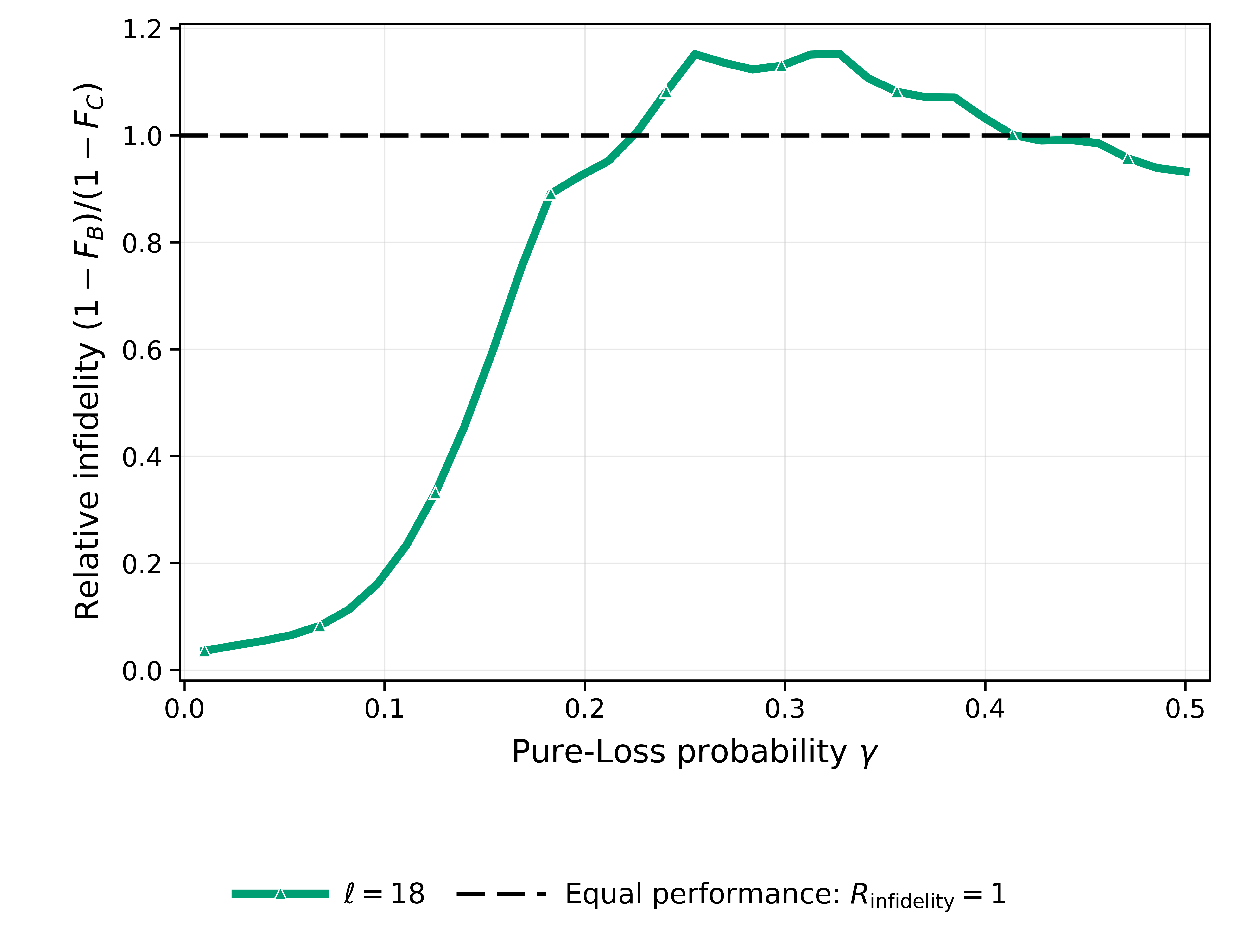}
\caption{Hex1}
\label{appendix-method2-hex1-qsc}
\label{appendix-method2-hex1-qcc}
\end{figure}

\begin{figure}[H]
\centering
\includegraphics[width=0.495\linewidth]{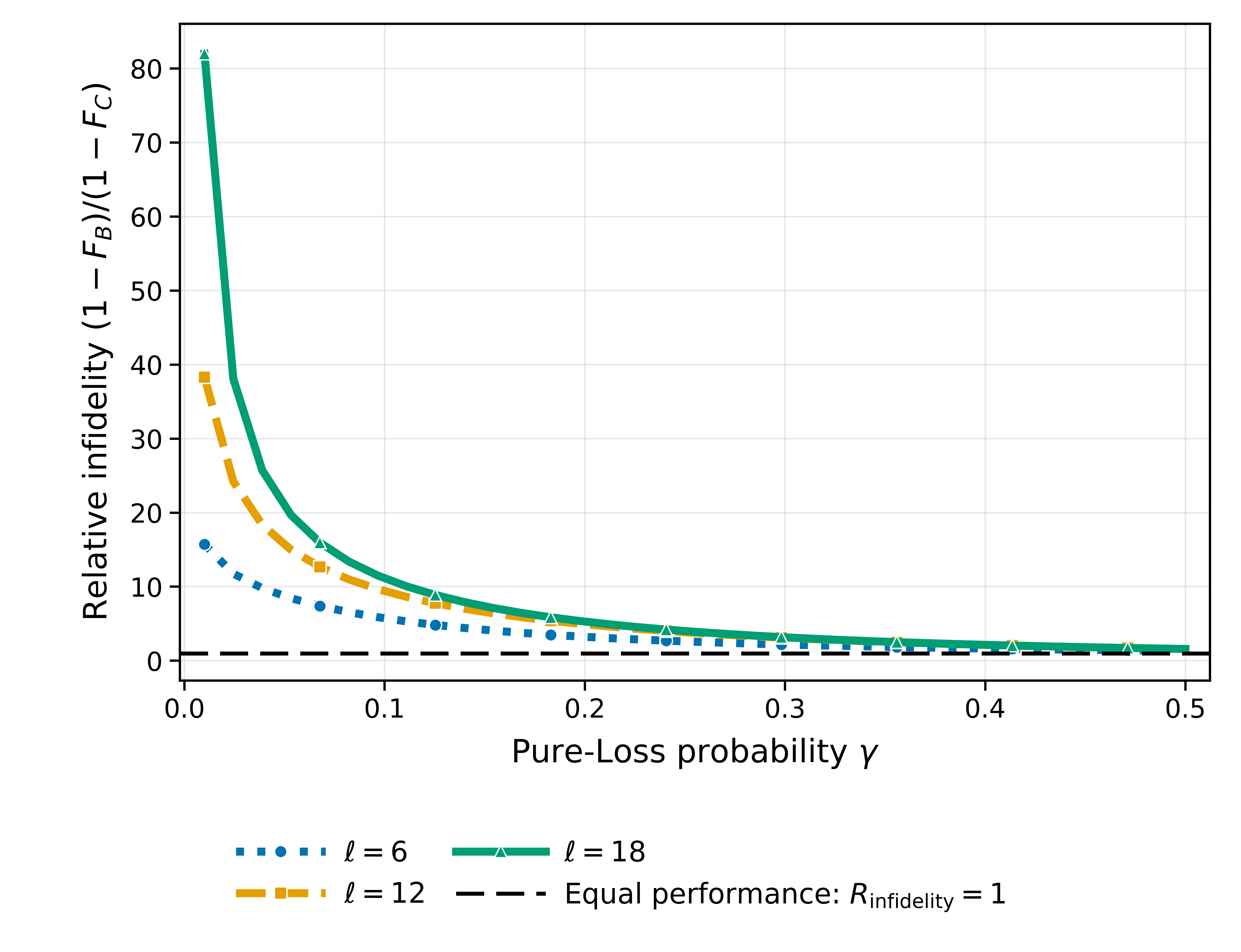}\hfill%
\includegraphics[width=0.495\linewidth]{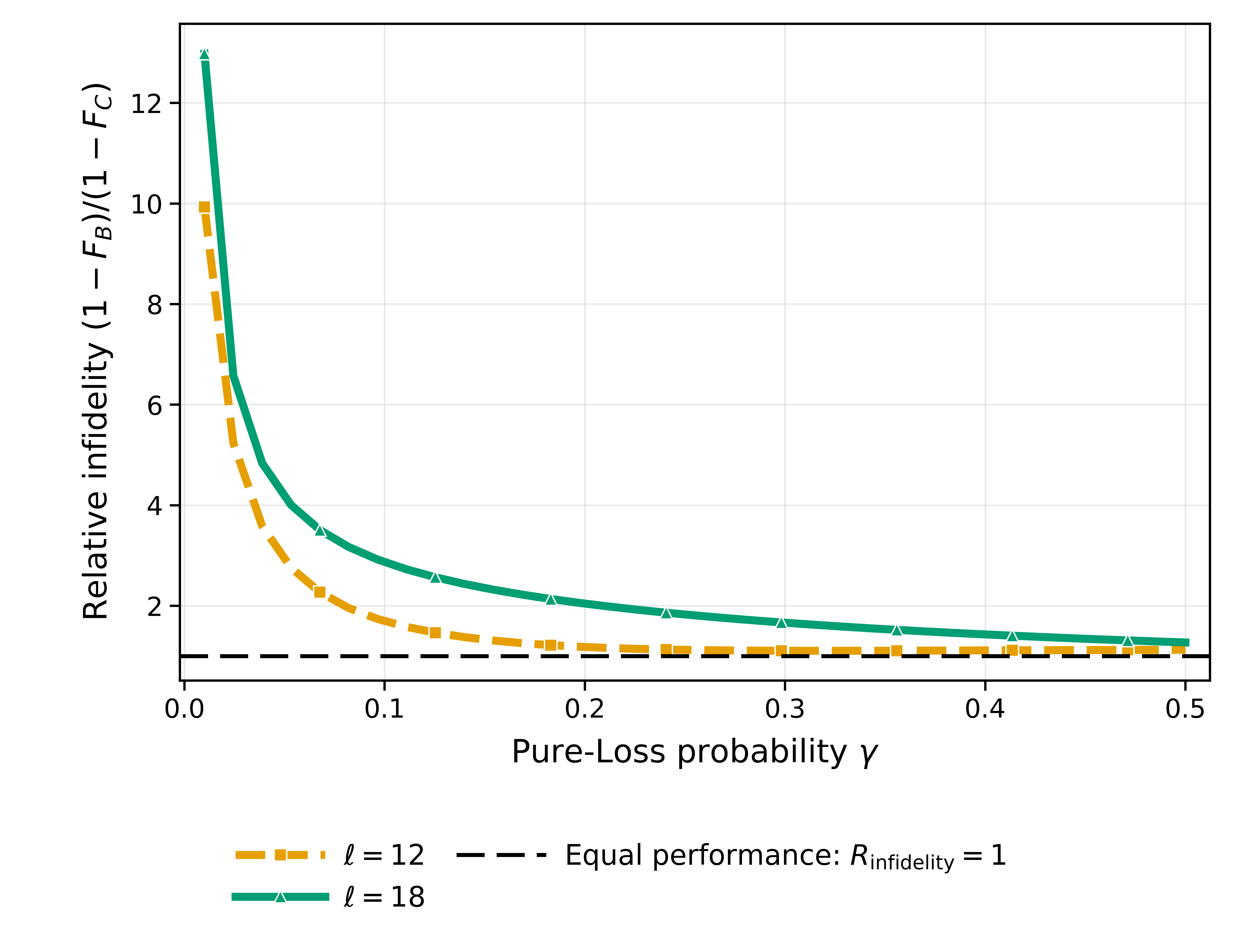}
\caption{ Elliptical Hex3 }
\label{appendix-method2-elliptical-hex3-qsc}
\label{appendix-method2-elliptical-hex3-qcc}
\end{figure}

\begin{figure}[H]
\centering
\includegraphics[width=0.495\linewidth]{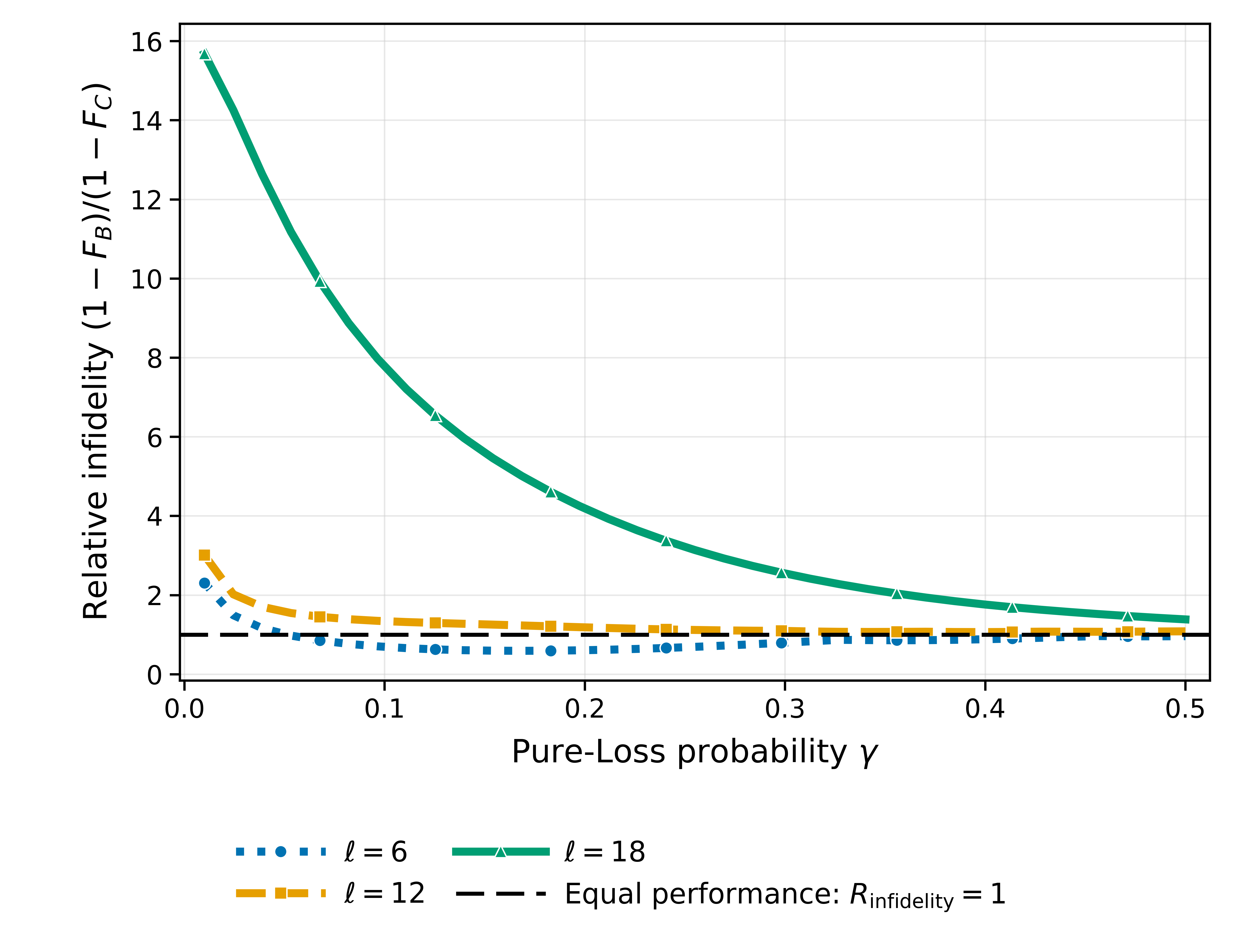}\hfill%
\includegraphics[width=0.495\linewidth]{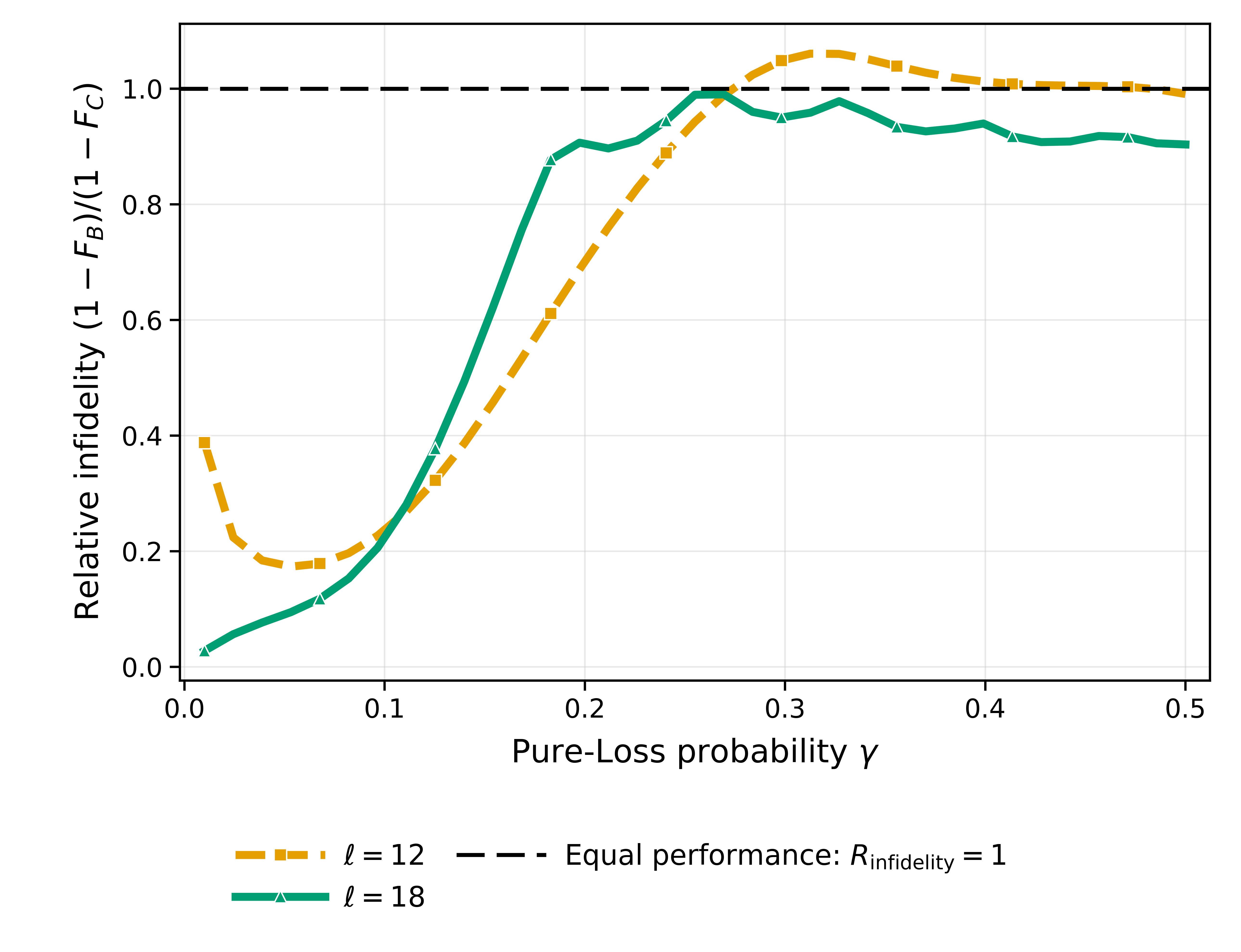}
\caption{ Elliptical Hex2 }
\label{appendix-method2-elliptical-hex2-qsc}
\label{appendix-method2-elliptical-hex2-qcc}
\end{figure}

\begin{figure}[H]
\centering
\includegraphics[width=0.495\linewidth]{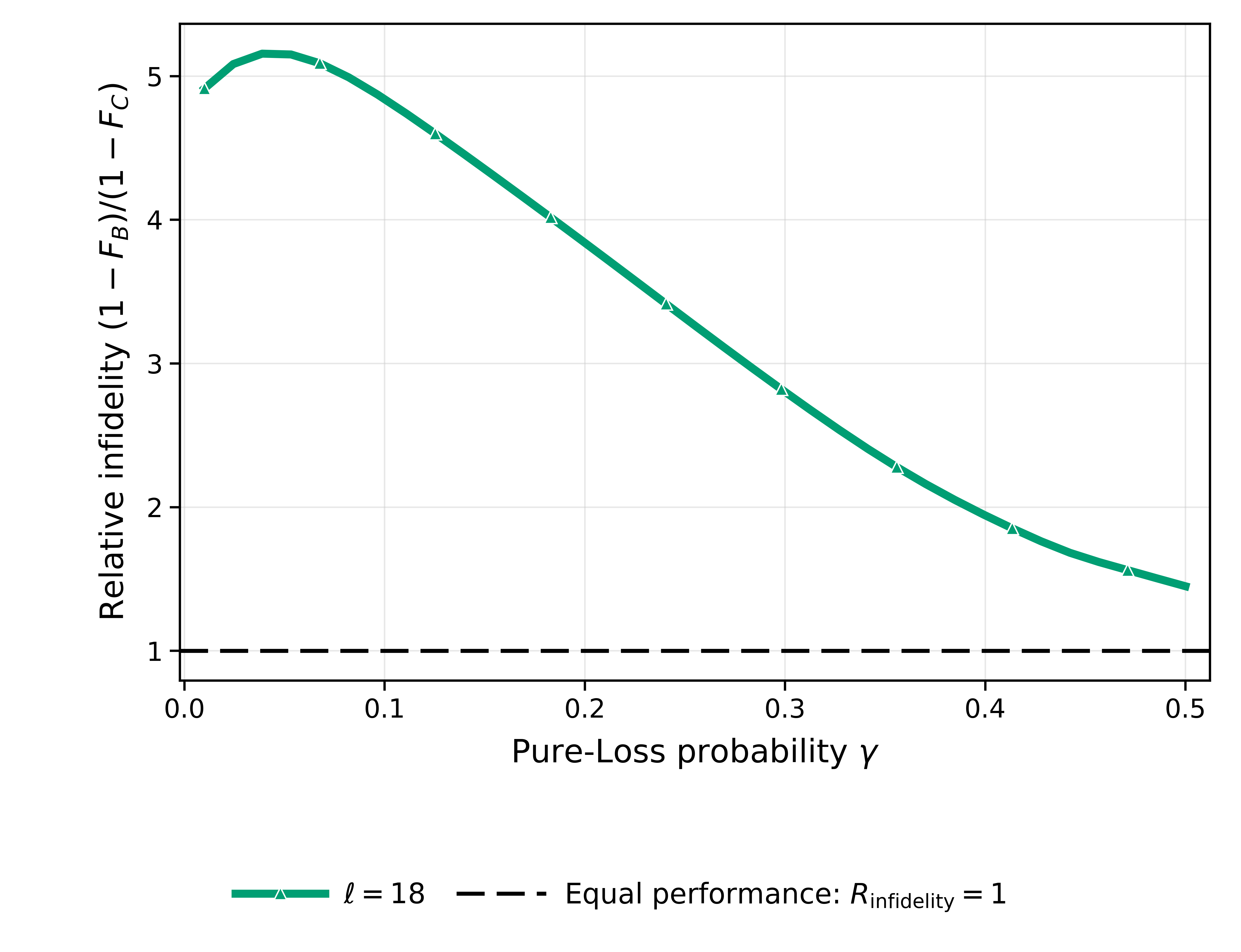}\hfill%
\includegraphics[width=0.495\linewidth]{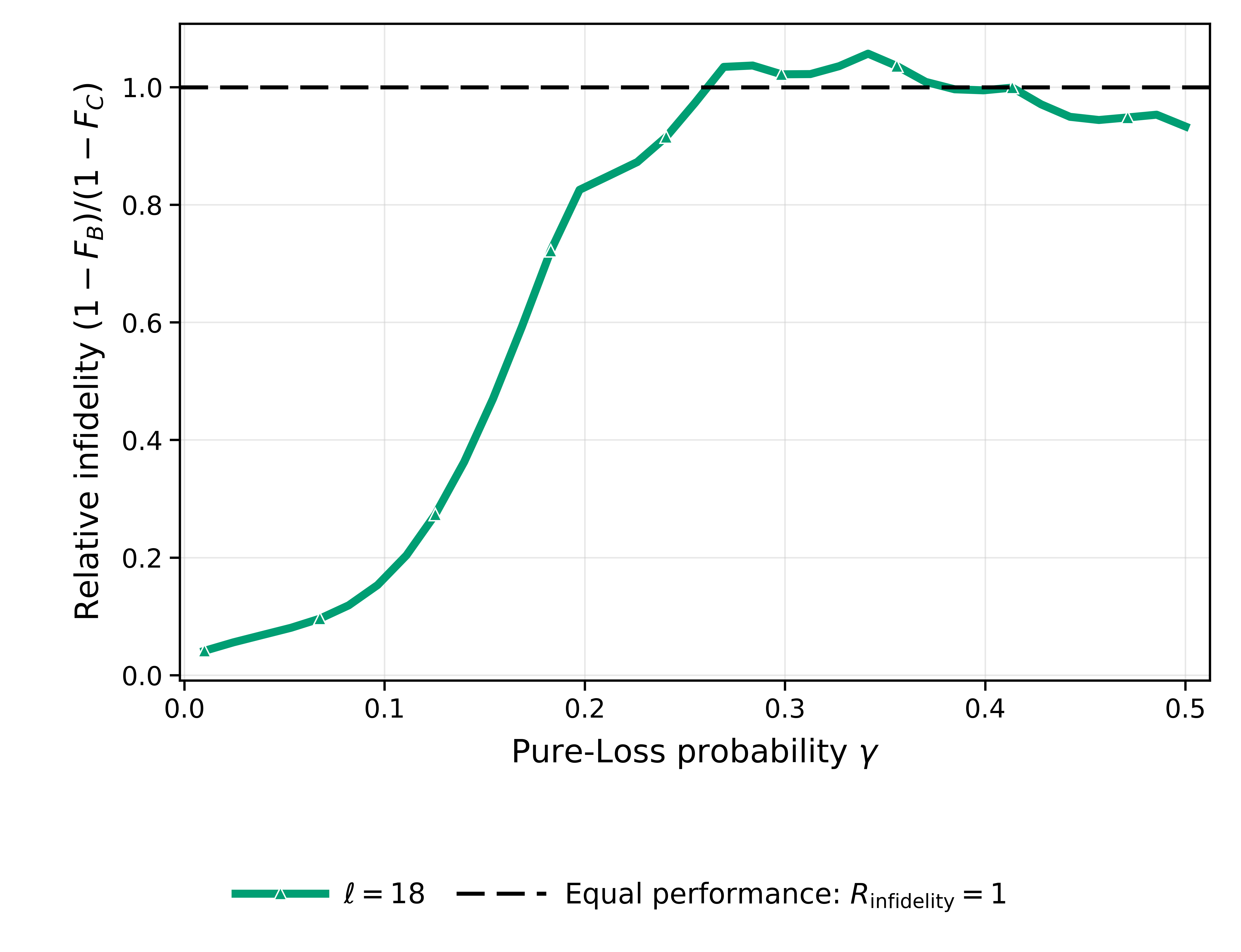}
\caption{Elliptical Hex1 }
\label{appendix-method2-elliptical-hex1-qsc}
\label{appendix-method2-elliptical-hex1-qcc}
\end{figure}

\begin{figure}[H]
\centering
\includegraphics[width=0.495\linewidth]{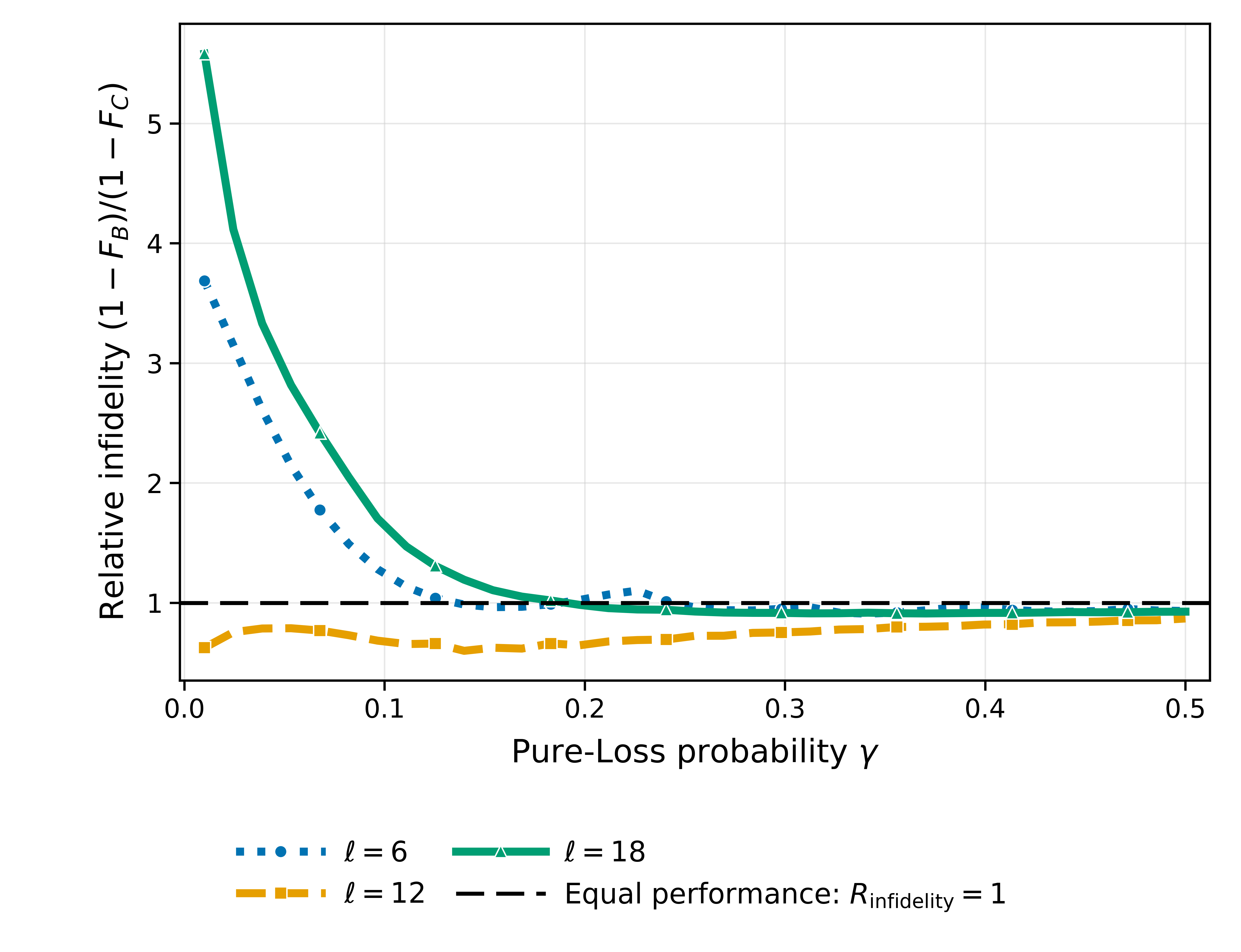}\hfill%
\includegraphics[width=0.495\linewidth]{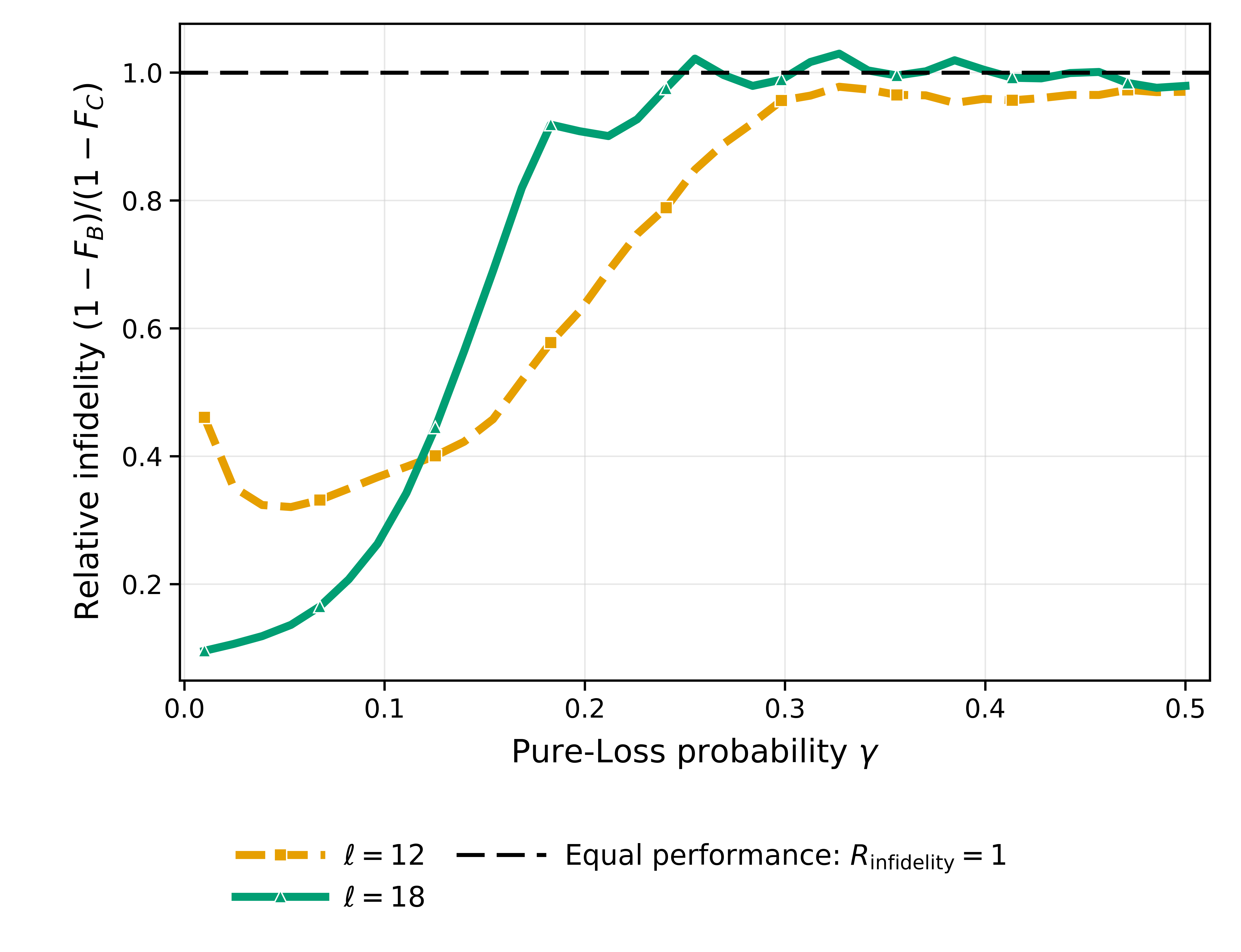}
\caption{Ellipse }
\label{appendix-method2-ellipse-qsc}
\label{appendix-method2-ellipse-qcc}
\end{figure}
\FloatBarrier

\section{Numerical results and code parameters}
\label{sec:numerical-code-parameters}

This section records the conventions and numerical operating points used in
Sections~\ref{sec:method1-results} and~\ref{sec:method2-results}. All benchmarked codes
encode $K=2$ logical states in a single bosonic mode, and $\ell$ denotes the
number of coherent components in each literal logical word. The tabulated
parameters describe the independently selected operating points used in the
equal-size comparison. They should therefore be distinguished from the
reference amplitudes selected for the photon-number-controlled comparison.

\paragraph{Literal coherent-state words.}
Let $z_{ki}\left(\alpha\right)\in\mathbb C$ denote the coherent amplitude of component
$i\in\left\{1,\ldots,\ell\right\}$ in literal word $k\in\left\{0,1\right\}$. The normalized word
is
\begin{equation}
 \ket{\psi_k\left(\alpha\right)}
 =
 \frac{\sum_{i=1}^{\ell}c_{ki}\ket{z_{ki}\left(\alpha\right)}}
 {\sqrt{c_k^\dagger S_k\left(\alpha\right)c_k}},
 \qquad
 \left(S_k\right)_{ij}=\braket{z_{ki}}{z_{kj}}.
 \label{eq:parameter-literal-word}
\end{equation}
The coherent-state overlap convention is
\begin{equation}
 \braket{z}{w}
 =
 \exp\!\left[-\frac{\abs{z}^2+\abs{w}^2}{2}+z^\ast w\right].
 \label{eq:parameter-coherent-overlap}
\end{equation}
The seven candidate BARC code families and QSCs use uniform coefficients. QCC coefficients are the square roots of their cubature probabilities, consistent
with the weighted-superposition convention of Ref.~\cite{Yang2025}. The
literal words are normalized but are not replaced by orthogonalized logical
states; their nonzero overlap is reported explicitly below.

\paragraph{Mean photon number.}
For the normalized literal words, the mean photon number used throughout the
comparisons is
\begin{equation}
 \bar n_C\left(\alpha\right)
 =
 \frac12\sum_{k=0}^1
 \frac{c_k^\dagger H_k\left(\alpha\right)c_k}
      {c_k^\dagger S_k\left(\alpha\right)c_k},
 \qquad
 \left(H_k\right)_{ij}
 =z_{ki}^{\ast}z_{kj}\left(S_k\right)_{ij}.
 \label{eq:equalell-energy}
\end{equation}
We write
$\bar n_C^\star=\bar n_C\!\left(\alpha_C^\star\right)$.
This is the quantum expectation value for the normalized coherent-state
superpositions, averaged over the two literal words. Coherent-state
interference is retained in both the numerator and normalization. Therefore,
$\bar n_C^\star$ is distinct from the squared radius of a point and from an
unnormalized geometric second moment. Because this amplitude--energy mapping
depends on the point geometry and superposition weights, values of $\alpha$
should not be compared across families as though they were a common energy
coordinate.

\paragraph{Phase-space resolution.}
The minimum physical separation between coherent components at the selected
operating point is
\begin{equation}
 \Delta_C^\star
 =
 \min_{\left(k,i\right)\ne\left(k',j\right)}
 \abs{z_{ki}\!\left(\alpha_C^\star\right)
       -z_{k'j}\!\left(\alpha_C^\star\right)}.
 \label{eq:equalell-resolution}
\end{equation}
This is an unsquared Euclidean distance over the union of both literal words,
including pairs within a word. A
larger $\Delta_C^\star$ means that the closest coherent components are farther
apart.

\paragraph{Literal-word overlap.}
The overlap between the two normalized literal words is
\begin{equation}
 s_C^\star
 =
 \abs{\braket{\psi_0\left(\alpha_C^\star\right)}{\psi_1\left(\alpha_C^\star\right)}}
 =
 \frac{\abs{c_0^\dagger S_{01}c_1}}
 {\sqrt{\left(c_0^\dagger S_0c_0\right)
              \left(c_1^\dagger S_1c_1\right)}},
 \qquad
 \left(S_{01}\right)_{ij}=\braket{z_{0i}}{z_{1j}}.
 \label{eq:equalell-word-overlap}
\end{equation}
Here $s_C^\star=0$ corresponds to orthogonal literal words, while values near
one indicate nearly indistinguishable prepared words. 

\paragraph{Relative infidelity.}

For a candidate $C$ and equal-size reference $B$, we use
the relative-infidelity ratio $R_{C|B}\left(\gamma\right)$ defined in
Eq.~\eqref{eq:equalell-ratio} in Section~\ref{sec:comparison-method-1-desc}.

Values above one mean that the candidate has lower infidelity; values below
one mean that the reference has lower infidelity.

\paragraph{Photon-number-controlled references.}
For the resource-controlled comparison, the candidate remains at
$\alpha_C^\star$ and defines the target energy $\bar n_C^\star$. 
The reference amplitude is selected by Eq.~\eqref{eq:energy-match}
in Section~\ref{sec:comparison-method-2-desc}.

\begin{table}[H]
\centering\small
\begin{tabular}{lrrrrrr}
\toprule
Family & $\ell$ & $\alpha_C^\star$ & $\bar n_C^\star$ &
$\Delta_C^\star$ & $F_{\mathrm B,C}^\star\left(\gamma_0\right)$ & $s_C^\star$ \\
\midrule
Hex3 & 6  & 2.2449 & 2.6444 & 1.2961 & 0.983188 & 0.0550 \\
Hex3 & 12 & 3.7959 & 2.6163 & 1.1317 & 0.977090 & 0.1232 \\
Hex3 & 18 & 4.7653 & 2.3231 & 0.9437 & 0.975359 & 0.0038 \\
Hex2 & 6  & 3.2143 & 6.5685  & 0.6614 & 0.955293 & 0.0767 \\
Hex2 & 12 & 6.8980 & 15.4504 & 0.7546 & 0.961114 & 0.1181 \\
Hex2 & 18 & 9.4184 & 19.3719 & 0.9068 & 0.968614 & 0.1258 \\
Hex1 & 18 & 8.8367 & 17.4795 & 0.6835 & 0.932726 & 0.3165 \\
Elliptical Hex3 & 6  & 2.8265 & 3.3661 & 1.0430 & 0.979646 & 0.0044 \\
Elliptical Hex3 & 12 & 4.3776 & 2.7339 & 0.8342 & 0.972513 & 0.0634 \\
Elliptical Hex3 & 18 & 5.7347 & 2.6817 & 0.7258 & 0.969632 & 0.0180 \\
Elliptical Hex2 & 6  & 4.3776 & 8.9431 & 0.9411 & 0.966477 & 0.0380 \\
Elliptical Hex2 & 12 & 8.4490 & 17.2261 & 0.9656 & 0.959587 & 0.0459 \\
Elliptical Hex2 & 18 & 10.0000$^{\dagger}$ & 17.0433 & 0.9883 & 0.949309 & 0.2610 \\
Elliptical Hex1 & 18 & 10.0000$^{\dagger}$ & 16.3251 & 0.4944 & 0.921857 & 0.4114 \\
Ellipse & 6  & 2.2449 & 13.3502 & 0.9729 & 0.975994 & 0.1574 \\
Ellipse & 12 & 4.1837 & 45.3251 & 0.8660 & 0.975315 & 0.0548 \\
Ellipse & 18 & 5.9286 & 90.9561 & 0.8246 & 0.959035 & 0.1932 \\
\midrule
QSC & 6  & 3.9898 & 15.9075 & 2.0653 & 0.989740 & 0.0249 \\
QSC & 12 & 6.5102 & 42.3055 & 1.6995 & 0.982982 & 0.0124 \\
QSC & 18 & 10.0000$^{\dagger}$ & 99.9698 & 1.7431 & 0.982735 & 0.0378 \\
QCC & 12 & 3.0204 & 9.1034 & 2.0087 & 0.990976 & 0.0577 \\
QCC & 18 & 4.1837 & 17.4944 & 2.0843 & 0.989487 & 0.0555 \\
\bottomrule
\end{tabular}
\caption{Independently selected operating points at $\gamma_0=0.095$.
All codes contain two literal words in one bosonic mode. The amplitude
$\alpha_C^\star$ is the best point on the 50-sample scan rather than a
certified continuous optimum. The mean photon number $\bar n_C^\star$ is
computed from the normalized coherent-state superpositions;
$\Delta_C^\star$ is the minimum unsquared separation over the union of both
words; and $s_C^\star$ is the normalized literal-word overlap. The fidelity
$F_{\mathrm B,C}^\star\left(\gamma_0\right)$ is optimized over recovery at the selected
amplitude. 
A dagger ($\dagger$) marks a best sampled amplitude at an endpoint of the scan; while the actual optimal amplitude may lie beyond the sampled limit.}
\label{tab:equalell-parameters}
\end{table}

\begin{table}[H]
\centering\small
\begin{tabular}{lp{0.69\linewidth}}
\toprule
Family & Construction parameters \\
\midrule
Hex3 & Alternating coloring of sixfold triangular-lattice orbits with
$\ell/3$ seed orbits; target moment $G_{\mathrm{iso}}=2I_2$; uniform
coefficients. \\
Hex2 & Opposite deep-hole shifts of alternating triangular subsets followed
by rewhitening to $G_{\mathrm{iso}}=2I_2$; uniform coefficients. \\
Hex1 & Origin-excluded radius-two triangular-lattice patch with $18$ points
per word, target moment $G_{\mathrm{iso}}=2I_2$, and relative rotation
$\pi/2$; uniform coefficients. The stored filename parameter $t=2$ is the
patch radius, giving $\ell=3t\left(t+1\right)=18$; it is not an error-correction order. \\
Elliptical Hex3 & Hex3 construction with $\ell=6,12,18$ and target moment
$G_{\mathrm{aniso}}$; uniform coefficients. \\
Elliptical Hex2 & Hex2 construction with $\ell=6,12,18$ and target moment
$G_{\mathrm{aniso}}$; uniform coefficients. \\
Elliptical Hex1 & Hex1 construction with $\ell=18$ and target moment
$G_{\mathrm{aniso}}$; uniform coefficients. \\
Ellipse & Circular points transformed directly by
$L=\left(\begin{smallmatrix}2&0.5\\0.5&1\end{smallmatrix}\right)$; uniform
coefficients. \\
QSC & Unit-circle constellation with angles
$2\pi r/\ell+2\pi k/\left(2\ell\right)$ and uniform coefficients, following the
single-mode QSC construction of Ref.~\cite{Jain2023}. \\
QCC & Three radial shells with radii ratio $1:2:3$;
$\left(t,p,m\right)=\left(7,3,4\right)$ for $\ell=12$ and $\left(9,3,6\right)$ for $\ell=18$.
The probability-weighted geometric radius squared is one, the coefficients
are square roots of the point probabilities, and the second word is rotated
by $\pi/m$~\cite{Yang2025}. \\
\bottomrule
\end{tabular}
\caption{Geometry and coefficient conventions for the benchmarked families.
All point-set centroids vanish. The isotropic lattice families use
$G_{\mathrm{iso}}=2I_2$, while the elliptically deformed lattice families use
$G_{\mathrm{aniso}}=\left(\begin{smallmatrix}2&0.5\\0.5&1\end{smallmatrix}\right)$.
For a word with real phase-space coordinates
$\boldsymbol\xi_{ki}\in\mathbb R^2$, $G$ is the prescribed summed centered
second moment,
$G=\sum_i\left(\boldsymbol\xi_{ki}-\bar{\boldsymbol\xi}_k\right)
\left(\boldsymbol\xi_{ki}-\bar{\boldsymbol\xi}_k\right)^{\mathsf T}$.
After whitening, the coordinates are mapped using the positive square root
$G^{1/2}$. Hence $G$ is a target moment matrix, not a direct coordinate map
and not the normalized quantum photon number. The Ellipse family is distinct:
$L$ acts directly on its circular coordinates.}
\label{tab:equalell-geometry}
\end{table}

Table~\ref{tab:equalell-geometry} specifies how the common scalar scan
parameter acts on each construction. If
$\boldsymbol\xi_{ki}=\left(x_{ki},p_{ki}\right)^{\mathsf T}$ denotes the final
dimensionless phase-space coordinate produced by the listed construction,
then
\begin{equation}
 z_{ki}\left(\alpha\right)=\alpha\left(x_{ki}+i p_{ki}\right).
 \label{eq:parameter-coordinate-scaling}
\end{equation}
The matrices $G_{\mathrm{iso}}$ and $G_{\mathrm{aniso}}$ prescribe moments
for the lattice constructions, whereas $L$ acts directly on the circular
coordinates of the Ellipse family. This distinction is required when
reconstructing the tabulated resolution, photon number, and word overlap.
For a centered set $\left\{w_i\right\}$ with $\sum_iw_iw_i^{\mathsf T}=I_2$, applying a coordinate map $L$ gives
\begin{equation}
    \sum_i\left(Lw_i\right)\left(Lw_i\right)^{\mathsf T}=LL^{\mathsf T}.
\end{equation}
Thus realizing a prescribed second-moment matrix $G$ requires $LL^{\mathsf T}=G$, for example $L=G^{1/2}$; applying $G$ itself would instead give $G^2$.
For the unscaled unit-circle $\ell$-gon, $\sum_iw_iw_i^{\mathsf T}=\left(\ell/2\right)I_2$, so a direct ellipse map gives $\left(\ell/2\right)LL^{\mathsf T}$, with any common amplitude factor squared multiplying this matrix.
For example, the matrix listed for the direct Ellipse map has
\begin{equation}
    G_\mathrm{aniso}G_\mathrm{aniso}^\top =
    \begin{bmatrix}2&0.5\\0.5&1\end{bmatrix}
    \begin{bmatrix}2&0.5\\0.5&1\end{bmatrix}^{\mathsf T}
    =\begin{bmatrix}4.25&1.5\\1.5&1.25\end{bmatrix}.
\end{equation}
These identities distinguish a coordinate map from a target moment matrix; identifying the actual inputs used for each numerical family requires the stored node arrays and their scaling conventions.

\FloatBarrier

\bibliographystyle{unsrt}
\bibliography{references}

\end{document}